\documentclass[a4paper,11pt]{article}
\usepackage[margin=1in]{geometry}
\usepackage[utf8]{inputenc}

\usepackage{lineno}

\usepackage{soul}
\usepackage{fontaxes}
\usepackage{trimspaces}
\usepackage{nccfoots}
\usepackage{setspace}
\usepackage{inconsolata}
\usepackage{libertine}

\usepackage{booktabs}
\usepackage{tabularx}
\usepackage{multirow}
\usepackage{makecell}
\usepackage{hhline}

\usepackage[T1]{fontenc}
\usepackage{amsmath, amssymb}
\usepackage{xargs}
\usepackage{mathtools}
\usepackage{amsthm}
\usepackage{stmaryrd}
\usepackage[sort&compress,numbers,square,comma,longnamesfirst]{natbib}
\usepackage{graphicx}
\usepackage[hidelinks]{hyperref}
\usepackage{xcolor}
\usepackage{comment}
\usepackage[disable]{todonotes}
\usepackage[linesnumbered,lined,commentsnumbered,noend,boxed]{algorithm2e}
\usepackage{enumitem}
\usepackage{framed}
\usepackage[noabbrev,capitalise]{cleveref}
\usepackage[framemethod=TikZ]{mdframed}
\usepackage{tikz}
\usetikzlibrary{positioning,decorations.markings,decorations,decorations.pathmorphing,arrows,arrows.meta,calligraphy}
\usetikzlibrary{shapes.geometric}
\usetikzlibrary{backgrounds}
\usetikzlibrary{ipe}
\usepackage{pgfplots}
\usepgflibrary{decorations.pathreplacing}
\usepackage{xspace}
\usepackage{subcaption}
\usepackage{thm-restate}

\usepackage[libertine]{newtxmath}

\numberwithin{equation}{section}
\numberwithin{figure}{section}

\makeatletter
\newtheoremstyle{linkedrestate}
  {}%
  {}%
  {\itshape}%
  {}%
  {\bfseries}%
  {.%
  \ifthmt@thisistheone{%
  \@ifundefined{r@\thisthm}{}{\textnormal{ 
  [\hyperlink{restated:\thisthm}{See
   page~\pageref{\thisthm}.}]}}%
  }
  \fi
}%
  { }%
  {%
    \ifthmt@thisistheone%
    \hyperlink{restated:\thisthm}{\thmname{#1} %
    \thmnumber{#2}\unskip}%
    \thmnote{ (#3)}%
    \myhypertarget{orig:\thisthm}{}%
    \else%
    \hyperlink{orig:\thisthm}{\thmname{#1} %
    \thmnumber{#2}\unskip}%
    \thmnote{ (#3)}%
    \myhypertarget{restated:\thisthm}{}%
    \fi
    }%

\NewDocumentCommand{\NewLinkedTheorem}{s m o m o}{
    \let\current@thmstyle\thmt@outerstyle %
    \IfBooleanTF{#1}{
        \newtheorem*{#2}{#4}
        \theoremstyle{linkedrestate}
        \newtheorem*{linked-#2}{#4}
        \ExpandArgs{e}\theoremstyle{\current@thmstyle} %
        }{
        \IfNoValueTF{#3}{
            \IfNoValueTF{#5}{
                \newtheorem{#2}{#4}
                }{
                \newtheorem{#2}{#4}[#5]
                }
            }{
            \newtheorem{#2}[#3]{#4}
            }
        \theoremstyle{linkedrestate}
        \newtheorem{linked-#2}[#2]{#4}
        \ExpandArgs{e}\theoremstyle{\current@thmstyle} %
        }
    }
\makeatother

\NewCommandCopy{\oldrestatable}{\restatable}
\NewCommandCopy{\endoldrestatable}{\endrestatable}

\NewDocumentEnvironment{linkedrestatable}{O{} m m +b}{%
    \def\thisthm{#3}%
    \begin{oldrestatable}[#1]{linked-#2}{#3}%
    #4
    \end{oldrestatable}%
	\ignorespacesafterend%
    \AddToHook{cmd/#3/before}{\def\thisthm{#3}}%
    }{}

\renewenvironment{restatable}{%
    \begin{linkedrestatable}%
}{%
    \end{linkedrestatable}%
    \ignorespacesafterend%
}

\NewLinkedTheorem{theorem}{Theorem}[section]
\NewLinkedTheorem{definition}[theorem]{Definition}
\NewLinkedTheorem{lemma}[theorem]{Lemma}
\NewLinkedTheorem{corollary}[theorem]{Corollary}
\NewLinkedTheorem{claim}[theorem]{Claim}
\NewLinkedTheorem{proposition}[theorem]{Proposition}
\NewLinkedTheorem{observation}[theorem]{Observation}
\NewLinkedTheorem{conjecture}[theorem]{Conjecture}
\NewLinkedTheorem{fact}[theorem]{Fact}
\NewLinkedTheorem{remark}[theorem]{Remark}
\NewLinkedTheorem{problem}{Open Problem}
\NewLinkedTheorem{reduction}[theorem]{Reduction Rule}
\NewLinkedTheorem{property}[theorem]{Property}
\NewLinkedTheorem{hypothesis}[theorem]{Hypothesis}

\crefname{linked-theorem}{Theorem}{Theorems}
\Crefname{linked-theorem}{Theorem}{Theorems}
\crefname{lemma-theorem}{Lemma}{Lemmas}
\Crefname{lemma-theorem}{Lemma}{Lemmas}

\makeatletter 
\newcommand{\myhypertarget}[1]{\Hy@raisedlink{\hypertarget{#1}{}}} 
\makeatother

\newenvironment{claimproof}[1][\unskip]
  {\begin{proof}[Proof of Claim #1.]}
  {\end{proof}}
\newcommand{\claimqedhere}{\qedhere}

\crefname{obs}{Observation}{Observations}
\Crefname{obs}{Observation}{Observations}
\crefname{fact}{Fact}{Facts}
\Crefname{fact}{Fact}{Facts}
\crefname{problem}{Problem}{Problems}
\Crefname{problem}{Problem}{Problems}
\crefname{conjecture}{Conjecture}{Conjectures}
\Crefname{conjecture}{Conjecture}{Conjectures}
\crefname{claim}{Claim}{Claims}
\Crefname{claim}{Claim}{Claims}

\newcommand\blfootnote[1]{%
	\begingroup
	\renewcommand\thefootnote{}\footnote{#1}%
	\addtocounter{footnote}{-1}%
	\endgroup
}

\newcommand{\Oh}{\mathcal{O}}

\newcommand{\Ii}{\mathcal{I}}
\newcommand{\Cc}{\mathcal{C}}
\newcommand{\Pp}{\mathcal{P}}

\newcommand{\Tt}{\mathcal{T}}

\newcommand{\Ll}{\mathcal{L}}
\newcommand{\nat}{\mathbb{N}}
\newcommand{\real}{\mathbb{R}_{+}}

\DeclareMathOperator{\poly}{poly}

\newcommandx{\unsure}[2][1=]{\todo[linecolor=green,backgroundcolor=green!25,bordercolor=green,#1]{\normalsize #2}}
\newcommandx{\improvement}[2][1=]{\todo[inline,linecolor=blue,backgroundcolor=blue!05,bordercolor=blue,#1]{\normalsize #2}}
\newcommandx{\info}[2][1=]{\todo[linecolor=yellow,backgroundcolor=yellow!25,bordercolor=yellow,#1]{#2}}
\newcommandx{\floatmodel}[2][1=]{\todo[inline,linecolor=red,backgroundcolor=yellow!25,bordercolor=yellow,#1]{#2}}
\newcommandx{\thiswillnotshow}[2][1=]{\todo[disable,#1]{#2}}

\newcommand{\ifab}[1]{\todo[inline,color=blue!40]{\textbf{Fabian:} 
#1}}

\newcommand{\SigTwoP}{\ensuremath{\Sigma_2^\mathsf{P}}\xspace}

\newcommand{\NP}{\ensuremath{\textsf{NP}}\xspace} %
\newcommand{\sharpP}{\ensuremath{\textsf{\#P}}\xspace} %

\newcommand{\yes}{\textsf{Yes}\xspace}
\newcommand{\no}{\textsf{No}\xspace}

\newcommand{\OurTriangle}{\triangle}
\newcommand{\OurEdge}{K_2}
\newcommand{\OurSquare}{\square}
\newcommand{\OurqClique}{K_q}

\renewcommand{\OurTriangle}{\text{Triangle}}
\renewcommand{\OurEdge}{\text{Edge}}
\renewcommand{\OurSquare}{\text{Square}}
\renewcommand{\OurqClique}{q\text{-Clique}}

\newcommand{\OurCycle}{\text{Cycle}}

\newcommand{\Pack}[1]{\textsc{\ensuremath{#1}-Packing}\xspace}
\newcommand{\EdgePack}{\Pack{\OurEdge}}
\newcommand{\TriPack}{\Pack{\OurTriangle}}

\newcommand{\Hit}[1]{\textsc{\ensuremath{#1}-Hitting}\xspace}
\newcommand{\EdgeHit}{\Hit{\OurEdge}}
\newcommand{\TriHit}{\Hit{\OurTriangle}}

\newcommand{\UndelHitPack}[1]{\textsc{\ensuremath{#1}-HitPack}\xspace}
\newcommand{\EdgeUndelHitPack}{\UndelHitPack{\OurEdge}}
\newcommand{\TriUndelHitPack}{\UndelHitPack{\OurTriangle}}
\newcommand{\SquareUndelHitPack}{\UndelHitPack{\OurSquare}}
\newcommand{\qCliqueUndelHitPack}{\UndelHitPack{\OurqClique}}
\newcommand{\CycleUndelHitPack}{\UndelHitPack{\OurCycle}}

\newcommand{\TriHitPart}{\textsc{\ensuremath{\OurTriangle}-HitPart}\xspace}

\newcommand{\QthreeDNFtwo}{\textsc{Q3DNF\ensuremath{{}_2}}\xspace}
\newcommand{\QthreeCNFtwo}{\textsc{Q3CNF\ensuremath{{}_2}}\xspace}
\newcommand{\textscup}[1]{\textnormal{\textsc{#1}}}

\newcommand{\SUS}{\textscup{SUS}\xspace}
\newcommand{\ThreeCNFSUS}{\textscup{3CNF-SUS}\xspace}

\DeclarePairedDelimiter{\abs}{\lvert}{\rvert}

\newcommand{\deff}{\coloneqq}
\newcommand{\from}{\colon}

\newcommand{\dimond}{\diamondsuit}
\newcommand{\nopack}{\bot}%
\newcommand{\seopack}{\bot\hspace{-0.629em}\raisebox{.21em}{\ensuremath{\looparrowright}}}%
\renewcommand{\seopack}{{\curlyveeuparrow}}

\def\emptyset{\varnothing}
\def\epsilon{\varepsilon}
\def\phi{\varphi} %

\def\lneg{\overline} %
\def\true{\mathsf{true}}
\def\false{\mathsf{false}}

\def\tilde{\widetilde}
\def\hat{\widehat}
\def\bar{\overline}

\newcommand{\Codomain}{\numbZ{n}} %
\newcommand{\CodomCyc}{\numbZ{2\tw+2} \cup \{\bot,\seopack\}} %
\newcommand{\TIntro}{\mathsf{Intro}} %
\newcommand{\TIntroV}{\mathsf{IntroVtx}} %
\newcommand{\TIntroE}{\mathsf{IntroEdge}} %
\newcommand{\TForget}{\mathsf{Forget}} %
\newcommand{\TJoin}{\mathsf{Join}} %
\newcommand{\dltd}{\mathsf{del}} %
\newcommand{\notdltd}{\mathsf{not{\text-}del}} %
\newcommand{\IntroVtxNodeDel}    {\ensuremath{\TIntroV_\dltd}\xspace}
\newcommand{\IntroVtxNodeNotDel} {\ensuremath{\TIntroV_\notdltd}\xspace}
\newcommand{\IntroEdgeNodeDel}   {\ensuremath{\TIntroE_\dltd}\xspace}
\newcommand{\IntroEdgeNodeNotDel}{\ensuremath{\TIntroE_\notdltd}\xspace}
\newcommand{\ForgetNodeDel}      {\ensuremath{\TForget_\dltd}\xspace}
\newcommand{\ForgetNodeNotDel}   {\ensuremath{\TForget_\notdltd}\xspace}
\newcommand{\JoinNode}           {\ensuremath{\TJoin}\xspace}

\newcommand{\List}{\mathcal L}%
\newcommand{\clq}{K_q} %

\newcommand{\Types}{\mathcal T}
\DeclareMathOperator{\prtn}{\mathsf{part}}
\newcommand{\extend}[3]{\textsf{extend}(#1,#2\mapsto #3)}
\newcommand{\remove}[2]{\textsf{remove}(#1,#2)}
\newcommand{\combine}[2]{#1 \oplus #2}
\newcommand{\reduce}{\mathsf{reduce}}
\newcommand{\packs}{\mathcal P} %
\newcommand{\SelGadget}{\mathsf{Sel}} %
\newcommand{\LitGadget}{\mathsf{Lit}} %

\newcommand{\AuxGadget}{\ensuremath{\mathsf{TriCyc}}}
\newcommand{\UniGadget}{\ensuremath{\mathsf{UniVar}}}

\newcommand{\Ptrue}{P_\mathsf{true}}
\newcommand{\Pfalse}{P_\mathsf{false}}

\newcommand{\pos}[1]{[#1]} %
\newcommand{\numb}[1]{[#1]} %
\newcommand{\numbZ}[2][0]{[#1..#2]} %

\newcommand{\KPack}{\nu_{\clq}} %
\newcommand{\CyComp}{\nu^\circ} %

\newcommand{\tw}{\textup{\textsf{tw}}} %
\newcommand{\pw}{\textup{\textsf{pw}}} %

\newcommand{\bin}[1]{\textup{bin}(#1)}

\newcommand{\dd}{\Upsilon}
\newcommand{\avail}{\mathbf{avail}}
\newcommand{\hit}{\mathbf{hit}}
\newcommand{\branch}{\textnormal{\textsf{Get-Candidates}}\xspace}
\newcommand{\good}{+}
\newcommand{\bad}{-}
\newcommand{\threedist}{\operatorname{dist}_{\ge 3}}

\newcommand{\lca}{\textnormal{\textsf{lca}}}

\newcommand{\degree}{\Oh(|F|^3+k|F|^2)}
\newcommand{\dist}{\textnormal{\textsf{dist}}}

\newcommand{\clit}{\mathsf{CLit}}
\newcommand{\csel}{\mathsf{CSel}}
\newcommand{\CID}{\mathsf{ClauseID}}

\graphicspath{{img/}}
\newcommand{\executeiffilenewer}[3]{%
\ifnum\pdfstrcmp{\pdffilemoddate{#1}}%
{\pdffilemoddate{#2}}>0%
{\immediate\write18{#3}}\fi%
}

\title{Hitting Meets Packing: How Hard Can it Be?}

\usepackage{authblk}
\author[1]{Jacob Focke}
\author[1]{Fabian Frei}
\author[1]{Shaohua Li}
\author[1]{D\'aniel Marx}
\author[1]{\\Philipp Schepper}
\author[2]{Roohani Sharma}
\author[2,3]{Karol W\k{e}grzycki}
\affil[1]{CISPA Helmholtz-Center for Information Security}
\affil[2]{Max Planck Institute for Informatics, SIC}
\affil[3]{Saarland University}

\date{}

\begin{document}

\maketitle

\begin{abstract}
We study a general family of problems
that form a common generalization of classic hitting
(also referred to as covering or transversal)
and packing problems.
An instance of \UndelHitPack{\mathcal{X}} asks:
Can removing $k$ (deletable) vertices of a graph $G$
prevent us from packing $\ell$ vertex-disjoint objects of type $\mathcal{X}$?
This problem captures a spectrum of problems
with standard hitting and packing on opposite ends.
Our main motivating question is
whether the combination \UndelHitPack{\mathcal X}
can be significantly harder than these two base problems.
Already for one particular choice of $\mathcal X$,
this question can be posed for many different complexity notions,
leading to a large, so-far unexplored domain
at the intersection of the areas of hitting
and packing problems.

At a high level, we present two case studies:
(1) $\mathcal{X}$ being all cycles,
and (2) $\mathcal{X}$ being all copies of a fixed graph $H$.
In each, we explore the classical complexity
as well as the parameterized complexity
with the natural parameters $k+\ell$ and treewidth.
We observe that the combined problem can be drastically harder
than the base problems:
for cycles or for $H$ being a connected graph on at least 3 vertices,
the problem is \SigTwoP-complete and requires double-exponential dependence
on the treewidth of the graph (assuming the Exponential-Time Hypothesis).
In contrast, the combined problem admits qualitatively similar running times
as the base problems in some cases,
although significant novel ideas are required.
For $\mathcal{X}$ being all cycles,
we establish a $2^{\poly(k+\ell)}\cdot n^{\Oh(1)}$ algorithm
using an involved branching method, for example.
Also, for $\mathcal X$ being all edges
(i.e., $H = K_2$;
this combines \textsc{Vertex Cover} and \textsc{Maximum Matching})
the problem can be solved in time
$2^{\poly(\tw)}\cdot n^{\Oh(1)}$ on graphs of treewidth $\tw$.
The key step enabling this running time relies on a combinatorial bound
obtained from an algebraic (linear delta-matroid)
representation of possible matchings.
\end{abstract}

\section{Introduction}\label{sec:intro}\blfootnote{A table of
contents is provided on page~\pageref{toc} at the \hyperref[toc]{end
of this work}.}%
In the combinatorial optimization literature, many algorithmic
problems can be classified into one of two dual classes: either as a
packing problem or as a hitting problem.
In \emph{packing problems}, the goal is to find a large pairwise
independent collection of objects of certain type.
For example, one of the most-studied
problems,
\textsc{Maximum Matching}, can be described as finding a pairwise
vertex-disjoint collection of at least $\ell$ edges.
Network flow problems require finding a large collection of
edge-disjoint paths from $s$ to $t$.
More generally, one can define the {\Pack{\mathcal{X}}} problem for
any type $\mathcal{X}$ of objects.
In \emph{hitting problems} (sometimes referred to as
\textit{transversal} or \textit{covering} problems), the task is to
find
a small set of
elements that
\emph{hits} (i.e., intersects) every object of a certain type
$\mathcal{X}$.
For example, the \textsc{Vertex Cover} problem can be described as
finding a set of at most $k$ vertices that intersect every edge
(i.e., contains at least one endpoint of each edge).
The minimum $s$-$t$ cut problem can be interpreted as finding a set
of edges that intersects every $s$-$t$ path.
A quick note on terminology is in order here.
The name of a covering problem typically refers to the type of
objects used to cover, rather than the type of objects being
covered.
For example, \textsc{Cycle Cover} usually refers to the problem of
covering the vertices of the graph with few (not necessarily
disjoint) cycles, and \emph{not} the
problem of hitting every cycle with a small set of vertices (which is
usually called \textsc{Feedback Vertex Set}).
For this reason, we prefer to use \Hit{\mathcal{X}} for the
problem where the task is to find a set of elements or vertices that
intersect every object of type $\mathcal{X}$.

There is a well-known duality phenomenon connecting hitting and
packing  problems.
If there are $\ell$ disjoint objects of type $\mathcal{X}$, then
clearly we need at least $\ell$ elements to hit every such object.
In other words, the optimum of \Hit{\mathcal{X}} is at least the
optimum of \Pack{\mathcal{X}}.
For some type of objects (such as edges in bipartite graphs and
$s$-$t$ paths),
celebrated duality theorems demonstrate that there is always equality
between the two optimum values.
These duality results and their variants underlie many of the
polynomial-time exact algorithms in combinatorial optimization.
For problems where the two optimum values do not coincide, it is
natural to ask how large the gap
can be.
Erd\H os and P\'osa~\cite{MR0175810} showed that
if $\ell$ is the maximum number of vertex-disjoint cycles,
then all the cycles can be hit by a set of $k=\Oh(\ell\log \ell)$
vertices.
More generally, we say that a type $\mathcal{X}$ of objects
has the \emph{Erd\H os--P\'osa property} if the hitting optimum
can be bounded by a function of the packing optimum.
For example, it is known that the Erd\H os--P\'osa property
holds for undirected cycles passing through a set $S$
\cite{DBLP:journals/jct/KakimuraKM11,DBLP:journals/jct/PontecorviW12}
or directed cycles~\cite{DBLP:journals/combinatorica/ReedRST96},
but it does not hold for cycles of odd
length~\cite{DBLP:journals/combinatorica/Reed99}.

In this paper, we study a different, algorithmic, question that
connects hitting and packing problems. Let $\mathcal{X}$ be a type of
objects in graphs, and consider the following problem.
Given a graph $G$
and integers $k$ and $\ell$, the task is to find a set $S$ of at most
$k$ vertices such that $G-S$ does not contain $\ell$ disjoint copies
of objects of type $\mathcal{X}$. This unified formulation
captures both \Hit{\mathcal{X}} and \Pack{\mathcal{X}} problems:
for $\ell=1$, it asks if every object can be hit with $k$ vertices;
for $k=0$, it asks whether it is impossible to find $\ell$ disjoint
objects. Therefore, \UndelHitPack{\mathcal{X}} is at least as hard as
\Hit{\mathcal{X}} and the complement of \Pack{\mathcal{X}}.
Note that deterministic algorithms, on which we focus in this paper,
always work for the complement of a problem as well.
The main meta-question that we explore is how hard such a combination
of two %
problems may become:
\medskip

\mdfdefinestyle{highlight}{frametitlebackgroundcolor=gray!40,backgroundcolor=gray!20}

\begin{mdframed}[style=highlight,userdefinedwidth=0.7\linewidth,align=center]
If some type of algorithm exists for both \Hit{\mathcal{X}} and
\Pack{\mathcal{X}}, then is there such an algorithm for
\UndelHitPack{\mathcal{X}} as well?
\end{mdframed}

The main message of this paper is that the formulation of this
question leads to a whole new unexplored continent of interesting and
challenging questions. As we shall see, in some settings the combined
problem is indeed strictly harder, while in other settings a
qualitatively similar algorithm can be obtained for the combined
problem, albeit only after developing significantly more involved
techniques.

To make the problem statement more robust,
we extend the problem by assuming that the input graph contains a set
$U$
of undeletable vertices and the solution $S$ has to be disjoint from
$U$. Such undeletable vertices may be needed to express problems
where the objects $\mathcal{X}$ are, say, $s$-$t$ paths or paths
between terminals, and we do not want to allow the deletion of
terminals.
Note that this generalization makes our algorithmic results slightly
stronger,
while it makes the lower bound results sightly weaker.
Formally, for a type $\mathcal{X}$ of objects, the problem is defined
as follows.

\begin{center}
  \fbox{\parbox{0.7\linewidth}{
      \UndelHitPack{\mathcal{X}}
      \smallskip

      \begin{tabularx}{\linewidth}{rX}
          \textbf{Input:} & Graph $G$, set $U\subseteq V(G)$,
          integers $k$ and $\ell$\\
          \textbf{Question:} & Is there a set $S\subseteq
          V(G)\setminus U$ of size $\le k$ such that $G-S$ does not
          contain $\ell$ vertex-disjoint objects of type
          $\mathcal{X}$?
        \end{tabularx}
      }}
\end{center}

There are two ways of looking at the \UndelHitPack{\mathcal{X}}
problem. It can
be considered as a \emph{weaker} version of hitting: the solution
does not have to destroy all sets in $\mathcal{X}$, but up to
$\ell$ disjoint sets are allowed to survive in $G-S$. An alternative
view is to
interpret it as a more stable version of packing. We have to
decide not only whether $\ell$ disjoint objects exist but whether the
graph may lose this property even if up to $k$ arbitrary vertices are
removed. Such a robust version of packability is clearly desirable in
many situations, and the problem of detecting this property is
precisely the complement of \UndelHitPack{\mathcal{X}}.
Frei et al.~\cite{FHR22} recently initiated a systematic complexity
study
of a related stability notion, where a graph
is \emph{vertex-stable} if some parameter cannot change upon deletion
of a single vertex.
Our \UndelHitPack{\mathcal{X}} models an even more robust notion of
stability, however, since multiple vertices may be deleted depending
on the input.

Returning to the overarching question how hard
\UndelHitPack{\mathcal{X}} may become compared to \Hit{\mathcal{X}},
and \Pack{\mathcal{X}}, let us start with examples where all of them
are polynomial-time solvable. If $\mathcal{X}$ is simply the set of
edges
of a bipartite graph, then the hitting problem (\textsc{Vertex
Cover}) and the packing problem (\textsc{Maximum Matching}) are both
known to be polynomial-time solvable and the size of the minimum
vertex cover and the maximum matching are known to be the same. Let
$d$ be this value. A set of $k$ vertices can decrease the size of the
maximum matching only by at most $k$, thus the answer is no if
$d-k\ge \ell$. Otherwise, if $d-k<\ell$, then deleting any $k$
vertices of a minimum vertex cover decreases the size of a maximum
matching by $k$ (as deleting $d-k$ further vertices of the vertex
cover decreases this size to 0), showing that the answer is yes. This
argument works for other objects where exact duality theorems are
known, for example for internally vertex-disjoint $s$-$t$ paths
(since the maximum number of disjoint paths is equal to the minimum
$s$-$t$ separator).

In general, however, no such exact duality theorem is available. In
such a case, \Hit{\mathcal{X}} and \Pack{\mathcal{X}} are two very
different problems that may require different techniques. Then,
solving \UndelHitPack{\mathcal{X}} would require combining the two
solution techniques in a nontrivial way, and it very well may be the case that  \UndelHitPack{\mathcal{X}} is a qualitatively harder problem that both
\Hit{\mathcal{X}} and \Pack{\mathcal{X}}. Let us point out that one can explore different aspects of hardness (NP-hardness, exact running times, parameterized complexity, approximation, etc.), thus already for one particular choice of $\mathcal{X}$, one can ask many different question. This means that understanding the \UndelHitPack{\mathcal{X}} problem is a two-dimensional question: one dimension is the choice of $\mathcal{X}$ and the other dimension is the notion of complexity.
We present two case studies
(the objects $\mathcal{X}$ being cycles of arbitrary size or
subgraphs isomorphic to a fixed graph $H$) and explore different
aspects of the complexity of \UndelHitPack{\mathcal{X}}.
See \cref{tab:overview} for an overview of our results.

\begin{figure}

\begin{center}\def\svgwidth{0.54\linewidth}
\begingroup%
  \makeatletter%
  \providecommand\color[2][]{%
    \errmessage{(Inkscape) Color is used for the text in Inkscape, but the package 'color.sty' is not loaded}%
    \renewcommand\color[2][]{}%
  }%
  \providecommand\transparent[1]{%
    \errmessage{(Inkscape) Transparency is used (non-zero) for the text in Inkscape, but the package 'transparent.sty' is not loaded}%
    \renewcommand\transparent[1]{}%
  }%
  \providecommand\rotatebox[2]{#2}%
  \newcommand*\fsize{\dimexpr\f@size pt\relax}%
  \newcommand*\lineheight[1]{\fontsize{\fsize}{#1\fsize}\selectfont}%
  \ifx\svgwidth\undefined%
    \setlength{\unitlength}{247.99995704bp}%
    \ifx\svgscale\undefined%
      \relax%
    \else%
      \setlength{\unitlength}{\unitlength * \real{\svgscale}}%
    \fi%
  \else%
    \setlength{\unitlength}{\svgwidth}%
  \fi%
  \global\let\svgwidth\undefined%
  \global\let\svgscale\undefined%
  \makeatother%
  \begin{picture}(1,0.69758074)%
    \lineheight{1}%
    \setlength\tabcolsep{0pt}%
    \put(0,0){\includegraphics[width=\unitlength,page=1]{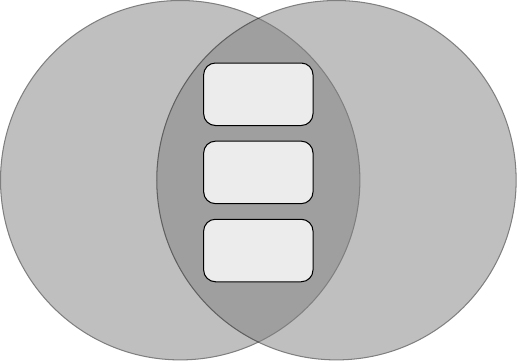}}%
    \put(0.50000013,0.5272176){\color[rgb]{0,0,0}\makebox(0,0)[t]{\lineheight{1.25}\smash{\begin{tabular}[t]{c}duality\\results\end{tabular}}}}%
    \put(0.50000013,0.37600784){\color[rgb]{0,0,0}\makebox(0,0)[t]{\lineheight{1.25}\smash{\begin{tabular}[t]{c}Erd\H os-P\'osa\\property\end{tabular}}}}%
    \put(0.50000013,0.22479825){\color[rgb]{0,0,0}\makebox(0,0)[t]{\lineheight{1.25}\smash{\begin{tabular}[t]{c}\UndelHitPack{\mathcal{X}}\\problem\end{tabular}}}}%
    \put(0.16733881,0.34879037){\color[rgb]{0,0,0}\makebox(0,0)[t]{\lineheight{1.25}\smash{\begin{tabular}[t]{c}\LARGE Hitting\\\LARGE Problems\end{tabular}}}}%
    \put(0.83266128,0.34879037){\color[rgb]{0,0,0}\makebox(0,0)[t]{\lineheight{1.25}\smash{\begin{tabular}[t]{c}\LARGE Packing\\\LARGE Problems\end{tabular}}}}%
    \put(0,0){\includegraphics[width=\unitlength,page=2]{venn.pdf}}%
  \end{picture}%
\endgroup%
\end{center}
\caption{The areas of hitting and packing problems intersect in
combinatorial duality and Erd\H os--P\'osa property results, and in
the algorithmic study of combined hitting and packing problems.}
\end{figure}

\paragraph*{Case Study 1: Cycles.}
Let us first consider the case when the type $\mathcal{X}$ of objects
are the cycles in the given graph $G$.
Thus, \Hit{\mathcal{X}} is exactly \textsc{Feedback Vertex Set}
(the problem of deciding whether there is a set $S$ of up to $k$
vertices such that $G-S$ has no cycle),
\Pack{\mathcal{X}} is \textsc{Cycle Packing} (asking whether are
there $\ell$ vertex-disjoint cycles),
and \UndelHitPack{\mathcal{X}} is the problem of asking if it is
possible to
remove a set $S$ of $k$ (deletable) vertices such that
the remaining graph does not contain $\ell$ vertex-disjoint cycles.

\textsc{Feedback Vertex Set} and \textsc{Cycle Packing} are both
known to be \NP-complete.
As we have seen, \CycleUndelHitPack generalizes both \textsc{Feedback
Vertex Set} and the \emph{complement} of \textsc{Cycle Packing}, thus
it is unlikely to be in \NP.
Indeed, a solution $S$ of size $k$ is \emph{not} a good certificate,
as it is hard to verify due to the \NP-hardness of \textsc{Cycle
Packing}.
We show that \CycleUndelHitPack is in fact located further up in the
polynomial hierarchy.

\begin{restatable}{theorem}{thmcyclepwLBSigP}
  \label{thm:sig2p:completeness:cycles}
\CycleUndelHitPack is \SigTwoP-complete.
\end{restatable}
Another jump in complexity can be
observed if we consider how the
problems behave on graphs of bounded treewidth.
The study of parameterized algorithms and complexity on such graphs
has been a fruitful area of research,
as many \NP-hard problems become tractable when restricted to graphs
of small treewidth.
In many cases, the problems are fixed-parameter tractable (FPT)
parameterized by treewidth:
there is an algorithm solving the problem in time $f(\tw)\cdot
n^{\Oh(1)}$ for some function $f$.
By now, there is a strong understanding on how the complexity of
different problems depend on the treewidth of the graph,
with a number of nontrivial algorithmic techniques and tight
conditional lower bounds appearing in the literature
\cite{DBLP:journals/siamcomp/OkrasaR21,DBLP:conf/esa/OkrasaPR20,DBLP:conf/stacs/EgriMR18,DBLP:conf/soda/CurticapeanLN18,
DBLP:conf/iwpec/BorradaileL16,DBLP:journals/dam/KatsikarelisLP19,
DBLP:conf/icalp/MarxSS21,
DBLP:conf/soda/CurticapeanM16,DBLP:conf/soda/FockeMR22,DBLP:conf/iwpec/MarxSS22,DBLP:conf/soda/FockeMINSSW23,cut-and-count,DBLP:journals/iandc/BodlaenderCKN15,DBLP:conf/icalp/MarxM16,DBLP:journals/siamcomp/LokshtanovMS18,DBLP:conf/sat/LampisMM18,DBLP:conf/iwpec/LampisM17}.
The function $f$ is typically of the form $2^{\poly(\tw)}$,
and we know of only a handful of problems where a double- or even
triple-exponential dependence on treewidth is necessary
assuming the Exponential-Time Hypothesis (ETH)
\cite{DBLP:conf/icalp/MarxM16,DBLP:conf/iwpec/LampisM17,DBLP:journals/corr/abs-1302-4266,DBLP:conf/lics/PanV06,galby-talk}.
One may see a pattern that problems at higher level of the polynomial
hierarchy
may need more than exponential dependence on treewidth,
but note that there are problems in \NP that need double-exponential
dependence
(metric dimension~\cite{galby-talk}) and there are \sharpP-hard
counting problems%
\footnote{By Toda's theorem~\cite{Toda1991}, it is known that \sharpP
contains the entire polynomial hierarchy.}
that can be solved with single-exponential
dependence~\cite{DBLP:conf/soda/FockeMINSSW23,DBLP:conf/soda/FockeMR22,DBLP:conf/soda/CurticapeanM16,DBLP:conf/icalp/MarxSS21}.
Both \textsc{Feedback Vertex Set} and \textsc{Cycle Packing} can be
solved in time $2^{\poly(\tw)}\cdot n^{\Oh(1)}$
(it is known that, assuming ETH, the optimal dependence on treewidth
is $2^{\Oh(\tw)}$ and $2^{\Oh(\tw\log\tw)}$ for the two problems,
respectively).
Does \CycleUndelHitPack, the common generalization of the two
problems, also admit an algorithm of this running time?
We answer this question in the negative: the dependence becomes
double-exponential on treewidth.
\begin{restatable}{theorem}{thmcycletwUB}
  \label{thm:twUpper:cycle}
  \CycleUndelHitPack can be solved in time
  $2^{2^{\Oh(\tw \log \tw)}}\cdot n^{\Oh(1)}$,
  where \tw\ is the treewidth of the input graph.
\end{restatable}

We complement this upper bound by a matching lower bound
which not only proves the double-exponential dependence
for the larger parameter pathwidth
but also that the additional logarithmic factor cannot be avoided.
\begin{restatable}{theorem}{thmcyclepwLB}
\label{thm:tw:lower:cycles}
  Assuming ETH, \CycleUndelHitPack has no
  $2^{2^{o(\pw \log \pw)}}\cdot n^{\Oh(1)}$ time algorithm,
  where \pw\ is the pathwidth of the input graph.
\end{restatable}

\textsc{Feedback Vertex Set} is known
to be fixed-parameter tractable
(FPT) parameterized by $k$,
in fact, it can be solved in time $2^{\Oh(k)}\cdot n^{\Oh(1)}$
\cite{DBLP:conf/soda/Cao18,DBLP:journals/ipl/KociumakaP14,DBLP:journals/mst/DehneFLRS07,DBLP:journals/jcss/GuoGHNW06}.
\textsc{Cycle Packing} is FPT parameterized by $\ell$, but it is an
open question whether there is a $2^{\Oh(\ell)}\cdot n^{\Oh(1)}$ time
algorithm:
the current best algorithm has running time $2^{\Oh(\ell \log ^2
\ell/\log\log \ell)}$ \cite{DBLP:journals/siamdm/LokshtanovMSZ19}.
As \CycleUndelHitPack is \NP-hard for $k=0$ and also for $\ell=1$,
the natural parameter for the problem is $p \deff k+\ell$.

We can try to use the following approach to show that
\CycleUndelHitPack is FPT parameterized by $p$.
Suppose that the input graph $G$ has $p=k+\ell$ disjoint cycles.
Then for every set $S$ of $k$ vertices,
the graph $G-S$ has $\ell$ disjoint cycles,
implying that $G$ is a no-instance of \CycleUndelHitPack.
Therefore, we can assume that $G$ has at most $p$ disjoint cycles.
Then the Erd\H os--P\'osa Theorem implies that $G$ has a feedback
vertex set of size $\Oh(p\log p)$, which also implies that $G$ has
treewidth $\Oh(p\log p)$.
Now we can try to use an algorithm for \CycleUndelHitPack
parameterized by treewidth. However, in light of
\cref{thm:tw:lower:cycles},
any such algorithm would give a running time with double-exponential
dependence on $p$.
Can we improve this running time to $2^{\poly(p)}\cdot n^{\Oh(1)}$,
to qualitatively match the running times of the \textsc{Feedback
Vertex Set} and \textsc{Cycle Packing} algorithms?
We show that this is indeed possible, and hence we do not see such a
drastic jump in complexity similar to the \SigTwoP-completeness of
the problem and the double-exponential dependence on treewidth.
\begin{restatable}{theorem}{cyclekl}\label{thm:cycle-k-l}
  \CycleUndelHitPack can be solved in time $2^{\poly(p)}\cdot
  n^{\Oh(1)}$
  (where $p\deff k+\ell$).
\end{restatable}
The proof exploits that $G$ has a feedback vertex set $F$ of size
$\Oh(p\log p)$.
By a simple branching argument, we assume that the solution $S$
is disjoint from $F$.
Then we interpret a packing of cycles as a collection of paths
connecting some neighbors of $F$ in the forest $G-F$.
Our goal is to hit every such collection of paths that would lead to
a collection of $\ell$ cycles.
With a branching algorithm, we collect paths that have to be hit,
until we can conclude that our collection of paths cannot be hit by
$k$ vertices.

\paragraph*{\boldmath Case Study 2: $H$-subgraphs.}
Next, let us consider the setting where the type $\mathcal{X}$ of
objects
are the (not necessarily induced) subgraphs isomorphic to a fixed
graph $H$.
Thus, \Hit{H} is the problem of removing a set $S$ of $k$ vertices
such that no subgraph isomorphic to $H$ remains,
\Pack{H} is finding $\ell$ disjoint copies of $H$ as subgraphs,
and \UndelHitPack{H} is the problem of removing a set $S$ of $k$
vertices
such that the remaining graph does not contain $\ell$ disjoint copies
of $H$.

For every fixed connected graph $H$ with at least 3 vertices,
\Hit{H} \cite{DBLP:journals/siamcomp/KirkpatrickH83}
and \Pack{H} \cite{DBLP:journals/jcss/LewisY80} are \NP-complete.
Similarly to the case of \CycleUndelHitPack, the \UndelHitPack{H}
problem lies on the second level of the polynomial hierarchy.

\begin{restatable}{theorem}{thmHSigTwoPLower}
  \label{thm:sig2p:completeness:H}\label{lem:sig2p:hardness:H}
  For any fixed connected graph $H$ with at least three vertices,
  \UndelHitPack{H} is \SigTwoP-complete.
\end{restatable}
Similarly to \CycleUndelHitPack, we again see a jump to
double-exponential dependency on treewidth and pathwidth.
For the case of general graphs $H$, we provide an algorithm
whose running time asymptotically matches the one
of the algorithm for \CycleUndelHitPack.
\begin{restatable}{theorem}{thmUpperTWGeneral}
  \label{thm:twUpper:arbitrary}
  For any fixed connected graph $H$, \UndelHitPack{H} can be solved
  in time $2^{2^{\Oh(\tw\log \tw)}}\cdot n^{\Oh(1)}$, where \tw\ is
  the treewidth of the input graph.
\end{restatable}

In the case when $H$ is a clique,
we exploit that we are only packing complete graphs
which leads to an improvement
where we remove the logarithmic factor from the exponent.
\begin{restatable}{theorem}{thmUpperTWClique}
  \label{thm:twUpper:clique}
  For any fixed integer $q \ge 2$,
  \qCliqueUndelHitPack can be solved in time
  $2^{2^{\Oh(\tw)}} \cdot n^{\Oh(1)}$,
  where $\tw$ is the treewidth of the input graph.
\end{restatable}

By designing the matching lower bound for \CycleUndelHitPack in such
a way
that the construction already works for \SquareUndelHitPack,
we obtain a matching lower bound for \SquareUndelHitPack.
\begin{restatable}{theorem}{thmTwLowerCFour}
  \label{thm:tw:lower:C4}
  Assuming ETH, \SquareUndelHitPack has no
  $2^{2^{o(\pw \log \pw)}}\cdot n^{\Oh(1)}$ time algorithm,
  where $\pw$ is the pathwidth of the input graph.
\end{restatable}

For the case of general $H$ we provide a separate reduction.
In contrast to the matching lower bound for the case of
\SquareUndelHitPack,
we present a lower bound which is only matching for the case of
cliques
as for cliques we provide an improved algorithm.
For the case when $H$ is neither a clique nor a $C_4$,
it remains open to remove the logarithmic factor from the running time
or to improve the lower bound accordingly.
\begin{restatable}{theorem}{thmHpwLB}
\label{th:H-pw-lower-bound}\label{thm:tw:lower:H}\label{thm:tw:lower}
  Assuming ETH, for any fixed connected graph $H$ with at least three
  vertices,
  \UndelHitPack{H} has no $2^{2^{o(\pw)}}\cdot n^{\Oh(1)}$ time
  algorithm,
  where \pw\ is the pathwidth of the input graph.
\end{restatable}
It is known that \Hit{H} parameterized by $k$ and \Pack{H}
parameterized by $\ell$ are both fixed-parameter tractable (FPT),
in fact, for fixed $H$, they can be solved in time $2^{\Oh(k)}\cdot
n^{\Oh(1)}$
and $2^{\Oh(\ell)}\cdot n^{\Oh(1)}$, respectively.
For \Hit{H}, this follows from a simple bounded-depth search tree
algorithm, while color coding~\cite{DBLP:journals/jacm/AlonYZ95} or
representative set techniques~\cite{DBLP:journals/jacm/FominLPS16}
can be used for \Pack{H}.
There is an easy bounded-tree search tree algorithm
showing that \UndelHitPack{H} is FPT parameterized by $p \deff
k+\ell$.
\begin{theorem}
  \label{thm:general:k-plus-ell}
  For any fixed graph $H$ that might be unconnected,
  \UndelHitPack{H} can be solved in time $2^{\Oh(p\log p)}\cdot
  n^{\Oh(1)}$
  (where $p \deff k + \ell$).
\end{theorem}
\begin{proof}
First, we use the $2^{\Oh(\ell)}\cdot n^{\Oh(1)}$ \Pack{H}
algorithm
to find $\ell$ copies of $H$ if they exist
\cite{DBLP:journals/jacm/AlonYZ95,DBLP:journals/jacm/FominLPS16}.
The solution $S$ has to contain at least one of the
$|V(H)|\cdot \ell=\Oh(\ell)$ vertices of this packing.
Thus, we can branch on choosing one of these vertices,
delete that vertex from the graph,
and decrease our quota $k$ of deletions by one.
We repeat this process until either $k=0$
or there are no $\ell$ disjoint copies of $H$ in the graph.

This results in a search tree of size $\ell^{\Oh(k)}=2^{\Oh(p\log p)}$
and thus, concludes the proof.
\end{proof}
We leave open the potentially challenging question of
whether this running time of the algorithm
can be improved to $2^{\Oh(p)}\cdot n^{\Oh(1)}$:
none of the techniques used for \Hit{H} and \Pack{H}
seem directly relevant for such improvement.
In other words, it is easy to show that \UndelHitPack{H} is also FPT,
but whether quantitatively same FPT algorithms can be obtained
for this more general problem is a far from trivial question.

\paragraph*{\boldmath The Curious Case of \EdgeUndelHitPack.}
Let us consider now the case of \UndelHitPack{H} where $H$ consists
of only a single edge, i.e., $H=K_2$.
Then \EdgeHit is the \NP-complete problem \textsc{Vertex Cover},
while \EdgePack is the polynomial-time solvable \textsc{Maximum
Matching} problem.
This implies that, unlike in the cases where $H$ has at least 3
vertices, \EdgeUndelHitPack is in \NP:
given a solution $S$, we can verify in polynomial time that $G-S$ has
no matching of size $\ell$.

While we do not have a $2^{\Oh(k+\ell)}\cdot n^{\Oh(1)}$ time
algorithm for \UndelHitPack{H} for general $H$, we present a very
simple algorithm solving \EdgeUndelHitPack in time $3^{k+\ell}\cdot
n^{\Oh(1)}$.
The algorithm essentially relies on an augmenting-path argument,
hence it gives no indication on how other \UndelHitPack{H} problems
could be solved with a similar running time.

\begin{restatable}{theorem}{edgeparaalg}
  \label{th:k2-param-alg}
  \EdgeUndelHitPack can be solved in time $3^{k+\ell}\cdot
  n^{\Oh(1)}$.
\end{restatable}

\Cref{th:H-pw-lower-bound} showed that double-exponential dependence
on
treewidth is needed to solve the \UndelHitPack{H} problem when $H$ is connected
and has at least 3
vertices, that is, already for the \TriUndelHitPack problem (we
will also denote this problem as \TriUndelHitPack).
However, \EdgeUndelHitPack can be solved with only exponential
dependence on treewidth.
\begin{restatable}{theorem}{edgetwalg}\label{th:k2-tw-alg}
   \EdgeUndelHitPack can be solved in time $2^{\poly(\tw)}\cdot
   n^{\Oh(1)}$,
  where $\tw$ is the treewidth of the graph.
\end{restatable}

Let us give an intuitive explanation for this difference in running
time between
\EdgeUndelHitPack and \TriUndelHitPack.
There is a well-understood methodology for designing algorithms on
tree
decompositions: for each rooted subtree of the tree decomposition, we
define a
certain number of subproblems, each asking for the existence of a
certain class
of partial solutions. The running time typically depends on how many
equivalence
classes of partial solutions we need to consider. For example, in the
\TriPack problem, the class of partial packings is described by the
subset of the bag that is covered by the packing, so there are
$2^{\poly(\tw)}$
different classes of partial solutions. For the \TriHit problem, a
partial solution is a set of vertices that destroys every triangle in
a rooted subtree of the tree decomposition, and its class is
described by its intersection with the bag.

For the combination, \TriUndelHitPack, a partial solution is a set of
vertices that does not necessarily destroy every triangle in the
subtree of the tree decomposition, but may still leave some triangle
packings of size $<\ell$ in the graph.
Therefore, a partial solution $S$ can be described by what kind of
triangle packings survive after deleting $S$, that is, by describing
which subsets of the bag can be covered/avoided by triangle packings
of a certain size. This means that the class of a partial solution is
described by a set system over a bag of the decomposition.
As the set systems arising this way can be fairly arbitrary in the
\TriUndelHitPack problem, there are up to $2^{2^{\poly(\tw)}}$ such
set systems and hence up to that many different classes of partial
solutions.
This is the intuitive reason why a double-exponential dependence on
treewidth is
needed for the problem \TriUndelHitPack.

In the case of \EdgeUndelHitPack, the set systems describing a
partial solution show how the bag can be covered by matchings of a
certain size.
Such set systems have lots of structure and cannot be completely
arbitrary; in particular, they are related to (delta)-matroids.
Inspired by an argument of Wahlstr\"
om~\cite{DBLP:journals/corr/abs-2306-03605},
we give a combinatorial bound showing that such set systems can be
represented algebraically with $\Oh(\tw^3)$ bits,
hence there are at most $2^{\Oh(\tw^3)}$ different set systems that
can arise.
Interestingly, our proof is not algorithmic, but it is sufficient to
bound the running time of our algorithm. In fact, we have to make no
adjustment to the algorithm of \cref{thm:twUpper:arbitrary}
solving \UndelHitPack{H} when $H$ is a clique: the algorithm was
designed in a way that a combinatorial proof on the relevant set
systems immediately bounds the running time of the algorithm.

\begin{table}[t]
  \newcommand{\polyfac}{\ensuremath{\cdot n^{\Oh(1)}}}
\newcolumntype{Y}{>{\centering\arraybackslash}X}
\newcolumntype{Z}{>{\raggedright\arraybackslash}X}
\newcolumntype{W}{>{\raggedright\arraybackslash}X}
\newcommand\T{\rule{0pt}{5ex}}       %
\newcommand\B{\rule[-1.2ex]{0pt}{0pt}} %
\noindent

\aboverulesep = 0pt
\belowrulesep = 0pt
\renewcommand{\arraystretch}{1.4}
\crefformat{theorem}{\hfill\footnotesize{Th.~#2#1#3}}
\crefformat{linked-theorem}{\hfill\footnotesize{Th.~#2#1#3}}

\begin{tabularx}{\textwidth}
	{
		>{\hsize=.22\hsize}X||
		>{\hsize=.21\hsize}Z|
		>{\hsize=.215\hsize}Z|
		>{\hsize=.26\hsize}Z|
		>{\hsize=.19\hsize}W}
	{{Object} $\mathcal X$}
    & \multicolumn{1}{c|}{{UB  $p=k+\ell$}}
    & \multicolumn{1}{c|}{{UB Treewidth}}
    & \multicolumn{1}{c|}{{LB Treewidth}}
    & \multicolumn{1}{c }{{Completeness}}
    \\
    \hhline{=::====}
		{Edge}
      & $3^p$ {\cref{th:k2-param-alg}}
      & $2^{\poly(\tw)}${\cref{th:k2-tw-alg}}
      & no $2^{o(\tw)}${\hfill\footnotesize\cite[Th.~1]{LokshtanovMS18}}
      &\NP{\hfill\footnotesize\cite[Th.~3.3]{garey-johnson-guide}}
      \\\hhline{-||----}
		{Triangle}
      & %
      &
      \multirow{2}{=}[-0.0ex]{$2^{2^{\Oh(\tw)}}${\cref{thm:twUpper:clique}}}
      & %
      & $\SigTwoP${\cref{thm:sig2p:hardness:tripart}}
      \\\hhline{-||~|~|~|-}
		\raisebox{.2ex}{{$q$-Clique}}
      & $2^{\Oh(p \log
      p)}${\cref{thm:general:k-plus-ell}}
      & %
      & no $2^{2^{o(\tw)}}${\cref{th:H-pw-lower-bound}}
      & %
      \\\hhline{-||~|~|~|-}
		{Conn. $H$, 3+ vert.}
      & %
      & \multirow{2}{=}[-0.0ex]{$2^{2^{\Oh(\tw \log
      \tw)}}${\cref{thm:twUpper:arbitrary}}}
      & %
      & $\SigTwoP${\cref{lem:sig2p:hardness:H}}
      \\\hhline{-||-|~|-|~}
		\raisebox{.1ex}{{Square}}
      & $2^{\Oh(p \log p)}${\cref{thm:cycle-k-l}}
      & %
      & no $2^{2^{o(\tw \log \tw)}}${\cref{thm:tw:lower:C4}}
      & %
      \\\hhline{-||-|-|-|-}
		{Cycles}
      & $2^{\poly(p)}${\cref{thm:cycle-k-l}}
      & $2^{2^{\Oh(\tw \log
      \tw)}}${\cref{thm:twUpper:cycle}}
      & no $2^{2^{o(\tw \log
      \tw)}}${\cref{th:H-pw-lower-bound}}
      & $\SigTwoP${\cref{thm:sig2p:completeness:cycles}}
      \\
\end{tabularx}

  \caption{%
  An overview of the main results for \UndelHitPack{\mathcal X}.
  For ease of comparability, we omit the common factor
  $n^{\mathcal{O}(1)}$ from the FPT running times.
  }
  \label{tab:overview}
\end{table}
\belowrulesep = 0.984mm

\crefformat{theorem}{Theorem~#2#1#3}
\crefformat{linked-theorem}{Theorem~#2#1#3}

\paragraph*{Discussion and Open Problems.}
We have initiated the study of a natural common generalization of
hitting problems and packing problems. Certain basic techniques for
hitting and packing problems can be lifted to this generalization,
but we have seen that the generalization can be significantly harder
and more challenging, requiring us to revisit classic problems from a
new perspective. The familiar landscape of hitting and packing
problems with their known properties and well-established techniques
is replaced by a strange world where many of the known techniques are
inapplicable, new techniques have to be brought in, and the problem
has to be approached with a completely different mindset that takes
into account the more complicated quantifier structure of the problem
definition.

We have presented a selection of algorithmic results and lower bounds
for \UndelHitPack{\mathcal{X}} problems, but they probably just
scratch the surface of a rich family of unexplored challenging
problems. We list a few open questions and potential research
directions to stimulate further work in this area.

\begin{itemize}
\item Is there a $2^{\Oh(k+\ell)}\cdot n^{\Oh(1)}$ time algorithm for
\UndelHitPack{H} for every fixed (connected) graph $H$?
\item Is the \UndelHitPack{\mathcal{X}} problem FPT in $k$ and $\ell$
where $\mathcal{X}$ are the odd cycles in the graph?
Note that the corresponding hitting problem \textsc{Odd Cycle
Transversal}
is well-known to be FPT by different techniques
\cite{DBLP:journals/orl/ReedSV04,DBLP:journals/talg/LokshtanovNRRS14,DBLP:journals/talg/RamanujanS17},
and \textsc{Odd Cycle Packing} is also FPT using an extension of
graph minor algorithms with parity conditions
\cite{DBLP:conf/focs/KawarabayashiRW11,DBLP:conf/stoc/KawarabayashiR10}.
\item Is the \UndelHitPack{\mathcal{X}} problem FPT in $k$ and $\ell$
where $\mathcal{X}$ are the induced cycles of length at least 4 in
the graph?
The hitting problem \textsc{Chordal Deletion}
\cite{DBLP:journals/talg/AgrawalLMSZ19,DBLP:journals/siamdm/JansenP18,DBLP:journals/algorithmica/CaoM16,DBLP:journals/algorithmica/Marx10}
and the packing problem \textsc{Chordless Cycle Packing}
\cite{DBLP:conf/esa/Marx20} are both FPT.
\item In general, one could explore if induced versions of the
\UndelHitPack{H} problems are different compared to the case when we
are considering not necessarily induced subgraphs isomorphic to $H$.
  \item We defined our framework in terms of removing vertices and
  vertex-disjoint packings, but one could analogously study a problem
  defined by removing edges and edge-disjoint packings. This setting
  may pose very different challenges compared to the problems studied
  in this paper.
  \item A natural generalization of \CycleUndelHitPack is to consider
  the
  \UndelHitPack{\mathcal{X}} problem where $\mathcal{X}$ is the set
  of all
  minor models of a fixed graph $H$ (\CycleUndelHitPack is equivalent
  to
  the case when $H$ is $K_3$). Observe that if we denote by
  $\ell\cdot H$ the graph consisting of $\ell$ disjoint copies of
  $H$, then the \UndelHitPack{\mathcal{X}} problem defined for minor
  models of
  $H$ is equivalent to removing $k$ vertices such that the resulting
  graph does not contain $\ell\cdot H$ as a minor. Problems of this form
  where intensively studied
  \cite{DBLP:conf/soda/AdlerGK08,DBLP:conf/stoc/FellowsL89,DBLP:conf/stoc/FominLP0Z20}.
   Thus this gives a way of solving the problem for fixed $k$,
  $\ell$, and $H$, but understanding the optimal form of the running
  time could be an interesting question. The same problem can be
  studied also in the context of topological minors.
  \item One could ask how the Erd\H os-P\'osa Property relates to the
  complexity of the \UndelHitPack{\mathcal{X}} problem, but it is not obvious how to formulate this question in a way that leads to meaningful results. Note first
  that for
  problems involving copies of a fixed graph $H$, the Erd\H os-P\'osa
  Property trivially holds (in some sense, this was implicitly used
  by the simple algorithm of \cref{thm:general:k-plus-ell}). The
  algorithm of Theorem~\ref{thm:cycle-k-l} explicitly used the Erd\H
  os-P\'osa Property for cycles as a starting step. This might be a
  useful starting step in other cases where $H$-minor models satisfy
  this property (which is known to be the case exactly when $H$ is
  planar). However, note that the argument sketched in the previous
  item works irrespective of whether $H$-minor models satisfy
  the Erd\H os-P\'osa Property, although it may affect the running
  time.

  \item Tournaments (i.e., directed graphs with exactly one directed
  edge between any pair of vertices)
form a well-studied class of directed graphs where many hitting and
packing problems are more tractable compared to general directed
graphs
\cite{DBLP:journals/algorithmica/BessyBKSSTZ21,DBLP:journals/talg/LokshtanovMMPPS21,DBLP:journals/algorithmica/Bang-JensenMS16,DBLP:conf/stacs/KumarL16,DBLP:journals/algorithmica/XiaoG15,DBLP:journals/mst/MisraRRS13,DBLP:journals/jcss/BessyFGPPST11,DBLP:conf/aaai/FominLRS10,DBLP:journals/jda/DomGHNT10,DBLP:journals/tcs/RamanS06,DBLP:journals/jct/ChudnovskyFS12,DBLP:journals/jct/ChudnovskySS19}.
Which of these results generalize to the combined hitting and packing
problem?
\item Investigating the approximability of \UndelHitPack{\mathcal{X}}
problems is another completely unexplored area. The proper notion of
approximation for these kind of problems seems to be the following:
if there is a solution $S$ of size $k$ such that $G-S$ has no $\ell$
disjoint objects of type $\mathcal{X}$, can we find a set $S'$ of
size at most $c\cdot k$ such that $G-S'$ has no $c\cdot\ell$ disjoint
objects of type $\mathcal{X}$. That is, it this approximate sense, it is ok to find a somewhat larger set $S'$ that has the somewhat weaker property that it prevents only packings of $c\cdot\ell$ disjoint objects.
\end{itemize}

\section{Technical Overview}\label{sec:technical-overview}
In this section, we give a brief overview of our results,
highlighting the main technical ideas and putting them in context.
The remaining sections of the paper prove these results in the order
presented below. Note that, besides these individual technical
contributions, it can be considered an equally important conceptual
contribution that we demonstrate that the combination of hitting and
packing can lead to a wide range of interesting and challenging
problems.

\subsection{Algorithmic Results}

\subparagraph*{\boldmath $2^{\poly(k+\ell)}\cdot n^{\Oh(1)}$ Time
Algorithm for \CycleUndelHitPack.}
As noted earlier, we may assume in this problem that the graph $G$
has a
feedback vertex set $F$ of size $\Oh((k+\ell)\poly(k+\ell))$,
otherwise the
Erd\H os--P\' osa Theorem implies that the answer is no. Instead of
using the
fact that this gives a bound on the treewidth and trying to use a
general algorithm
parameterized by treewidth, we present an algorithm with running time
$2^{\poly(k+|F|)}\cdot n^{\Oh(1)}$ where $F$ is a feedback vertex set.

With a standard branching step, we can guess which vertices of $F$
are in the solution, remove these vertices from $G$, adjust $k$
appropriately, and then assume that the feedback vertex set $F$ is
undeletable.
To sketch the main ideas of the proof,
let us assume that $G-F$ is not only a forest,
but every component of $G-F$ is a path.
If $C$ is a cycle in $G$, then it contains at least one vertex of $F$,
and $C-F$ consists of one or more (sub-)paths in $G-F$.
If a graph contains a packing of $\ell$ cycles,
then we may assume that each cycle is an induced cycle
(this can be achieved by possibly shortening some cycles).
If $C$ is an induced cycle,
then every path $P$ of $C-F$ is of the following form:
$P$ goes from a neighbor of some $f_1\in F$
to a neighbor of some $f_2\in F$ (possibly $f_1=f_2$) such that
the internal vertices of $P$ are adjacent to neither $f_1$ nor $f_2$.
Let us call this a \textit{usable path}.

Suppose that we have a packing of $\ell$ (induced) cycles in $G$.
The solution has to contain a vertex of a cycle $C$ of this packing,
that is, a vertex of one of the paths $C-F$ (as $F$ is undeletable).
We branch on choosing a cycle $C$ of the packing
and choosing a path $P$ of $C-F$ that is broken by the solution,
but we \textit{do not} choose a vertex of $P$.
Instead, we put $P$ into a collection $\mathcal{P}$ of forbidden
paths that need to be broken by the solution.
Then we find a packing of $\ell$ cycles that does not use any of the
forbidden paths.
Such a collection can be found by branching on the number and type of
paths in the packing and then by a dynamic programming algorithm that
scans paths in $G-F$ in a left-to-right order and tries to find
disjoint paths of these types that are not on the forbidden list
$\mathcal{P}$.
Once we have such a collection, we once again branch on a path that
has to be broken by the solution and put it into the collection
$\mathcal{P}$.
We repeat this procedure as long as we are able to find an
appropriate packing.

If the algorithm is not able to find a packing of $\ell$ cycles that
does not use any forbidden path, then we need to check if there is a
set of $k$ vertices that can break every forbidden path in
$\mathcal{P}$. This can be done by a simple polynomial-time algorithm
(find a minimum number of points covering a set of intervals). If
there is such a set $S$, then it forms a solution; if there is no
such set, then this is an incorrect branch of the algorithm.

To bound the running time, we need to bound the depth of the search
tree, that is, the number of paths we put into the solution. The key
observation is that a vertex $v$ can cover at most $|F|^2$ different
usable paths. If $v$ covers a useful path $P$ from $u_1$ to $u_2$,
then $u_1$ should be the last vertex before $v$ that is the neighbor
of some $f_1\in F$ and $u_2$ is the first vertex after $v$ that is
the neighbor of some $f_2\in F$. Therefore, if
$|\mathcal{P}|>k|F|^2$, then surely there is no set $S$ of $k$
vertices intersecting all these paths. This observation gives a
$\poly(k+|F|)$ bound on the height of the search tree. As we branch
into $\poly(|F|)$ cases in each step, the claimed running time
follows.

With additional work, this algorithmic idea can be extended to the
case when $G-F$ is not a collection of paths, but a general forest.
The situation becomes significantly more complicated due to high
degree vertices in the forest, paths with many branch nodes, and
other issues, but the difficulties can be overcome by additional
layers of arguments (see Section~\ref{sec:cycle-fvs} for details).

\subparagraph*{\boldmath $3^{k+\ell}\cdot n^{\Oh(1)}$ Time Algorithm
for \EdgeUndelHitPack.}
Let us sketch a very simple branching algorithm. We measure our
progress by $k+\ell+1-\nu(G[U])$, where $\nu(G[U])$ is the size of
the maximum matching in the graph induced by the undeletable
vertices. Let us find a maximum matching $M$ in $G[U]$. We have
$\nu(G[U])<\ell \le \nu(G)$, hence there is an augmenting path
increasing the size of $M$. Let $u$ and $v$ be the two endpoints of
the augmenting path. We branch into three directions:
\begin{itemize}
\item $u$ is in the solution: remove $u$, decrease $k$ by one.
\item $v$ is in the solution: remove $v$, decrease $k$ by one.
\item neither $u$ nor $v$ is in the solution: put $u$ and $v$ into
$U$.
\end{itemize}
As the last branch strictly increases the size of the maximum
matching in $U$, we can conclude that the measure
$k+\ell+1-\nu(G[U])$ strictly decreases in each branch, giving a
bound of $k+\ell$ on the depth of the search tree.

  \subparagraph*{\boldmath $2^{2^{\Oh(\tw)}}\cdot n^{\Oh(1)}$ Time
  Algorithm for \UndelHitPack{H} When $H$ Is a Clique.}
  A typical way of designing algorithms for bounded-treewidth graphs
  is the following. Let us recall the definition of tree
  decompositions.
A {\em{tree decomposition}} of graph $G$ is a rooted tree $T$ with a
collection $\{X_t\subseteq V(G) \mid t\in V(T)\}$ of sets called
{\em{bags}}. The conditions for a
tree decomposition are:
\begin{itemize}
    \item For any vertex $u$ in $G$, the nodes in $T$ with bags
    containing $u$
        form a connected subtree of $T$.
    \item For any edge $uv$ in $G$, there exists a node in $T$ with a
    bag
        containing both $u$ and $v$.
\end{itemize}
The {\em{width}} of a tree decomposition $(T,B)$ is $\max_{t\in
V(T)}|X_t|-1$.
The {\em{treewidth}} $\tw$ of $G$ is the minimum possible width of a
tree
decomposition.
\emph{Pathwidth} is defined similarly, with $T$ restricted to be a
path.
We denote by $V_t$ the set of vertices
appearing in the bags of the nodes in the subtree rooted at some node
$t$.
Similarly, we define $G_t$ as the graph induced by the vertices
or rather given by the edges introduced in the subtree rooted at node $t$.

For \UndelHitPack{H},
the solution $S$ has a part $S\cap V_t$ that somehow influences
packings of cliques that intersect $V_t$.
A key observation is that a clique $K$ is either fully contained in
$V_t$, or intersects only the root bag $X_t$, i.e., $K\cap V_t=K\cap
X_t$.
Based on this observation, we can argue that the effect of $S\cap
V_t$ can be described by the following information:
  \begin{itemize}
  \item The intersection $S\cap X_t$.
  \item For every $D\subseteq X_t$, the maximum size of a packing in
  $G[V_t\setminus S]-D$.
  \end{itemize}
  Note that if the maximum packing size in $G[V_t\setminus S]$ is
  $m$, then the
  maximum packing size in $G[V_t\setminus S]-D$ is between $m-|D|$
  and $m$. Thus
  all the relevant information about $S\cap V_t$ can be described by
  the set
  $S\cap X_t$ ($2^{\tw+1}$ possibilities) and by a sequence of
  $2^{\tw+1}$
  integers between 0 and $\tw+1$ ($(\tw+2)^{2^{\tw+1}}$
  possibilities), leading
  to a bound of $2^{2^{\Oh(\tw)}}\cdot n$ different ways $S\cap X_t$
  can
  behave. With this bound at hand, we can follow the standard
  methodology of
  designing algorithms on tree decompositions: we define subproblems
  at each node
  $t$ corresponding to the different behaviors of $S$ and solve these
  subproblems in a bottom-up manner. The dominating factor of the
  running time
  is the number of subproblems at each node of the tree
  decomposition, leading
  to a $2^{2^{\Oh(\tw)}}\cdot n^{\Oh(1)}$ time algorithm.

\subparagraph*{\boldmath $2^{\poly(\tw)}\cdot n^{\Oh(1)}$ Time
Algorithm for \EdgeUndelHitPack.}
When $H$ is a single edge, we can improve the running time the
following way. As
noted above, the behavior of the set $S\cap V_t$ can be described by
the set
$S\cap X_t$ and by a function showing how the size of the maximum
matching
decreases if we remove a subset $D\subseteq X_t$, that is, by the
function
$g(D)=\nu(G[V_t\setminus S])-\nu(G[V_t\setminus S]-D)$. Because of
the highly
structured nature of the matching problem, this function cannot be
arbitrary and
can be compactly described, hence the number of possibilities is much
smaller
than $2^{2^{\Oh(\tw)}}$. Let us first consider the related function
$f(D)$,
which is 0 or 1 depending on whether $G[V_t\setminus S]-D$ has a
perfect
matching or not. This function describes a so-called
\emph{delta-matroid} and
has an algebraic representation as a skew-symmetric matrix. Following
a proof
sketch\footnote{The given proof sketch does not treat the question of
what field to choose and the obvious way of handling this issue does
not lead to the claimed bound (as confirmed by the author). Our proof
needs additional arguments to ensure the existence of a
representation over a suitable field.} of Wahlstr\"om
\cite{DBLP:journals/corr/abs-2306-03605}, this matrix can
be turned into a representation with $\Oh(\tw^3)$ bits. Formally, in~\cref{sec:matroid} we prove the following lemma.

\begin{restatable}{lemma}{Lemmatchfunction}\label{lem:matchfunctionbound}
Let $G$ be an $n$-vertex graph over a  vertex set $V\supseteq [k]$ for some integer $k$.
 Let
$f_{G,k} \from 2^{[k]}\to \mathbb{Z}^+$  be the function defined by
$f_{G,k}(S)=\nu(G-S)$. For each $k$ and $n$, there are $n\cdot 2^{\Oh(k^3)}$ functions $f_{G,k}$ that can arise this way.
\end{restatable}

Then a simple graph-theoretic
construction can be used to compactly describe the function $g$ using
the compact representation of the function $f$. Interestingly, our
proof of obtaining the algebraic representation is not algorithmic,
but it is sufficient for our purposes: the dynamic-programming
algorithm on the tree decomposition can be designed in a way that it
needs only a combinatorial bound on the number of different
subproblems that has a solution in the current graph, but does not
need to be able to compute which subproblems have no solution in any
graph. Therefore, the algorithm for \UndelHitPack{H} with $H$ being a
clique does not need any modification at all to achieve this running
time.

The crucial combinatorial insight that allows us to achieve $\Oh(\tw^3)$-bits
and prove~\cref{lem:matchfunctionbound} is the following statement.

\begin{linkedrestatable}{lemma}{Lemperfmatchbound}
    \label{lem:perfmatchbound}
    Let $G$ be a graph over a vertex set $V\supseteq [k]$ for some
    integer $k$.
  Let us define the function $h_{G,k}\from 2^{[k]}\to \{0,1\}$ the
  following way:
  \[
      h_{G,k}(S)=\begin{cases}
      1 & \text{if $G-S$ has a perfect matching,}\\
      0 & \text{otherwise.}
    \end{cases}
  \]
  For each integer $k$, the number of distinct functions $h_{G,k}$ is
  $2^{\Oh(k^3)}$.
\end{linkedrestatable}
Note, that naively the number of such functions is double-exponential
in $k$.
Now, we
present the proof of~\cref{lem:perfmatchbound} and use algebraic tools, in
particular,
representation of linear delta-matroids and multivariate polynomials
over finite fields.
For a field $\mathbb{F}$, we denote by $\mathbb{F}[x_1,\dots, x_n]$
the ring of $n$-variable polynomials with coefficients from
$\mathbb{F}$.
We have to be careful to make a distinction between the \emph{zero
polynomial}
of $\mathbb{F}[x_1,\dots, x_n]$,
which is the polynomial where every coefficient is zero,
and a \emph{vanishing polynomial} over $\mathbb{F}$,
which is a polynomial that is $0$ for every substitution of values
from $\mathbb{F}$ to the variables.
For example, $x^{|\mathbb{F}|}-x$ is a nonzero vanishing polynomial
over the field $\mathbb{F}$.
The following observation can be used to argue that some nonzero
polynomial is not vanishing.

\begin{observation}\label{obs:poly}
If $P\in \mathbb{F}[x_1,\dots,x_n]$ is an $n$-variable polynomial
over the field $\mathbb{F}$ and every variable $x_i$ has degree less
than $|\mathbb{F}|$ in $P$, then $P$ is a vanishing polynomial if and
only if it is the zero polynomial.
\end{observation}

The proof of~\cref{lem:perfmatchbound} follows from the fact that
polynomials $P\in\mathbb{F}[x_1,\dots,x_n]$
where every variable has degree less than $\mathbb{F}$
are in one-to-one correspondence with functions
$f\from \mathbb{F}^n\to \mathbb{F}$.
Indeed, both sets have size exactly $|\mathbb{F}|^{|\mathbb{F}|^n}$
and Lagrange interpolation shows that for any function
$f\from \mathbb{F}^n\to \mathbb{F}$,
there is a corresponding polynomial where the degree of every
variable is less than $|\mathbb{F}|$.
\begin{proof}[Proof of~\cref{lem:perfmatchbound}]
  For notational convenience, let us assume that $V=[n]$.
 Let $\mathbb{F}$ be a field of size $2^{k+2}$. Let $A$ be the Tutte
 matrix
 corresponding to $G$, that is, an element $a_{i,j}$ is defined as
  \[
    a_{i,j}\coloneqq\begin{cases}
      x_{i,j} &\text{if $i$ and $j$ are adjacent and $i<j$,}\\
      -x_{i,j} &\text{if $i$ and $j$ are adjacent and $i>j$,}\\
      0 &\text{if $i$ and $j$ are not adjacent.}
    \end{cases}
  \]
 We consider each entry of $A$ as an $n^2$-variate polynomial from
 $\mathbb{F}[x_{1,1},\dots,x_{n,n}]$. For $X\subseteq [n]$, we denote
 by $A[X]$
 the principal submatrix of $A$ corresponding to the rows and columns
 described
 by $X$. It is well known that $G[X]$ has a perfect matching if and
 only if
 $\det(A[X])$ is a nonzero polynomial.
 Let us denote by $S^c$ the complement of $S$,
 i.e., the set of rows/columns \emph{not} indexed by $S$.
 Thus for every $S\subseteq [k]$, we have
 \[
     \text{$G-S$ has a perfect matching } \Leftrightarrow \text{
     $\det(A[S^c])$
     is nonzero} \tag{$\star$}
 \]

 We would like to obtain a matrix $A'$ with the same property
 ($\star$), but over $\mathbb{F}$ (so the elements of $A'$ are
 \textit{not} polynomials).
 Let $S_1$, $\dots$, $S_t$ be the subsets of $[k]$ such that $G-S$
 has a
 perfect matching. For $\ell\in [t]$, the polynomial
 $P_\ell=\det(A[S^c_\ell])$ is a nonzero polynomial where every
 variable
 $x_{i,j}$ has degree at most 2 (as every variable appears at most
 twice in the
 Tutte matrix). The product $P=\prod_{\ell=1}^t P_\ell$ is also a
 nonzero
 polynomial (as $\mathbb{F}$ is a field,
 $\mathbb{F}[x_{1,1},\dots,x_{n,n}]$ is an
 integral domain) where every variable has degree at most $2t\le
 2\cdot
 2^k<|\mathbb{F}|$. Therefore, \cref{obs:poly} implies that $P$ is
 not vanishing. This means that we can substitute values to the
 variables such
 that $P$ evaluates to a nonzero value, which also means that every
 $P_\ell$ is
 nonzero under this substitution. Let $A'$ be obtained from $A$ by
 this
 substitution. It is easy to see that this matrix $A'$ over
 $\mathbb{F}$ has the
 desired property ($\star$): if $S$ has no perfect matchings, then
 $\det(A'[S^c])=0$ (as this already holds for $A$),
 while if $S$ has a perfect matching,
 then the determinant is one of the polynomials $P_\ell$,
 and hence it evaluates to a nonzero value under the substitution.

 Now we argue that we can obtain a $k\times k$ matrix $A''$ that also
 has the
 property ($\star$) Matrix $A'$ represents the set system
 $\mathcal{X}=\{X \mid
 \det(A'[X])\neq 0\}$ over $[n]$, which is known to be a
 delta-matroid. The set system
 $\mathcal{X}'=\{X\in [k] \mid X\cup ([n]\setminus [k])\in
 \mathcal{X}\}$ is a
 set system over $[k]$, which is called the \textit{contraction} of
 $\mathcal{X}$. It is known (see, e.g., Wahlstr\"om
 \cite{DBLP:journals/corr/abs-2306-03605}) that given some $[n]\times
 [n]$
 matrix $A'$ over $\mathbb{F}$ representing $\mathcal{X}$, we can use
 algebraic
 operations to compute a $k\times k$ matrix $A''$ representing
 $\mathcal{X}'$.
 Now we can verify that this matrix $A''$ indeed satisfies property
 ($\star$)

 As $A''$ has property ($\star$) we can deduce from $A''$ the value
 of the function
$f$ for any $S\subseteq [k]$. As $A''$ is an $k \times k$ matrix over
a field
of size $2^{k+2}$, it can be described by $k^2(k+2)$ bits. Therefore,
$\Oh(k^3)$ bits are sufficient to describe the function $h_{G,k}$,
that
is, there are at most $2^{\Oh(k^3)}$ such functions.
\end{proof}

\subparagraph*{\boldmath $2^{2^{\Oh(\tw\log \tw)}}\cdot n^{\Oh(1)}$
Time Algorithm for \UndelHitPack{H} for Arbitrary Connected $H$.}
The dynamic programming approach becomes significantly more
complicated if we
generalize it to $\UndelHitPack{H}$ where $H$ is not a clique. The
main issue is that now
it is no longer true that in every packing every copy of $H$
intersecting $V_t$
is either fully contained in $V_t$ or intersects $V_t$ only in $X_t$.
As $H$ is
not a clique, it can be the case that a copy is split by $X_t$ (see
Figure~\ref{fig:partialpack}).
Therefore, we need to argue about \emph{partial packings} in $V_t$
that may
contain some partial copies of $H$ split by $X_t$. When reasoning
about such a
partial packing, we need to describe not only
which vertices of $X_t$ the partial packing covers,
but also how the partial copies partition $X_t$.
We formalize this intuitive idea by the notion of \emph{types}.
The type of a (partial) packing contains the following information:
\begin{itemize}
  \item
  A set of those vertices in $X_t$ that are not covered by the packing
  including the vertices which are deleted.
  \item
  A partition of the vertices in $X_t$
  describing which vertices contribute to the same copy of $H$.
  \item
  For each part of the partition,
  a mapping between the vertices of $X_t$
  and the vertices of the partial copy of $H$,
  so we can determine which vertices of $H$ have been packed
  and which vertices of $H$ still need to be packed.
\end{itemize}
Clearly there are at most $2^{\tw+1}$ choices
for the set of uncovered vertices.
Since we consider a fixed graph $H$,
the precise mapping for each of the at most $\tw+1$ parts
can be described by $\abs{H}^{\tw+1}$ possible functions
although involving a significant notational overhead in the formal
description.
However, the partition of the vertices into the different parts
dominates the number of possible types
for which there are $\tw^{\Oh(\tw)}$ possibilities for each node.
Due to this larger number of types,
the running time of the algorithm involves an additional logarithmic
factor
which we avoid for the case when $H$ is a clique.

Similar to the algorithm when packing cliques,
we also have to remember the size of the largest packing for each
type.
As this number could potentially range from $0$ to $\ell$,
a naive bound for the number of possible states is
$\ell^{\tw^{\Oh(\tw)}}$
which does not depend on the treewidth only.
However,
since each graph $H$ has only a fixed size,
the packing number for two different types cannot differ by too much.
Indeed, the size of two optimal partial packings with different types
can differ by at most $\Oh(\abs{H} \cdot \tw)$
as each partial packing of $H$ can ``block'' at most $\abs{H}$
vertices
from being packed in the other packing.
This observation drastically reduces the number of states for each
node
to $\tw^{\tw^{\Oh(\tw)}}$
which then determines the running time of the algorithm.

\begin{figure}
\begin{center}\def\svgwidth{0.5\linewidth}
\begingroup%
  \makeatletter%
  \providecommand\color[2][]{%
    \errmessage{(Inkscape) Color is used for the text in Inkscape, but the package 'color.sty' is not loaded}%
    \renewcommand\color[2][]{}%
  }%
  \providecommand\transparent[1]{%
    \errmessage{(Inkscape) Transparency is used (non-zero) for the text in Inkscape, but the package 'transparent.sty' is not loaded}%
    \renewcommand\transparent[1]{}%
  }%
  \providecommand\rotatebox[2]{#2}%
  \newcommand*\fsize{\dimexpr\f@size pt\relax}%
  \newcommand*\lineheight[1]{\fontsize{\fsize}{#1\fsize}\selectfont}%
  \ifx\svgwidth\undefined%
    \setlength{\unitlength}{250.3559035bp}%
    \ifx\svgscale\undefined%
      \relax%
    \else%
      \setlength{\unitlength}{\unitlength * \real{\svgscale}}%
    \fi%
  \else%
    \setlength{\unitlength}{\svgwidth}%
  \fi%
  \global\let\svgwidth\undefined%
  \global\let\svgscale\undefined%
  \makeatother%
  \begin{picture}(1,0.59164543)%
    \lineheight{1}%
    \setlength\tabcolsep{0pt}%
    \put(0,0){\includegraphics[width=\unitlength,page=1]{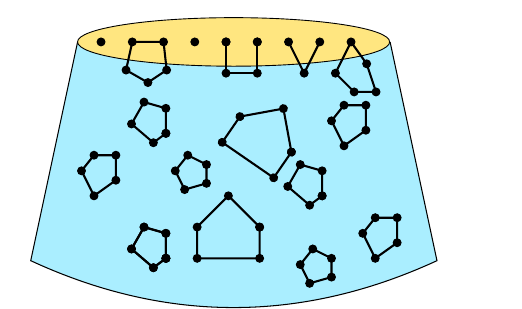}}%
    \put(0.0889015,0.4961356){\color[rgb]{0,0,0}\makebox(0,0)[rt]{\lineheight{1.25}\smash{\begin{tabular}[t]{r}$X_t$\end{tabular}}}}%
    \put(0,0){\includegraphics[width=\unitlength,page=2]{partialpack.pdf}}%
    \put(0.91037955,0.27982645){\color[rgb]{0,0,0}\makebox(0,0)[lt]{\lineheight{1.25}\smash{\begin{tabular}[t]{l}$V_t$\end{tabular}}}}%
    \put(0,0){\includegraphics[width=\unitlength,page=3]{partialpack.pdf}}%
  \end{picture}%
\endgroup%
\end{center}
  \caption{If $H$ is not a clique (e.g., $H$ is a cycle on 5
  vertices), then restricting a packing to $V_t$ may result in
  partial copies of $H$ that contain vertices from $V_t\setminus
  X_t$. Therefore, the description of a partial packing needs to
  include how these partial copies interact with $X_t$.}
  \label{fig:partialpack}
\end{figure}

\subparagraph*{\boldmath $2^{2^{\Oh(\tw\log \tw)}}\cdot n^{\Oh(1)}$
Time Algorithm for \CycleUndelHitPack.}
For the case of \CycleUndelHitPack the situation is different from
the cases before.
Since we are not dealing with a fixed graph $H$
but an infinitely sized family of graphs,
we cannot decide a priori which graphs are packed where.
However, as we can exploit that all possible graphs must be cycles,
we can still characterize the interaction with a bag $X_t$ in a
compact way.
The main observation is that it does not matter in which cycle
(i.e., of which length) the vertices appear in a partial packing
as we allow cycles of all possible lengths.
Instead, it suffices to classify the vertices of a bag $X_t$
into three groups which are ``uncovered'', ``covered with one
incident edge'',
and ``covered with two incident edges''.
Despite this partition of the vertices is possible,
it is not sufficient to fully describe
the behavior of a partial packing with a bag.
We also have to be able to recognize when adding a new edge
to the packing would close a cycle.
For this we need to remember
which vertices (from the bag) appear in the same cycle.
However, we do not have to store this information for all vertices in
the bag
but only for those vertices that are incident to exactly one edge
as only these vertices are eligible to obtain another edge.
This information between the endpoints of the partial cycles, i.e.,
paths,
can be efficiently described by a perfect matching between these
vertices.
Hence, the type of a partial cycle-packing
is a four-tuple with the following information:
\begin{itemize}
  \item
  A set of vertices that are not covered by the packing
  which includes the deleted vertices.
  \item
  A set of vertices that have one incident edge in the packing.
  \item
  A set with the remaining vertices
  which have two incident edges in the packing,
  that is, these vertices are already fully covered.
  \item
  A perfect matching for the vertices with one incident edge
  abstracting the paths between the endpoints.
\end{itemize}
When treating the perfect matching as a special form of a partition,
the connection to the algorithm for the general case become apparent;
for general $H$, the type also contains a partition of the vertices
describing which vertices appear in the same copy of $H$.
From this perspective it comes as no surprise that
for \CycleUndelHitPack the number of types is at most
$\tw^{\Oh(\tw)}$ as well.

Similar to the algorithm for \UndelHitPack{H} when $H$ is one fixed
graph,
the interaction of a solution with bag
further includes the information of how many cycles can be packed
for each possible type of partial packing.
Interestingly, when considering two different types,
there is no obvious bound for the difference in the number of
packable cycles.
However, if for some type the number of cycles that can be packed is
too low
then there is no benefit in proceeding with this type.
More formally, consider a type $T$ where we can pack at least $\tw$
cycles less
than for the optimal packing where no vertex of the bag is covered.
Even if all partial cycles for a packing of type $T$ could be closed,
the total number of cycles is still lower than
for the packing where none of the vertices from the bag are covered.
Hence, proceeding with this type $T$ never results in a maximal
packing
and therefore, all packings that have this specific type $T$ at this
node
can be discarded immediately.

With this observation it then follows that even for \CycleUndelHitPack
the total number of subproblems at each node
can be bounded by $2^{2^{\Oh(\tw \log \tw)}} \cdot n$.

\subsection{Lower Bounds}

\subparagraph*{\boldmath\SigTwoP-completeness of \TriUndelHitPack.}
The containment in $\SigTwoP$
for \TriUndelHitPack is clear, as we can guess the set of deleted
vertices and use
an \NP-oracle to determine the size of the maximum packing.
To establish the \SigTwoP-hardness,
we reduce from a special \SigTwoP-complete satisfiability problem,
the \textsc{Smallest Unsatisfiable Subformula} problem
(\SUS)~\cite{sigact-column-37,Umans99}.
For \SUS the input is a CNF-formula $\phi$ together with a parameter
$k$
and the task is to decide
if there \emph{exists} a collection of at most $k$ clauses of $\phi$
(we refer to this as subformula)
such that, \emph{for all} assignments, the subformula is not
satisfied,
i.e., the subformula is unsatisfiable.
This ``exists, for all'' formulation of \SUS already gives
a hint for the reduction to \TriUndelHitPack.

We construct a graph consisting of clause gadgets,
variable gadgets, and literal edges.
Each of the clause gadgets consist of a single triangle
with a distinguished deletable vertex.
The interpretation of this vertex is as follows:
if the vertex is deleted, then the clause is activated/selected,
and otherwise, the clause remains inactive
(which we consider as the default state).
Each variable gadget consists of a cycle of triangles
such that there are exactly two maximum triangle packings,
one corresponding to setting the variable to $\true$
and the other packing corresponding to setting the variable to
$\false$.
As a last component the graph contains literal edges
which connect the variable gadgets to the clause gadgets
depending on the occurrence of the literals in the clauses.
Intuitively the idea is that once the distinguished vertex
of a clause gadget is deleted,
which corresponds to a selected clause,
it should not be possible to find a large packing
as otherwise we could transform this packing (for the variable
gadgets)
into an assignment which satisfies all selected clauses.

To strengthen the result, we prove the hardness for tripartite graphs
by reducing from \ThreeCNFSUS (the restriction of \SUS to
3CNF-formulas),
for which we also provide the \SigTwoP-completeness
as, to our knowledge, no proof appeared in the literature,
although the completeness was claimed
(see \cref{sec:sig2p:hardness:sat} for further discussion).

\subparagraph*{\boldmath\SigTwoP-completeness of \UndelHitPack{H}.}
For connected graphs $H$ other than the triangle, the
\SigTwoP-completeness result can be obtained from \TriUndelHitPack by
a clean reduction. Kirkpatrick and
Hell~\cite{DBLP:journals/siamcomp/KirkpatrickH83} showed how to
reduce a triangle packing problem to arbitrary $H$-packing problems
for connected $H$ with at least 3 vertices. However, they considered
the problem of finding a packing that covers every vertex of the
graph and the arguments do not readily work for problems where not
every vertex needs to be covered. Nevertheless, we show that with
additional arguments and by extending their construction, a reduction
can be obtained from \TriUndelHitPack to \UndelHitPack{H}, showing
the \SigTwoP-completeness of the latter problem.

\subparagraph*{\boldmath\SigTwoP-completeness of \CycleUndelHitPack.}
This hardness proof is obtained by observing that in the
\SigTwoP-completeness proof of \TriUndelHitPack, every cycle relevant
for a packing is a triangle.

\begin{figure}[t]
  \centering
  \includegraphics[width=.7\textwidth]{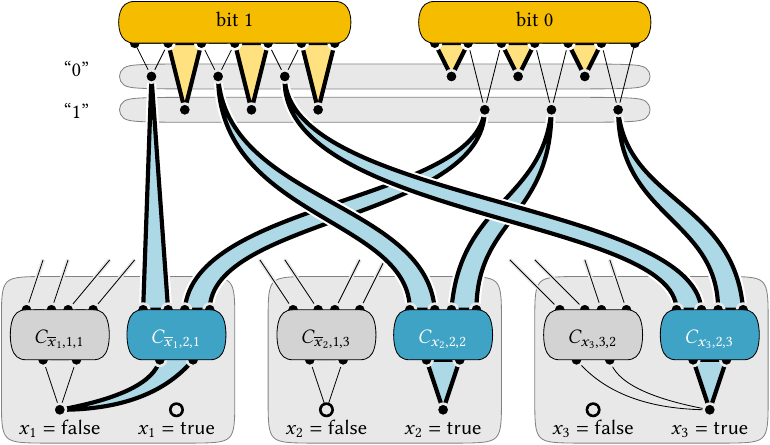}
  \caption{%
  An illustration of the construction
  for the double exponential lower bound for \UndelHitPack{H}.
  \\
  The two topmost gadgets determine the clause number we consider.
  Each gadget $C_{\lambda,j,p}$ at the bottom
  encodes that literal $\lambda$ appears in the $j$th clause at
  position $p$.
  The hollow vertices indicate deleted vertices
  representing an assignment.
  \\
  The packing represented by the colored triangles
  represents the evaluation of the second clause that is
  $C_2 = (\bar x_1 \lor x_2 \lor x_3)$,
  which is not satisfied by the chosen assignment of
  $x_1=\mathsf{true}$ and $x_2=x_3=\mathsf{false}$.
  }
\end{figure}
\subparagraph*{\boldmath Double-exponential Lower Bounds for
\UndelHitPack{H} Parameterized by Treewidth.}
We reduce from an instance of $3$-SAT with $n$ variables
and $m$ clauses to
achieve a double-exponential lower bound
by constructing a graph with pathwidth $\Oh(\log m)$.
Although the starting point is again a 3CNF-formula
as for the \SigTwoP-hardness,
the interpretation and the basic ideas differ
because of the change in the quantification.
For \SUS the task is to check
if there \emph{exists} a set of \emph{clauses}
such that, \emph{for all assignments} to the variables,
the formula is not satisfiable.
For $3$-SAT the task is to check
if there \emph{exists an assignment} for the variables
such that \emph{for all clauses} at least one literal of the clause
is satisfied,
i.e., for each clause not all three literals are false.
While as before, we use gadgets consisting of a long cycle
with attached triangles,
the overall construction of how these gadgets interact
is substantially different.
The critical difference to the $\SigTwoP$-hardness proof
lies in the way we verify the satisfying assignment.
Here, we do not explicitly construct a different gadget for every
clause
as this would result in a construction
with treewidth linear in the number of clauses.

Our constructed instance comprises three parts, which we refer to as
left, middle, and right. The deleted vertices on the
right correspond to a choice of the satisfying assignment.  The left
part
consists of the so-called \emph{selector gadgets}. Each selector
gadget models one bit of the binary encoding of some clause. So, if
there are $m$ clauses, we have roughly $\log m$ selector gadgets. For
each selector gadget we introduce three pairs of vertices in the
middle. Intuitively, there are three pairs since each clause contains
three literals, and there are pairs to encode whether the bit
corresponding to this selector gadget is $0$ or $1$. Then a packing
of the vertices
on the left, i.e., a packing for the selector gadgets, can interact
in $m$ different ways with the $\Oh(\log(m))$ vertices in the middle.
Each of these possibilities corresponds to a different selection of
(the encoding of) a clause.

We verify the satisfying assignment by ensuring that no matter how we
might choose a
maximum packing on the left (i.e., no matter which clause we look
at), the maximum packing on the right is small. This
corresponds to verifying that each clause is satisfiable. A crucial
difference
to the $\SigTwoP$-hardness proof lies in the fact that, here, the
small packing is ensured by the variable gadgets (if a solution
exists)
and not by the clause gadget.

The treewidth of the construction is $\Oh(\log(m))$ since this is the
number of
vertices with which the gadgets interact (and each gadget separately
has constant
treewidth).

\subparagraph*{\boldmath Double-exponential Lower Bound for
\SquareUndelHitPack and \CycleUndelHitPack Parameterized by
Treewidth.}
We first prove the lower bound for \SquareUndelHitPack
and then extend it to \CycleUndelHitPack as follows.
We define the reduction for \SquareUndelHitPack in a way
such that whenever a cycle packing $\packs$
contains a \emph{large} cycle of length at least five,
then we can repack $\packs$ to obtain a packing $\packs'$
by replacing this large cycle with a $C_4$.
This keeps the number of cycles the same (or might actually increase
it)
while reducing the number of large cycles.
Hence, the maximum possible cycle packing only contains cycles of
length four.

The basic high-level idea for the lower bound for \SquareUndelHitPack
is similar to the previous one for the general case.
However, instead of having a middle part with $\Oh(\log m)$ vertices
(which dominates the pathwidth),
we have to find a different way of encoding the clause numbers
by only using $\Oh(\log m / \log\log m)$ vertices.

To make this possible, we have to change the way the gadgets
on the left and right side interact with the vertices in the middle.
For the previous construction
the gadgets only cover a single vertex
while the remaining part of the triangle
is entirely contained either in the left or the right half.
When packing four-cycles we can change this
and allow that those cycles which contain vertices in the middle
have exactly one vertex on the left side
and exactly one vertex on the right side.
With this method,
the position of the cycle can be described by a matching
for the vertices in the middle.
The idea is to introduce two groups of~$t$ vertices each where~$t!
\approx m$.
Then there are~$t!$~possible perfect matchings
between the vertices from the first group
and the vertices from the second group.
With this it is possible to associate with each clause
a unique perfect matching on these vertices.
The gadgets on the left and right side are then adjusted
such that they connect each pair of vertices from the matching
by a path of length two.

Since choosing $t \approx \log m/\log \log m$ satisfies the above
property,
there are $\Oh(\log m/\log\log m)$ vertices in the middle part
which then determines the pathwidth of the graph.

\section{\texorpdfstring{\boldmath \CycleUndelHitPack\ Parameterized by $k+\ell$}
{Cycle-HitPack Parameterized by k+l}}
\label{sec:cycle-fvs}

In this section we prove the following \cref{thm:cycle-k-l}.

\cyclekl*\label\thisthm

Recall that an instance $\Ii = (G,U,k,\ell)$ of \CycleUndelHitPack
consists of a graph $G$,
a set $U \subseteq V(G)$ of undeletable vertices,
a budget $k$ for the number of vertices that can be deleted,
and an integer $\ell$ as a strict upper bound on the size of the cycle packing in the graph after removing the solution,
that is, the size of the cycle packing in the resulting graph can be only strictly smaller than $\ell$.

We will make use of a known close relationship between the parameter $k+\ell$ and the feedback vertex set size of $G$.

\paragraph*{\boldmath Step 0: From $k+\ell$ to Feedback Vertex Sets.}
\addcontentsline{toc}{subsection}{Step 0: From $k+\ell$ to Feedback Vertex Sets}

The main work of this section is to show the following result. 
\begin{restatable}{theorem}{fvsAlgo}
	\label{thm:fvs:algo}
	An instance $(G,U,k,\ell)$ of \CycleUndelHitPack given with a feedback vertex set $F$ of the $n$-vertex graph $G$ can be solved in time
	$2^{\poly(|F|+k)} \cdot n^{\Oh(1)}$.
\end{restatable}

With \cref{thm:fvs:algo} in hand, it is then straight-forward to prove \cref{thm:cycle-k-l}.

\begin{proof}[Proof of \cref{thm:cycle-k-l}]
From the Erd\H os--P\'osa
Theorem~\cite{MR0175810}, either $G$ has $k + \ell$ vertex-disjoint
cycles, or it has a feedback vertex set of size $\Oh((k+\ell) \log
(k+\ell))$.  In particular, from~\cite[Theorem~2.3.2]{diestel2005graph},
if the size of a minimum feedback vertex set in $G$ is larger than
$
	k^*= 4(k+\ell)(\log (k+\ell)+ \log \log (k+\ell) +4) + k+\ell -1,
$
then $G$
has at least $k +\ell$ vertex disjoint cycles.
Using the $2^{\Oh(k^*)} \cdot n^{\Oh(1)}$ algorithm for finding a feedback
vertex set of size at most
$k^*$~\cite{DBLP:journals/ipl/KociumakaP14}, we can compute such a feedback vertex set, or otherwise conclude that $G$ has at least $(k+\ell)$ vertex-disjoint cycles.

If the algorithm computes a feedback vertex set $F$ of size at most $k^*$ then, by \cref{thm:fvs:algo}, we can obtain a solution for the instance $(G,U,k,\ell)$ in total time
\[
  2^{\poly(k^*+k)} \cdot n^{\Oh(1)}= 2^{\poly(k+\ell)} \cdot n^{\Oh(1)}.
\]
Otherwise, there are $(k+\ell)$ disjoint cycles in $G$. Then
$(G,U,k,\ell)$ is a \no-instance (as deleting any set of at most $k$ vertices results in a
graph which contains at least $\ell$ vertex-disjoint cycles).
\end{proof}

So it remains to prove \cref{thm:fvs:algo}. To shorten notation, from now on, whenever we consider an instance $(G,U,k,\ell)$ of \CycleUndelHitPack that is given with a feedback vertex set $F$ of $G$, we will use $(G,U,F,k,\ell)$ to denote this instance.
Suppose we are given an instance $\Ii = (G,U,F,k,\ell)$ of \CycleUndelHitPack.
The idea behind the algorithm is as follows.
\begin{enumerate}
	\item
	We guess the intersection of the solution with the feedback vertex set $F$.
	The intersection is then removed and the remaining vertices
	from the feedback vertex set are made undeletable.

	\item
	We reduce the degree of the instance such that
	it is bounded in terms of $k$ and $\abs{F}$.
	For this we first mark parts of the graph that should not be removed
	and then iteratively delete parts that are not marked.
	This marking procedure ensures that for all possible solutions
	there is a way to tweak the cycle packing to avoid the deleted vertices
	by only using vertices from the parts that have been marked.

	\item
	We introduce the notion of usable paths and cycles
	and show that it suffices to rule out large cycle packings
	that consists of usable cycles only.

	\item
	We use a branching procedure to generate instances of a simpler problem
	where we only have to hit paths from a given candidate set by the solution.
	We design the candidate sets in a way such that
	the equivalence to the original instance is preserved.

	\item
	For each such candidate set of paths we check whether there is a set
	of deletable vertices such that every path that should be hit is indeed hit.
	That is, we solve the easier problem
	and therefore, solve the original problem.
\end{enumerate}

Based on this high-level overview of the algorithm
we now formalize the individual steps of the algorithm.

\paragraph*{Step 1: Making the Feedback Vertex Set Undeletable.}
\addcontentsline{toc}{subsection}{Step 1: Making the Feedback Vertex Set Undeletable}

In the first step, the algorithm guesses the intersection of an optimal solution with the set $F$.
For each such guess $F' \subseteq F$, the algorithm creates a new instance
$(G-F',F\setminus F',U \cup (F\setminus F'),k-|F'|,\ell)$.
The following result immediately follows.
\begin{lemma}
	\label{lem:fvs:fvsIsUndeletable}
	The instance $(G,U,F,k,\ell)$ is a \yes-instance of \CycleUndelHitPack,
	if and only if
	there is some $F' \subseteq F$ such that
	$(G-F', U \cup (F\setminus F'), F \setminus F', k-|F'|, \ell)$
	is a \yes-instance.
\end{lemma}
The rest of the algorithm solves each of these $2^{|F|}$
instances and reports \yes\ if at least one of these instances reports \yes.
Therefore, for the remaining section we assume
that the instance $(G,U,F,k,\ell)$ satisfies $F \subseteq U$,
that is, the vertices in $F$ are undeletable.

\paragraph*{Some Useful Notation.}
Before we proceed, we first introduce some notations.
For each $i \in F$, we denote by $N_i$ the set of neighbors of $i$ in
$G-F$.
Note that every cycle of $G$ intersects $F$.
We
denote the forest $G-F$ by $G'$.

Let $\Cc$ be a cycle packing in $G$.
For a cycle $C\in \Cc$ and a sub-path $P_1$ of $C$
such that $V(P_1)\subsetneq V(C)$,
by $\Cc\setminus P_1$, we denote the set of subgraphs obtained from $\Cc$
by replacing the cycle $C$ with the path $C-P_1$.
By $V(\Cc \setminus P_1)$ we denote the vertex set of $\Cc\setminus P_1$.
Let $u_0, v_0$ be the endpoints of $C-P_1$
and let $u_1,v_1$ be the endpoints of $P_1$.
Let $P_2$ be a path with endpoints $u_2,v_2$
such that
$V(P_2) \cap V(\Cc\setminus P_1) =\emptyset$ and %
$u_0u_2$ and $v_0v_2$ are edges of $G$.
By $(\Cc\setminus P_1) \cup P_2$ we
denote the cycle packing obtained from $\Cc\setminus P_1$
by replacing the path $C-P_1$ with the cycle formed by $C-P_1$
together with $P_2$ and the edges $u_0u_2$ and $v_0v_2$.

For any cycle packing $\Cc$ of $G$, the \emph{$G'$-paths of} $\Cc$ are
the paths obtained by deleting $F$ from the cycles of $\Cc$.
Since $F$ is a feedback vertex set of $G$, we have $|\Cc|\le |F|$.
As deleting one vertex of $F$ increases the number of paths (obtained from $\Cc$) by at most one, we get the following observation, which we use frequently.
\begin{observation} \label{obs:number-of-G'-paths}
	Let $\Cc$ be a cycle packing of $G$. Then the number of $G'$-paths of
	$\Cc$ is at most $|F|$.
\end{observation}

For a graph $G$ and two vertex sets $X,Y \subseteq V(G)$,
an $(X,Y)$-path in $G$ is a path from some $x \in X$ to some $y \in
Y$;
that is, one endpoint belongs to $X$ and the other endpoint belongs to $Y$.
To simplify notation we also say that a $(\{u\},\{v\})$-path in $G$
is a $(u,v)$-path in $G$,
that is, a path with the two endpoints $u$ and $v$.
In the following, when referring to an $(N_i,N_j)$-path $P$
we also implicitly require that $V(P) \cap F = \emptyset$.

We say that a path $P$ is contained in a tree $T$ if $V(P)\subseteq V(T)$.
Similarly, we say that a cycle $C$ contains a path $P$ if $V(P)\subseteq V(C)$.

\paragraph*{\boldmath Step 2: Bounding the Degree of $G'$.}
\addcontentsline{toc}{subsection}{Step 2: Bounding the Degree of $G'$}

As a next step we prove a bound on the degree of the vertices.
To this end, we design a procedure to modify the given instance
until the resulting instance has low degree,
that is, the largest degree is at most $\degree$.
Formally, we set
\begin{equation}
	\dd \deff |F|^3 + (k+3) |F|^2 + (k+2)|F|
\end{equation}
and then remove vertices from the instance until the maximum degree of $G'$
it at most $\dd$.
We say that an instance \emph{is of low degree with respect to $\dd$}
if, for all $v \in V(G')$, it holds that $\deg_{G'}(v) \le \dd$.

\begin{lemma}
	\label{lem:fvs:reduceDegree}
	Let $\Ii = (G, U, F, k, \ell)$ be an instance of \CycleUndelHitPack
	with $F \subseteq U$.
	In time $\Oh(\poly(\abs{G})$
	we can construct an equivalent instance $\Ii_1 = (G_1, U_1, F, k, \ell)$
	of low degree with respect to $\dd$,
	i.e., for all $v \in V(G_1) \setminus F$,
	the new graph $G_1$ satisfies $\deg_{G_1-F}(v) \le \dd$,
	and $E(G_1) \subseteq E(G)$ and $F \subseteq U_1 \subseteq U$.
\end{lemma}
\begin{proof}
	Consider an arbitrary vertex $v \in V(G')$
	such that $\deg_{G'}(v) \geq \dd+1$.
	Let $T$ be the tree of $G'$ containing $v$ and root $T$ at $v$.
	For all neighbors $u \in N_{G'}(v)$ of $v$,
	we denote by $T_u$ the subtree of $T$ rooted at the vertex $u$.
	We perform the following \textsc{Marking Procedure} to mark
	certain subtrees rooted at the children of $v$.

	\begin{enumerate}
		\item
		Initially, all trees in $\{T_u \mid u \in N(v)\}$ are unmarked.

		\item For every pair
		$(i,j) \in F \times F$, repeat the following. If there are
		at least $k+|F|+2$ unmarked trees in $\{T_u \mid u \in N_{G'}(v)\}$
		each of which contains an $(N_i,N_j)$-path,
		then mark an arbitrary set of $k+|F|+2$ of them;
		otherwise, mark all of them.

		Once this marking step is performed for every pair $(i,j) \in F \times F$
		(or if all trees rooted at the children of $v$ are already marked),
		we continue marking the remaining unmarked trees in the following step.

		\item
		For every $i \in F$,
		if there are at least $k+|F|+2$ unmarked trees
		in $\{T_u \mid u \in N_{G'}(v) \}$
		that contain some
		vertex of $N_i$, then mark an arbitrary set of $k+|F|+2$ of them; otherwise,
		mark all of them.
	\end{enumerate}

	Observe that the number of marked trees is at most
	\[
		|F|^2 (k+|F|+2) + |F|	(k+|F|+2) = |F|^3 + (k+3) |F|^2 + (k+2)|F|
		=\dd
		.
	\]

	\begin{claim}%
		\label{clm:fvs:reduceDegree}
		Let $T_x$ be an unmarked tree in the set $\{T_u \mid u \in N_{G'}(v)\}$.
		Pick an arbitrary leaf $w$ of $T_x$.
		Then the original instance $\Ii$ is equivalent to
		the reduced instance $\Ii' = (G-\{w\}, U\setminus\{w\}, F, k,\ell)$.

	\end{claim}

	\begin{claimproof}
		If $S \subseteq V(G) \setminus U$ is a solution
		to the input instance $(G, U, F, k,\ell)$,
		then $S\setminus \{w\}$ is a solution to
		the instance $(G-\{w\}, U\setminus\{w\}, F, k,\ell)$.
		This proves the forward direction of the equivalence.

		We now prove the backward direction. Let
		$S\subseteq V(G)\setminus \{w\}$ be a solution of the reduced instance $\Ii'$.
		For the sake of a contradiction,
		assume there is a cycle packing $\Cc$ in $G-S$ of size at least
		$\ell$.
		This implies that there must be a cycle $C \in \Cc$ which contains $w$.
		Since $F$ is a feedback vertex set, it intersects all cycles
		and hence, there
		exist $i,j \in F$ (where $i$ could be the same as $j$) such that $C$
		contains an $(N_i,N_j)$-path.
		Moreover, there is an $(N_i, N_j)$-path $P$ in $C$
		that is also contained in $T$
		and goes through $w$ and therefore, intersects $T_x$
		(as $T$ is rooted at $v$ and $x$ is a neighbor of $v$).

		In the following we prove that (due to the marking procedure),
		we can replace the path $P$ by another path $Q$ that does not use $w$.
		For this we argue that there is at least one marked tree $T_u$
		for some $u \in N_{G'}(v)$
		that does not intersect with the given cycle cover of the set $S$.
		We distinguish two cases depending on whether
		the path $P$ is entirely contained in $T_x$ or also goes through $v$.

		To simplify notation we set $\Tt \deff \{T_u \mid u \in N_{G'}(v)\}$
		in the following.

		\begin{description}
			\item[Case 1.]
			The path $P$ is contained in $T_x$, that is, $V(P) \subseteq V(T_x)$.
			By the \textsc{Marking Procedure} and the fact that $T_x$ is unmarked,
			there are $k + \abs{F} + 2$ trees in $\Tt$,
			denote this set by $\Tt_0$,
			that have been marked for the pair $(i,j) \in F \times F$.

			Note that at most $k$ trees in $\Tt_0$ are hit by $S$.
			Let $\Tt_1$ denote the set of the trees in $\Tt_0$
			that are not hit by $S$.
			From the size bound on $S$,
			we directly get that $\abs{\Tt_1} \ge \abs{\Tt_0} - k \ge \abs{F} + 2$.

			For the next step, let $\Pp$ be the set of $G'$-paths of $\Cc$.
			Let $\Tt_2$ be the set of trees from $\Tt_1$
			that do not intersect a path from $\Pp$.
			Suppose that a path $Q$ of $\Pp$ goes through $v$.
			In this case $Q$ can intersect up to two trees in $\Tt_1$.
			Since the paths of $\Pp$ are vertex-disjoint,
			each path in $\Pp \setminus \{Q\}$
			can intersect at most one tree in $\Tt_1$.
			By \cref{obs:number-of-G'-paths}, we know that
			the number of $G'$-paths of $\Cc$ is at most $\abs{F}$.
			Hence, the size of $\Tt_2$ is at least $\abs{\Tt_1} - (\abs{F}+1) \ge 1$.
			So, let $T_y$ be an arbitrary tree in $\Tt_2$.

			Since $T_y$ was marked for the pair $(i,j)$
			by the above \textsc{Marking Procedure},
			it contains an $(N_i,N_j)$-path $P'$
			which is disjoint from $\Cc \setminus P$ and the solution $S$.
			Therefore, $(\Cc \setminus P) \cup P'$ forms a cycle packing
			of size at least $\ell$ in $(G-\{w\})-S$.
			But this immediately contradicts our assumption that $S$
			is a solution for the instance $\Ii'$.

			\item[Case 2.]
			The path $P$ passes through both $v$ and $w$.
			Without loss of generality, we can assume that $P$ ends in $T_x$
			at some vertex of $N_i$.

			By the \textsc{Marking Procedure} there are at least
			$k+\abs{F}+2$ marked trees in $\Tt$ that contain some vertex of $N_i$.
			Denote this set of trees by $\Tt_0$.

			At most $k$ trees in $\Tt_0$ are intersected by the solution $S$.
			Moreover, at most $|F|$ trees in $\Tt_0$ are intersected by $\Cc$
			(recall that $P$ also intersects $T_x$ which is not marked).
			Let $\Tt_1$ be the set of the remaining trees in $\Tt_0$.
			We directly get that $\abs{\Tt_1} \ge \abs{\Tt_0} - k - \abs{F} \ge 2$.
			From the size bound on $\Tt_1$,
			we get that there is at least one tree in $\Tt_1$, say $T_y$,
			that is disjoint from $\Cc$ and $S$
			and contains some vertex of $N_i$.

			Denote the unique $(N_i,y)$-path from a neighbor of $i$ to $y$ by $P_1$.
			Let $P_0$ be the subpath of $P$ from $w$ to $x$ in the tree $T_x$.
			Since $P_1$ is disjoint from $S$ and $\Cc$,
			we can use $P_1$ to replace the path $P_0$ in the cycle packing.
			Formally, $(\Cc \setminus P_0) \cup P_1$
			is a cycle packing of size at least $\ell$
			in the reduced graph after deleting the solution $S$.
			But this immediately contradicts our assumption that $S$ is a solution.
			\qedhere
		\end{description}

	\end{claimproof}
	After applying the above \textsc{Marking Procedure}
	and \cref{clm:fvs:reduceDegree} for all vertices of high degree exhaustively,
	the instance satisfies $\deg_{G'}(v)\leq \dd$ for every $v\in V(G')$.

	It remains to prove the running time of the above procedure.
	As for each vertex the procedure runs in time polynomial
	in the size of the graph, the stated runtime directly follows.
\end{proof}

\paragraph*{Step 3: Usable Paths and Cycle Packings.}
\addcontentsline{toc}{subsection}{Step 3: Usable Paths and Cycle Packings}

As a next step we define the notion of \emph{usable paths}
and \emph{usable cycle packings}.
We then show that it is enough to restrict our attention
to usable cycle packings (and not care about other cycle packings).
This notion is crucial to bound the depth of the branching tree
that we describe a little later.

For any $i,j \in F$ (where $i$ could be the same as $j$),
we denote by $G^{i,j}$
the sub-forest of $G'$ obtained as follows: repeatedly delete leaves that do not belong to $N_i \cup N_j$ from $G'$
until all leaves are contained in $N_i \cup N_j$.

For any two vertices $u,v \in N_i \cup N_j$,
we define $\threedist^{i,j}(u,v)$ as the number of internal vertices
on the unique $(u,v)$-path in $G^{i,j}$
whose degrees are
at least $3$ in $G^{i,j}$.
If there is no $(u,v)$-path in $G^{i,j}$,
we define $\threedist^{i,j}(u,v)=+\infty$.

\begin{definition}[Usable Paths and Cycle Packings]
	For any two vertices $i,j\in F$ ($i = j$ is allowed),
	an $(N_i,N_j)$-path $P$ from $u$ to $v$ is called \emph{usable}
	if the following three conditions hold:
	\begin{enumerate}[label=(\roman*)]
		\item
		path $P$ is contained in $G^{i,j}$,
		\item
		path $P$ is a minimal $(N_i,N_j)$-path,
		meaning no subpath of $P$ of length less than $|P|$
		is an $(N_i,N_j)$-path, an $(N_i,N_i)$-path, or an $(N_j,N_j)$-path,
		and
		\item
		$\threedist^{i,j}(u,v) \le |F|+1$.

	\end{enumerate}

	We say that a cycle packing $\Cc$ is \emph{usable}
	if all its $G'$-paths are usable.
\end{definition}

With this definition we prove as a next step
that it suffices to focus on usable cycle packings only.
Hence, all other cycle packings can be ignored.

\begin{lemma}\label{lem:usable-cp}
	If a graph $G$ with a feedback vertex set $F$
	has a cycle packing of size $\ell$,
	then $G$ has a \emph{usable cycle packing} of size $\ell$.
\end{lemma}
\begin{proof}
	Let $\Cc$ be a cycle packing in $G$ of size $\ell$ such that the
	number of $G'$-paths of $\Cc$ that are not usable is minimized.
	Let $P$ be a $G'$-path of $\Cc$ which is not a usable path
	and let $u$ and $v$ be the endpoints of $P$.
	Let $i,j \in F$ be the neighbors of $u,v$
	such that the cycle of $\Cc$ that contains $P$
	also contains the edges $ui$ and $vj$.
	This especially implies that $P$ is an $(N_i,N_j)$-path in $G'$.

	If $P$ is not a minimal $(N_i,N_j)$-path, then let
	$P'$ be a subpath of $P$ that is a minimal $(N_i,N_j)$-path (minimal $(N_i,N_i)$-path or $(N_j,N_j)$-path).
	Hence, $\Cc'=(\Cc\setminus P) \cup P'$ is also a cycle packing of
	size $\ell$ in $G$.
	Therefore, we can assume that all $G'$-paths of $\Cc$ are minimal.

	Next we assume that $P$ is a minimal $(N_i,N_j)$-path
	and that $\threedist^{i,j}(u,v)\ge |F|+2$.
	Let $x_1 , \dots, x_{|F|+1}$ be 
	the internal vertices of $P$ traversed from $u$ in order
	such that each of them is a vertex of degree at least $3$ in $G^{i,j}$.
	For all $q \in \numb{\abs{F}+1}$,
	we define the unique tree $T(q)$
	as the tree in $G^{i,j}-E(P)$ containing $x_q$.
	From \cref{obs:number-of-G'-paths}, we know that $\Cc$
	contains at most $|F|$ distinct $G'$-paths.
	Hence, there exists some $q_0 \in \numb{\abs{F}+1}$
	such that $T(q_0) \setminus \{x_{q_0}\}$ does not intersect $\Cc$.
	Moreover, since $x_{q_0}$ is a vertex of degree at least $3$ in $G^{i,j}$,
	tree $T(q_0)$ contains at least one vertex from $N_i \cup N_j$.

	Now we distinguish two cases
	depending on how many vertices from $N_i \cup N_j$ appear in $T(q_0)$.

	\begin{description}
		\item[\boldmath One vertex from $N_i\cup N_j$ appears in $T(q_0)$.]
		Formally, we have that $|V(T(q_0)) \cap (N_i \cup N_j)| =1$.
		Let $w$ be this vertex, that is, $w$ is the leaf of $T(q)$.

		In this case $T(q_0)$ is a path in $G^{i,j}$
		and let $P'$ be the unique $(u,w)$-path in $G^{i,j}$ (and thus in $G'$).
		By assumption about $u$ and $v$,
		we get that $\threedist^{i,j}(u,w) \leq |F|+1$
		and $P'$ is a usable path.

		In conclusion we get that $\Cc'=(\Cc\setminus P) \cup P'$
		is also a cycle packing of size $\ell$ in $G$.

		\item[\boldmath Multiple vertices from $N_i\cup N_j$ appear in $T(q_0)$.]
		In this case we formally get that
		$|V(T(q_0)) \cap (N_i \cup N_j)| \geq 2$.
		Let $u'$ and $v'$ be two distinct vertices in $T(q_0)$
		such that they are also contained in $N_i \cup N_j$.
		Moreover, we require that the unique path $P'$ from $u'$ to $v'$
		does not contain any other vertex from $N_i \cup N_j$.
		Since $T(q_0)$ is either a path or a tree with at least two leaves
		and, by definition of $G^{i,j}$, all leaves belong to $N_i \cup N_j$,
		we can choose $u'$ and $v'$ such that $\threedist^{i,j}(u',v') \le 1$.

		We get that $\Cc'=(\Cc\setminus P) \cup P'$ is also a cycle packing
		of size $\ell$ in $G$
		or $\Cc'$ contains a cycle packing of size $\ell$
		(after removing the path from $i$ to $j$)
		in case $u',v' \in N_i$ or $u',v' \in N_j$
		(note that a path starting from $i$ and ending with $i$ is also a cycle).
	\end{description}

	In both of the above cases,
	we get a cycle packing of size $\ell$
	with strictly larger number of usable $G'$-paths.
	This contradicts the choice of $\Cc$
	and thus, completes the proof for the lemma.
\end{proof}

\subparagraph*{Rich Pairs and a Bound for the Number of Usable Paths.}
For all $i,j \in F$, we bound the number of usable $(N_i,N_j)$-paths
that are hit by a fixed vertex $v \in G'$ by some value $\Gamma$
or prove that this pair is rich.
We say that a pair $(i,j) \in F \times F$ is a \emph{rich pair}
if, for any cycle packing $\Cc$ of $G$,
there are at least $k+1$ vertex-disjoint $(N_i \cup N_j,N_i \cup N_j)$-paths
in $G'$ that are disjoint from $\Cc$.

We show that, if a pair $(i,j)$ is not rich,
then $\Gamma$ is the maximum number of usable $(N_i,N_j)$-paths
that can be hit by any fixed vertex, 
where we define $\Gamma$ as
\begin{equation}
	\Gamma \deff \left(\dd (2|F|+k)(|F|+3)\right)^2
	.
\end{equation}

This bound is the key property to obtain a small depth of the branching tree
of the branching procedure that we describe later. See \cref{fig:usable-paths} for an illustration of \cref{lem:intersect}.

\begin{figure}[tp]
	\centering

\begin{tikzpicture}[%
  scale=.85,
  node distance = 1cm and 1cm,
  on grid,
  ]
  \definecolor{selColor}{RGB}{164,210,225}
\definecolor{freeColor}{RGB}{243,180,48}

\tikzset{%
vertex/.style={
  draw=white,
  fill=black,
  label={below:{\large\ensuremath{#1}}},
  circle,
  line width = 1.5pt,
  inner sep = 2.25pt,
  outer sep = 0pt,
},
vertex/.default=\null,
vertexDist/.style={
  vertex = #1,
  diamond,
},
vertexDist/.default=\null,
edge/.style={
  line width=0.5pt,
  draw=white,
  double distance = 1.5pt,
  double = black,
  rounded corners,
},
selVtx/.style={
  draw = #1,
},
selVtx/.default=selColor,
freeVtx/.style={
  selVtx = freeColor,
},
selEdg/.style={
  line width = 2pt,
  double = black,
  double distance = 2pt,
  draw = #1,
  rounded corners,
},
selEdg/.default=selColor,
freeEdg/.style={
  selEdg = freeColor,
},
}

  \node[vertexDist,label={above:{\large$v$}}] (v0) at (9.5,9.5) {};

  \node[vertexDist] (v1) at (7.5,7.5) {};
    \node[vertexDist] (v11) at (4.5,4.5) {};
      \node[vertexDist,selVtx] (v111) at (2,2) {};
        \node[vertex=i,selVtx] (l1) at (1,1) {};
        \node[vertex=j,selVtx] (l2) at (3,1) {};
      \node[vertexDist] (v112) at (6,3) {};
        \node[vertexDist,selVtx] (v1121) at (5,2) {};
          \node[vertex=i,selVtx] (l3) at (4,1) {};
          \node[vertex=i,selVtx] (l4) at (6,1) {};
        \node[vertex=i,vertex] (l5) at (8,1) {};
    \node[vertexDist] (v12) at (9.5,5.5) {};
      \node[vertex=i] (l11) at (7.5,3.5) {};
      \node[vertexDist,freeVtx] (v122) at (10.5,4.5) {};
        \node[vertex=j,freeVtx] (l12) at (9.5,3.5) {};
        \node[vertex=i,freeVtx] (l13) at (11.5,3.5) {};
  \node[vertexDist] (v2) at (14.5,4.5) {};
    \node[vertexDist] (v21) at (13,3) {};
      \node[vertexDist,selVtx] (v211) at (12,2) {};
        \node[vertex=j,selVtx] (l6) at (11,1) {};
        \node[vertex=j,selVtx] (l7) at (13,1) {};
      \node[vertex=i,vertex] (l8) at (15,1) {};
    \node[vertexDist,selVtx] (v22) at (17,2) {};
      \node[vertex=j,selVtx] (l9) at (16,1) {};
      \node[vertex=i,selVtx] (l10) at (18,1) {};

  \draw[edge] (v0) --
    (v1) -- (v11) -- (v111)
            (v11) -- (v112) -- (v1121)
                     (v112) -- (l5)
    (v1) -- (v12) -- (l11)
            (v12) -- (v122)
    (v0) -- (v2) -- (v21) -- (v211)
                    (v21) -- (l8)
            (v2) -- (v22)
    ;

    \draw[edge,selEdg] (l1) -- (v111) -- (l2);
    \draw[edge,selEdg] (l3) -- (v1121) -- (l4);
    \draw[edge,selEdg] (l6) -- (v211) -- (l7);
    \draw[edge,selEdg] (l9) -- (v22) -- (l10);
    \draw[edge,freeEdg] (l12) -- (v122) -- (l13);

    \draw[line width = 1.5pt,double=freeColor,white,double distance=2pt,rounded corners=5pt]
      (l11)++(-0.5,-1) -- ++(5,0) -- ++(0,3.5) -- ++(-5,0) -- cycle;

\end{tikzpicture}
	\caption{An illustration from the proof of \cref{lem:intersect}.
		The four highlighted paths at the bottom
		belong to the assumed maximal collection $\Pp$.
		For convenience, we use $i$ and $j$ to denote the neighbors of vertex $i,j\in F$ respectively. 
		The nodes with diamond shape belong to $A=\bigcup_{P \in \Pp} A_P$,
		and after removing these nodes,
		there is one component which has more than one leaf (appearing in the box).
		Then, we can add the highlighted path in this box to $\Pp$
		and thus, strictly increase the size of $\Pp$.
	}
	\label{fig:usable-paths}
\end{figure}
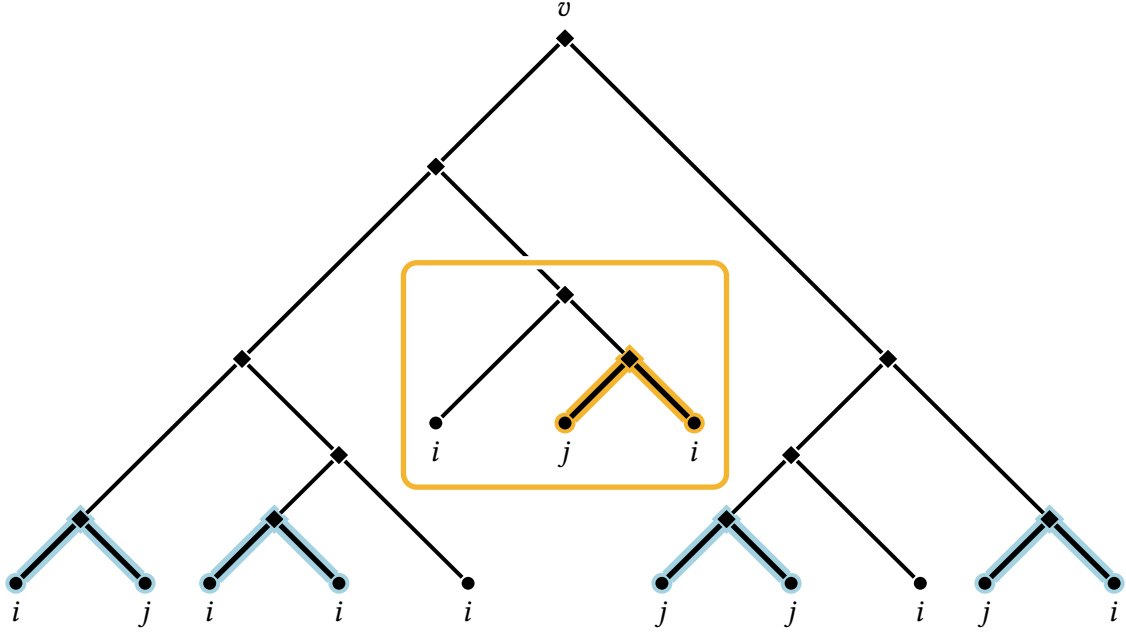

\begin{lemma}\label{lem:intersect}
	Let $G$ be a graph
	of low degree with regard to $\dd$.
	For all $v \in V(G')$ and $i,j \in F$, either $(i,j)$ is a rich pair,
	or the number of usable $(N_i,N_j)$-paths that are hit by $v$,
	i.e., intersected by vertex $v$, is at most $\Gamma$.
\end{lemma}
\begin{proof}
	Fix some vertex $v \in V(G')$, some $i,j\in F$
	and let $T_v$ be the tree of $G^{i,j}$ that contains $v$.
	We define $T$ as a maximal subtree of $T_v$ containing $v$
	such that every $(N_i,N_j)$-path contained in $T$ is usable
	and every leaf of $T$ is in $N_i \cup N_j$.
	We root $T$ at $v$.
	Observe that the vertices of $N_i \cup N_j$ appear only at the leaves of $T$,
	as otherwise some $(N_i,N_j)$-paths in $T$ would not be minimal (and thus not usable).

	From the definition of usable paths,
	it follows that every usable $(N_i,N_j)$-path in $G'$ that is hit by $v$
	has its endpoints in the set $(N_i \cup N_j) \cap V(T)$,
	which is precisely the set of leaves of $T$.
	Hence, the number of usable $(N_i,N_j)$-paths hit by $v$
	is at most $|(N_i \cup N_j) \cap V(T)|^2$.

	Based on this, our new goal is to prove that either $(i,j)$ is a rich pair,
	or the number of leaves of $T$ is at most
	$\sqrt \Gamma = \dd(2\abs{F}+k)(\abs{F}+3)$.

	We say that two vertices of $T$ are \emph{incomparable} if neither vertex is a descendant of the other one in $T$. 
	For a path $P$ in $T$, we define $\lca(P)$ as the lowest common ancestor
	of the two endpoints of $P$ in $T$.
	Consider a maximal collection $\Pp$
	of pairwise disjoint $((N_i \cup N_j), (N_i \cup N_j))$-paths in $T$
	such that the following two conditions hold:
	\begin{itemize}
		\item
		For two distinct paths $P_1, P_2 \in \Pp$,
		$\lca(P_1)$ is incomparable with $\lca(P_2)$.

		\item
		For every path $P\in \Pp$,
		there is no other $((N_i \cup N_j), (N_i \cup N_j))$-path $P'$ in $T$
		such that $\lca(P')$
		is a descendant of $\lca(P)$.
	\end{itemize}

	Observe that any path in {$G'$} can
	intersect at most two paths in $\Pp$.
	In particular,
	since any cycle	packing $\Cc$ of $G$ has at most $|F|$ different $G'$-paths,
	at most $2|F|$ paths in $\Pp$ can intersect $\Cc$.
	Therefore, if $|\Pp| \ge 2|F|+k+1$,
	then $(i,j)$ is a rich pair.
	For the remainder of the proof we assume that $|\Pp| \leq 2|F|+k$.
	Recall that we aim to bound the number of leaves of $T$.

	For each $P \in \Pp$,
	we additionally define the set $A_P$
	as the set of vertices on the unique path from $\lca(P)$ to $v$
	that have degree at least three in $G^{i,j}$.
	We directly get that $|A_P| \le \threedist^{i,j}(v,\lca(P))+2 \le |F|+3$,
	since all $(N_i,N_j)$-paths in $T$ are usable.

	We set $A = \bigcup_{P \in \Pp} A_P$
	and get $\abs{A} \le (2|F|+k)(|F|+3)$	as $|\Pp| \leq 2|F|+k$.
	As a next step we define the graph $T' = T \setminus A$.
	Since $T$ is connected and $\deg_{G'}(u)\leq \dd$ for every $u\in V(G')$,
	the graph $T'$ has at most $\dd (2|F|+k)(|F|+3)$ components.
	In particular, every connected component of $T'$
	contains at most one vertex of $N_i \cup N_j$.
	Suppose for contradiction that one component of $T'$
	has at least two vertices of $N_i \cup N_j$.
	Then in this component we can find an $((N_i \cup N_j), (N_i \cup N_j))$-path $\hat{P}$ satisfying that (i) $\lca(\hat{P})$ does not belong to $A$;
	(ii) there is no other $((N_i \cup N_j), (N_i \cup N_j))$-path $P'$ in $T$
	such that $\lca(P')$ is descendant of $\lca(\hat{P})$.
	Condition (i) implies that for any $P\in \Pp$, $\lca(P)$ is incomparable with $\lca(\hat{P})$.
	Then we could strictly increase the size of $\Pp$,
	contradicting the maximality of $\Pp$.

	Therefore, the number of vertices of $(N_i \cup N_j) \cap V(T)$,
	i.e., the number of leaves of $T$, is at most $\dd (2|F|+k)(|F|+3)$.
\end{proof}

\paragraph*{Step 4: The Branching Procedure.}
\addcontentsline{toc}{subsection}{Step 4: The Branching Procedure}

For a given instance
$(G,U,F,k,\ell)$, we now design a branching procedure, which we call \branch. This procedure
generates $2^{\poly(|F|+k)}$ candidate instances of a different (relatively easier) problem.
We design these instances such that
(i) it is enough to look for a solution for these candidate instances, and
(ii) each of these candidate instances can be solved in time $2^k n^{\Oh(1)}$.

We first present an informal description of the branching procedure.
First,
observe that if the input instance is a \yes-instance
and $\Cc$ is a cycle-packing of size $\ell$,
then every solution, i.e., set of vertices that are deleted,
hits at least one of the $G'$-paths of $\Cc$.
The branching procedure stores two complementary sets $\avail$ and $\hit$
of $G'$-paths
such that it first finds a usable cycle packing $\Cc$ of size $\ell$
in $G$ which only uses paths in $\avail$ and then branches on
\emph{which} particular $G'$-path of $\Cc$ is hit by the solution,
that is, which path is added to $\hit$ and removed from $\avail$.
Note that unlike typical branching algorithms,
this procedure does not branch on
the choice of a fixed vertex into the solution.
Instead, in one branch, it
considers a set of vertices (namely a $G'$-path) that is guaranteed to be hit
by the solution.

We first describe the procedure to find a usable cycle packing
which uses paths from a specified set only.
For this we make use of the subroutine in \cref{lem:find-usable-cp}.

\begin{restatable}{lemma}{findUsableCp}
	\label{lem:find-usable-cp}
	Let $H$ be a forest $H$ on $n$ vertices
	and $\mathcal{P}_1, \dots, \mathcal{P}_f$ be $f$ collections of paths of $H$.
	In time $2^{\Oh(f)} \cdot n^{\Oh(1)}$,
	we can find $f$ pairwise vertex-disjoint paths $P_1, \dots, P_f$
	such that, for each $i \in [f]$,
	it holds that $P_i \in \mathcal{P}_i$
	whenever such paths exist or correctly answer \no otherwise.
\end{restatable}
Note that the notation used in \cref{lem:find-usable-cp}
is independent of the notation of the algorithm that we have built so far.
As the proof of this lemma
uses standard bottom-up dynamic programming over rooted trees,
we defer the proof to \cref{sec:appendix}.

With this subroutine,
we can design the first ingredient of the branching procedure.

\begin{lemma}\label{lem:step-one}
	Let $(G, U, F, k, \ell)$ be an instance of \CycleUndelHitPack
	and let $\avail$ be some set of $G'$-paths.
	In time $2^{\Oh(|F| \log |F|)} \cdot n^{\Oh(1)}$
	we can find a usable cycle packing $\Cc$ of size $\ell$ in $G$
	whose $G'$-paths are contained in $\avail$.
\end{lemma}
\begin{proof}
	\newcommand{\fedge}{\to_{\mathsf{edge}}}
	\newcommand{\fpath}{\to_{\mathsf{path}}}
	\newcommand{\bedge}{\curvearrowleft_{\mathsf{edge}}}
	\newcommand{\bpath}{\curvearrowleft_{\mathsf{path}}}
	Our goal is to use \cref{lem:find-usable-cp} to find disjoint paths
	in the graph $G'$ that can be extended
	to a usable cycle packing of size $\ell$ in $G$.
	As input of the algorithm in \cref{lem:find-usable-cp}
	we use $f$ collections of $(N_{i_1},N_{j_1})$-paths,
	\dots, $(N_{i_f},N_{j_f})$-paths
	where $(i_1,j_1), \dots ,(i_f,j_f)$ are distinct pairs of $F\times F$.
	However, these pairs cannot be chosen arbitrarily
	as we must be able to extend the paths to a cycle packing.
	To find these pairs, we start with some preprocessing.

	Observe that a cycle packing $\Cc$
	naturally induces a collection of such pairs from $F\times F$
	based on the vertices from $F$ that appear in each cycle.
	But it can also happen that vertices in $F$ appear in only one pair
	because they are connected to some other vertex in $F$ by an edge.
	Based on these observations we show how to compute all possible sets
	of pairs that have to be considered.

	We first consider all possible permutations $\pi$ of $F$.
	Note that this does not provide any information about the cycles
	and how the vertices in $F$ are connected.
	To resolve this, we introduce four special symbols
	$\fedge$, $\fpath$, $\bedge$, and $\bpath$
	and insert one of them after each position in the permutation $\pi$.
	The permutation is then a string $\pi'$ of length $2\abs{F}$.
	The first two symbols, i.e., $\fedge$ and $\fpath$,
	indicate how two consecutive vertices in the permutation are connected,
	either by an edge or by a path in $G'$.
	The other two symbols, i.e., $\bedge$ and $\bpath$,
	indicate that the cycle is closed by an edge or a path in $G'$
	to the first vertex of the cycle,
	the previous vertex in $\pi'$ following a symbol $\bedge$ or $\bpath$.

	For example the string $1\fpath 3 \fedge 4\bpath 5 \fedge 2 \bpath$
	represents two cycles;
	the first cycle starts at $1$ and is then connected to $3$ by a path in $G'$,
	vertex $3$ has an edge to $4$
	and from $4$ there is a path in $G'$ back to $1$;
	the second cycle starts at $5$ and has an edge to $2$
	and from there a path in $G'$ back to $5$.

	By first choosing the permutation of $F$ and then inserting the symbols,
	we get that there are at most
	$\abs{F}! \cdot 4^{\abs{F}} \in 2^{\Oh(\abs{F} \log \abs{F})}$
	possible strings to consider.
	For each generated string,
	we check for each occurrence of $\fedge$ and $\bedge$
	whether the corresponding edge exists
	(for ease of notation we assume in this step
	that the vertices in $F$ have self-loops).
	For each occurrence of $\fpath$ and $\bpath$,
	we add the corresponding pair $(i,j)$ to a list $\Ll_{\pi'}$.
	We additionally discard the string if it represents less than $\ell$ cycles.

	It remains to define the sets $\Pp_{i,j}$ for all $(i,j) \in \Ll_{\pi'}$
	that serve as input for the algorithm from \cref{lem:find-usable-cp}.
	For each $(i,j) \in \Ll_{\pi'}$,
	we pick the $(N_i,N_j)$-paths in the set $\avail$
	which is given as input.

	Now for each collection of sets of candidates
	we invoke the algorithm from \cref{lem:find-usable-cp}.
	If the algorithm outputs a set of disjoint paths,
	we complete it to a usable cycle packing
	according to the above preprocessing.
	If no set of disjoint paths was found,
	we proceed to the next possible set of pairs,
	i.e., string $\pi'$, and repeat the procedure.

	Eventually we find a usable cycle packing of size $\ell$ for $G$
	whose $G'$-paths are contained in $\avail$
	or can conclude that no such packing exists.
	\qedhere

	\end{proof}

Now we have everything ready to state the branching procedure \branch.
Recall that we are working with the instance $(G,U,F,k,\ell)$, where
set $U$ contains the undeletable vertices,
set $F$ is a feedback vertex set of $G$,
and $F \subseteq U$.
\emph{Throughout the following branching procedure,
the instance $(G,U,F,k,\ell)$ is fixed}.
We can imagine that each node of the branching tree for \branch
has a pair $(\avail, \hit)$ associated with it,
where $\avail$
and $\hit$ are sub-collections of usable $(N_i,N_j)$-paths for some $i,j \in F$.
The set $\hit$ corresponds to the collection of $G'$-paths
that have been guessed so far to be hit by the solution
and the set $\avail$ contains all $G'$-paths that can still be used
for a cycle packing.
Therefore we also get $\avail \cap \hit = \emptyset$.

The idea of the branching procedure is that at each node of the branching tree,
one finds a usable cycle packing of size $\ell$
with a set of $r$ different $G'$-paths that are contained in $\avail$.
Then we branch into $r$ directions (nodes)
where each of the $r$ nodes includes one of the $r$ $G'$-paths
into the set $\hit$ in the associated pair $(\avail, \hit)$.

Now, we present the branching procedure formally.

\begin{lemma}[\branch]
	\label{lem:fvs:branching}
	Let $(G, U, F, k, \ell)$ be an instance of \CycleUndelHitPack of low degree.
	In time $2^{\Oh((|F|^2 k\Gamma) \log |F|)} \cdot n^{\Oh(1)}$
	we can construct a list of candidates  $\{\Ii_1, \ldots, {\Ii_p}\}$ (for some $p \in 2^{\Oh((|F|^2 k\Gamma) \log |F|)}$), where for each {$q \in [p]$}, $\Ii_q=(\avail_q,\hit_q)$, 
	 such that	the instance $(G,U,F,k,\ell)$ is a \yes-instance if and only if
		there exists an index $q_0 \in [p]$ and a set $X \subseteq V(G)\setminus U$
		of size at most $k$ such that $X$ hits all paths in $\hit_{q_0}$.

\end{lemma}
\begin{proof}
	We describe the branching procedure
	by constructing the labeled branching tree.
	For the root node of the branching tree
	initialize $\avail$ as the set of all usable $(N_i,N_j)$-paths
	in $G'$ for every $i,j \in F$, and set $\hit=\emptyset$.

	At a node labelled $(\avail, \hit)$,
	the \branch\ procedure performs the following
	two steps:
	\begin{description}

		\item[Step 1.]
		Use \cref{lem:step-one} to find a usable cycle packing $\Cc$
		of size $\ell$ in $G$ whose $G'$-paths are contained in $\avail$.

		\item[Step 2.]
		Let $\Pp$
		be the $G'$-paths of $\Cc$.
		Branch into $r \deff \abs{\Pp}$ directions,
		by creating new child nodes
		associated with the label $(\avail \setminus \{P\}, \hit \cup \{P\})$
		for all $P \in \Pp$.

		\item[First Stopping Criterion.]
		The {branching procedure} stops when $G$ has no usable cycle packing of size $\ell$ whose $G'$-paths are totally contained in $\avail$.

		\item[Second Stopping Criterion.]
The branching procedure also stops
when the depth of the current branch in the branching tree
is strictly more than $|F|^2 \cdot k\Gamma$.
	\end{description}

At the end of the branching procedure, let the
set of \emph{candidates} be $\{\Ii_1, \ldots, {\Ii_p}\}$ ,
where for each {$q \in [p]$}, $\Ii_q=(\avail_q,\hit_q)$. Let
$\Ll_{\good}$ be the list of \emph{promising
candidates}, in which there are no usable cycle packing of size $\ell$ whose $G'$-paths are totally contained in $\avail$,
and let
$\Ll_{\bad}$, called the \emph{discarded} candidates, be the
remaining candidates. 
The branching procedure returns $\Ll_{\good}$ as the output.

	As a first step we analyze the running time of the branching procedure.

	\begin{claim}[Running time of \branch]\label{lem:branching-time}
		The running time of the branching procedure \branch\ is
		$2^{\Oh((|F|^2 k\Gamma) \log |F|)} \cdot n^{\Oh(1)}$.
	\end{claim}
	\begin{claimproof}
		First, observe that for any pair $(\avail,\hit)$ that is associated with some node of the branching tree of \branch,
		the size of the sets
		$\avail$ and $\hit$ is at most $n^2$, where $n$ is the number of vertices of
		the input graph $G$. This is because the sets $\avail$ and $\hit$ contain
		paths of $G'$, which is a forest. Since there is a unique path between any
		two vertices in a tree, the total number of paths in $G'$ and hence in
		$\avail$ and $\hit$ is at most quadratic in the input size.

		By \cref{lem:step-one},
		Step~1 of \branch\ can be executed in time
		$2^{\Oh(|F| \log |F|)} \cdot n^{\Oh(1)}$.
		Further from the second stopping condition of \branch,
		the depth of the branching tree is $|F|^2 k \Gamma$.
		Moreover, the width is at most $\abs{F}$
		as the algorithm branches into $\abs{\Pp}$ directions in Step~2
		and, by \cref{obs:number-of-G'-paths}, it holds that
		$\abs{\Pp} \le \abs{F}$.

		We conclude that the number of nodes of the branching tree of \branch
		is $2^{\Oh(|F|^2 k \Gamma \log |F|)}$.
		The running time follows directly.
	\end{claimproof}

	As a next step we show that the output of the branching procedure
	can be used to recover the solution for the original instance.

	\begin{claim}%

		If there exists $(\avail,\hit) \in \Ll_{\good}$
		and a set $X \subseteq V(G)\setminus U$ of size at most $k$
		such that $X$ hits all paths in $\hit$,
		then the instance $(G,U,F,k,\ell)$ is a \yes-instance.
	\end{claim}
	\begin{claimproof}
		By the design of the branching procedure \branch,
		at least one path of the $G'$-paths of every usable cycle packing
		of size $\ell$ are hit by $X$
		as otherwise $\hit$ would not have been added to $\Ll_{\good}$
		in the first stopping criterion.
		Hence, in $G\setminus X$, there are no usable cycle packings of size $\ell$.
		From \cref{lem:usable-cp},
		we know that if $G$ has a cycle packing of size $\ell$,
		then $G$ also has a \emph{usable} cycle packing of size $\ell$.
		It follows that in $G\setminus X$ there are no cycle packings of size $\ell$.
		Since $X \subseteq V(G)\setminus U$ and $|X|\leq k$,
		set $X$ is a solution to $(G,U,F,k,\ell)$.
	\end{claimproof}

	It remains to show that a solution for the original instance
	propagates to the new instances.
	\begin{claim}
		If the instance $(G,U,F,k,\ell)$ is a \yes-instance,
		then there exists $(\avail,\hit) \in \Ll_{\good}$
		and a set $X \subseteq V(G)\setminus U$ of size at most $k$
		such that $X$ hits all paths in $\hit$.
	\end{claim}
	\begin{claimproof}
		Suppose that $X$ is a solution to the instance $(G,U,F,k,\ell)$,
		that is, the vertex set $X
		\subseteq V(G) \setminus U$ is of size at most $k$ such that $X$ hits every
		cycle packing of size $\ell$ in $G-X$.

		To construct the set $\hit$, we analyze the branching procedure.
		In every step of \branch
		where we can find a usable cycle packing $\Cc$ of size $\ell$ in $G$,
		there is at least one path $P$ of the $G'$-paths that is hit by $X$.
		Assume that this path is an $(N_i,N_j)$-path
		for some $i,j \in F$.
		We claim that we can choose $P$ such that
		the corresponding $(i,j)$ pair is not rich.

		Suppose for contradiction that all of the $G'$-paths hit by $X$
		correspond to rich pairs of $F\times F$.
		Then by the definition of rich pairs, there are for every cycle packing
		at least $k+1$ vertex-disjoint $(N_i,N_j)$-paths
		that are disjoint from the cycle packing.
		Hence, we can pick one of these paths that is not hit by $X$
		and replace the original path by it
		to get a new cycle packing $\Cc'$.
		Thus $X$ cannot hit all cycle packings of size $\ell$,
		contradicting that $X$ is a solution to the instance $(G,U,F,k,\ell)$.
		Therefore, our claim above holds and we can choose $P$
		such that the corresponding pair is not rich.

		Then following the branching procedure \branch,
		we can find a path in the branching tree
		which leads to a leaf node $t$ of the branching tree,
		that is, we go to the child of the current node
		where $P$ is included in corresponding set $\hit$.
		By our choice of the paths in the set $\hit$,
		all paths in the set $\hit$ associated with the leaf node $t$
		are hit by $X$ by our claim above.

		It remains to show that $t$ is a promising candidate.
		Suppose that the depth of $t$ in the branching tree is less than
		$|F|^2 \cdot k\Gamma$.
		In this case it is true that we cannot find a usable cycle packing $\bar{\Cc}$
		at node $t$ whose $G'$-paths are contained in $\avail$.
		Hence, candidate $t$ is promising.

	Suppose that the depth of $t$ in the branching tree is equal to $|F|^2 \cdot k\Gamma$ and we can still find a usable cycle packing $\bar{\Cc}$ at node $t$ whose $G'$-paths are contained in $\avail$.
This implies that $\bar{\Cc}$ is not hit by $X$.
By \cref{lem:intersect},
$X$ can hit at most $|F|^2 \cdot k\Gamma$ usable paths
(every path in $\hit$ does not correspond to a rich pair by our claim).
This contradicts that $X$ is a solution to the instance $(G,U,F,k,\ell)$.
Thus $t$ corresponds to a promising candidate, say $(\avail,\hit)$.
By the way we branch among all children of a node in the branching tree, the paths of $\hit$ are all hit by $X$.
This completes the proof for this claim.		\qedhere

	\end{claimproof}

	This concludes the proof of the branching procedure \branch.
\end{proof}

\paragraph*{Step 5: Solving Instances with Promising Sets.}
\addcontentsline{toc}{subsection}{Step 5: Solving Instances with Promising Sets}

According to \cref{lem:fvs:branching}, to complete the algorithm,
it suffices to consider an instance together with a promising candidate $(\avail,\hit)$
and find a set of at most $k$ vertices
that intersect all paths in $\hit$.
Recall that the set $\hit$ contains paths of the forest $G'$.

If we do not have any undeletable vertices, i.e., if $U =\emptyset$, then this
problem is equivalent to \textsc{Vertex Multicut on Trees}, which is
polynomial-time solvable~\cite[Folklore]{book1}.
If there are undeletable vertices,
i.e., when $U\neq \emptyset$, then this problem resembles
\textsc{Edge Multicut on Trees} which can be solved in $\Oh(2^k \cdot n)$
time~\cite{DBLP:journals/networks/GuoN05}.
We give this algorithm dealing with promising sets
in \cref{lem:hitting-problem}.
If one can find a solution for any of the promising sets, then the
algorithm reports that $(G,U,F,k,\ell)$ is a \yes-instance; otherwise,
it reports a \no.

\begin{lemma}[Path-Hitting on Forests]\label{lem:hitting-problem}
	\renewcommand{\hit}{\mathcal P}
	Given a forest $G'$, a set $U \subseteq V(G')$ of undeletable vertices,
	a collection of paths $\hit$ of
	this forest and an integer $k$.
	In time $2^k \cdot n^{\Oh(1)}$,
	one can find a set $X\subseteq V(G') \setminus
	U$ of size at most $k$, if it exists, such that $X$ {hits} all paths in $\hit$.
\end{lemma}
\begin{proof}
	\renewcommand{\hit}{\mathcal P}
	Without loss of generality we can assume that $G'$ is a tree
	that is rooted at some arbitrary vertex
	as otherwise we consider the trees of the forest sequentially
	and root each tree arbitrarily.

	Let $P_{u,v}$ be a path from $u$ to $v$ in $\hit$
	such that the
	lowest common ancestor (lca) of two endpoints of any other path in $\hit$ is not a descendant of the lca of $u$ and $v$ in $G'$.
	Let the lca of $u$ and $v$ be denoted as $\lca(u,v)$.%
	footnote{Observe that $\lca(u,v) \in \{u,v\}$ is possible.}
	We define two distinguished vertices on the path $P_{u,v}$ as follows.
	Let $u'$, $v'$ be the first vertices on the path
	from $\lca(u,v)$ to $u$, $v$ such that $u',v' \notin U$, respectively.
	We claim that there exists a minimum-sized set $X$ disjoint from $U$ that hits
	all paths in $\hit$ such that $X$ contains either $u'$ or $v'$.
	Once the claim is
	proven, it yields a straightforward branching algorithm
	with running time $2^k \cdot n^{\Oh(1)}$.

	We now show that our claim is true. Let $\bar{X}$ be
	a minimum-sized set that is disjoint from $U$ and intersects all paths
	in $\hit$. In particular, $\bar{X}$ intersects the $(u,\lca(u,v))$-path or
	the $(\lca(u,v),v)$-path of $P_{u,v}$. Without loss of
	generality, assume that $\bar{X}$ intersects the $(u,\lca(u,v))$-path of
	$P_{u,v}$ (the other case is symmetric)
	and let $\bar{u} \in \bar{X}$ be a vertex on
	the $(u, \lca(u,v))$-path. We claim that
	$X=\bar{X} \setminus \{\bar{u}\} \cup \{u'\}$ is also a set
	that hits all paths in $\hit$ of minimum size.
	Suppose not, then there exists another path $P_1$ in $\hit$,
	say with endpoints $u_1$ and $v_1$,
	which is not hit by $X$ and goes through $\bar{u}$.
	By the choice of $P_{u,v}$ as the path
	where the lca of the two endpoints is lowest in the tree,
	the path $P_1$ also goes through $\lca(u,v)$, and hence $u'$.
	Thus, $P_1$ is hit by $X$ which proves the claim.
	This completes the proof for the lemma.
\end{proof}

Using \cref{lem:hitting-problem}, for each promising candidate $(\avail,\hit)$,
a set of size at most $k$ that hits all paths in
$\hit$ can be found in $2^k \cdot n^{\Oh(1)}$ time, if it exists. We
are now ready to complete the proof of \cref{thm:fvs:algo}
which we restate here for convenience.

\fvsAlgo*

\begin{proof}
	Let $\Ii = (G,U,F,k,\ell)$ be the input instance of \CycleUndelHitPack.
	The overall algorithm works as follows.
	\begin{enumerate}
		\item
		For each subset $D \subseteq F$ where $\abs{D} \le k$,
		we construct a new instance
		$\Ii_D = (G-D, U \cup (F \setminus D), F\setminus D, k-\abs{D}, \ell)$,
		that is, we guess the subset of the feedback vertex set that is deleted.
		where the feedback vertex set is undeletable.

		\item
		For each instance $\Ii_1$ from the first step
		we use the algorithm from \cref{lem:fvs:reduceDegree}
		to create a new instance $\Ii_2(\Ii_1)$ of low degree with regard to $\dd$.

		\item
		For each instance $\Ii_2$ from the second step,
		we use the branching procedure \branch from \cref{lem:fvs:branching}
		to generate a list $\Ll_{\good}(\Ii_2)$ of promising candidates.

		\item
		For each instance $\Ii_2$ from the second step
		and each candidate $(\avail,\hit)$ in the corresponding list $\Ll_{\good}(\Ii_2)$
		from the third step,
		we check if there is a set $X$ of deletable vertices
		that hits every path in $\hit$ by using \cref{lem:hitting-problem}.
		If this is the case, then output \yes.
		Otherwise we proceed with the next set in $\Ll_{\good}(\Ii_2)$
		or with the next instance from the first step.

		\item
		If no branch outputs \yes, then the algorithm outputs \no.
	\end{enumerate}

	We first prove the correctness of the algorithm.
	The correctness of the branching in the first step follows from
	\cref{lem:fvs:fvsIsUndeletable}.
	For the second step the correctness follows directly from the
	properties of the algorithm in \cref{lem:fvs:reduceDegree}.
	The correctness of the remaining steps follows
	from the design of the procedure \branch in \cref{lem:fvs:branching}
	and by \cref{lem:hitting-problem}.

	It remains to analyze the running time of the entire algorithm.
	The first step takes time $2^{\abs{F}} \cdot n^{\Oh(1)}$.
	By \cref{lem:fvs:reduceDegree},
	the second step can be applied in time polynomial in the size of the instance.
	According to \cref{lem:fvs:branching},
	the branching procedure \branch runs in time
	$2^{\Oh((|F|^2 k\Gamma)\log |F|)} \cdot n^{\Oh(1)}$,
	and the number of promising sets
	it generates is $2^{\Oh((|F|^2 k\Gamma) \log |F|)}$.
	The last step takes time $2^k \cdot n^{\Oh(1)}$.
	As by assumption we have $\Gamma=\poly(|F|+k)$,
	the entire procedure runs in time
	\[
		2^{\abs{F}} \cdot 2^{\Oh((|F|^2 k\Gamma)\log |F|)} \cdot 2^k \cdot n^{\Oh(1)}
		= 2^{\poly(\abs{F}+k)} \cdot n^{\Oh(1)}
	\]
	which concludes the proof.
	\qedhere

\end{proof}

\subsection{Proof of 
Lemma~\ref{lem:find-usable-cp}:
Finding Vertex-Disjoint Paths}
\label{sec:appendix}

We now prove 
\cref{lem:find-usable-cp}, 
which we restate here for convenience.

\findUsableCp*

\begin{proof}
\newcommand{\UndelHitPackath}{\mathsf{CoverPath}}
\newcommand{\child}{\mathsf{child}}
\newcommand{\ppath}{\mathsf{path}}

We first show that it suffices to consider the case when $H$ is a tree.

\subparagraph*{\boldmath The Graph $H$ is a Forest.}
Assume that $H$ is a forest
and let $H_1, \dots, H_r$ be the $r$ different trees of $H$.
For a subforest $H'$ of $H$ and a set of indices $I \subseteq \numb{f}$,
we define $\UndelHitPackath(H',I)$ to be $\true$ if and only if
$H'$ contains a pairwise vertex-disjoint collection of paths
$\{P_i \in {\mathcal{P}_i} \mid i\in I\}$.

Suppose that we can answer $\UndelHitPackath(H',I)$
for any $I \subseteq \numb{f}$ when $H'$ is tree.
For ease of notation, for all $j\in \numb{r}$,
we denote by $\mathcal H_j$ the union of the first $j$ trees of $H$,
formally, we define $\mathcal{H}_j \deff H_1 \uplus \dots \uplus H_j$.
It vacuously follows that $\UndelHitPackath(H',\emptyset)=\true$
for any forest $H'$.
With these assumptions,
for all $j \in \numb{r}$ and all $I \subseteq \numb{f}$,
we can use the following dynamic program
to compute $\UndelHitPackath(H_j, I)$:
\begin{align*}
  \UndelHitPackath(\mathcal{H}_j,I)
  \deff \bigvee_{I' \subseteq I}
    \UndelHitPackath(H_j,I') \land \UndelHitPackath(\mathcal{H}_{j-1}, I \setminus I')
  .
\end{align*}
Then the final output of the algorithm is $\UndelHitPackath(H, \numb{f})$.
The correctness of this step follows because all $H_j$'s are disjoint trees
and $H=\mathcal{H}_r$.

If we assume that we can compute $\UndelHitPackath(H_j, I)$
for each $j \in \numb{r}$ and $I \subseteq \numb{f}$
in time $2^{\Oh(f)} \cdot n^{\Oh(1)}$,
then the above dynamic program solves the problem on forests in time
$2^{\Oh(f)} \cdot n^{\Oh(1)}$.
As a next step we design the claimed algorithm for trees.

\subparagraph*{\boldmath The Graph $H$ is a Tree.}
Based on the above reasoning,
we now focus on the case when $H$ is a tree.
We prove this lemma using bottom-up dynamic programming on $H$.
To begin, we arbitrarily root the tree $H$. For each $v \in V(H)$, let $H_v$ be
the subtree rooted at $v$ in $H$.
For any vertex set $W \subseteq V(H)$
(of vertices that are incomparable,
i.e., no vertex in $W$ is a descendant of any other vertex in $W$),
let $H_W$ be the disjoint union of the subtrees
rooted at each vertex $v \in W$.
For a vertex $v$ in $H$,
we denote by $\child(v)$ be the set of children of $v$ in $H$.

Our goal is to compute three tables $T$, $T_\child$, and $T_\ppath$
such that, for all $I \subseteq [f]$, vertices $v \in V(H)$,
and paths $P$ in collection with $v \in V(P)$,
\begin{align*}
  T[v, I] &= \true
    \iff \text{there are vertex-disjoint paths }
    \{P_i \in \mathcal{P}_i \mid i \in I\} \text{ in } H_v
    ,
    \\
  T_\child[v, I] &= \true
    \iff \text{there are vertex-disjoint paths }
    \{P_i \in \mathcal{P}_i \mid i \in I\} \text{ in } H_{\child(v)}
    , \text{ and}
    \\
  T_\ppath[v,P,I] &= \true
    \iff \text{there are vertex-disjoint paths }
    \{P_i \in \mathcal{P}_i \mid i \in I\} \text{ in } H_v \setminus V(P).
\end{align*}
We compute the table entries $T[v,I], T_\child[v,I]$ and $T_\ppath[v,P,I]$
using a bottom-up dynamic programming on the rooted tree $H$.

We first we define the base cases.
If $I = \emptyset$, then we set $T[v,\emptyset] = \true$
for every $v \in V(H)$.
Similarly, if $v$ is a leaf and $I=\{i\}$,
then we set $\true$ if and only if $\{v\} \in \mathcal{P}_i$
(otherwise we define the table entry to be $\false$).
Finally, if $v$ is a leaf and $|I| \ge 2$,
then we can directly set $T[v,I]$ to $\false$.

For the remaining case, when $v \in V(H)$ is not a leaf and $\abs{I} \ge 2$,
then we compute the table entry by setting
\begin{align*}
  T[v,I] \deff T_\child[v,I] \vee \bigvee_{i \in I} \bigvee_{\substack{P \in
  \mathcal{P}_i, \\ v \in V(P)}} T_\ppath[v,P, I\setminus \{i\}].
\end{align*}

Observe that $T_\child[v,I]$ and $T_\ppath[v,P,I]$ can be computed with the
analogous dynamic programming to $\UndelHitPackath$. Note that in $T_\ppath[v,P,I]$ table we again end in the
disjoint case. To address that we can add a dummy-vertex adjacent to every root
of the trees in the forest.

Next we prove the correctness of the table entries.
\begin{claim}
  For all vertices $v$ of $H$,
  all sets $I\subseteq \numb{f}$,
  all paths $P$ in collection with $v \in V(P)$,
  the table entries $T[v,I]$, $T_\child[v,I]$, and $T_\ppath[v,P,I]$
  are computed correctly.
\end{claim}
\begin{claimproof}
  We give a proof by induction on the depth of $v$ in $H$.
  We focus on the proof of correctness for the table $T$.
  The correctness proofs for $T_\child$ and $T_\ppath$ follow analogously.

  Suppose there exists a collection of pairwise vertex-disjoint paths $\{P_i \in
  \mathcal{P}_i \mid i \in I \}$ in $H_v$.
  If $v$ does not appear in any of these paths,
  i.e., $v \not \in \bigcup_{i \in I} V(P_i)$,
  then the paths in $\{P_i \in \mathcal{P}_i \mid i \in I \}$ are also
  contained in $H_{\child(v)}$ and by induction $T_\child[v,I]=\true$.
  Moreover, if there is an index $i \in I$ such that
  path $P_i \in \mathcal{P}_{i}$ contains $v$, i.e., $v \in V(P_{i})$,
  then $\{P_{i'} \mid i' \in I \setminus \{i\}\}$ are paths in
  $H_v - V(P_i)$.
  By induction this implies $T_\ppath[v,P,I\setminus \{i\}] = \true$,
  which finishes this direction of the proof.

  For the other direction, suppose that $T[v,I]=\true$.
  We show that in this case there exists a collection
  of pairwise vertex-disjoint paths $\{P_i \in \mathcal{P}_i \mid  i \in I\}$
  in $H_v$.
  It is straightforward to verify the bases cases.
  Therefore, we focus on the case when $v$ is not a leaf and $|I| \ge 2$.

  From the definition of the table entry it follows that if $T[v,I]=\true$,
  then (1) table entry $T_\child[v,I]$ is $\true$ or%
  \footnote{Case (1) and (2) are not necessarily disjoint.}
  (2) there is some index $\hat i \in I$
  and some path $P \in \mathcal P_{\hat i}$ with $v \in V(P)$
  such that $T_\ppath[v,P, I\setminus \{\hat i\}]$ is $\true$.
  In case (1), by induction,
  $H_{\child(v)}$ contains the desired collection of paths.
  For case (2),
  the induction hypothesis implies that $H_v - V(P)$
  contains paths $\{P_i \mid i \in I \setminus \{\hat i\}\}$
  that are pairwise vertex-disjoint.
  From the fact that $P \in \mathcal P_{\hat i}$,
  it directly follows that
  $\{P_i \mid i \in I \setminus \{\hat i\}\} \cup \{P\}$
  is the desired collection of paths in $H_v$.
\end{claimproof}

It remains to analyze the running time of the algorithm.
Observe that the number of states in the tables $T$, $T_\child$, and
$T_\ppath$ is $2^{\Oh(f)} \cdot n^{\Oh(1)}$. Moreover, to compute a single
table entry $2^{\Oh(f)} \cdot n^{\Oh(1)}$ operations are required.
\end{proof}

\section{\texorpdfstring%
{\boldmath \EdgeUndelHitPack Parameterized by $k+\ell$: Single-Exponential Algorithm}
{Edge-HitPack Parameterized by k+l: Single-Exponential Algorithm}}
\label{sec:algo:k2}

In this section, we prove \cref{th:k2-param-alg}.

\edgeparaalg*\label\thisthm

\begin{proof}\label{pf:edgeparaalg}
For a graph $G$, we denote by $\nu(G)$ the matching number of $H$,
that is, the size of the maximum matching in $H$.
It is well-known that a maximum matching (and thus the matching number)
for a graph $H$ can be computed in polynomial time
using the algorithm by Edmonds~\cite{edmonds1965paths}
or by Micali and Vazirani~\cite{DBLP:conf/focs/MicaliV80} for example.

Let $I=(G,U,k,\ell)$ be the given instance of \EdgeUndelHitPack.
The algorithm starts by some preprocessing of the instance
to identify trivial \yes and \no instances.
\begin{itemize}
  \item
  Compute the matching number $\nu(G)$ of $G$.
  If $\nu(G)$ is smaller than $\ell$, then directly answer \yes.

  \item
  Otherwise, compute the matching number $\nu(G[U])$ of $G[U]$,
  that is, the subgraph induced by the undeletable vertices.
  If $\nu(G[U])$ is at least $\ell$, then directly output \no.

  \item
  Additionally, the algorithm outputs \no if $k < 0$.
\end{itemize}
After this preprocessing, the algorithm proceeds as follows.
\begin{enumerate}
  \item
  Compute a maximum matching $M$ of $G[U]$.

  \item
  Invoke Edmonds' blossom algorithm \cite{edmonds1965paths}
  to compute an augmenting path $P$ for $M$ in $G$, that is,
  a path that starts and ends with vertices not covered by $M$
  and alternates between edges not in $M$ and edges in $M$.

  If no such augmenting path exists, then output \yes.

  \item
  Let $u$ and $v$ be the two endpoints of $P$.
  Note that a single edge with both endpoints not in $M$
  is also an augmenting path
  and that at least one of $u$ and $v$ is deletable
  as otherwise $M$ would not be the largest matching in $G[U]$.

  \item
  Construct three instances
  $I_1 = (G\setminus \{u\}, U, k-1, \ell)$ if $u \notin U$,
  $I_2 = (G\setminus \{v\}, U, k-1, \ell)$ if $v \notin U$,
  and $(G, U\cup\{u,v\}, k, \ell)$,
  corresponding to the case that $u$ is deleted, that $v$ is deleted
  and that $u$ and $v$ are made undeletable.
  Solve these three instances recursively
  and output \yes if the algorithm return \yes for at least one instance.
\end{enumerate}

To analyze the running time of the algorithm,
we define a branching measure $\mu$ with $\mu(I) \deff k+\ell-1-\nu(G[U])$.
By the preprocessing, $\mu(I)$ is always non-negative
and from the definition of $\mu$ we get $\mu(I) \le k + \ell$.

If the algorithm deletes $u$ or $v$,
then the measure $\mu(I_1)$ and $\mu(I_2)$ decreases by one
as $k$ decreases by one.
If the algorithm puts $\{u,v\}$ into $U$,
then the measure $\mu(I_3)$ decreases by one as $\nu(G[U])$ increases by one
because $u$ and $v$ are the endpoints of the augmenting path
for the previously largest matching.

Observe that all steps of the algorithm can be computed in polynomial time.
Since $\mu(I)$ decreases in each step
and the algorithm branches into three cases,
the bound on the runtime follows.

The correctness of the algorithm follows from the definition of our problem
and the well-known Berge's theorem \cite{berge1957two},
which says that a matching $M$ in a graph $G$ is maximum
if and only if there is no augmenting path with $M$.
\end{proof}

\section{\texorpdfstring%
{\boldmath \qCliqueUndelHitPack Parameterized by 
Treewidth:\\Double-Exponential Algorithm}
{Clique-HitPack Parameterized by Treewidth: Double-Exponential Algorithm}}
\label{sec:twUpper:clique}
In this \lcnamecref{sec:twUpper:clique} we consider the \UndelHitPack{H} problem
when parameterizing by treewidth and provide an algorithmic upper bound when $H$ is a
complete graph. The algorithm in this section serves two purposes. First, it is
a smooth introduction to the more involved algorithm for \UndelHitPack{H} for
general $H$ in~\cref{sec:twUpper} with the additional advantage
that the algorithm in this section has a slightly better dependence on $\tw$.
The second purpose is that it can
be used as a black-box in \cref{sec:matroid} to get a single-exponential
dependence on $\tw$ in the case when $H = K_2$.

Formally, in this section we prove the following theorem.
\thmUpperTWClique*\label\thisthm

Although the running time is double-exponential in treewidth, we provide a
standard dynamic program based on the tree decomposition of the graph. Note that
the claimed runtime allows us to omit the requirement that such a tree
decomposition has to be given with the input graph because we can compute an
optimal tree decomposition as the first step of the algorithm.

\subparagraph*{Intuition.}
For each bag, we consider the different ways in which a partial solution $D
\subseteq V_t \setminus U$ could interact with the bag $X_t$. We describe this
interaction by a triple $(k_0, D_0, f_0)$ with the following interpretation.
\begin{itemize}
  \item
      The set $D$ contains exactly $k_0$ vertices not in $X_t$. Moreover,
      inside a bag $X_t$ the set $D$ deletes all the vertices in $D_0$.
  \item
  For every set $A \subseteq X_t \setminus D_0$
  the function $f_0(A)$ stores the following integer.
  We interpret $A$ as the set of vertices we want to avoid,
  and $f(A)$ is the size of the largest packing for $G_t-D$
  where additionally no vertex from $A$ is covered.
\end{itemize}

The runtime is essentially bounded by the number of functions $f_0$ we have to
consider. As there are $2^{\tw+1}$ vertex subsets and for each such subset we
want to store a value between $0$ and $n$, the number of functions is at most
$n^{2^{\tw+1}}$. We later show that we can get a significantly better bound on
the number of functions, which then leads to the claimed running time.

To formally exploit these properties, we have to design a dynamic program whose
style deviates from the classical approaches for such algorithms. In most cases,
we first aggregate the information from the appropriate states for the child
nodes and then compute the entry for the parent node. However, we proceed in a
different direction. When considering a state for a child node, we check for
which states of the parent node this information is relevant. By this approach,
it is not necessary to consider all possible table entries but only those where
we store a non-zero entry (i.e., an entry corresponding to a valid partial
solution).

\subparagraph*{Classes.}
Before we state the algorithm, we first introduce some notation.
For any graph $G'$ we let $\KPack(G')$ be the $\clq$-packing number of $G'$
(i.e., the number of vertex-disjoint complete copies of $\clq$).
Consider a node $t_0$ of the tree decomposition,
an integer $k_0 \in \numbZ{k}$,
a subset $D_0 \subseteq X_t$ of vertices,
and a function $f_0 \from 2^{X_t \setminus D_0} \to \Codomain$.
We say that a set $D \subseteq V_{t_0} \setminus U$
\emph{is of class $(k_0, D_0, f_0)$ for $t_0$}
if the following conditions are satisfied:
\begin{enumerate}%
  \item $\abs{D \setminus D_0} = k_0$ and $D \cap X_{t_0} = D_0$.
  \item
  for all $A \subseteq X_{t_0} \setminus D_0$
  we have $f_0(A) = \KPack(G_{t_0}-(D \cup A))$.

\end{enumerate}
For each node $t_0$ of the tree decomposition,
we compute a list $\List(t_0)$ of classes for $t_0$
such that $(k_0, D_0, f_0) \in \List(t_0)$
if and only if
there is a set $D \subseteq V_{t_0} \setminus U$ that is of class $(k_0, D_0, f_0)$.

Before we state the algorithm,
we first define five procedures to extend the classes for a child of node
to a class for the parent node.
The first two procedure correspond to the introduce node,
the third and fourth to the forget node,
and the last one to the join node.

\begin{lemma}[Introduce 1]
  \label{lem:twUpper:clique:introDel}
  Let $t_0$ be an introduce node with the unique child $t_1$
  and let $v \notin U$ be the vertex introduced.

  There is a procedure $\TIntro_1$
  that, for a given class $c_1 = (k_1, D_1, f_1)$ for $t_0$,
  computes in time $\Oh(2^\tw)$
  a class $c_0 = (k_0, D_0, f_0)$ for $t_0$
  such that the following holds:
  For all sets $D \subseteq V_{t_1} \setminus U$,
  if $D$ is of class $c_1$ for $t_1$,
  then $D \cup \{v\}$ is of class $c_0$ for $t_0$.
\end{lemma}
\begin{proof}
  We set $c_0 = (k_1, D_1 \cup \{v\}, f_1)$ as the class for $t_0$.
  Since $v$ not contained in $G_{t_1}$
  but contained in $D_0 = D_1 \cup \{v\}$,
  the two graph $G_{t_1}-D_1$ and $G_{t_0}-(D_1 \cup \{v\})$ are identical.
  Hence, it trivially follows that $D \cup \{v\}$ is of class $c_0$ for $t_0$.
\end{proof}

Next we consider the case when the vertex introduced is not deleted.

\begin{lemma}[Introduce 2]
  \label{lem:twUpper:clique:introNotDel}
  Let $t_0$ be an introduce node with the unique child $t_1$
  and let $v$ be the vertex introduced.

  There is a procedure $\TIntro_2$
  that, for a given class $c_1 = (k_1, D_1, f_1)$,
  computes in time $\Oh(2^\tw)$
  a class $c_0 = (k_1, D_1, f_0)$ for $t_0$
  such that the following holds:
  For all sets $D \subseteq V_{t_1} \setminus U$,
  if $D$ is of class $c_1$ for $t_1$,
  then $D$ is of class $c_0$ for $t_0$.
\end{lemma}
\begin{proof}
  We define the function
  $f_0 \from 2^{X_{t_0}\setminus D_1} \to \Codomain$
  where, for all $A \subseteq X_{t_0} \setminus D_1$, we set
  \begin{align}
    f_0(A) \deff \begin{dcases}
        f_1(A \setminus \{v\}) & \text{if } v \in A, \\
      \max\left(
        f_1(A), 1+
        \max_{\substack{C \subseteq X_t \setminus (D_0 \cup A)\\C \text{ is }
        K_q \text{ containing } v}}
          f_1(A \cup C \setminus \{v\})
      \right) & \text{otherwise.} \\
    \end{dcases}
  \end{align}
  Then the new class is defined as $c_0 = (k_1, D_1, f_0)$.
  Consider some set $A \subseteq X_{t_0} \setminus D_1$.
  If $v \in A$,
  then $G_{t_1}-(D \cup A\setminus \{v\})$ is identical to $G_{t_0}-(D \cup A)$
  since $v$ is introduced at $t_0$.
  Thus, $\KPack(G_{t_0}-(D\cup A))
  = \KPack(G_{t_1}-(D \cup A\setminus \{v\})) = f_1(A\setminus\{v\}) = f_0(A)$.

  For the case when $v \notin A$ we split the proof into two parts.
  \begin{description}
      \item[\boldmath$f_0(A) \le \KPack(G_{t_0}-(D \cup A))$.]
    First, consider the case when $f_0(A) = f_1(A)$.
    Since every packing for $G_{t-1}-(D \cup A)$
    is also a packing for $G_{t_0}-(D \cup A)$,
    we get $f_1(A) = \KPack(G_{t_1}-(D \cup A)) \le \KPack(G_{t_0}-(D \cup A))$.

    Let $C \subseteq X_{t_0} \setminus (D_0 \cup A)$
    be a $q$-clique containing $v$.
    Every packing for $G_{t_0}-(D \cup A \cup C)$ can be extended
    to a packing for $G_{t_0}-(D \cup A)$ by covering the vertices in $C$
    with the same copy of $\clq$.
    Hence, $\KPack(G_{t_0}-(D \cup A)) \ge \KPack(G_{t_0}-(D \cup A \cup C))+1$.
    Since $G_{t_0}-(D \cup A \cup C)$ does not contain $v$,
    the graphs $G_{t_0}-(D \cup A \cup C)$ and
    $G_{t_1}-(D \cup A \cup C \setminus \{v\})$ are identical,
    which implies that the packing number is the same.
    Hence, $\KPack(G_{t_0}-(D \cup A \cup C)) \ge f_1(A \cup C \setminus \{v\})$
    follows
    by the assumption about $f_1$.

\item[\boldmath$f_0(A) \ge \KPack(G_{t_0}-(D \cup A))$.]
    Let $P$ be the packing for $G_{t_0}-(D \cup A)$
    containing exactly $\KPack(G_{t_0}-(D \cup A))$ cliques.
    If $v$ is not covered by $P$,
    then $P$ is also a packing for $G_{t_1}-(D \cup A)$
    and contains at most $\KPack(G_{t_1}-(D \cup A))$ cliques,
    which is exactly $f_1(A)$ by assumption.
    Hence, $P$ contains at most $f_0(A)$ cliques by the definition of $f_0$.

    In the case when $v$ is covered,
    there must be a clique $C \subseteq X_{t_0} \setminus (D_0\cup A)$
    containing $v$ such that $C$ is covered by the same copy of $\clq$.
    Hence, the number of cliques in $P$
    is at most $\KPack(G_{t_0}-(D \cup A \cup C))+1$.
    Since $v$ must be avoided in this case, $P$ contains at most
    $\KPack(G_{t_1}-(D \cup A\cup C \setminus \{v\}))+1$ cliques
    and this is equal to $f_1(A \cup C \setminus \{v\})+1$
    which is bounded by $f_0(A)$ by the definition of $f_1$.
  \end{description}
  This concludes the proof as the function $f_0$
  can be constructed in time $\Oh(2^{\abs{X_t \setminus D_1}})$.
\end{proof}

The following two cases correspond to the forget nodes.
The first procedure deals with the setting
when the forgotten vertex is deleted.

\begin{lemma}[Forget 1]
  \label{lem:twUpper:clique:forgetDel}
  Let $t_1$ be a forget node with the unique child $t_0$
  and let $v$ be the vertex forgotten.

  There is a procedure $\TForget_1$
  that, for a given class $c_1 = (k_1, D_1, f_1)$,
  computes in time $\Oh(2^\tw)$
  a class $c_0 = (k_1+1, D_1 \setminus \{v\}, f_1)$
  such that the following holds:
  For all sets $D \subseteq V_{t_1} \setminus U$ with $v \in D$,
  if $D$ is of class $c_1$ for $t_1$,
  then $D$ is of class $c_0$ for $t_0$.
\end{lemma}
\begin{proof}
  From $\abs{D \setminus D_1} = k_1$
  and the assumption that $v$ is forgotten at $t_0$,
  we get that
  $\abs{D \setminus (D_1 \setminus \{v\})} = \abs{D\setminus D_1}+1 = k_1+1$.
  Since $v$ is deleted, we get that $G_{t_0}-D = G_{t_1}-D$.
  Thus, the bound on the packing number follows trivially,
  as we consider the same graph in both settings.
\end{proof}

In contrast to the previous case,
we now assume that the vertex forgotten is not deleted.

\begin{lemma}[Forget 2]
  \label{lem:twUpper:clique:forgetNotDel}
  Let $t_1$ be a forget node with the unique child $t_0$
  and let $v$ be the vertex forgotten.

  There is a procedure $\TForget_2$
  that, for a given class $c_1 = (k_1, D_1, f_1)$,
  computes in time $\Oh(2^\tw)$
  a class $c_0 = (k_1, D_1, f_1)$
  such that for following holds:
  For all sets $D \subseteq V_{t_1} \setminus U$ with $v \notin D$,
  if $D$ is of class $c_1$ for $t_1$,
  then $D$ is of class $c_0$ for $t_0$.
\end{lemma}
\begin{proof}
  Let $A \subseteq X_{t_0} \setminus D_1$.
  Observe that every packing for $G_{t_0}-(D \cup A)$
  is also a packing for $G_{t_1}-(D \cup A)$ as both graphs are identical.
  As the packing number for $G_{t_1}-D-A$ is equal to $f_1(A)$,
  this is also the packing number for $G_{t_0}-D-A$.
\end{proof}

The last procedure covers the case of the join node.

\begin{lemma}[Join Node]
  \label{lem:twUpper:clique:join}
  Let $t_0$ be a join node
  and let $t_1$ and $t_2$ be the two children of $t_0$.

  There is a procedure $\TJoin$,
  that, for a given class $c_1 = (k_1, D_1, f_1)$ for $t_1$
  and a given class $c_2 = (k_2, D_1, f_2)$ for $t_2$,
  computes in time $2^{\Oh(\tw)}$
  a class $c_0 = (k_1+k_2, D_1, f_0)$
  such that the following holds:
  For all sets $D' \subseteq V_{t_1} \setminus U$
  and for all sets $D''\subseteq V_{t_2} \setminus U$,
  if $D'$ is of class $c_1$ for $t_1$
  and $D''$ is of class $c_2$ for $t_2$,
  then $D' \cup D''$ is of class $c_0$.
\end{lemma}
\begin{proof}
  We define a function $f_0 \from 2^{X_{t_0} \setminus D_1} \to \Codomain$
  where, for all $A \subseteq X_{t_0} \setminus D_0$, we set
  \begin{align}\label{eq:def-of-f0}
    f_0(A) \deff \max_{A_1 \uplus A_2 = X_{t_0} \setminus (D_0 \cup A)}
      f_1(A \cup A_1) + f_2(A \cup A_2)
    .
  \end{align}
  To simplify notation, we set $D = D' \cup D''$ in the following.
  By the classes of $D'$ and $D''$ we get $D' \cap D'' = D_1$,
  which implies that $\abs{D \setminus D_1}
  = \abs{D'\setminus D_1} + \abs{D'' \setminus D_1} = k_1+k_2$

  Consider some set $A \subseteq X_{t_0}\setminus D_1$
  and let $P$ be a maximum packing for $G_{t_0}-(D \cup A)$.
  By the definition of a tree decomposition,
  the packing $P$ cannot contain a clique that
  covers vertices in $G_{t_1}-X_{t_1}$ and $G_{t_2}-X_{t_2}$ simultaneously.
  Hence, there exists partition $\tilde A_1 \uplus \tilde A_2 = X_{t_0} \setminus (D_1 \cup A)$
  such that $\KPack(G_{t_0}-(D \cup A))
  = \KPack(G_{t_1}-(D \cup A \cup \tilde A_1))
  + \KPack(G_{t_2}-(D \cup A \cup \tilde A_2))$.

  From the definition of $f_0$ and the properties of $f_1$ and $f_2$,
  it follows that
  \begin{align*}
    f_0(A) &\ge f_1(A \cup \tilde A_1) + f_2(A \cup \tilde A_2) \\
    &=  \KPack(G_{t_1}-(D' \cup A \cup \tilde A_1))
      + \KPack(G_{t_2}-(D''\cup A \cup \tilde A_2)) \\
    &=  \KPack(G_{t_0}-(D \cup A))
    .
  \end{align*}
  Now let $\bar A_1$ and $\bar A_2$ be the sets
  maximizing the sum in \eqref{eq:def-of-f0}.
  Then it directly follows that
  \begin{align*}
    f_0(A)
    &=  f_1(A \cup \bar A_1) + f_2(A \cup \bar A_2) \\
    &=  \KPack(G_{t_1}-(D  \cup A \cup \bar A_1))
      + \KPack(G_{t_2}-(D  \cup A \cup \bar A_2)) \\
    &\le\KPack(G_{t_1}-(D' \cup A \cup \tilde A_1))
      + \KPack(G_{t_2}-(D''\cup A \cup \tilde A_2)) \\
    &=  \KPack(G_{t_0}-(D \cup A)
    .
  \end{align*}

  This concludes the proof that~\eqref{eq:def-of-f0} is a correct definition.
  To construct the function $f$ for each $A$,
  we just need to iterate over at most $2^\tw$ possible subsets.
  Hence, the output can be computed in time $2^{\Oh(\tw)}$.
\end{proof}

Now we have everything ready to state the algorithm solving 
\qCliqueUndelHitPack.

\begin{proof}[Proof of \Cref{thm:twUpper:clique}]
Given the instance $I = (G, U, k, \ell)$,
we first compute an optimal tree decomposition of $G$
and then transform the decomposition into a nice tree decomposition
where the root and leaf nodes have an empty bag.
For all leaf nodes $t_0$ of the tree decomposition,
we set $\List(t_0) = \{(0, \emptyset, \emptyset \mapsto 0)\}$.

We traverse the nodes of the tree decomposition in post-order
and for each node $t_0$ with at least one child
we perform the following actions (depending on the type of node $t_0$).
\begin{description}
  \item[Introduce Node.]
  Let $t_1$ be the unique child of $t_0$
  and let $v$ be the vertex introduced at $t_0$,
  that is, $X_{t_0} = X_{t_1} \cup \{v\}$.
  Repeat the following for all classes $c_1 \in \List(t_1)$:
  If $v \in V(G) \setminus U$, then use $\TIntro_1$ and $\TIntro_2$ on $c_1$
  to compute two classes $c_0$ and $c_0'$
  and add both to the list $\List(t_0)$.
  Otherwise, we have $v \notin V(G) \setminus U$ and apply $\TIntro_2$ on $c_1$ to compute the class $c_0''$ and add it to the list $\List(t_0)$.

  \item[Forget Node.]
  Let $t_1$ be the unique child of $t_0$
  and let $v$ be the vertex forgotten,
  that is, $X_{t_0} = X_{t_1} \setminus \{v\}$.
  Repeat the following for all classes $c_1 \in \List(t_1)$:
  If $v \in D_1$, then use $\TForget_1$ on $c_1$ to compute the class $c_0$ and add it to the list $\List(t_0)$.
  Otherwise, we have $v \notin D_1$ and use $\TForget_2$ on $c_1$ to compute the class $c_0'$ and add it to the list $\List(t_0)$.

  \item[Join Node.]
  Let $t_0$ be the unique parent of the nodes $t_1$ and $t_2$.
  Repeat the following for all pairs of classes
  $(c_1,c_2) \in \List(t_1) \times \List(t_2)$
  where $c_1=(k_1,D_1,f_1)$ and $c_2=(k_2,D_2,f_2)$:
  Check that $D_1 = D_2$ and if so apply $\TJoin$ on $(c_1,c_2)$
  to get a class $c_0$ for $t_0$
  and add it to the list $\List(t_0)$.
\end{description}
It remains to define the output of the procedure.
For this let $r$ be the root of the tree decomposition.
Then, the algorithm outputs \yes
if there is a class $(k_0, \emptyset, \emptyset \mapsto \ell_0) \in \List(r)$
for some $k_0 \in \numbZ{k}$ and $\ell_0 \in \numbZ{\ell-1}$.

\subparagraph*{Correctness.}
Next we prove that this dynamic program is correct.
\begin{claim}[Correctness]
  \label{clm:twUpper:clique:dpCorrectness}
  For all nodes $t_0$,
  integers $0 \le k_0 \le k$,
  vertex sets $D_0 \subseteq X_t \setminus U$,
  and
  functions $f_0\from 2^{X_t \setminus D_0} \to \Codomain$,
  the following two statements are equivalent:
  \begin{itemize}%
    \item
    There is a set $D \subseteq V_t\setminus U$
    of class $(k_0, D_0, f_0)$ for $t_0$.

    \item
    $(k_0, D_0, f_0) \in \List(t_0)$.
  \end{itemize}
\end{claim}
\begin{claimproof}
  By the definition of the dynamic program, it suffices to show that the first
  statement implies the second statement.  We prove the correctness inductively
  based on the type of the node in the tree decomposition.

\begin{description}
    \item[Leaf Node.]
  Since the bags of the leaf nodes do not contain any vertices,
  no vertices can be deleted.
  By assumption, the leaf nodes do not have children
  and hence, there is only the empty packing
  which contains no copy of $\clq$.

  \item[Introduce Node.]
  If the set $D$ is of class $c_0$ and $v \in D$,
  then $D\setminus \{v\}$ is of \emph{some} class $c_1$ for $t_1$.
  By induction, we get $c_1 \in \List(t_1)$.
  From the algorithm and the properties of $\TIntro_1$,
  we obtain a class $c$ such that $D$ is of class $c$ for $t_0$.
  Moreover, the algorithm adds $c$ to $\List(t_0)$.
  Since each set $D$ has exactly one class for each node,
  we have $c_0 = c$ and thus, $c_0 \in \List(t_0)$.

  Note, that if $v \notin D$, then we similarly get that $c_0 \in \List(t_0)$ by using $\TIntro_2$.
  \item[Forget Node.] The result follows analogously to the introduce node.
  \item[Join Node.]
  Assume that $D$ is of class $c_0$.
  If we consider $D' = D \cap V_{t_1}$ and $D'' = D \cap V_{t_2}$,
  we have that $D'$ and $D''$ are of classes $c_1$ and $c_2$
  for $t_1$ and $t_2$, respectively.
  By the induction hypothesis it follows that
  $c_1 \in \List(t_1)$ and $c_2 \in \List(t_2)$.
  From the definition of the algorithm and the properties of $\TJoin$,
  it follows that $D' \cup D'' = D$ is of some class $c$ for $t_0$.
  By the uniqueness of classes, we conclude $c = c_0$
  and therefore, $c_0 \in \List(t_0)$.
  \claimqedhere
\end{description}
\end{claimproof}

As the last step we prove the running time of the algorithm.
\begin{claim}\label{claim:k2twrunningtime}
  Let $L$ denote the maximum length of a list $\List(t)$ for all nodes $t$.
  Then, the algorithm terminates in time
  $L^2 \cdot 2^{\poly(\tw)} \cdot \poly(n)$.
\end{claim}
\begin{claimproof}
  Computing a nice tree decomposition is possible in time
  $2^{\poly(\tw)} \cdot \poly(n)$ \cite[Chapter~7]{book1}.
  Observe that handling the join nodes dominates the running time
  as we have to consider up to $L^2$ different pairs of classes.
  Because each new class can be computed in time $2^{\poly(\tw)}$,
  the claim follows.
\end{claimproof}

We conclude the proof by bounding the number of classes $(k_0, D_0, f_0)$ that
can appear for each node $t_0$.  There are $k$ choices for $k_0$ and at most
$2^{\tw+1}$ choices for $D_0$.  By a naive bound, the number of choices of
function $f_0$ is $n^{2^{\tw+1}}$.  We claim that the number of functions
appearing in the algorithm is actually $2^{2^{\Oh(\tw)}}$.  This would conclude
the proof, as it means that the running time of the algorithm is $2^{2^{\Oh(\tw)}}
\poly(n)$.

Recall, that $f_0(\emptyset)$ is the maximum integer
in the image of $f_0$.  Moreover, the smallest integer is at most $\tw+1$
smaller because avoiding single additional vertex can decrease the packing
number by at most one. Therefore, the maximum and the minimum in the image differ
by at most $\tw+1$.  In conclusion, the number of functions is bounded by $n
\cdot (\tw+1)^{2^{\tw+1}}$ as there are $n$ choices for the maximum (obtained at
$f_0(\emptyset)$) and $\tw+1$ choices for the values of the remaining
$2^{\tw+1}$ subsets.
\end{proof}

\section{\texorpdfstring%
{\boldmath\EdgeUndelHitPack Parameterized by Treewidth:\\ 
Single-Exponential Algorithm}
{Edge-HitPack Parameterized by Treewidth: Single-Exponential Algorithm}}
\label{sec:matroid}

In this section, we show that for $q=2$,
the algorithm of \cref{sec:twUpper:clique} for \qCliqueUndelHitPack
runs in time $2^{\poly(\tw)}\cdot n^{\Oh(1)}$
without \textit{any changes whatsoever} to the algorithm.

\edgetwalg*\label\thisthm

In light of \cref{claim:k2twrunningtime},
it is sufficient to give an upper bound on the number of different classes
that subsets $D\subseteq V_t\setminus D$ can have
at each node $t$ of the tree decomposition.
As $|D\setminus D_0|$ can take at most $n$ different values
and $D\cap X_t$ can take at most $2^{\tw+1}$ different values,
it boils down to bounding the number of different functions
$f(A)=\nu(G_t-(D\cup A))$ that can arise for a fixed graph $G_t$
and different sets $D$
(note that function $\nu_{K_2}=\nu$ denotes the size of the maximum matching).
The main combinatorial result of this section is precisely such a bound.
This immediately shows that the algorithm of \cref{thm:twUpper:clique}
for $q=2$ runs in time $2^{\poly(\tw)}\cdot n^{\Oh(1)}$,
proving \cref{th:k2-tw-alg}.

\Lemmatchfunction*

Before, we prove~\cref{lem:matchfunctionbound}
we restate~\cref{lem:perfmatchbound}.
Note, that we have already proven \cref{lem:perfmatchbound}
 in~\cref{sec:technical-overview}.

\Lemperfmatchbound*

The following lemma proves a stronger statement, giving a bound on the number of possibilities for a more expressive function that describes how that size of the maximum matching changes when removing a set $S$. We prove this generalization using a simple purely graph-theoretical construction. Let $\nu(G)$ be the size of the maximum matching in $G$. Observe that removing a vertex cannot increase this value and can decrease it only by at most 1. Thus $\nu(G)-\nu(G-S)$ is always between 0 and $|S|$.
\begin{lemma}\label{lem:deffunctionbound}
Let $G$ be a graph over a  vertex set $V\supseteq [k]$ for some integer $k$. Let
$g_{G,k}\from 2^{[k]}\to \{0,1,\dots,k\}$ be the function defined by
$g_{G,k}(S)=\nu(G)-\nu(G-S)$. For each $k$, there are $2^{\Oh(k^3)}$ functions $g_{G,k}$ that can arise this way.
\end{lemma}
\begin{proof}
  Let $d=\nu(G)$. Let us construct the graph $G^*$ from $G$ by introducing
  $n-2d+2k$ independent vertices that are adjacent to every original vertex of
  $G$. Graph $G^*$ has as set $V^*$ of $n^*=2n-2d+2k$ vertices. Let
  $t=\min\{3k,n-2d+2k\}\le 3k$. For notational convenience, let us rename the
  vertices of $G^*$ such that $t$ of the newly introduced vertices form the set
  $\{k+1,\dots,k+t\}$. Let $k_0:=k+t$, we have $[k_0]\subseteq V^*$. Let
  $h_{G^*,k_0}\from 2^{[k_0]}\to \{0,1\}$ be the function defined for $G^*$ and $k_0$ as in
  \cref{lem:perfmatchbound}. We claim that $g_{G,k}(S)$ for any $S\subseteq [k]$
  can be deduced from the function $h_{G^*,k_0}$. As \cref{lem:perfmatchbound}
  asserts that there are at most $2^{\Oh(k_0^3)}=2^{\Oh(k^3)}$ possible functions $h_{G^*,k_0}$, the same bound also holds for the number of possible functions $g_{G,k}$.

  Let $S\subseteq [k]$ be an arbitrary subset and let $s=|S|$.
  Then $G-S$ has a matching of size at least $d-s$, which leaves at most
  $n-s-2(d-s)=n+s-2d$ vertices uncovered. As every matching in $G-S$ has size at
  most $d$, every matching leaves at least $n-s-2d$ vertices uncovered. Let us
  show how the existence of a matching in $G-S$ that leaves at most a certain number of vertices uncovered can be deduced from the function $f$. Let $x\in \{0,1,\dots,2k\}$ such that $n-s-2d+x\ge 0$ (hence $s-x\le n-2d$) and let $c=2k+s-x$. Note that $s\le k$ implies that $c\le 3k$ and
$s-x\le n-2d$ implies $c\le n-2d+2k$. Thus $c\le t$ and  let
$S^*=S\cup\{k+1,\dots, k+c\}$. We claim that $G-S$ has a matching leaving at
most $n-s-2d+x$ vertices uncovered if and only if $G^*-S^*$ has a perfect
matching. Consider a matching $M$ of $G-S$ that leaves $n-s-2d+x$ vertices
uncovered. In $G^*-S^*$, the number of newly introduced vertices is exactly
$n-2d+2k-c=n-2d+x-s$, thus $M$ can be completed to a perfect matching $M^*$ of
$G^*-S^*$. Similarly, if $G-S$ has a perfect matching, then $n-2d+2k-c=n+x-s-2d$
newly introduced vertices are covered in this matching. If we remove the edges
of $M^*$ incident to these edges, then we obtain a matching of $G-S$ avoiding
exactly this number $n+x-s-2d$ of vertices of $G-S$. Thus $f_{G,k}(S)\le i$ for any
$S\subseteq [k]$ and any $0\le i \le k$ can be deduced given the function $h_{G^*,k_0}$.
\end{proof}

Now the proof of \cref{lem:matchfunctionbound} follows immediately.

\begin{proof}[Proof of \cref{lem:matchfunctionbound}]
  In the $n$ vertex graph $G$,
  there are $n$ possibilities for the value of $\nu(G)$.
  Let us consider those functions $f_{G,k}$ that arise from some $n$-vertex graph $G$
  with $\nu(G)=d$ for some fixed integer $d$.
  Then by \cref{lem:deffunctionbound},
  there are $2^{\Oh(k^3)}$ possibilities for the function $g_{G,k}=d-\nu(G-S)$,
  which implies that there are only that many possibilities for the function $f_{G,k}$.
  Considering every possible $d$,
  this proves the bound of $n\cdot 2^{\Oh(k^3)}$.
\end{proof}

\section{\texorpdfstring%
{\boldmath\UndelHitPack{H} Parameterized by Treewidth:\\ Double-Exponential Algorithm}
{H-HitPack Parameterized by Treewidth: Double-Exponential Algorithm}}
\label{sec:twUpper}

In \cref{sec:twUpper:clique} we have seen the algorithm for \UndelHitPack{H}
parameterized by treewidth when $H$ is a clique.
In the following we consider the general case
when $H$ is an arbitrary connected graph with at least three vertices.
Note that if $H$ has only two vertices,
the problem is precisely \UndelHitPack{K_2}.
In this case the result follows by \cref{thm:twUpper:clique}
or rather by \cref{th:k2-tw-alg} proving the improved running time.
Formally, we prove \cref{thm:twUpper:arbitrary}.

\thmUpperTWGeneral*\label\thisthm

Recall that for the case when $H$ is a clique
we heavily exploited that, for each clique appearing in the final packing,
there is one bag containing all vertices covered by this clique.
By this property we could describe the states for each node
by a function based on the subsets of the bag.

When now considering the more general case
we do not have this assumption anymore.
Instead, a copy of the graph $H$ in the packing
might cover vertices from many different bags
(think of a long path for example).
Hence, the algorithm does not ``see'' all vertices of the copy at the same time
and therefore, also has to deal with \emph{partial packings}.
Such a partial packing is some variant of a packing where we also allow
that \emph{subgraphs} of $H$ appear in the packing (in a controlled way).
For each such partial packing we define a \emph{type}
which describes how the packing interacts with the bag.
See \cref{fig:twUpper:general:partial} for an illustration of these concepts.

Then the idea of the algorithm is as follows.
For each node $t$ we consider all possible types
that a packing could have with respect to this node.
For each such type $T$,
we store a bound on the maximum number of copies of $H$
that can appear in any partial packing of type $T$.

We start by introducing the notation and concepts needed
to formally state the dynamic program.
The algorithm is then presented in \cref{sec:twUpper:general:algo},
its correctness is proven in \cref{sec:twUpper:general:correctness},
and the runtime is analyzed in \cref{sec:twUpper:general:runtime}.

\begin{figure}[t]
  \centering
  \begin{subfigure}[b]{.16\textwidth}
    \centering
    \includegraphics[page=1,width=.5\textwidth]{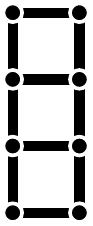}

    \vspace*{1cm}

    \caption{The graph $H$.}
  \end{subfigure}%
  \hfill
  \begin{subfigure}[b]{.40\textwidth}
    \centering
    \includegraphics[page=2,width=\textwidth]{img/partial}
    \caption{%
      An example for a (partial) $H$-packing for $G_t$.
      Note that in the upper left only a subgraph of $H$ is packed.
    }
  \end{subfigure}%
  \hfill
  \begin{subfigure}[b]{.40\textwidth}
    \centering
    \includegraphics[page=4,width=\textwidth]{img/partial}
    \caption{%
      A graphical representation of the type
      $(\prtn, H_1, h_1, H_2, h_2)$
      for the partial packing.
      \\
      The three left-most vertices are not covered and form part $0$,
      the highlighted five-vertex graph in the middle forms $H_1$,
      and the highlighted four-vertex graph on the right corresponds to $H_2$.
      The functions $h_1$ and $h_2$
      map the vertices in $H_1$ and $H_2$ as depicted.
    }
  \end{subfigure}
  \caption{
    An illustration for a node $t$
    of how a partial packing and its type relate to each other.
    \\
    The black vertices and edges correspond to the edges and vertices of $G$.
    The deleted vertices are indicated by hollow dots.
    The highlighted vertices and edges show how the copies of $H$
    are packed to the vertices of $G$.
    A highlighted vertex or edge with a white filling,
    indicates that we do not know
    to which vertices and edges of $G$ the ones of $H$ correspond.
  }
  \label{fig:twUpper:general:partial}
\end{figure}

\subparagraph*{\boldmath Partial $H$-packings.}
Let $I=(G,U,k,\ell)$ be an instance of \UndelHitPack{H}
and let $\tw$ be the treewidth of $G$.
As a first step the algorithm computes an optimal tree decomposition of $G$
Now consider a node $t$ of the computed tree decomposition of width $\tw$.
It is well-known that this can be done in time $2^{\poly(\tw)} \poly(n)$
\cite{KorhonenL23,Bodlaender96}.
We additionally introduce some notation when working with tree decompositions.
When $t$ denotes some node of a tree decomposition,
then we define $X_t$ as the associated bag.
We denote by $V_t$ the vertex set of the subtree of the tree decomposition
that is rooted in $t$, and define $G_t=G[V_t]$ as the corresponding graph.

	For some integers $p,q \ge 0$,
	a tuple
	$P=(h_1, \dots, h_p, \bar H_1, \bar h_1, \bar H_2, \bar h_2,
	\dots, \bar H_q, \bar h_q)$
	is a \emph{partial $H$-packing for $G_t$ (or just $t$ for short)} if
  the following holds:
	\begin{itemize}
		\item
		Each $h_i$ is an injective homomorphism from $H$ to $G_t\setminus X_t$. We refer to these as the \emph{complete
		copies} of $H$ (in $P$).
		\item
		Each $\bar H_j$ is an induced subgraph of $H$, and
		each $\bar h_j$ is an injective homomorphism from $\bar H_j$ to $G_t$ whose image has a non-empty intersection with $X_t$.
		We refer to these as the \emph{partial copies} of $H$ (in $P$).%
    \footnote{Note that $\bar H_j=H$ is possible and this counts as a partial copy as long as the image intersects the bag $X_T$.}
		\item
		The images of all $h_i, \bar h_j$ are pairwise vertex disjoint.
		\item
		For each $j \in \numb{q}$, we refer to a vertex of $\bar H_j$
    that has a neighbor in $H$ outside of $\bar H_j$
    as a \emph{border vertex} of $\bar H_j$.
		All such border vertices are mapped to a vertex from $X_t$ by $\bar h_j$.
	\end{itemize}

\subparagraph*{Types.}
Next, we define so-called \emph{types} for $t$.
Intuitively, such a type describes how a partial $H$-packing for $G_t$
interacts with the vertices from the bag $X_t$.
For this we first consider a function $\prtn$, which partitions the vertices in the bag $X_t$. Here, we also introduce a part $0$, which we interpret to include those vertices that are not covered by any copy of $H$.
All other parts of the partition are associated with some induced subgraph $H_i$ of $H$ together with a function $h_i$
that specifies which vertex from the bag belongs to which vertex of $H_i$,
that is, $h_i$ points out $H_i$ as a (not necessarily induced) subgraph
on the vertices in $X_t$.
Formally, we represent these (labeled) subgraphs
by injective homomorphisms from $H_i$ to $G[X_t]$.
In addition, $h_i$ includes information about the vertices of $H$
that are not in $H_i$.
These vertices are labeled using two special symbols
$\uparrow$ and $\downarrow$ to indicate
that they lie above $X_t$, i.e., not in $V_t$,
or below $X_t$, i.e., in $V_t\setminus X_t$, respectively.

\begin{definition}[Type for $t$, $D$-avoiding]
For some $0 \le w \le \abs{X_t}$,
a \emph{type} $T$
for $t$
is defined as a tuple $T=(\prtn, H_1, h_1, \dots, H_{w}, h_{w})$
such that
\begin{itemize}
  \item
  $\prtn$ is a function $X_t \to \numbZ{w}$
  and, for each $i\in [w]$, we set $X_t(i)\coloneqq \prtn^{-1}(i)$,
  \item
  for all $i \in \numb{w}$,
  $H_i$ is a non-empty induced subgraph of $H$ with $|X_t(i)|$ vertices,
  and
  \item
  for all $i \in \numb{w}$,
  $h_i$ is a function from $V(H)$
  to $X_t(i) \cup \{\downarrow,\uparrow\}$
  such that
  $h_i\vert_{V(H_i)}$ is an injective homomorphism from $H_i$ to $G[X_t(i)]$.
\end{itemize}
If for some set $D \subseteq X_t$
it holds that $\prtn(v) = 0$ for all $v \in D$,
then we say that $T$ is \emph{$D$-avoiding}.
\end{definition}

In the following definition we define the type of a partial $H$-packing.

\begin{definition}[Type of a partial $H$-packing]
  Let $P = (h_1, \dots, h_p, \bar H_1, \bar h_1, \dots, \bar H_q, \bar h_q)$
  be a partial $H$-packing for $G_t$ with regard to $X_t$.

  \begin{itemize}
    \item
    Let $\prtn \from X_t \to \numbZ{q}$ be the function
    with $\prtn(v) = j \in \numb{q}$
    if $v\in X_t$ is in the image of $\bar h_j$;
    and $\prtn(v) = 0$ otherwise.
    \item
    For all $j \in \numb{q}$,
    we set $\hat H_j \deff \bar H_j\left[\bar h_j^{-1}(X_t)\right]$.
    \item
    For all $j \in \numb{q}$,
    we define $\hat h_j \from V(H) \to X_t \cup \{\uparrow,\downarrow\}$
    with
    \[
      \hat h_j(a) \deff \begin{cases}
        \uparrow,    & a\notin \bar h_j^{-1}(V_t),\\ %
        \bar h_i(a), & a\in \bar h_j^{-1}(X_t),\\ %
        \downarrow,  & a\in \bar h_j^{-1}(V_t\setminus X_t). %
      \end{cases}
    \]
  \end{itemize}
Then $T=(\prtn, \hat H_1, \hat h_1, \dots, \hat H_q, \hat h_q)$ is the \emph{type of $P$}.
\end{definition}

\begin{remark}
	It is straightforward to check that the type of a partial $H$-packing for $G_t$ with regard to $X_t$ is a type for $t$ with $w=q$.
\end{remark}

With the definition of partial $H$-packings of a certain type,
we define the $T$-completion number.

\begin{definition}[$T$-completion number]
	Let $T$ be a type for $t$. Consider the partial $H$-packings for $G_t$ with respect to $X_t$ that have type $T$. The \emph{$T$-completion number} of $G_t$ (with respect to $X_t$) is the maximum number of complete copies of $H$ obtained for such a partial $H$-packing of type $T$.
\end{definition}

We prove a bound on the number of types
in terms of $\tw$ and the size of $H$.

\begin{lemma}
  \label{lem:tw:upper:numberOfStates}
  Let $H$ be a fixed graph.
  For all nodes $t$ and all sets $D \subseteq X_t$,
  let $\Types_t(D)$ denote the set of all possible $D$-avoiding types for $t$.

  Then we get that
  $\abs{\Types_t(D)} \le \tw^{\Oh(\tw \cdot \poly(\abs{H}))}$.

\end{lemma}
\begin{proof}
  As $w\le |X_t|\le \tw+1$ there are at most $(\tw+2)^{\tw+1}$ choices for the function $\prtn$.
  For each $H_i$ there are at most $2^{\abs{V(H)}}$ choices,
  and for each $h_i$ there are at most $(\tw+3)^{\abs{V(H)}}$ choices.
  Hence, the total number of types is bounded by
  \[
  (\tw+2)^{\tw+1} \cdot \left(
    2^{\abs{V(H)}} \cdot (\tw+3)^{\abs{V(H)}}
  \right)^{\!\tw+1}
  \]
  which is upper bounded by $\tw^{\Oh(\tw \cdot \poly(\abs{H}))}$
  and thus proves the claim.
\end{proof}

As a next step we define some notation
which we use to modify types. We will use the following notation for neighborhoods:
For a vertex $v$ of $G$, we use $N_G(v)$ to denote the neighborhood of $v$ in $G$, that is, the set of all vertices that share an edge with $v$. For a subset $U$ of the vertices of $G$, we use $N_U(v)$ to denote the set of neighbors of $v$ in $U$.

\subparagraph*{Extending Types by a Vertex.}
Let $T=(\prtn, H_1, h_1, \dots, H_w, h_w)$ be some type for $t$
and let $v$ be a vertex not in $G_t$ and thus, also not in $X_t$.
For all $i \in \numbZ{w}$,
we define $\extend Tvi$ as the set of all possible types for $t'$ with $X_{t'}=X_t\cup \{v\}$ that extend $T$
in the sense that now the vertex $v$ is contained in the $i$th part
(which might be a new one).
These extensions also consider all possible ways in which $v$ could appear in the corresponding (partial) copy of $H$.

Formally, we define $\extend{T}{v}{i}$ as follows:

\begin{itemize}
	\item
	If $v$ is not covered by a (partial) copy of $H$ then it is included in part $0$:
	\begin{equation}
	  \extend Tv0 \deff \{
	    (\prtn\vert_{v\mapsto 0}, H_1, h_1, \dots, H_w, h_w) \}.
	    \label{eq:twUpper:typesExtend:zero}
	\end{equation}
	\item
  Suppose $v$ is covered by a vertex $a$ of a copy of $H$
  for which the corresponding induced subgraph $H_i$ of $H$
  already covers other vertices of $X_t$
  (this implies that $a$ is not in $H_i$),
  i.e., the corresponding (partial) copy of $H$
  is already considered as part of the type $T$.
  Then $v$ is included in part $i$,
  let $\prtn'$ be the corresponding extension of $\prtn$.
  Moreover, $H_i$ is adjusted to include the vertex $a$,
  so let $H_{i,a}\deff H[V(H_i) \cup \{a\}]$.
  Finally, $h_{i,a}$ is a mapping from $V(H)$ to
  $X_{t'}\cup \{\uparrow,\downarrow\}$ that is identical to $h_i$,
  with the exception that $h_{i,a}(a)=v$.
  Ultimately, as we care about subgraph copies of $H$ in $G$,
  edges of $H$ should be preserved,
  and so we need only consider the set $\mathcal{A}$
  containing those vertices $a$ of $H$ with $N_{H_i}(a)\in N_{X_t}(v)$.
  So, for each $i\in [w]$, we define
	\begin{equation}
		\extend Tvi \deff \bigcup_{a\in \mathcal{A}}
		\{
		(\prtn', H_1, h_1, \dots,
		H_{i,a}, h_{i,a}, \dots, H_w, h_w)
		\}.
		\label{eq:twUpper:typesExtend:mid}
	\end{equation}
	\item Finally, suppose that $v$ is covered by some copy of $H$
  that has not been considered yet.
  Then a new part $w+1$ is introduced, and $v$ is mapped to $w+1$ by $\prtn''$,
  the corresponding extension of $\prtn$.
  The corresponding partial copy of $H$ then only holds the single vertex $a$,
  and a corresponding mapping $h_a$ maps every vertex of $H$ to $\uparrow$,
  except for $a$, which is mapped to $v$.
  We set
\end{itemize}
\begin{equation}
  \extend Tv{w+1} \deff \bigcup_{a \in V(H)} \{
    (\prtn'', H_1, h_1, \dots, H_w, h_w, H[\{a\}],
      h_a)
    \}.
    \label{eq:twUpper:typesExtend:last}
\end{equation}
For the case when $i=0$, we know that there is only one type $T' \in \extend{T}{v}{0}$.
Hence, we abuse notation and directly write $\extend Tv0$
whenever we refer to $T'$.

\subparagraph*{Removing Vertices from Types.}
Now we define the converse of the previous operation,
that is, the modification of the types when we remove a vertex from a bag $X_t$.
Let $T=(\prtn, H_1, h_1, \dots, H_w, h_w)$ be some type for $t$,
and let $v$ be a vertex of the considered bag. Let $X_{t'}=X_t\setminus \{t\}$.
We define a set of types $\remove{T}{v}$ for $t'$ as follows:

\begin{itemize}
	\item If $v$ was not covered by any copy of $H$, i.e., if $\prtn(v)=0$, then we can simply remove it by restricting $\prtn$.
  Formally we set
	\begin{equation} \label{eq:twUpper:typesRemove:zero}
		\remove{T}{v} =\{(\prtn\vert_{X_{t'}}, H_1, h_1, \dots, H_w, h_w)\}.
	\end{equation}
	\item
	Suppose $v$ is the last vertex in its part of $X_t$,
  that is, there is an $i\in [w]$ with $\prtn(v)=i$ and $|V(H_i)|=1$.
  Then remove this part from the type.
  It turns out that in this situation we can assume that the corresponding copy of $H$ is completely below $t'$,
  i.e., if $h_i$ maps every vertex of $H$ to $\downarrow$,
  except for $v$.
  The reason for this is that otherwise some vertices
  covered by the respective copy of $H$ would supposedly be not in $G_{t'}$
  (that is, ``above'' $t'$), and some vertices covered by $H$
  would be in $G_{t'}$ but none of them would be in $X_{t'}$,
  which is a contradiction to the fact that $H$ is connected
  and $X_{t'}$ a separator.
	In this case we set
	\begin{equation} \label{eq:twUpper:typesRemove:oneVtx}
		\remove{T}{v} =\{(\prtn\vert_{X_{t'}}, H_1, h_1, \dots, H_{i-1}, h_{i-1},
		H_{i+1}, h_{i+1}, \dots, H_w, h_w)\}.
	\end{equation}
	\item
	Finally, suppose that $\prtn(v)=i$ for some $i\in [w]$,
  and that $|V(H_i)|>1$.
  In this case $v$ is removed from $\prtn$,
  and its preimage $h_i^{-1}(v)$ is removed from $H_i$,
  while the updated homomorphism $h'_i$ is identical to $h_i$,
  with the exception that $h_i^{-1}(v)$ is now mapped to $\downarrow$
  (since it covers a vertex ``below'' the bag $X_{t'}$).
	By the same argument as in the previous case,
  we can assume that none of the vertices in $N_H(h_i(v))$
  are mapped to $\uparrow$ by $h_i$.
	We set
	\begin{equation} \label{eq:twUpper:typesRemove:moreVtcs}
    \remove{t}{v} = \{
		(\prtn\vert_{X_{t'}}, H_1, h_1, \dots, %
		H_i - h_i^{-1}(v), h'_i,
		\dots, H_w, h_w)
    \}.
	\end{equation}
\end{itemize}

\subparagraph*{Combining Types.}
The last operation is used to combine two types.
We use this later for the join nodes.
Let $T_1=(\prtn, H_1, h_1, \dots, H_w, h_w)$
and $T_2=(\prtn', H_1', h_1', \dots, H_w', h_w')$ be two types
for the same node $t$.
We define the addition of $T_1$ and $T_2$ as $\combine{T_1}{T_2}$
if the following requirements are met (otherwise, $\combine{T_1}{T_2}$ is undefined):
\begin{itemize}
	\item $\prtn=\prtn'$,
	\item $H_i=H_i'$ (for each $i\in [w]$),
	\item $h_i\vert_{H_i} = h'_i\vert_{H'_i}$ (for each $i\in [w]$),
and
	\item $\{h_i(a),h_i'(a)\}\neq \{\downarrow\}$ for all $a\in V(H)$.
\end{itemize}

Then we define, for all $i \in \numb{w}$,
the combined function $(\combine{h_i}{h_i'})$ with
\begin{align}
  \label{eq:twUpper:typesCombine:first}
  (\combine{h_i}{h_i'})(a) \deff
  \begin{cases}
    \uparrow,
       & h_i(a) = h_i'(a) = \uparrow, \\
    v, & h_i(a) = h_i'(a) = v, \\
    \downarrow,
      & (h_i(a),h_i'(a)) \in \{(\uparrow,\downarrow), (\downarrow,\uparrow)\},
  \end{cases}
  \qquad
  \forall a \in V(H).
\end{align}
Here the idea is that the way the vertices in $X_t$ are covered does not change.
If for one type some vertex appears below and for the other above,
then this vertex appears below for the combined type.
If both types require that the vertex appears above,
then this also holds for the combined type.

Formally, we define the final type as
\begin{align}
  \combine{T_1}{T_2} \deff
    (\prtn, H_1, \combine{h_1}{h_1'}, \dots, H_w, \combine{h_w}{h_w'})
  \label{eq:twUpper:typesCombine:second}
\end{align}

Before we give the algorithm we prove one structural result that allows us to get an improved running time.

\begin{lemma}
  \label{lem:twUpper:packingNumbersAreNotTooDifferent}
  Let $G$ be a graph and $t$ be a node of its tree decomposition.
  Moreover, consider $D \subseteq V_t$ and $D_0 = D\cap X_t$.
  Let $L$ be the maximum $T$-completion number taken over all types $T$ of some partial $H$-packing for $G_t - D$ with respect to $X_t\setminus D$.
  Then, for all $T \in \Types_t(D_0)$,
  if there is a partial $H$-packing of type $T$,
  then there is a partial $H$-packing of type $T$ with at least
  $\max\{L - \abs{X_t} \cdot \abs{V(H)}, 0\}$ complete copies of $H$.
\end{lemma}
\begin{proof}
  Let $P$ be some partial $H$-packing for $G_t-D$ with respect to $X_t\setminus D$
  containing $L$ complete copies of $H$.

  Fix some arbitrary type $T\in \Types_t(D_0)$
  and some partial $H$-packing $Q$ of type $T$ for $G_t-D$. Note that the complete copies of $H$ contained in some partial packing do not affect its type. So, by removing all complete copies of $H$ in $Q$,
  we obtain another partial $H$-packing $Q'$ of type $T$.
  Since each vertex in $X_t$ can be covered
  by at most one vertex from a copy of $H$ in $Q'$,
  the partial packing $Q'$ contains (partial) copies of $H$ with a total of at most
  $R \deff \abs{X_t} \cdot \abs{V(H)}$ vertices.

  Now let $P'$ be the packing obtained from $P$
  by removing all (partial) copies of $H$ that cover some vertex of $G_t$ that is also covered by a partial copy in $Q'$.
  Since there are at most $R$ vertices covered by partial copies of $Q'$,
  the partial packing $P'$ contains at least $L - R$ complete copies of $H$.
  By the construction of $P'$, the packings $P'$ and $Q'$ are vertex-disjoint
  and therefore, $P' \cup Q'$ is a partial $H$-packing of type $T$
  which contains at least $L - R$ complete copies of $H$.
  This concludes the proof.
\end{proof}

Now we are ready to state the algorithm for \UndelHitPack{H}
parameterized by treewidth.

\subsection{Dynamic Program}
\label{sec:twUpper:general:algo}

Now, we present an algorithm behind the proof of \cref{thm:twUpper:arbitrary}.
  The algorithm is a dynamic program
  based on the tree decomposition of the input graph
  and fills a table entry $A[t_0, k_0, D_0, \ell_0, f_0]$ for all
  \begin{itemize}
    \item
    nodes $t_0$, integers $k_0 \in \numbZ{k}$ and $\ell_0 \in \numbZ{\ell}$,
    \item
    subsets $D_0 \subseteq X_{t_0} \setminus U$,
    \item
    functions $f_0 \from \Types_t(D_0)
    \to \numbZ{(\tw+1) \cdot \abs{V(H)}} \cup \{\nopack\}$.
  \end{itemize}

  Now, we define featured sets
  which are precisely the sets the algorithm counts.

  \begin{definition}[Featured set]
      We say that a vertex set $D
      \subseteq V_{t_0} \setminus U$ with $D_0 \subseteq D$ and $\abs{D \setminus D_0} = k_0$
      is \emph{featured in $(t_0,k_0,
      D_0,\ell_0,f_0)$} if the following two conditions hold:

      \begin{itemize}
          \item
      For every type $T \in \Types_t(D_0)$ with $f_0(T)
      \neq \nopack$, the $T$-completion number for $G_{t_0} - D$
      (i.e., the maximum number of complete copies of $H$ in a partial
      $H$-packing of type $T$) is exactly $\ell_0 - f_0(T)$.
  \item For every type $T \in \Types_t(D_0)$ with
      $f_0(T) = \nopack$,
      there is no partial $H$-packing for $G_{t_0}-D$ of type $T$.
      \end{itemize}
  \end{definition}

  Henceforth, we define $A[t_0, k_0, D_0, \ell_0, f_0]$ to be the number of sets
  that are \emph{featured} in $(t_0, k_0, D_0, \ell_0, f_0)$.
  This concludes the definition of the dynamic programming table. Note that here
  we solve the counting version.
  For the decision version, it would suffice if table
  $A$ stores Boolean values.

  Let $r$ be the root node of the tree decomposition such that the corresponding
  bag $X_r$ does not contain any vertices. Note that the only valid type for $r$
  is the empty type, which we denote by $(\emptyset)$. Then the algorithm
  returns
  \begin{align*}
    \sum_{k_0 = 0}^{k}
    \sum_{\ell_0 = 0}^{\ell-1}
        A[r, k_0, \emptyset, \ell_0, (\emptyset) \mapsto 0],
  \end{align*}
  as the number of solutions for the given instance of \UndelHitPack{H}.

  As the program deals with each table entry separately,
  fix some node $t$,
  some integer $k_0 \in \numbZ{k}$,
  some set $D_0 \subseteq X_t$,
  some integer $\ell_0 \in \numbZ{\ell}$,
  and some function $f\from \Types_t(D_0) \to
  \numbZ{(\tw+1) \cdot \abs{V(H)}} \cup \{\nopack\}$.
  For ease of notation we define $R \deff (\tw+1) \cdot \abs{V(H)}$
  for the range of the integers in the image of $f$.

   \paragraph*{Leaf Node.}
   If $t$ is a leaf node we set $A[t, 0, \emptyset, 0, (\emptyset) \mapsto 0] \deff 1$.
    For all other combinations we set the table entry to be $0$.

    \paragraph*{Introduce Node.}
    Let $t'$ be the child of $t$ and let $v$ be the vertex introduced at $t$,
    that is, $X_t = X_{t'} \cup \{v\}$.

    \subparagraph*{Case 1: $v \in D_0$.} Then, we first define
    the function
    $f' \from \Types_{t'}(D_0 \setminus \{v\}) \to \numbZ{R} \cup \{\nopack\}$ such as
    \begin{align}
        f'(T') &= f(\extend{T'}{v}{0}).
      \label{eq:twUpper:intro:funcDel}
    \intertext{Using this we set the table entry as}
      A[t, k_0, D_0, \ell_0, f] &\deff
        A[t', k_0, D_0 \setminus \{v\}, \ell_0, f']
        \label{eq:twUpper:intro:entryDel}
      .
    \end{align}
    \subparagraph*{Case 2: $v \notin D_0$.} We first perform two checks for the function $f$.

    First, we iterate over all types $T \in \Types_t(D_0)$. Let $i = \prtn(v)$
    be the partition of $v$ in $T$. If $i = 0$, then the check succeeds directly
    for this type. Otherwise, we let $a \in V(H_i)$ be the vertex such that
    $h_i(a) = v$ in $T$. For all $b \in N_{H_i}(v)$, we check if $h_i(b) \in
    N_G(v)$. If this condition is true, we proceed to the next type. Otherwise,
    we check whether $f(T) = \nopack$. If this test fails for any type, we
    define the table entry of $A$ to be $0$.

    As a second step, for all $T' \in \Types_{t'}(D_0)$, we check if there
    exists $c_{T'} \in \numbZ{R}$ such that, for all $T \in \bigcup_{i =
    0}^{\tw+1} \extend{T'}{v}{i}$ with $f(T) \neq \nopack$, it holds that $f(T)
    = c_{T'}$. If this condition is satisfied, we define the function
    $f' \from \Types_{t'}(D_0) \to \numbZ{R} \cup \{\nopack\}$ with
        $f'(T') = c_{T'}$ and set
    \begin{align}
      A[t, k_0, D_0, \ell_0, f] \deff
        A[t', k_0, D_0, \ell_0, f']
        \label{eq:twUpper:intro:entryNotDel}
    \end{align}
    Otherwise, the table entry is set to $0$.

    \paragraph*{Forget Node.}
    Let $t'$ be the unique child of $t$ and let $v$ be the vertex forgotten,
    that is, $X_t = X_{t'} \setminus \{v\}$.
    First, we define function $f_1 \from \Types_{t'}(D_0 \cup \{v\}) \to \numbZ{R} \cup \{\nopack\}$ as
    \begin{align} \label{eq:twUpper:forget:funcDel}
        f_1(T') &= f(\remove{T'}{v}).
  \intertext{Next, we define function $f_2 \from \Types_{t'}(D_0) \to \numbZ{R+1} \cup \{\nopack\}$ as}
      f_2(T') &=
      \begin{cases}
          \nopack, & \text{if } f(\remove{T'}{v}) = \nopack, \\
          f(\remove{T'}{v})+1, & \text{if }\prtn(v) \neq 0 \text{ and } \abs{H_{\prtn(v)}} = 1, \\
        f(\remove{T'}{v}), & \text{otherwise.}
      \end{cases}
      \label{eq:twUpper:forget:funcNotDel}
    \end{align}
    We define
    \begin{align*}
      r_{\min} = \min_{T' \in \Types_{t'}(D_0): f_2(T') \neq \nopack} f_2(T')
      &&\text{ and } &&
      r_{\max} = \max_{T' \in \Types_{t'}(D_0): f_2(T') \neq \nopack} f_2(T')
    \end{align*}
    as the smallest and largest value in the image of $f_2$ respectively.
    Now, we have two cases. If $v \in U$, then we can assume that
    $r_{\max} - r_{\min} \le R$
    (otherwise the table entry can be set to $0$ directly)
    and we fill out the table
    \begin{align}
      \label{eq:twUpper:forget:undel}
          A[ t, k_0, D_0, \ell_0, f] &\deff A[t', k_0, D_0, \ell_0-r_{\min},
          f_2-r_{\min}]
        .\\
     \intertext{Otherwise, we know that $v \notin U$. If $r_{\max} - r_{\min} \le R$, then we fill the table entry as}
      A[ t, k_0, D_0, \ell_0, f] &\deff
        A[t', k_0-1, D_0 \cup \{v\}, \ell_0, f_1]
        + A[t', k_0, D_0, \ell_0-r_{\min}, f_2-r_{\min}]
        .
      \label{eq:twUpper:forget:entryBoth}
    \intertext{Otherwise we ignore the second part of the definition and set}
      A[t, k_0, D_0, \ell_0, f] &\deff
        A[t', k_0-1, D_0 \cup \{v\}, \ell_0, f_1]
      \label{eq:twUpper:forget:entryOnlyFirst}
      .
    \end{align}

    \paragraph*{Join Node.}
        Now, we describe the join operation.
    Let $t_1,t_2$ be the two children of $t$ with the same bag as~$t$.

    We need to guarantee that $\nopack$ propagates correctly. Therefore, we
    introduce the following definition.
    \begin{definition}[Compatible Tuples]
    We say that three integers
    $\ell_0,\ell_1,\ell_2 \in \numbZ{\ell}$ and
    three functions
    $f,f_1,f_2$ from types to $\numbZ{R} \cup \{\nopack\}$
    are \emph{compatible}
    if the following conditions are satisfied:
    \begin{itemize}[nosep]
      \item
          exist $T_1$ and $T_2$ with $f_1(T_1) = f_2(T_2) = 0$, and
      \item
          $f(\combine{T_1}{T_2}) = \nopack$ if and only if $f_1(T_1) = \nopack$ or $f_2(T_2) =
          \nopack$ for any types $T_1,T_2$, and
      \item
          $\ell_0 - f(\combine{T_1}{T_2}) = \ell_1 - f_1(T_1) + \ell_2 -
          f_2(T_2)$ for all types $T_1,T_2$ (if none of the terms is $\nopack$).
    \end{itemize}
    \end{definition}

    Now, we let $A[t, k_0, D_0, \ell_0, f]$ be
    \begin{align}
        \sum_{
          \substack{
          k_1, k_2 \in \numbZ{k}:
          \\
          k_0 = k_1 + k_2
          }
        }
        \sum_{
          \substack{
              \ell_1,\ell_2,f_1,f_2\\
              \text{compatible with }
              \ell_0,f_0
          }
        }
          A[t_1, k_1, D_0, \ell_1, f_1] \cdot A[t_2, k_2, D_0, \ell_2, f_2]
          .
      \label{eq:twUpper:join:entry}
    \end{align}

    This concludes the description of the dynamic programming recursion.

  \subsection{Correctness}
  \label{sec:twUpper:general:correctness}
  It remains to prove the correctness of this dynamic program.
  \begin{lemma}[Correctness]
    \label{clm:tw:upper:dpCorrectness}
    For all nodes $t$,
    integers $k_0 \in \numbZ{k}$,
    vertex sets $D_0 \subseteq X_t$,
    integers $\ell_0 \in \numbZ{\ell}$,
    functions $f\from \Types_t(D_0) \to \numbZ{R} \cup\{\nopack\}$,
    and integers $c \ge 0$,
    the following two statements are equivalent:
    \begin{enumerate}[label=(P.\arabic*)]
      \item
      \label{prop:tw:upper:dpCorrectness:first}
      There are exactly $c$ pairwise different vertex sets $D \subseteq V_t\setminus U$
      with $D_0 \subseteq D$ and $\abs{D \setminus D_0} = k_0$
      such that, for all types $T \in \Types_t(D_0)$,
      the $T$-completion number for $G_t - D$
      is exactly $\ell_0 - f(T)$ if $f(T) \neq \nopack$
      or there is no partial $H$-packing of type $T$.
      \item
      \label{prop:tw:upper:dpCorrectness:second}
      $A[t, k_0, D_0, \ell_0, f] = c$.
    \end{enumerate}
  \end{lemma}
  \begin{proof}
    We prove the statement inductively by handling each possible type
    of a node $t$ of the tree decomposition individually.
    \paragraph*{Leaf Node.}
      There is only one way to select no vertices from the empty set
      and hence, the statement is trivially true.

    \paragraph*{Introduce Node.}
    We distinguish two subcases based on whether
    $v$ is in $D_0$.

      \subparagraph*{Case 1: $v \in D_0$.}
      We first consider the case when $v \in D_0$.
      Recall that $v$ is not part of any (partial) covering in this case.
      Recall the definition of function $f'$
      from \cref{eq:twUpper:intro:funcDel}:
      \begin{align*}
        f' \from \Types_{t'}(D_0 \setminus \{v\} ) \to \numbZ{R} \cup
        \{\nopack\} \text{ where } f'(T') = f(\extend{T'}{v}{0}).
      \end{align*}
      Our goal is to show that \ref{prop:tw:upper:dpCorrectness:first} for $(t,
      k_0, D_0, \ell_0, f, c)$ is equivalent to
      \ref{prop:tw:upper:dpCorrectness:first} for $(t', k_0, D_0 \setminus \{v\},
      \ell_0, f', c)$. Then, by the induction hypothesis, this is equivalent to
      \ref{prop:tw:upper:dpCorrectness:second} for $(t', k_0, D_0 \setminus \{v\},
      \ell_0, f', c)$ and, by the definition of the table entry from
      \cref{eq:twUpper:intro:entryDel}, this is equivalent to
      \ref{prop:tw:upper:dpCorrectness:second} for $(t, k_0, D_0, \ell_0, f,
      c)$.

      To prove the above equivalence, we define a bijection $\psi$ from the
      partial solutions for $(t,k_0,D_0,\ell_0,f)$ to the partial solutions for
      $(t',k_0,D_0 \setminus\{v\},\ell_0,f')$. For a partial solution $D$ for $t$
      with the claimed properties from \ref{prop:tw:upper:dpCorrectness:first},
      $\psi$ maps $D$ to $D' \deff D\setminus\{v\}$ as the corresponding partial
      solution for $t'$.

      We first show that $\psi$ is well-defined. For some partial solution $D$
      for $t$ and for some type $T' \in \Types_{t'}(D_0 \setminus \{v\})$, let $P$
      be a partial $H$-packing for $G_{t'}-(D\setminus\{v\})$
      of type $T'$ with regard to
      $X_{t'}$ with the maximum number of complete copies of $H$. Observe that
      $G_{t'}-(D\setminus\{v\}) = G_{t}-D$ as vertex $v$ is introduced at $t$.
      Moreover, the type of $P$ with regard to $X_t$ is $\extend{T'}{v}{0}$.
      Hence, there are exactly $\ell_0 - f(\extend{T'}{v}{0})$ copies of $H$ in
      $P$, which is equal to $\ell_0 - f'(T')$ by the definition of $f'$.

      As injectivity follows from the definition of $\psi$, it remains to show
      that $\psi$ is surjective. Consider some partial solution $D'$ for $t'$
      with the properties from \ref{prop:tw:upper:dpCorrectness:first}. Fix some
      type $T \in \Types_{t}(D_0)$ and a partial $H$-packing $P$ for
      $G_{t}-(D'\cup \{v\})$ of type $T$
      with the maximum number of complete copies of $H$. Since $P$
      does not contain $v$, packing $P$ has type $\remove{T}{v}$ with regard to
      $X_{t'}$. Hence, it contains at most $\ell_0 - f'(\remove{T}{v}) = \ell_0
      - f(\extend{\remove{T}{v}}{v}{0})$ copies of $H$. Moreover, as it is a
      packing with the maximum number of complete copies of $H$, it contains exactly that
      many copies of $H$. From
      \cref{eq:twUpper:typesExtend:zero,eq:twUpper:typesRemove:zero}, we get
      that $\extend{\remove{T}{v}}{v}{0}=T$ because we assumed that $v \in D_0$.
      Hence, the right-hand side is equal to $\ell_0 - f(T)$, and thus, finishes
      the proof of the surjectivity.

      \subparagraph*{Case 2: $v \notin D_0$.} First, consider the case where for
      some type $T' \in \Types_{t'}(D_0)$, there are two types $T_1 \neq T_2 \in
      \bigcup_{i=0}^{\tw+1} \extend{T}{v}{i}$
      with $f_1(T_1),f_2(T_2) \neq \nopack$ such
      that $f(T_1) \neq f(T_2)$. We claim that in this case, there is no
      solution for $t$ with the claimed properties from
      \ref{prop:tw:upper:dpCorrectness:first}. Assume otherwise and let the
      solution be $\tilde D$. Without loss of generality, it suffices to
      consider the case when $f(T_1) > f(T_2)$. Let $P_2$ be a partial
      $H$-packing for $G_t - \tilde D$ of type $T_2$
      with the maximum $T$-completion number.
      After removing $v$ from $P_2$, we get a packing for $G_{t'}-\tilde D$.
      This packing $P_2-v$ can be extended to a partial $H$-packing for
      $G_t - \tilde D$ of type $T_1$
      by choosing the packing for $v$ appropriately. Let $P_2'$
      be this new packing. Since $P_2'$ and $P_2$ contain the same number of
      complete copies of $H$,
      packing $P_2'$ contains $\ell_0 - f(T_2) > \ell_0 - f(T_1)$ copies
      of $H$. But this contradicts the assumption that the $T_1$-completion number for $G_t-\tilde D$ is exactly $\ell_0 - f(T_1)$.
      Thus, in this case, we have correctly set the value of the table entry to
      $0$.

      Hence, for each type $T' \in \Types_{t'}(D_0)$,
      there is a constant $c_{T'}$ such that,
      the function $f'$ is defined as
      \begin{align*}
          f' \from \Types_{t'}(D_0) \to \numbZ{R} \cup \{\nopack\} \text{ where }
          f'(T') = c_{T'}
          .
      \end{align*}

      As for the first case, our goal is to show that
      \ref{prop:tw:upper:dpCorrectness:first} for $(t, k_0, D_0, \ell_0, f, c)$
      is equivalent to \ref{prop:tw:upper:dpCorrectness:first} for $(t', k_0,
      D_0, \ell_0, f', c)$. Then, the induction hypothesis implies that the
      latter result is equivalent to \ref{prop:tw:upper:dpCorrectness:second}
      for $(t', k_0, D_0, \ell_0, f', c)$, and by the definition of the table
      entry in \cref{eq:twUpper:intro:entryNotDel}, we get that this is
      equivalent to \ref{prop:tw:upper:dpCorrectness:second} for $(t, k_0, D_0,
      \ell_0, f, c)$.

      In the following, we show that every partial solution $D$ for $(t, k_0,
      D_0, \ell_0, f)$ is also a partial solution for $(t', k_0, D_0, \ell_0,
      f')$ and vice versa, that is, we prove a bijection between the two sets of
      solutions.

      Let $D$ be a partial solution for $t$. Consider some type $T' \in
      \Types_{t'}(D_0)$.
      Fix some arbitrary type $\tilde T \in \bigcup_{i=0}^{\tw+1}
      \extend{T'}{v}{i} \subseteq \Types_t(D_0)$.
      We can extend all partial $H$-packings for $G_{t'}-D$ of type $T'$
      to a partial $H$-packing for $G_t-D$ of type $\tilde T$
      by extending the packing for $v$ according to $\tilde T$ because
      $v$ is not yet covered by the packing. Hence, the $T'$-completion
      number for $G_{t'}-D$ is at most the $\tilde T$-completion number
      for $G_t-D$, which is precisely $\ell_0 - f(\tilde T) = \ell_0 - f'(T')$.

      Now, let $\tilde P$ be the partial $H$-packing of type $\tilde T$ containing the
      largest number of complete copies of $H$. By our assumptions about $f$, the packing
      $\tilde P$ contains exactly $\ell_0 - f(\tilde T)$ copies of $H$. Removing
      the vertex $v$ from $\tilde P$,
      we get a new partial $H$-packing for $G_{t'}-D$
      of type $\remove{\tilde T}{v}$ with regard to $X_{t'}$. By our
      choice of $\tilde T$, we get $T' = \remove{\tilde T}{v}$, and thus, the
      packing also have type $T'$.
      Therefore, the $T'$-completion number is exactly $\ell_0-f'(T')$.

      Next, we prove the reverse direction. Consider some partial solution $D'$
      for $(t', k_0, D_0, \ell_0, f')$.
      Fix some type $T \in \Types_t(D_0)$
      and let $P$ be a partial $H$-packing for $G_t-D'$ of type $T$
      with the maximum number of complete copies.
      Observe that the number of complete copies of $H$ does not change
      if we remove $v$ from $P$.
      Moreover, $P-v$ has type $\remove{T}{v}$ with regard to $X_{t'}$.
      Hence, the number of complete copies of $H$ in $P$
      is at most $\ell_0 - f'(\remove{T}{v})$.
      From \cref{eq:twUpper:typesExtend:zero,eq:twUpper:typesExtend:mid,%
      eq:twUpper:typesExtend:last,eq:twUpper:typesRemove:zero,%
      eq:twUpper:typesRemove:oneVtx},
      we get that this is at most $\ell_0 - f(T)$.
      We conclude that $D'$ is a partial solution for $(t, k_0, D_0, \ell_0,f)$.

      Let $P'$ be some partial $H$-packing of type $\remove{T}{v}$ containing the
      largest number of complete copies of $H$. By our assumptions about $f'$, the
      packing $P'$ contains exactly $\ell_0 - f'(T')$ copies of $H$. Since $v$
      is not covered in $P'$, we can extend $P'$ to a new partial
      $H$-packing $P$ of type $T$ by choosing the packing for $v$ according to $T$. We see
      that $P$ and $P'$ contain the same number of complete copies of $H$, which is equal
      to $\ell_0 - f'(\remove{T}{v})$, and by the definition of $f'$ it is also
      equal to $\ell_0 - f(T)$.

      \paragraph*{Forget Node.} Assume out of the $c$ partial solutions from
      \ref{prop:tw:upper:dpCorrectness:first}
      $c_1$ additionally satisfy $v \in D$ and
      the remaining $c_2 \deff c - c_1$ satisfy $v \notin D$.
      Recall,
      from \cref{eq:twUpper:forget:funcDel,eq:twUpper:forget:funcNotDel},
      that the function  $f_1 \from \Types_{t'}(D_0 \cup \{v\}) \to \numbZ{R}
      \cup \{\nopack\}$ satisfies
    \begin{align*}
        f_1(T') &= f(\remove{T'}{v}),
  \intertext{and the function $f_2 \from \Types_{t'}(D_0) \to \numbZ{R+1} \cup \{\nopack\}$ satisfies}
      f_2(T') &=
      \begin{cases}
          \nopack, & \text{ if } f(\remove{T'}{v}) = \nopack, \\
          f(\remove{T'}{v})+1, & \text{ if }\prtn(v) \neq 0 \text{ and } \abs{H_{\prtn(v)}} = 1, \\
        f(\remove{T'}{v}), & \text{otherwise.}
      \end{cases}
    \end{align*}
      Also recall that
      $r_{\min}$ is the smallest integer from the image of $f_2$
      and $r_{\max}$ is the largest integer in the image.

      We claim that if $r_{\max} - r_{\min} > R$, then $c_2 = 0$.
      Assume otherwise and let $D$ be such a solution
      satisfying the constraints in \ref{prop:tw:upper:dpCorrectness:first}.
      Let $T_{\min} \in \Types_t(D_0)$ be a type
      such that $f(T_{\min}) = r_{\min}$
      and $T_{\max} \in \Types_{t}(D_0)$ analogously for $r_{\max}$.
      Since there is a partial $H$-packing for $G_t-D$ of type $T_{\max}$
      and a partial $H$-packing for $G_t-D$ of type $T_{\min}$
      with $\ell_0 - f(T_{\min})$ copies of $H$,
      by \cref{lem:twUpper:packingNumbersAreNotTooDifferent},
      the $T_{\max}$-completion number is at least
      $\ell_0 - f(T_{\min}) - R > \ell_0 - r_{\max}$.
      But this contradicts the assumption that the $T_{\max}$-completion number is exactly $\ell_0 - f(T_{\max}) = \ell_0 - r_{\max}$.
      Hence, we can safely assume in the following that
      $r_{\max} - r_{\min} \le R$.
      Moreover, observe that subtracting $r_{\min}$ from $\ell_0$ and $f_2$
      is only needed to adjust the range of the function $f_2$
      to be in $\numbZ{R}$.
      Hence, we omit this from the proof in the following
      to keep notation simple.

      Assume that $v \notin U$ first (the case when $v \in U$ is analogous).
      Our goal is to show that \ref{prop:tw:upper:dpCorrectness:first}
      for $(t, k_0, D_0, \ell_0, f, c)$
      is equivalent to the two combined properties
      \ref{prop:tw:upper:dpCorrectness:first}
      for $(t', k_0, D_0, \ell_0, f_1, c_1)$
      and \ref{prop:tw:upper:dpCorrectness:first}
      for $(t', k_0 -1, D_0 + v, \ell_0, f_2, c_2)$ together.
      Then we can use the induction hypothesis
      to get the equivalence to \ref{prop:tw:upper:dpCorrectness:second}
      for both cases
      and then by the definition of the dynamic programming table $A$ this is equivalent to
      \ref{prop:tw:upper:dpCorrectness:second}
      for $(t, k_0, D_0, \ell_0, f, c)$.

      For both settings we show that the partial solutions for $t$
      precisely coincide with the partial solutions for $t'$,
      that is, there is a bijection between these partial solutions.
      Nevertheless, we handle both cases separately.

        \subparagraph*{Case 1: $v$ is deleted.}
        Fix some partial solution $D$ with $v \in D$
        and the properties from \ref{prop:tw:upper:dpCorrectness:first}.
        Consider some type $T' \in \Types_{t'}(D_0 \cup \{v\})$.
        Observe that $G_{t'}-D = G_t-D$ and hence,
        every partial $H$-packing for $G_{t'}-D$ of type $T'$
        is also a partial $H$-packing for $G_t-D$ of type $\remove{T'}{v}$
        with regard to $X_t$.
        Therefore, the $T'$-completion number for $G_{t'}-D$
        is at most the $\remove{T'}{v}$-completion number for $G_t-D$
        which is by assumption $\ell_0 - f(\remove{T'}{v})$
        and by definition of $f_1$ this is equal to $\ell_0 - f_1(T')$.

        Moreover,
        let $P$ be the partial $H$-packing for $G_t-D$ of type $\remove{T'}{v}$
        with the maximum number of complete copies of $H$.
        Since $P$ is also a partial $H$-packing for $G_{t'}-D$ of type $T'$,
        the $T'$-completion number of $G_{t'}-D$
        is at least $\ell_0 - f(\remove{T'}{v}) = \ell_0 - f_1(T')$.

        Now consider some partial solution $D'$ for
        $(t', k_0-1, D_0 \cup \{v\}, \ell_0, f_1)$.
        Pick some arbitrary type $T \in \Types_t(D_0)$.
        Let $P$ be some partial $H$-packing for $G_t-D'$ of type $T$.
        Then, since $v \in D'$, we get that $P$ is a partial
        $H$-packing for $G_{t'}-D'$ of type $\extend{T}{v}{0}$.
        By assumption the number of complete copies of $H$ in $P$
        is at most $\ell_0 - f_1(\extend{T}{v}{0})$
        which is, by the definition of $f_1$,
        equal to $\ell_0 - f(\remove{\extend{T}{v}{0}}{v})$
        and by
        \cref{eq:twUpper:typesExtend:zero,eq:twUpper:typesRemove:zero}
        equal to $\ell_0 - f(T)$.

        Let $P'$ be some partial $H$-packing for $G_{t'}-D'$
        of type $\extend{T}{v}{0}$
        with the maximum number of complete copies of $H$.
        Since $v$ is not contained in the packing,
        this is also a partial $H$-packing for $G_t-D'$ of type $T$
        containing $\ell_0 - f_1(\extend{T}{v}{0})$ copies of $H$.
        Hence, by the definition of $f_1$,
        the $T$-completion number is exactly $\ell_0 - f(T)$.

        \subparagraph*{Case 2: $v$ is not deleted.}
        Consider some partial solution $D$ with $v \notin D$
        for $(t, k_0, D_0, \ell_0, f)$.
        We show that $D$ is also a partial solution for
        $(t', k_0, D_0, \ell_0, f_2)$.
        Consider some type $T' \in \Types_{t'}(D_0)$
        with $\prtn(v) \neq 0$ and $\abs{H_{\prtn(v)}}=1$.
        Pick some partial $H$-packing $P'$ for $G_{t'}-D$ of type $T'$.
        By our assumption we know that $v$ is covered by a copy of $H$
        that is entirely contained in $P'$.
        Hence, with respect to $X_t$ there is one more copy of $H$ in $P'$
        then with respect to $X_{t'}$.
        Moreover, the packing $P'$ has type $\remove{T'}{v}$
        with regard to $X_t$.
        Hence, the $T'$-completion number for $G_{t'}-D$
        is at most the $\remove{T'}{v}$-completion number
        for $G_t-D$ minus $1$;
        by assumption this number is equal to $\ell_0 - f(\remove{T'}{v})-1$
        which is equal to $\ell_0 - f_2(T')$ by the definition of $f_2$.

        On the other side,
        let $P$ be the partial $H$-packing of type $\remove{T'}{v}$
        containing exactly $\ell_0 - f(\remove{T'}{v})$ copies of $H$.
        We directly get that $P$ is a partial $H$-packing of type $T'$
        containing exactly $\ell_0 - f(\remove{T'}{v})-1$ copies of $H$.
        Hence, by the definition of $f_2$
        from \cref{eq:twUpper:forget:funcNotDel},
        the $T'$-completion number of $G_{t'}-D$
        is exactly $\ell_0 - f_2(T')$.

        Now consider some type $T' \in \Types_{t'}(D_0)$,
        with $\prtn(T') = 0$ or $\abs{H_{\prtn(v)}} > 1$.
        Again pick some partial $H$-packing $P'$ for $G_{t'}-D$ of type $T'$.
        Since $v$ is either not covered by any (partial) copy of $H$
        or part of a (partial) copy covering at least one more vertex from $X_{t'}$,
        the number of complete copies of $H$ in $P'$ with respect to $X_{t'}$
        is equal to the number of complete copies of $H$ in $P'$ with respect to $X_t$.
        Moreover, observe that $P'$ has type $\remove{T'}{v}$
        with regard to $X_t$.
        Hence, the $T'$-completion number for $G_{t'}-D$
        is at most the $\remove{T'}{v}$-completion number
        for $G_t-D$ which is, by assumption, exactly $\ell_0-f(\remove{T'}{v})$
        and this is equal to $\ell_0 - f_2(T')$ by the definition of $f_2$.

        On the other side,
        let $P$ be the partial $H$-packing of type $\remove{T'}{v}$
        containing exactly $\ell_0 - f(\remove{T'}{v})$ copies of $H$.
        Since $v$ is either not covered
        or part of a partial packing containing some other vertex from $X_t$,
        the packing $P$ is a partial $H$-packing of type $T'$
        containing exactly $\ell_0 - f(\remove{T'}{v})-1$ copies of $H$.
        Hence, by the definition of $f_2$
        from \cref{eq:twUpper:forget:funcNotDel},
        the $T'$-completion number of $G_{t'}-D$
        is exactly $\ell_0 - f_2(T')$.

        For the reverse direction consider some partial solution $D'$
        for $(t', k_0, D_0, \ell_0, f_2)$.
        Let $T \in \Types_t(D_0)$ be some type
        and $P$ some partial $H$-packing for $G_t-D'$ of type $T$.
        There is a unique type $T' \in \bigcup_{i=1}^{\tw+1} \extend{T}{v}{i}$
        such that $P$ is a partial $H$-covering for $G_{t'}-D'$ of type $T'$.
        The number of complete copies of $H$ in $P$ with respect to $X_{t'}$
        is at most $\ell_0 - f_2(T')$.
        Observe that when $v$ is not covered
        or covered by a copy of $H$ also covering other vertices from $X_{t'}$,
        then $P$ contains as many copies of $H$ with respect to $X_{t'}$
        as with respect to $X_t$.
        By the definition of $f_2$,
        we get $\ell_0 - f_2(T') = \ell_0 - f(\remove{T'}{v})$.
        Otherwise, we know that $P$ with respect to $X_{t}$
        contains one copy of $H$ more than $P$ with respect to $X_{t'}$.
        By assumption the latter is at most $\ell_0 - f_2(T')-1$
        which is equal to $\ell_0 - (f(\remove{T'}{v})+1)-1$.
        Using \cref{eq:twUpper:typesExtend:mid,eq:twUpper:typesRemove:oneVtx}
        we conclude that in both cases
        the $T$-completion number is at most $\ell_0 - f(T)$.

        On the other side,
        let $P'$ be a partial $H$-packing of type $T'$
        with $\ell_0 - f_2(T')$ complete copies of $H$.
        By an argument similar to the ones before,
        we get that $P'$ contains $\ell_0 - f_2(T)$ complete copies of $H$
        with respect to $X_t$.
        Hence, the $T$-completion number is exactly $\ell_0-f_2(T)$.

      \paragraph*{Join Node}
      Let $D$ be one of the $c$ different partial solutions. We define two new
      sets $D_1$ and $D_2$ as $D_1 \deff D \cap V_{t_1}$ and $D_2 \deff D \cap V_{t_2}$.
      By the structure of the tree decomposition, we have $D_1 \cap D_2 = D_0$.
      We choose $k_1$ and $k_2$ as $k_1 = \abs{D_1\setminus D_0}$ and $k_2 =
      \abs{D_2 \setminus D_0}$.

      Our first goal is to construct two functions, $f_1$ and $f_2$, and define
      two integers $\ell_1$ and $\ell_2$ in $\numbZ{\ell}$ such that
      $\ell_1,\ell_2, f_1, f_2$ are compatible with $\ell_0, f_0$.
      Moreover, we define them
      such that $D_1$ is one of the $c_1$ partial solutions
      for $(t_1, k_1, D_0, \ell_1, f_1)$
      from \ref{prop:tw:upper:dpCorrectness:first}
      and $D_2$ is one of the $c_2$ partial solutions
      for $(t_2, k_2, D_0, \ell_2, f_2)$
      from \ref{prop:tw:upper:dpCorrectness:first}.

      For all $T \in \Types_t(D_0)$,
      let $\packs(T)$ be the set of all possible partial $H$-packings
      for $G_t-D$ of type $T$.
      For all types $T_1 \in \Types_{t_1}(D_0)$ and $T_2 \in \Types_{t_2}(D_0)$,
      we define the following two sets of partial packings:
      \begin{align*}
        \packs(T \to T_1) \deff \{ P' = P \cap V_{t_1} \mid
          P \in \packs(T),
          P' \text{ has type \(T_1\) w.r.t.\ } X_{t_1} \}
          \\
        \packs(T \to T_2) \deff \{ P' = P \cap V_{t_2} \mid
          P \in \packs(T),
          P' \text{ has type \(T_2\) w.r.t.\ } X_{t_2} \}
      \end{align*}
      Based on this definition,
      for all $T_1 \in \Types_{t_1}(D_0)$,
      we consider all packings in $\packs(T \to T_1)$
      and define $\lambda(T \to T_1)$ as the maximum number of complete copies of $H$
      in these packings.
      Likewise, we define the value $\lambda(T \to T_2)$
      based on $\packs(T \to T_2)$ for all $T_2 \in \Types_{t_2}(D_0)$.
      We set
      \[
        \ell_1 \deff
          \max_{T_1 \in \Types_{t_1}(D_0)}
          \max_{T \in \Types_t(D_0)}
            \lambda(T \to T_1)
        \quad
        \text{ and }
        \quad
        \ell_2 \deff
          \max_{T_2 \in \Types_{t_2}(D_0)}
          \max_{T \in \Types_t(D_0)}
            \lambda(T \to T_2)
        .
      \]

      We define the two functions $f_1$ and $f_2$ as
      \begin{align*}
          f_1 \from \Types_{t_1}(D_0) \to \numbZ{\ell} &\text{  with  }
          f_1(T_1) = \ell_1 - \max_{T \in \Types_t(D_0)} \lambda(T \to T_1) \text{
          and}\\
        f_2 \from \Types_{t_2}(D_0) \to \numbZ{\ell} &\text {  with  }
        f_2(T_2) = \ell_2 - \max_{T \in \Types_t(D_0)} \lambda(T \to T_2)
        .
      \end{align*}
      We first argue that the image of $f_1$ and $f_2$ is actually contained in
      $\numbZ{R}$ as otherwise, there would be no solution. Assuming there is a
      solution, we directly deduce from
      \cref{lem:twUpper:packingNumbersAreNotTooDifferent} that the $T$-completion numbers for different types $T$ differ by at most $R$. Hence, if one type
      violated this property, it would lead to a contradiction.

      Next, we prove that $D_1$ is a solution for $(t_1, k_1, D_0, \ell_1, f_1)$
      and $D_2$ is a solution for $(t_2, k_2, D_0, \ell_2, f_2)$. Since the
      proofs for both cases are almost identical, we will focus only on the
      proof for $t_1$.

      Let $T_1$ be some type in $\Types_{t_1}(D_0)$, and let $P_1$ be a
      partial $H$-packing for $G_{t_1}-D_1$ of type $T_1$.
      Choose $T \in \Types_t(D_0)$
      such that $\ell_0 - f_1(T_1) = \ell_1 - \lambda(T \to T_1)$, and let $P$
      be the corresponding partial $H$-packing for $G_t - D$ of type $T$.
      Now, observe
      that $P' \deff (P \cap V_{t_2}) \cup P_1$ is also a partial packing for
      $G_t-D$ and has type $T_1$ with regard to $X_{t_1}$. Therefore, $P' \cap
      V_{t_1}$ is a packing of type $T_1$ and thus is contained in $\packs(T \to
      T_1)$. Moreover, this packing contains at most $\ell_1 - f_1(T_1)$ copies
      of $H$ according to the definition of $f_1$.
      Since
      \[
        P' \cap V_{t_1} = ((P \cap V_{t_2}) \cup P_1) \cap V_{t_1} = P_1
        ,
      \]
      this implies that the $T_1$-completion number is at most $\ell_0
      - f_1(T_1)$.

      By the definition of the function $f_1$, it directly follows that there is
      a partial $H$-packing of type $T_1$
      containing $\ell_0 - f_1(T_1)$ complete copies
      of $H$. Hence, combined with the previous result, the $T_1$-completion number is exactly $\ell_0 - f_1(T_1)$.

      For the reverse direction, fix some $k_1, k_2$ with $k_1+k_2 = k_0$,
      $\ell_1,\ell_2$, and $f_1$, $f_2$ such that, for all $T_1, T_2$ with $\combine{T_1}{T_2}
      = T$,
      it holds that $\ell_1,\ell_2,f_1,f_2$ are compatible with $\ell_0,f_0$.
      We show that if there is a partial solution $D_1$ for $(t_1, k_1,
      D_0, \ell_1, f_1)$ and a partial solution $D_2$ for $(t_2, k_2, D_0,
      \ell_2, f_2)$, then the set $D_1 \cup D_2$ is a partial solution for $(t,
      k_1+k_2, D_0, \ell_0, f)$.

      For ease of notation, we set $D = D_1 \cup D_2$. Now consider some type $T
      \in \Types_t(D_0)$
      and some partial $H$-packing for $G_t - D$ of type $T$.
      As a
      first step, we observe that $P \cap V_{t_1}$ is a partial packing for
      $G_{t_1}-D = G_{t_1} - D_1$. Since the type of $P \cap V_{t_1}$ is $T_1$
      for some $T_1 \in \Types_{t_1}(D_0)$, the number of complete copies of $H$ in $P
      \cap V_{t_1}$ is at most $\ell_1 - f_1(T_1)$. Similarly, the number of
      complete copies in $P \cap V_{t_2}$ is at most $\ell_2 - f_2(T_2)$, where $T_2$ is
      such that $\combine{T_1}{T_2} = T$. Hence, the number of complete copies
      of $H$ in $P$ is at most $\ell_1 - f_1(T_1) + \ell_2 - f_2(T_2) = \ell_0 -
      f(T)$.

      Next, we show that the bound for the $T$-completion number is
      tight. Let $T_1$ and $T_2$ be types such that $\combine{T_1}{T_2} = T$.
      Choose a partial $H$-packing $P_1$ of type $T_1$
      containing $\ell_1 - f_1(T_1)$ copies of $H$.
      Likewise, let $P_2$ be a partial $H$-packing of type $T_2$ with
      $\ell_2-f_2(T_2)$ complete copies of $H$. It follows that $P_1 \cup P_2$
      is a partial $H$-packing of type $T$
      with $\ell_0 - f_1(T_1) + \ell_2 - f_2(T_2)$
      complete copies of $H$. Hence, the $T$-completion number for
      $G_t-D$ is exactly $\ell_0 - f(T)$.

      Now, for all tuples $(t_1, k_1, D_0, f_1)$ and $(t_2, k_2, D_0, f_2)$, we
      can use the induction hypothesis to conclude that
      \ref{prop:tw:upper:dpCorrectness:first} is equivalent to
      \ref{prop:tw:upper:dpCorrectness:second}. By the definition of the table
      entry from \cref{eq:twUpper:join:entry}, this is equivalent to
      \ref{prop:tw:upper:dpCorrectness:second}. This concludes the
      proof of correctness of the dynamic program.
  \end{proof}

  \subsection{Running Time}
  \label{sec:twUpper:general:runtime}
  As the last step we prove the bound on the running time
  of the dynamic program.

  \begin{claim}
    \label{clm:twUpper:dpRuntime}
    For a fixed graph $H$, the above algorithm runs in time
    $2^{2^{\Oh(\tw \log \tw)}} \cdot \poly(n)$
    to fill the table $A$.
  \end{claim}
  \begin{claimproof}
    By \cref{lem:tw:upper:numberOfStates},
    for each node and each subset of the bag there are at most
    $\tau \deff \tw^{\Oh(\tw \poly\abs{H})}$
    different types.
    Hence, the size of the table $A$ is bounded by
    \[
      \Oh(n \cdot \tw) \cdot k \cdot 2^{\tw+1} \cdot (\ell+1) \cdot (R +1)^{\tau}
      \le 2^{\tau \cdot \Oh(\log (R+1))} \cdot \poly(n)
      .
    \]
    Now we analyze the time it takes to compute a single table entry.
    Observe, that is suffices to check the join nodes,
    as this running time dominates the running time for the other node types.
    Using a brute-force approach to compute the entry
    and a naive bound for the running time yields
    \[
      k^2 \cdot ((R+1)^{\tau})^2 \cdot \tau^2
      \le 2^{3 \cdot \tau \cdot \log (R+1)} \cdot \poly(n)
      .
    \]
    Combined with the first observation from above,
    we get that the algorithm has the following runtime:
    \begin{align*}
      \left( 2^{3 \cdot \tau \cdot \log(R+1)} \cdot \poly(n) \right)
        \cdot \left( 2^{\tau \cdot \log(R+1)} \cdot \poly(n) \right)
      &
      \le 2^{4 \cdot \tau \cdot \log(R+1)} \cdot \poly(n)
    \end{align*}
    Note, that this is at most
    $2^{2^{\Oh(\tw \cdot \log\tw \cdot \poly\abs{H})}} \cdot \poly(n)$.
    The claimed running time follows, since $H$ is a fixed graph.
  \end{claimproof}

  The proof of \cref{thm:twUpper:arbitrary}
  now follows from \cref{clm:tw:upper:dpCorrectness,clm:twUpper:dpRuntime}.

\section{\texorpdfstring%
{\boldmath \CycleUndelHitPack Parameterized by Treewidth:\\
Double-Exponential Algorithm}
{Cycle-HitPack Parameterized by Treewidth: Double-Exponential Algorithm}}
\label{sec:twUpper:cycle}

In this section we extend the result from \cref{sec:twUpper}
to the case where we want to pack cycles of arbitrary length
(at least length $3$).
Formally, we prove \cref{thm:twUpper:cycle}.

\thmcycletwUB*\label{\thisthm}

The underlying idea for this algorithm is the same as for the other algorithms
for cliques and finite graphs. However, as we are not packing a fixed, single
finite graph $H$ but rather an infinitely sized family of graphs, we need a
different way to represent the intersection of a partial packing with the
vertices and edges from a bag. For this description, we heavily exploit that we
are packing cycles and not some other family of graphs that could be
unstructured or significantly more complex to describe, for example, the family
of planar graphs.

The crucial observation is that every partial cycle-packing consists of a
collection of cycles and paths. Hence, the vertices in a bag (and actually all
other vertices) can be grouped into three groups, which can intuitively be
described as follows.
\begin{enumerate}
  \item
  The group of vertices that are not part of any cycle or path
  (this includes the deleted vertices).
  \item
  The group of vertices that are the endpoints of a path
  that is later extended to a cycle.
  \item
  The group of vertices that are contained in a cycle
  or contained in a path but not as their endpoints.
\end{enumerate}

As a next step, we first formalize this idea by introducing some notation. Since
the final algorithm operates on a tree decomposition, we define the notation
only in this context.
\begin{definition}[Partial Cycle-Packing]
  \label{def:twUpper:cycle:partialPack}
  Let $G$ be a graph and let $t$ be a node of its tree decomposition.

  A \emph{partial cycle-packing} for $G_t$ (or just $t$ for short)
  is a set $\packs$
  of vertex disjoint paths and cycles in $G_t$
  such that for every path $P$ in $\packs$,
  the endpoints of $P$ are contained in the bag $X_t$.
\end{definition}

Depending on how a partial cycle-packing interacts with the vertices in a bag,
we define different types to classify the different cycle-packings.
See \cref{fig:twUpper:cycle:partial} for an illustration
of an example for a type.

\begin{definition}[Types and $D$-avoidance]
  Let $G$ be a graph and let $t$ be a node of its tree decomposition.
  A \emph{type} $T$ for $t$ is defined as a tuple $T=(Y_0, Y_1, Y_2, M)$
  where $Y_0, Y_1, Y_2$ form a partition of $X_t$
  and $M$ is a perfect matching for the vertices in $Y_1$.

  Let $\packs$ be a partial cycle-packing for $G_t$.
  We set
  \begin{itemize}
    \item
    $Y_0 \deff \{ v \in V(\packs) \cap X_t \mid \deg_{\packs}(v) = 0 \}$
    as the set of all vertices in $X_t$
    that are not part of any path or cycle in $\packs$,
    \item
    $Y_1 \deff \{ v \in V(\packs) \cap X_t \mid \deg_{\packs}(v) = 1 \}$
    as the set of all vertices in $X_t$
    that are endpoints of a path in $\packs$,
    \item
    $Y_2 \deff \{ v \in V(\packs) \cap X_t \mid \deg_{\packs}(v) = 2 \}$
    as the set of all vertices in $X_t$
    that are either in a path (but not as endpoints) or in a cycle,
    and
    \item
    $M \deff \{ \{u,v\} \in Y_1 \times Y_1 \mid u \neq v, \exists P \in \packs: u,v \in P \}$
    as the perfect matching on $Y_1$
    describing which endpoints appear in the same path of the partial packing.
  \end{itemize}
  Then we say that \emph{$\packs$ is of type $T = (Y_0, Y_1, Y_2, M)$ for $t$}.
  If for some set $D \subseteq X_t$ it holds that $D \subseteq Y_0$,
  then we say that \emph{type $T$ is $D$-avoiding}.
  We denote the set of all $D$-avoiding types for $t$ by $\Types_t(D)$.
\end{definition}

When considering partial cycle-packings of a certain type,
we want to formally describe the number of cycles that potentially could be packed
in the remaining graph while preserving this type.
We formalize this by the $T$-completion number.

\begin{definition}[$T$-completion number]
  Let $G$ be a graph and let $t$ be a node of its tree decomposition.

  For a type $T$ for $t$,
  we define the \emph{$T$-completion number of $G_t$},
  denoted by $\CyComp(T)$,
  as the maximum number of \emph{complete} cycles in $\packs$
  over all possible cycle-packings $\packs$ of type $T$ for $G_t$.
  The number of partial cycles, i.e., paths, is not included in this count.
\end{definition}

Intuitively the completion number of different types can vary
and for some types its ``low'' while for others its ``high''.
Our next step is to define what we mean by a ``low'' completion number
which then transfers to the corresponding type.

\begin{definition}[Low Types]
  Let $G$ be a graph and let $t$ be a node of its tree decomposition.
  Consider two $D$-avoiding types $T, T' \in \Types_t(D)$ for $t$
  such that (1) type $\hat T$ maximizes the $\hat T$-completion number
  and (2) type $T$ is such that $\CyComp(T) < \CyComp(\hat T) - 2(\tw+1)$.
  Then we say that \emph{$T$ is low (for $t$)}.
\end{definition}

We first prove that low types are not relevant
for cycle-packings of maximum size.

\begin{lemma}
  Let $G$ be a graph and let $t$ be a node of its tree decomposition.
  Consider a type $T$ for $t$ that is low.
  Then, there is no \emph{maximum} cycle-packing $\packs$ for $G$
  such that $\packs$ has type $T$ restricted to $G_t$.
\end{lemma}
\begin{proof}
  Consider a maximum cycle-packing $\packs$ for the entire graph $G$.
  Assume for contradiction's sake
  that for some node $t$ of the tree decomposition
  the type of $\packs$ restricted to $G_t$ is $T$
  and $T$ is low for $t$.
  We claim that $\CyComp(T) < \CyComp(T_\emptyset) - (\tw+1)$
  where $T_\emptyset = (X_t, \emptyset, \emptyset, \emptyset)$
  is the type for $t$
  where no vertex of the bag is contained in the packing.

  First, let $\hat T$ be a type for $t$ proving that $T$ is low,
  that is, type $\hat T$ maximizes the $\hat T$-completion number
  and let $\hat \packs$ be a corresponding cycle-packing.
  Removing all paths and cycles from $\hat \packs$ that intersect $X_t$,
  yields a (partial) cycle-packing of type $T_\emptyset$.
  Because of the removed paths and cycles
  we get $\CyComp(T_\emptyset) \ge \CyComp(\hat T) - (\tw+1)$.
  Hence, from
  \[
    \CyComp(T) < \CyComp(\hat T) -2(\tw+1) \le \CyComp(T_\emptyset)- (\tw+1)
  \]
  the claim follows.

  Now we move back to proving the statement. 
  We define $\packs_0$ as the restriction of $\packs$ to the graph $G_t$
  and $\bar \packs$ as the restriction of $\packs$ to the graph $G-G_t$
  where we additionally remove all paths.
  For convenience, we denote the type of $\bar \packs$ for $G-G_t$ by $\bar T$.
  Additionally,
  consider a (partial) cycle-packing $\packs_\emptyset$ of type $T_\emptyset$
  containing $\CyComp(T_\emptyset)$ complete cycles.

  With these definitions at hand,
  we define a new packing $\packs'$ as the union
  of $\packs_\emptyset$ and $\bar \packs$.
  Clearly $\packs'$ is a cycle-packing for the graph $G$
  as the packing $\packs_\emptyset$ does not cover vertices from the bag.
  Next we analyze the relation between the size of $\packs$ and $\packs'$.
  From the above constructions and assumptions we get
  \begin{align*}
    \abs{\packs}
      &\le \CyComp(T_0) + \abs{\bar \packs} + (\tw+1) \\
      &< \CyComp(T_\emptyset) - (\tw+1) + \CyComp(\bar T) + (\tw+1) \\
      &= \CyComp(T_\emptyset) + \CyComp(\bar T) \\
      &= \abs{\packs'}
  \end{align*}
  which proves that $\packs'$ contains more cycles than $\packs$.
  This contradicts our assumption that $\packs$ is a maximum cycle-packing,
  no node of the decomposition has a type that is low.
\end{proof}

Since we are actually solving the \CycleUndelHitPack problem,
we are not interested in just packing cycles
but rather in deleting vertices such that
we cannot pack to many cycles in the remaining graph.
Therefore, it is necessary to argue about all partial cycle-packings
for some graph $G_t$ at the same time.
This idea is formalized by the concept of classes.

\begin{definition}[Classes]
  Let $G$ be a graph and let $t$ be a node of its tree decomposition.

  Moreover, let $k_0 \in \numbZ{\abs{V(G)}}$ be a non-negative integer,
  $D_0$ a subset of vertices from $X_t$,
  $\ell_0  \in \numbZ{\abs{V(G)}}$ be a non-negative integer,
  and $f_0 \from \Types_t(D_0) \to \CodomCyc$.

  We say that \emph{a set $D \subseteq V_t \setminus U$
  is of class $c_0 = (k_0, D_0, \ell_0, f_0)$ for $t$}
  if
  \begin{itemize}
    \item
    $\abs{D \setminus D_0} = k_0$ and $D \cap X_t = D_0$,
    \item
    there is a type $T \in \Types_t(D_0)$ with $f_0(T) = 0$,
    and
    \item
    for all types $T \in \Types_t(D_0)$,
    \begin{itemize}
      \item
      if there is no partial cycle-packing of type $T$ for $t$,
      then $f_0(T) = \nopack$,
      \item
      if $T$ is low for $t$
      or, for every partial packing $\packs$ of type $T$ for $t$,
      there is a descendant $t'$ of $t$
      such that the type of $\packs$ for $t'$ is low,
      then $f_0(T) = \seopack$,
      or
      \item
      otherwise if
      the $T$-completion number for $G_t$ is exactly $\ell_0 - a$
      for some $a \in \numbZ{2\tw+2}$,
      then $f_0(T) = a$.
    \end{itemize}
  \end{itemize}
  To simplify notation we assume that $\nopack$ and $\seopack$
  are invariant under arithmetic operations
  where $\nopack$ takes precedence over $\seopack$,
  that is, $\seopack + \nopack = \nopack$.
\end{definition}

\begin{figure}[t]
  \centering
  \begin{subfigure}[b]{.48\textwidth}
    \centering
    \includegraphics[page=1,width=\textwidth]{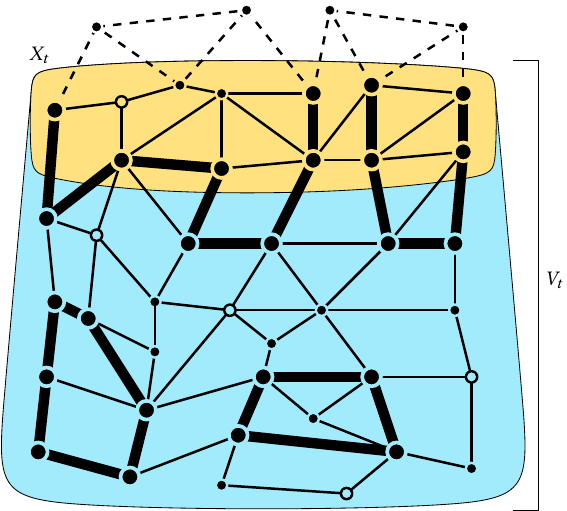}
    \caption{%
      An example for a partial cycle-packing for $G_t$.
      \\
      The bold edges correspond to the edges in the partial packing.
    }
  \end{subfigure}%
  \hfill
  \begin{subfigure}[b]{.48\textwidth}
    \centering
    \includegraphics[page=3,width=\textwidth]{img/partial-cycle}
    \caption{%
      A graphical representation of the type
      $(Y_0, Y_1, Y_2, M)$
      of the depicted partial packing.
      \\
      The small group of the three round vertices at the top middle corresponds to $Y_0$,
      the vertices at the bottom represent $Y_2$,
      and the remaining vertices form $Y_1$.
      The matching $M$ between the vertices in $Y_1$
      is represented by the dashed edges.
    }
  \end{subfigure}
  \caption{
    An illustration for a node $t$
    of how a partial cycle-packing and its type relate to each other.
    \\
    The black vertices and edges correspond to the edges and vertices of $G$.
    The deleted vertices are indicated by hollow dots.
  }
  \label{fig:twUpper:cycle:partial}
\end{figure}

Before stating the algorithm, we provide a bound for the number of $D$-avoiding types and classes for a fixed node $t$.
These bounds follow directly from the definition of the two concepts.
\begin{lemma}
  \label{lem:twUpper:cycle:numberOfTypesAndClasses}
  Let $G$ be a graph
  and let $t$ be a node of its tree decomposition of width $\tw$.

  For a set $D \subseteq X_t$, the number of $D$-avoiding types for $t$,
  i.e., the size of $\Types_t(D)$, is at most $\tau \in 2^{\Oh(\tw \log \tw)}$.

  The number of classes for $t$ is at most
  $2^{2^{\Oh(\tw \log \tw)}} \cdot n^2$.
\end{lemma}
\begin{proof}
  We start by bounding the number of $D$-avoiding types for a fixed set $D$
  and then use this result to bound the number of classes.

  The number of partitions of $X_t$ into the three sets $Y_0$, $Y_1$, and $Y_1$
  is bounded by $3^{\tw+1}$.
  The choices for the set $M$ which contains a perfect matching for $Y_1$,
  is bounded by the number of permutations on $\abs{Y_1}$ elements.
  Hence, the number of types is at most
  $\tau \deff 3^{\tw+1} \cdot (\tw+1)^{\tw+1} = 2^{\Oh(\tw \log \tw)}$.

  To bound the number of classes $(k_0, D_0, \ell_0, f_0)$ for $t$,
  observe that there are at most $n+1$ choices for $k_0$,
  exactly $2^{\tw+1}$ choices for the set $D_0$,
  at most $n+1$ choices for the integer $\ell_0$,
  and at most $(2(\tw+1)+3)^{\tau}$ choices for the function $f_0$.
  Combining this yields that there are at most
  \[
    n \cdot 2^{\tw+1} \cdot n \cdot (2\tw+5)^{2^{\Oh(\tw \log \tw)}}
    = 2^{2^{\Oh(\tw \log \tw)}} \cdot n^2
  \]
  different classes for a node $t$.
\end{proof}

\subsection{Auxiliary procedures}

The algorithm for \CycleUndelHitPack is a dynamic program
based on a given tree decomposition.
We state the program in a bottom-up manner.
For this we define five different procedures
each of which compute the value for the parent node
based on the values for the child nodes.
To simplify the procedures we make use of introduce edge nodes
and assume that each edge is introduced exactly once.

\begin{description}
  \item[\IntroVtxNodeDel.]
  The introduced vertex is deleted and
  thus, included in the part of the partition
  corresponding to uncovered vertices.
  See \cref{lem:twUpper:cycle:introVtxDel}.
  \item[\IntroVtxNodeNotDel.]
  The introduced vertex is not deleted.
  Since we are using introduce edge nodes, the vertex is isolated and
  must be contained in the set corresponding to uncovered vertices
  for all types.
  See \cref{lem:twUpper:cycle:introVtxNotDel}.

  \item[\IntroEdgeNodeDel.]
  The endpoints of the introduced edge are deleted
  and thus, the edge cannot be included in any packing.
  See \cref{lem:twUpper:cycle:introEdgeDel}.
  \item[\IntroEdgeNodeNotDel.]
  Both endpoints of the introduced edge are not deleted.
  Thus, this edge can be included in the possible packings
  but it can also not be included in a packing.
  See \cref{lem:twUpper:cycle:introEdgeNotDel}.

  \item[\ForgetNodeDel.]
  The vertex forgotten at the node is deleted.
  By forgetting this vertex we have to change the budget
  for the number of vertices being deleted.
  See \cref{lem:twUpper:cycle:forgetDel}.
  \item[\ForgetNodeNotDel.]
  The vertex forgotten at the node is not deleted.
  As the new class forgets whether the vertex was used in a packing
  or remained uncovered, two types for the child
  combine into one type for the parent.
  See \cref{lem:twUpper:cycle:forgetNotDel}.

  \item[\JoinNode.]
  The two classes for the two child nodes are combined into a new class
  for the parent where the existing paths might be joined to paths or cycles.
  See \cref{lem:twUpper:cycle:join}.
\end{description}

We start with the introduce vertex nodes where we have two different cases.
The first case corresponds to the choice that $v$ is deleted.

\begin{lemma}[Introduce Vertex Node With a Deleted Vertex]
  \label{lem:twUpper:cycle:introVtxDel}
  Let $t_0$ be an introduce vertex node with child $t_1$
  and let $v \in V_{t_0} \setminus U$ be the vertex introduced.

  There is a procedure \IntroVtxNodeDel that,
  for a given class $c_1=(k_1,D_1,\ell_1,f_1)$ for $t_1$,
  computes in time $\Oh(\tau)$
  a class $c_0 = (k_1, D_1 \cup \{v\}, \ell_1, f_0)$ for $t_0$
  such that the following holds:
  For all sets $D \subseteq V_{t_1}\setminus U$,
  if $D$ is of class $c_1$ for $t_1$,
  then $D\cup\{v\}$ is of class $c_0$ for $t_0$.
\end{lemma}
\begin{proof}
  We define the function $f_0$ as follows.
  For all types $T \in \Types_{t_1}(D_1)$
  with $T = (Y_0, Y_1, Y_2, M)$
  we set
  \[
    f_0(Y_0\cup\{v\}, Y_1, Y_2, M) \deff f_1(T)
    .
  \]

  Since vertex $v$ does not appear in the graph $G_{t_1}$,
  and is deleted from $G_{t_0}$ by considering $D_0=D_1\cup\{v\}$,
  the graphs $G_{t_1}-D_1$ and $G_{t_0}-D_0$ are identical.
  Hence, if the set $D$ is of class $c_1$ for $t_1$,
  then $D\cup\{v\}$ is of class $c_0$ for $t_0$.
\end{proof}

The second case for the introduce vertex node corresponds to the case
when the vertex is not deleted.

\begin{lemma}[Introduce Vertex Node With a Usable Vertex]
  \label{lem:twUpper:cycle:introVtxNotDel}
  Let $t_0$ be an introduce vertex node with child $t_1$
  and let $v \in V_{t_0}$ be the vertex introduced.

  There is a procedure \IntroVtxNodeNotDel that,
  for a given class $c_1=(k_1,D_1,\ell_1,f_1)$ for $t_1$,
  computes in time $\Oh(\tau)$ a class $c_0 = (k_1, D_1, \ell_1, f_0)$ for $t_0$
  such that the following holds:
  For all sets $D \subseteq V_{t_1}\setminus U$,
  if $D$ is of class $c_1$ for $t_1$,
  then $D$ is of class $c_0$ for $t_0$.
\end{lemma}
\begin{proof}
  Observe that $v$ is an isolated vertex in $G_{t_0}$.
  Hence, it especially does not have any neighbors
  in $G_{t_0}-X_{t_0}$.
  This implies that all partial solutions for $t_0$
  cannot include $v$ in a partial cycle-packing.
  Thus, all types for $t_0$ can be described by all types for $t_1$
  where we add $v$ to the set of uncovered vertices.

  Observe that the intuitive idea agrees with the setting
  from the previous case, where the vertex is deleted.
  However, since in the current case the introduced vertex
  might receive some edge at some point in the future,
  the sets of types are strictly speaking not identical
  and thus, both cases need to be handled separately.

  Based on this observation, we define the new function $f_0$ as follows.
  For all $T \in \Types_{t_1}(D_1)$
  with $T = (Y_0, Y_1, Y_2, M)$,
  we set
  \[
    f_0(Y_0 \cup \{v\}, Y_1, Y_2, M) \deff f_1(T)
    .
  \]

  With this definition it follows directly
  that every set $D\subseteq V_{t_1} \setminus U$ with class $c_1$ for $t_1$
  is also of class $c_0$ for $t_0$.
\end{proof}

Next we consider the introduce edge nodes
where we also provide two different procedures.
The first one corresponds to the case where the edge cannot be selected
in a packing because at least one of its endpoints is deleted.

\begin{lemma}[Introduce Edge Node With a Deleted Endpoint]
  \label{lem:twUpper:cycle:introEdgeDel}
  Let $t_0$ be an introduce edge node with child $t_1$
  and let $uv$ be the edge introduced.

  There is a procedure \IntroEdgeNodeDel that,
  for a given class $c_1 = (k_1, D_1, \ell_1, f_1)$ for $t_1$,
  computes in time $\Oh(\tau)$ a class $c_0 = (k_1, D_1, \ell_1, f_1)$ for $t_0$
  such that the following holds:
  For all sets $D \subseteq V_{t_1}\setminus U$
  with $D \cap \{u,v\} \neq \emptyset$,
  if $D$ is of class $c_1$ for $t_1$,
  then $D$ is of class $c_0$ for $t_0$.
\end{lemma}
\begin{proof}
  If (at least) one of the endpoints of the introduced edge is deleted,
  i.e., contained in the set $D_1$ and thus, $D$,
  the edge cannot be part of any (partial) cycle-packing.
  Hence, the class of any such set $D$ does not change.
\end{proof}

The next procedure handles the case when both endpoints of the introduced edge
are not deleted, meaning that the edge could be included in a packing.

\begin{lemma}[Introduce Edge Node With Usable Endpoints]
  \label{lem:twUpper:cycle:introEdgeNotDel}
  Let $t_0$ be an introduce edge node with child $t_1$
  and let $uv$ be the edge introduced.

  There is a procedure \IntroEdgeNodeNotDel that,
  for a given class $c_1 = (k_1, D_1, \ell_1, f_1)$ for $t_1$
  computes in time $\Oh(\tau)$ a class $c_0 = (k_1, D_1, \ell_0, f_0)$ for $t_0$
  such that the following holds:
  For all sets $D \subseteq V_{t_1}\setminus U$
  with $D \cap \{u,v\} = \emptyset$,
  if $D$ is of class $c_1$ for $t_1$,
  then $D$ is of class $c_0$ for $t_0$.
\end{lemma}
\begin{proof}
  It remains to define $f_0$.
  For this we first define a function $f$ which we derive $f_0$ later from.
  In the definition of $f$ we consider six possible cases
  where two of them are symmetric.
  For all types $T \in \Types_{t_1}(D_1)$ with $T = (Y_0, Y_1, Y_2, M)$,
  we define
  \newcommand{\tif}{\ \text{if }}
  \begin{alignat*}{6}
    &\text{Case 0:}&~&~f(
        Y_0,&&
        Y_1,&~&
        Y_2,&~&
        M&
      ) \deff& f_1(T), \\
    &\text{Case 1:}&&\tif u \in Y_0 \land v \in Y_0&\\
      &&&~f(
        Y_0\setminus\{u,v\},&&
        Y_1\cup\{u,v\},&&
        Y_2,&&
        M\cup\{uv\}&
      ) \deff& f_1(T), \\
    &\text{Case 2a:}&&\tif u \in Y_0 \land yv \in M&\\
      &&&~f(
        Y_0\setminus\{u\},&&
        Y_1\setminus\{v\}\cup\{u\},&&
        Y_2\cup\{v\},&&
        M\setminus\{vy\}\cup\{uy\}&
      ) \deff& f_1(T), \\
    &\text{Case 2b:}&&\tif xu \in M \land v \in Y_0&\\
      &&&~f(
        Y_0\setminus\{v\},&&
        Y_1\setminus\{u\}\cup\{v\},&&
        Y_2\cup\{u\},&&
        M\setminus\{xu\}\cup\{xv\}&
      ) \deff& f_1(T), \\
    &\text{Case 3: }&&\tif xu \in M \land yv \in M&\\
      &&&~f(
        Y_0,&&
        Y_1\setminus\{u,v\},&&
        Y_2\cup\{u,v\},&&
        M\setminus\{xu,yv\}\cup\{xy\}&
      ) \deff& f_1(T),\text{ and} \\
    &\text{Case 4: }&&\tif uv \in M&\\
      &&&~f(
        Y_0,&&
        Y_1\setminus\{u,v\},&&
        Y_2\cup\{u,v\},&&
        M\cup\{uv\}&
      ) \deff& f_1(T)-1
  \end{alignat*}
  where $x, y \in Y_1$ are such that $u, v, x, y$ are all distinct.
  Recall, that $\nopack$ and $\seopack$ are invariant
  under arithmetic operations,
  this especially implies $\nopack - 1 = \nopack$
  and $\seopack - 1 = \seopack$.

  Observe that $f$ is a function from $\Types_{t_0}(D_0)$
  to $\numbZ[-1]{2\tw+2} \cup \{\nopack, \seopack\}$.
  If there is no type $T \in \Types_{t_0}(D_0)$ such that $f(T) = -1$,
  then we define $f_0 = f$.
  Otherwise, we define $f_0$ based on $f$ where,
  for all types $T \in \Types_{t_0}(D_0)$,
  \[
    f_0(T) \deff
      \begin{cases}
        f(T),     & \text{if } f(T) \in \{\seopack, \nopack\}, \\
        \seopack, & \text{if } f(T) = 2(\tw+1), \\
        f(T)+1,   & \text{otherwise}.
      \end{cases}
  \]
  As a last step it remains to define $\ell_0$.
  If $f_0 = f$, then we set $\ell_0 = \ell_1$
  and otherwise we set $\ell_0 = \ell_1 + 1$
  which is complementing the definition of $f_0$.

  Now consider a set $D \subseteq V_{t_1} \setminus U = V_{t_0}\setminus U$
  with $D \cap \{u, v\} = \emptyset$
  that is of class $c_1$ for $t_1$.
  We claim that $D$ is of class $c_0$ for $t_0$.
  To prove this, fix some type $T_0 \in \Types_{t_0}(D_1)$.

  If a packing of type $T_0$ does not contain the edge $uv$,
  then this packing is also of type $T_0$ for $t_1$.
  This is Case~0 in the definition above.

  If the edge is selected,
  then there is a unique type $T_1$ for $t_1$ resulting from deleting the edge.
  Depending on how the inclusion of edge $uv$ changed the type,
  we consider four different cases:
  \begin{description}%
    \item[Case 1:]
    A new path is created.
    \item[Case 2:]
    A path is extended by this edge.
    \item[Case 3:]
    Two paths are connected to each other.
    \item[Case 4:]
    A cycle is closed.
  \end{description}

  For the first three cases, the number of cycles does not change.
  Only for Case~4, we obtain one more cycle in the type for $t_0$
  which was not present in the type for $t_1$.
\end{proof}

As a next step we consider forget nodes.
Depending on whether the forgotten vertex is deleted or not,
we consider two disjoint cases.

\begin{lemma}[Forget Node With a Deleted Vertex]
  \label{lem:twUpper:cycle:forgetDel}
  Let $t_0$ be a forget node with child $t_1$
  and let $v \notin U$ be the vertex forgotten.

  There is a procedure \ForgetNodeDel that,
  for a given class $c_1 = (k_1, D_1, \ell_1, f_1)$ for $t_1$ with $v \in D_1$,
  computes in time $\Oh(\tau)$
  a class $c_0 = (k_1+1, D_1\setminus\{v\}, \ell_1, f_0)$ for $t_0$
  such that the following holds:
  For all sets $D \subseteq V_{t_1}\setminus U$,
  if $D$ is of class $c_1$ for $t_1$,
  then $D$ is of class $c_0$ for $t_0$.
\end{lemma}
\begin{proof}
  The function $f_0$ basically agrees with $f_1$
  except that every type for $t_0$ does not include $v$.
  As every type for $t_1$ includes $v$,
  there is a simple one-to-one correspondence
  between the types for $t_0$ and $t_1$.

  Formally, we define the function
  $f_0\from \Types_{t_0}(D_1\setminus\{v\}) \to \CodomCyc$
  as follows.
  For all types $T \in \Types_{t_1}(D_1)$
  with $T = (Y_0, Y_1, Y_2, M)$,
  we set
  \[
    f_0(Y_0\setminus\{v\}, Y_1, Y_2, M) \deff f_1(Y_0, Y_1, Y_2, M) = f_1(T)
    .
  \]
  When considering a set $D \subseteq V_{t_1} \setminus U$
  of class $c_1$ for $t_1$,
  we immediately get, by our choice of $f_0$,
  that $D$ is of class $c_0$ for $t_0$.
\end{proof}

The next procedure covers the second case for a forget node
where the vertex forgotten is not deleted
and thus, could potentially appear in cycle-packings.

\begin{lemma}[Forget Node With a Usable Vertex]
  \label{lem:twUpper:cycle:forgetNotDel}
  Let $t_0$ be a forget node with child $t_1$
  and let $u$ be the vertex forgotten.

  There is a procedure \ForgetNodeNotDel that,
  for a given class $T_1 = (k_1, D_1, \ell_1, f_1)$
  computes in time $\Oh(\tau)$ a class $c_0 = (k_1, D_1, \ell_1, f_0)$
  such that the following holds:
  For all sets $D \subseteq V_{t_1}\setminus U$
  with $v \notin D$,
  if $D$ is of class $c_1$, then $D$ is of class $c_0$.
\end{lemma}
\begin{proof}
  We define $f_0$ as follows.
  For all types $T \in \Types_{t_0}(D_1)$
  with $T = (Y_0, Y_1, Y_2, M)$,
  we set
  \[
    f_0(T) \deff \ell_1 - \max\bigl(
        \ell_1 - f_1(Y_0\cup\{v\}, Y_1, Y_2, M),
        \ell_1 - f_1(Y_0, Y_1, Y_2\cup\{v\}, M)
      \bigr)
  \]
  where we assume that $a > \seopack > \nopack$ for all $a \in \nat$.

  Observe that when the vertex $v$ is not deleted,
  it can either be covered or uncovered.
  In both cases we can get the same type for $t_0$,
  although the packing has different types for $t_1$.
\end{proof}

The last procedure handles the case for join nodes.
Here we have to combine two partial cycle-packings to form a new one.
In order to make this combination formal,
we define a function $\reduce$
which operates on the union of two matchings and combines them
such that partial packings are joined together whenever possible
and cycles are closed as soon as two partial cycle-packings
form together a cycle.

\begin{definition}
  Let $G$ be a graph and let $t$ be a join node of its tree decomposition.
  We define the function
  \[
    \reduce
      \from 2^{X_{t} \times X_{t}} \times \nat
      \to 2^{X_{t} \times X_{t}} \times \nat
  \]
  as
  $\reduce(M, \lambda) \deff (M, \lambda)$ if, for all distinct $u,v,w \in X$,
  the set $M$ satisfies $\abs{M \cap \{ uv, vw \}} \le 1$.
  Otherwise, we define
  \[
    \reduce(M, \lambda) \deff
      \begin{cases}
        (M\setminus\{uv, vw\}\cup\{uw\}, \lambda  ),
          & \text{if } uv, vw \in M \land uw \notin M,\\
        (M\setminus\{uv, vw\},           \lambda+1),
          & \text{if } uv, vw \in M \land uw \in M.
      \end{cases}
  \]
\end{definition}

With this function $\reduce$ in mind,
we now formally define the combination of two types.

\begin{definition}[Combination of Types]
  Let $G$ be a graph and let $t_0$ be a join node of its tree decomposition
  with children $t_1$ and $t_2$.
  Let $T_1 \in \Types_{t_1}(D_1)$ be a type for $t_1$
  with $T_1 = (Y'_0, Y'_1, Y'_2, M')$
  and $T_2 \in \Types_{t_2}(D_1)$ be a type for $t_2$
  with $T_2 = (Y''_0, Y''_1, Y''_2, M'')$.

  The combination of $T_1$ and $T_2$ is undefined
  whenever $Y'_2 \not\subseteq Y''_0$ or $Y''_2 \not\subseteq Y'_0$,
  that is, when a vertex would be incident to at least three selected edges
  in the combined packing.
  Otherwise, we define a new type $T_0 = (Y_0, Y_1, Y_2, M)$
  with
  \begin{align*}
    Y_0 &\deff Y'_0 \cap Y''_0 \\
    Y_1 &\deff (Y'_1 \cap Y''_0) \cup (Y'_0 \cap Y''_1) \\
    Y_2 &\deff Y'_2 \cup Y''_2 \cup (Y'_1 \cap Y''_1)
  \end{align*}
  where the set $M$ is defined as
  the set such that $(M, \lambda)$ is the fix point of the function
  $\reduce$
  on input $((M' \setminus M'') \cup (M'' \setminus M'), \abs{M' \cap M''})$.

  We say that $T_1$ and $t_2$ can be combined to type $T_0$
  by creating $\lambda$ cycles
  and denote this by $\combine{T_1}{T_2} \overset{+\lambda}= T_0$.
\end{definition}

With this definition we can now state the procedure for the join node.

\begin{lemma}[Join]
  \label{lem:twUpper:cycle:join}
  Let $t_0$ be a join node with children $t_1$ and $t_2$.

  There is a procedure \JoinNode that,
  for given classes $c_1 = (k_1, D_1, \ell_1, f_1)$ for $t_1$
  and $c_2 = (k_2, D_2, \ell_2, f_2)$ for $t_2$
  with $D_1 = D_2$,
  computes in time $\Oh(\tau^3)$ a class
  $c_0 = (k_1+k_2, D_1, \ell_0, f_0)$ for $t_0$
  such that the following holds:
  For all sets $D' \subseteq V_{t_1}\setminus U$
  and all sets $D'' \subseteq V_{t_2}\setminus U$
  if $D'$ is of class $c_1$ for $t_1$
  and $D''$ is of class $c_2$ for $t_2$,
  then $D' \cup D''$ is of class $c_0$ for $t_0$.
\end{lemma}
\begin{proof}
  First observe that we assumed that each edge is introduced exactly once
  in the underlying tree decomposition.
  Hence, the edge cannot appear in the partial packing for the left subtree
  \emph{and} the partial packing for the right subtree.
  Moreover, every type for $t_0$ can be split into a component
  coming from $t_1$ and a part coming from $t_2$
  depending on the edges that have been introduced so far.

  We first define a function $f$
  which serves as the basis for defining $f_0$.
  For all types $T_0 \in \Types_{t_0}(D_1)$, we set
  \[
    f(T) \deff \max_{\substack{T_1 \in \Types_{t_1}(D_1) \\ T_2 \in \Types_{t_2}(D_1) \\ \combine{T_1}{T_2} \overset{+\lambda}= T_0}}
    \begin{cases}
      \ell_1 - f_1(T_1) + \ell_2 - f_2(T_2) + \lambda,
        & \text{if } f_1(T_1),f_2(T_2) \in \nat, \\
      \seopack,
        & \text{if } \seopack \in \{ f_1(T_1), f_2(T_2)\} \not\ni \nopack \\
      \nopack,
        & \text{if } \nopack \in \{f_1(T_1),f_2(T_2)\},
    \end{cases}
  \]
  where we assume that $a > \seopack > \nopack$ for all $a \in \nat$.

  Based on this we now define $f_0$ and $\ell_0$.
  For some type $\hat T$ maximizing $f(\hat T)$, we set $\ell_0 \deff f(\hat T)$.
  We construct $f_0$ based on $f$ where,
  for all types $T \in \Types_{t_0}(D_0)$, we set
  \[
    f_0(T) \deff
      \begin{cases}
        f(T),
          & \text{if } f(T) \in \{\seopack, \nopack\}, \\
        \seopack,
          & \text{if } f(T) < \ell_0 - 2(\tw+1), \\
        \ell_0 - f(T),
          & \text{otherwise.}
      \end{cases}
  \]

  With this definition it follows directly
  that for
  two given sets $D' \subseteq V_{t_1}$ of class $c_1$ for $t_1$
  and $D'' \subseteq V_{t_2}$ of class $c_2$ for $t_2$,
  the set $D' \cup D''$ is of class $c_0$ for $t_0$.
  \qedhere

\end{proof}

\subsection{Dynamic program}

Now we have everything ready to state the algorithm solving \CycleUndelHitPack,
that is, prove \cref{thm:twUpper:cycle} which we restate for convenience.

\thmcycletwUB*

\begin{proof}[Proof of \Cref{thm:twUpper:cycle}]
  Given an instance $I = (G, U, k, \ell)$ of \CycleUndelHitPack,
  we first compute an optimal tree decomposition of $G$
  and then transform this decomposition into a nice tree decomposition
  with introduce edge nodes where each edge is introduced once
  and the additional requirement that
  the root and leaf nodes have empty bags.

  The algorithm computes,
  for each node $t$ of the tree decomposition,
  a list $\List(t)$ of classes $c$ for this node
  such that there is a set $D \subseteq V_t \setminus U$ that is of class $c$.

  The dynamic program traverses the tree decomposition in post-order
  and for each node $t_0$ we perform the following actions
  (depending on the type of node $t_0$).

  \begin{description}
    \item[Leaf Node.]
    For each leaf node $t_0$ of the tree decomposition,
    we define the list $\List(t_0)$ explicitly by setting
    $\List(t_0) \deff \{(0, \emptyset, 0,
    (\emptyset,\emptyset,\emptyset,\emptyset) \mapsto 0) \}$.

    \item[Introduce Vertex Node.]
    Let $t_1$ be the unique child of $t_0$
    and let $v$ be the vertex introduced at $t_0$,
    that is, $X_{t_0} = X_{t_1} \cup \{v\}$.
    Repeat the following two steps for all classes $c_1 \in \List(t_1)$:
    \begin{itemize}
      \item
      Apply \IntroVtxNodeNotDel from \cref{lem:twUpper:cycle:introVtxNotDel}
      on $c_1$ to compute the class $c_0$
      and add $c_0$ to the list $\List(t_0)$.
      \item
      If $v \in V(G) \setminus U$, i.e., the vertex $v$ can be deleted,
      then use \IntroVtxNodeDel from \cref{lem:twUpper:cycle:introVtxDel}
      to compute a class $c'_0$
      and add $c'_0$ to the list $\List(t_0)$ as well.
    \end{itemize}

    \item[Introduce Edge Node.]
    Let $t_1$ be the unique child of $t_0$
    and let $uv$ be the edge introduced.
    Repeat the following for all classes $c_1 \in \List(t_1)$:
    \begin{itemize}
      \item
      If $u \in D_1$ or $v \in D_1$, then use \IntroEdgeNodeDel
      from \cref{lem:twUpper:cycle:introEdgeDel} on $c_1$
      to compute the class $c_0$ and add $c_0$ to the list $\List(t_0)$.
      \item
      Otherwise we get $u,v \notin D_1$.
      Then, we use \IntroEdgeNodeNotDel from
      \cref{lem:twUpper:cycle:introEdgeNotDel} on $c_1$
      to compute the class $c_0'$ and add $c'_0$ to the list $\List(t_0)$.
    \end{itemize}

    \item[Forget Node.]
    Let $t_1$ be the unique child of $t_0$
    and let $v$ be the vertex forgotten,
    that is, $X_{t_0} = X_{t_1} \setminus \{v\}$.
    Repeat the following for all classes $c_1 \in \List(t_1)$:
    \begin{itemize}
      \item
      If $v \in D_1$,
      then use \ForgetNodeDel from \cref{lem:twUpper:cycle:forgetDel} on $c_1$
      to compute the class $c_0$ and add $c_0$ to the list $\List(t_0)$.
      \item
      If $v \notin D_1$, then use \ForgetNodeNotDel
      from \cref{lem:twUpper:cycle:forgetNotDel} on $c_1$
      to compute the class $c_0'$ and add $c'_0$ to the list $\List(t_0)$.
    \end{itemize}

    \item[Join Node.]
    Let $t_0$ be the unique parent of the nodes $t_1$ and $t_2$.
    Repeat the following for all pairs of classes
    $(c_1,c_2) \in \List(t_1) \times \List(t_2)$:
    Check that $D_1 = D_2$ and if so apply the procedure \JoinNode
    from \cref{lem:twUpper:cycle:join} on $(c_1,c_2)$
    to get a class $c_0$ for $t_0$
    and add $c_0$ to the list $\List(t_0)$.
  \end{description}
  It remains to define the output of the algorithm.
  For this let $r$ be the root of the tree decomposition.
  Then, the algorithm outputs \yes if the list $\List(r)$ contains a class
  $(k_0, \emptyset,\ell_0, (\emptyset,\emptyset,\emptyset,\emptyset) \mapsto 0)$
  for some $k_0 \in \numbZ{k}$ and $\ell_0 \in \numbZ{\ell-1}$.

  \subparagraph*{Correctness.}
  It remains to prove the correctness and the running time of this algorithm.
  We first show correctness by proving the following claim.
  \begin{claim}[Correctness]
    \label{clm:twUpper:cycle:dpCorrectness}
    For all nodes $t_0$,
    integers $0 \le k_0 \le k$,
    vertex sets $D_0 \subseteq X_t \setminus U$,
    integers $0 \le \ell_0 < \ell$,
    and
    functions $f_0\from \Types_t(D_0) \to \CodomCyc$,
    the following two statements are equivalent:
    \begin{itemize}%
      \item
      There is a set $D \subseteq V_t\setminus U$
      of class $(k_0, D_0, \ell_0, f_0)$ for $t_0$.

      \item
      $(k_0, D_0, \ell_0, f_0) \in \List(t_0)$.
    \end{itemize}
  \end{claim}
  \begin{claimproof}
    By the definition of the procedures from
    \cref{lem:twUpper:cycle:introVtxDel,lem:twUpper:cycle:introVtxNotDel,%
    lem:twUpper:cycle:introEdgeDel,lem:twUpper:cycle:introEdgeNotDel,%
    lem:twUpper:cycle:forgetDel,lem:twUpper:cycle:forgetNotDel,%
    lem:twUpper:cycle:join}
    and the definition of the list for the leaf nodes,
    it directly follows that the second statement implies the first one.

    Therefore, we only have to prove
    that the first statement implies the second statement.
    We prove this inductively
    based on the type of the node $t_0$ in the tree decomposition.

    \begin{description}
      \item[Leaf Node.]
      Since the bags of the leaf nodes do not contain any vertices,
      no vertices can be deleted.
      By assumption, the leaf nodes do not have children
      and hence, the only possible class is the empty class,
      which we included in the list.

      \item[Introduce Vertex Node.]
      If the set $D$ is of class $c_0$ for $t_0$ and $v \in D$,
      then $D\setminus \{v\}$ is of \emph{some} class $c_1$ for $t_1$.
      By induction we get $c_1 \in \List(t_1)$.
      From the algorithm and the properties of \IntroVtxNodeDel,
      we obtain a class $c$ for $t_0$ such that $D$ is of class $c$ for $t_0$.
      Moreover, the algorithm adds $c_0$ to $\List(t_0)$.
      Since each set $D$ has exactly one class for each node,
      we have $c_0 = c$ and thus, $c_0 \in \List(t_0)$.

      Note, that if $v \notin D$,
      then we similarly get that $c_0 \in \List(t_0)$
      by using \IntroVtxNodeNotDel.

      \item[Introduce Edge Node.]
      The correctness follows analogously to the introduce vertex node
      by distinguishing the two cases
      whether an endpoint of the introduced edge was deleted
      or both endpoints are not deleted.

      \item[Forget Node.]
      The result follows analogously to the introduce vertex node
      by splitting into the two cases of $v$ being deleted or not.

      \item[Join Node.]
      Assume that $D$ is of class $c_0$.
      If we consider $D' = D \cap V_{t_1}$ and $D'' = D \cap V_{t_2}$,
      we have that $D'$ and $D''$ are of some classes $c_1$ and $c_2$
      for $t_1$ and $t_2$, respectively.
      By the induction hypothesis it follows that
      $c_1 \in \List(t_1)$ and $c_2 \in \List(t_2)$.

      From the definition of the algorithm
      and the properties of the procedure \JoinNode,
      it follows that $D' \cup D'' = D$ is of some class $c$ for $t_0$.
      By the uniqueness of classes, we conclude that $c = c_0$
      and therefore, $c_0 \in \List(t_0)$.
      \claimqedhere
    \end{description}
  \end{claimproof}

  As the last step we prove the running time of the algorithm
  depending on the size of the lists.
  \begin{claim}
    \label{clm:twUpper:cycle:metaRuntime}
    Let $L$ denote the maximum length of a list $\List(t)$ for all nodes $t$.
    Then, the algorithm terminates in time
    $L^2 \cdot 2^{\poly(\tw)} \cdot \poly(n)$.
  \end{claim}
  \begin{claimproof}
    Computing a nice tree decomposition is possible in time
    $2^{\poly(\tw)} \cdot \poly(n)$ \cite[Chapter~7.6]{book1}.
    Observe that handling the join nodes dominates the running time
    as we have to consider up to $L^2$ different pairs of classes.
    Since each new class can be computed in time $2^{\poly(\tw)}$,
    the claim follows.
  \end{claimproof}

  From \cref{lem:twUpper:cycle:numberOfTypesAndClasses} we know that
  for each node of the tree decomposition,
  the number of distinct classes is at most
  $2^{2^{\Oh(\tw \log \tw)}} \cdot \poly(n)$.
  Hence, the final running time directly follows from
  \cref{clm:twUpper:cycle:metaRuntime}
  which concludes the proof by \cref{clm:twUpper:cycle:dpCorrectness}.
\end{proof}

\section{\texorpdfstring%
{\boldmath\SigTwoP-Completeness Results}
{Sigma2P-Completeness Results}}
\label{sec:sig2p-complete}

In this section we show the completeness of the \CycleUndelHitPack
problem
and the \UndelHitPack{H} problem for
the second level of the polynomial hierarchy.
Formally, we prove 
\cref{thm:sig2p:completeness:cycles,thm:sig2p:completeness:H}.
\ifab{TODO: Unresolved references.}

Towards proving these results, we first show in
\cref{lem:sig2p:membership}
that the problems are contained in \SigTwoP.
Then, by the proof of
\cref{thm:sig2p:hardness:triangle},
which is given in \cref{sec:sig2p:hardness:triangle},
we show the \SigTwoP-hardness of \TriUndelHitPack
which we extend to \CycleUndelHitPack.
In \cref{sec:sig2p:hardness:H} we lift the hardness of
\TriUndelHitPack
to the general \UndelHitPack{H} problem by proving
\cref{lem:tw:triangleToAnyGraph}.

The \SigTwoP-completeness results for \CycleUndelHitPack and \UndelHitPack{H} from
\cref{thm:sig2p:completeness:cycles,,thm:sig2p:completeness:H}
then follow from \cref{lem:sig2p:membership} together with the
hardness results
from \cref{sec:sig2p:hardness:triangle,,sec:sig2p:hardness:H},
respectively.

\begin{lemma}
    \label{lem:sig2p:membership}
    For any connected graph $H$,
    the problems \UndelHitPack{H} and \CycleUndelHitPack are in \SigTwoP.
\end{lemma}

\begin{proof}
    We observe that for a given graph $G$, the problem of deciding if
    there is a
    $H$-packing of size at least $\ell$ is in \NP, as the $H$-packing
    itself constitutes a polynomial witness. We assume oracle access
    to the
    problem.

    Next, we can solve \UndelHitPack{H} by guessing the deleted vertices
    non-deterministically and
    subsequently verifying that an $H$-packing of size at least
    $\ell$ does not
    exist in the remaining graph. This can be achieved with an
    \NP-oracle by
    querying it and negating the answer. Therefore, \UndelHitPack{H} can
    be solved
    in non-deterministic polynomial time using an $\NP$ oracle. Hence,
    $\UndelHitPack{H} \in \NP^{\NP} = \SigTwoP$. An analogous argument
    shows that
    $\CycleUndelHitPack \in \SigTwoP$.
\end{proof}

\subsection{\texorpdfstring%
{\boldmath Satisfiability Problems Complete for \SigTwoP}
{Satisfiability Problems Complete for Sigma2P}}
\label{sec:sig2p:hardness:sat}

In this section, we introduce some satisfiability problems
which are complete for the second level of the polynomial hierarchy.

We begin with the formal definition of the prototypical
\SigTwoP-complete problem, which is based on Boolean formulas in
disjoint normal form (DNF).

Note that in this paper, for a variable
$x$, we let $\lneg x$ denote the negation of $x$.
Moreover, we say
that $x$ is a positive literal and $\lneg x$ is a negative
literal.
Finally, note that a maximal group of disjoint literals in a DNF is
commonly referred to as a
\emph{term}.
For example, the $3$-DNF formula
$\psi = (x_1 \land x_2 \land \lneg x_3) \lor (x_4 \land x_5 \land \lneg x_1)$
has 5 variables, 6 literals and 2 terms.
One of the terms of $\psi$ is $(x_1 \land x_2 \land \lneg x_3)$.

\begin{definition}
    Let
    \QthreeDNFtwo denote the problem of deciding
    whether a
    given formula $\phi$ in disjunctive normal form with terms
    of length at most three on the variable set $X\cup Y$, where
    $X=\{x_1,\dots,x_m\}$ and $Y=\{y_1,\dots,y_n\}$ for some natural
    numbers $m$ and $n$, has an assignment to the variables in $X$ such
    that the remaining formula is a tautology.
    That is,
    \[
        \QthreeDNFtwo
        =\{\phi(x_1,\dots,x_m,y_1,\dots,y_n)\in\textup{3DNF}\mid
        \exists x_1,\dots,x_m\,\forall
        y_1,\dots,y_n\colon \phi=1\}
        .
    \]
\end{definition}

In his seminal work,
Stockmeyer~\cite{polynomialhierarchy} proved this problem to be
\SigTwoP-complete.

\begin{theorem}[{\cite[Theorem~4.1.2]{polynomialhierarchy}}]\label{thm:qdnf}
    The problem \QthreeDNFtwo is log-space \SigTwoP-complete.
\end{theorem}

We define the related problem where the underlying SAT-formula is in CNF.

\begin{definition}
    Denote by
    \QthreeCNFtwo the
    problem of deciding
    whether a
    given formula $\phi$ in conjunctive normal form with clauses of
    length at most three on the variable set $X\cup Y$, where
    $X=\{x_1,\dots,x_m\}$ and $Y=\{y_1,\dots,y_n\}$ for some natural
    numbers $m$ and $n$, has an assignment to the variables in $X$ such
    that remaining formula is unsatisfiable.
    That is,
    \[
        \QthreeCNFtwo=\{\phi(x_1,\dots,x_m,y_1,\dots,y_n)\in\textup{3CNF}\mid
        \exists
        x_1,\dots,x_m\,\nexists
        y_1,\dots,y_n\colon \phi=1\}
        .
    \]
\end{definition}

By using De Morgan's laws and basic properties of quantification
and their negation,
we can extend the \SigTwoP-completeness of \QthreeDNFtwo
from \cref{thm:qdnf} to \QthreeCNFtwo.

\begin{corollary}\label{lem:qcnf}
    The problem
    \QthreeCNFtwo is
    log-space
    \SigTwoP-complete.
\end{corollary}
\begin{proof}
    For any given 3DNF-formula $\phi(x_1,\dots,x_m,y_1,\dots,y_n)$, let
    $\psi$ be the 3CNF-formula obtained from $\phi$ by replacing every
    literal by its negation and swapping every $\vee$ for $\wedge$ and
    vice versa.

    This yields a
    log-space
    reduction from \QthreeDNFtwo to
    \QthreeCNFtwo.
    On the one hand, logarithmic space is clearly sufficient to
    transform $\phi$ into $\psi$.
    On the other hand, $\phi$ is equivalent to
    $\neg\psi(x_1,\dots,x_m,y_1,\dots,y_n)$,
    by the laws of De Morgan and distributivity.
    Therefore, the formula
    \begin{align*}
        \exists x_1,\dots,x_m\,\forall y_1,\dots,y_n\colon
        \phi(x_1,\dots,x_m,y_1,\dots,y_n)=1&
    \intertext{is equivalent to}
        \exists x_1,\dots,x_m\,\forall y_1,\dots,y_n\colon
        \neg\psi(x_1,\dots,x_m,y_1,\dots,y_n)=1&,
    \intertext{which is in turn equivalent to}
        \exists x_1,\dots,x_m\,\nexists y_1,\dots,y_n\colon
        \psi(x_1,\dots,x_m,y_1,\dots,y_n)=1&.
    \qedhere
    \end{align*}
\end{proof}

As a next step we first formally define the size of a formula.

\begin{definition}[Formula Size]
    The \emph{size} $|\phi|$ of a CNF-formula $\phi$ is the number
    of its
    clauses. We say that $\phi$ is smaller than another formula $\phi'$
    if $|\phi|<|\phi'|$.
\end{definition}

The last two problems we introduce
consider subformulas of a given CNF-formula.
For this reason we first formally introduce this intuitive concept.

\begin{definition}[Subformula]
    Given a CNF-formula $\phi=C_1\wedge\dots\wedge C_m$ consisting of
    $m$ clauses, we call any formula $\phi'$ that results from deleting
    an arbitrary subset of these clauses a \emph{subformula} of $\phi$
    and denote this by $\phi'\subseteq \phi$.
    If additionally $\phi'\neq\phi$, that is, the deleted subset is
    nonempty, the resulting subformula is called
    \emph{proper} and we write $\phi'\subsetneq \phi$.
    If $\psi$ is a (proper) subformula of $\psi'$, then $\psi'$ is a
    (proper) superformula of $\psi$.
\end{definition}

With this definition, we define the problems \SUS
and its corresponding version \ThreeCNFSUS for 3CNFs.

\begin{definition}[Smallest Unsatisfiable Subformula]
    \SUS is the problem of deciding, given a formula $\phi$ and
    an integer $k$, whether $\phi$ has an unsatisfiable subformula of
    size at most $k$.
    \ThreeCNFSUS is the problem restricted to CNF formulas
    with
    at most three literals per clause.
\end{definition}

Umans, Fortnow and Killian have reportedly proven
the \SigTwoP-completeness of \SUS
(calling the problem MIN DNF
TAUTOLOGY)~\cite[p.~35,~L7, MIN DNF TAUTOLOGY]{sigact-column-37}, but
only Umans published a proof, and only for the version that does not
limit the number of literals in a clause~\cite[Theorem~2]{Umans99}.
Several years later and apparently without knowledge of Uman's proof,
Liberatore~\cite[Theorem~2]{Liberatore05} reproved this statement by a
more direct reduction, but still without the restriction to
at most three literals per clause.
Although not strictly necessary for our purposes, we extend this
result by providing the missing proof of
the \SigTwoP-hardness for \ThreeCNFSUS.

\begin{restatable}{theorem}{ThreeSUSCompleteness}\label{thm:3-sus-complete}
    \textup{3CNF-SUS} is \SigTwoP-complete.
\end{restatable}

\begin{proof}
    For membership in \SigTwoP, guess a subformula of size $k$
    and verify its
    unsatisfiability with the \NP-oracle.

    We now
    reduce \QthreeCNFtwo
    to \ThreeCNFSUS.
    Given an instance
    $\phi(x_1,\dots,x_m,y_1,\dots,y_n)
    =C_1\wedge\dots\wedge C_t$
    with $t$ clauses
    for the problem \QthreeCNFtwo,
    that is normalized
    to contain exactly three
    literals per
    clause by duplicating literals if necessary.
    First, replace every
    literal that is a positive occurrence of $x_i$
    by a negative
    occurrence of the new variable $z_i$, and write
    the $t$
    clauses of the resulting formula as
    $(\lambda_1^1\lor\lambda_1^2\lor\lambda_1^3)\land
    \dots\land(\lambda_t^1\lor\lambda_t^2\lor\lambda_t^3)$.
    Note that any of these literals might be a negative or
    positive occurrence of $y_j$ for any $j\in\numb n$ since we
    did not modify these variables at all.
    The instance for \ThreeCNFSUS is
    $(\psi,k)$, where, using the
    definition $s=\lceil\log m\rceil$, we have
    $%
    k=2^{s+1}-(m+1)+(2t+2)m+2t\\%
	$
    and
    $\psi=W\wedge
    V\wedge \bigwedge_{i=1}^m(X_i\wedge Z_i)\wedge \bigwedge_{j=1}^t
    \Gamma_j$ with
    \begin{align*}
        W&=\bigwedge_{i=1}^{2^s-1}(\lneg w_{2i}\vee
            \lneg
            w_{2i+1}\vee w_i),\\
        V&=\bigwedge_{i=2^s}^{2^{s+1}-(m+1)}(w_i),\\
        X_i&=(x_i^{2t})
        \wedge(\lneg x_i^1\vee x_i)
        \wedge(\lneg x_i\vee w_{2^{s+1}-i})
        \wedge\bigwedge_{j=1}^{2t-1}(\lneg x_i^{j+1}\vee x_i^j)
        ,\\
        Z_i&=(z_i^{2t})
        \wedge(\lneg z_i^1\vee z_i)
        \wedge(\lneg z_i\vee w_{2^{s+1}-i})
        \wedge\bigwedge_{j=1}^{2t-1}(\lneg z_i^{j+1}\vee z_i^j)
        ,\text{ and}\\
        \Gamma_i&=
            (\lneg w_1\vee \lambda_i^1\vee u_i)
            \wedge
            (\lneg u_i\vee \lambda_i^2\vee \lambda_i^3).
    \end{align*}
\begin{figure}[th!]
\centering
   \begin{tikzpicture}[baseline={(0,0)},
upperbrace/.style={decoration={calligraphic 
brace,amplitude=5pt,raise=5.5pt,aspect=.5},line 
width=1.2pt,xshift=2pt,decorate},
underbrace/.style={decoration={calligraphic 
brace,amplitude=5pt,mirror,raise=5.5pt,aspect=.5},line 
width=1.2pt,xshift=2pt,decorate},
simpbrace/.style={decoration={brace,amplitude=5pt,mirror,raise=5.5pt,aspect=.5},line
 width=1.2pt,xshift=2pt,decorate},
 implies/.style={-{Straight Barb[scale=.8]},line width=.6pt,},
]
\newcommand{\upperbraced}[3]{
\node[anchor=west] (v#1) at (#3) {#2};
\draw[upperbrace] (#3) -- node(n#1)[above=3.5pt] {}  +(1.17,0);
}
\newcommand{\underbraced}[3]{
\node[anchor=west] (v#1) at (#3) {#2};
\draw[underbrace] (#3) -- node(n#1)[below=3pt] {}  +(1.17,0);
}
\newcommand{\underbracedw}[3]{
\node[anchor=west] (v#1) at (#3) {#2};
\draw[underbrace] (#3) -- node(n#1)[below=3pt] {}  +(1.36,0);
}
\newcommand{\underbracedww}[3]{
\node[anchor=west] (v#1) at (#3) {#2};
\draw[underbrace] (#3) -- node(n#1)[below=3pt] {}  +(1.66,0);
}
\newcommand{\underbracedwww}[3]{
\node[anchor=west] (v#1) at (#3) {#2};
\draw[underbrace] (#3) -- node(n#1)[below=3pt] {}  +(1.9,0);
}
\newcommand{\underbracedwwww}[3]{
\node[anchor=west] (v#1) at (#3) {#2};
\draw[underbrace] (#3) -- node(n#1)[below=3pt] {}  +(2.15,0);
}
\newcommand{\underbracedwwwww}[3]{
\node[anchor=west] (v#1) at (#3) {#2};
\draw[underbrace] (#3) -- node(n#1)[below=3pt] {}  +(2.75,0);
}
\newcommand{\underbracedwwwwww}[3]{
\node[anchor=west] (v#1) at (#3) {#2};
\draw[underbrace] (#3) -- node(n#1)[below=3pt] {}  +(3,0);
}

\newcommand{\fromright}[2]{
\draw[implies]
(n#1)[shift={(2ex,-2.7ex)}]edge[transform 
canvas={xshift=-2.5pt,yshift=1pt}](n#2);}
\newcommand{\fromleft}[2]{
\draw[implies] (n#1)[shift={(-2ex,-2.7ex)}]edge[transform
canvas={xshift=2.5pt,yshift=1pt}](n#2);}
\newcommand{\toright}[2]{
\draw[implies] (n#1)edge[transform
canvas={xshift=2.5pt,yshift=1pt}]([shift={(-2ex,3.3ex)}]n#2);}
\newcommand{\toleft}[2]{
\draw[implies] (n#1)edge[transform
canvas={xshift=-2.5pt,yshift=1pt}]([shift={(2ex,3.3ex)}]n#2);}

\newcommand{\torightw}[2]{
\draw[implies] (n#1)edge[transform
canvas={xshift=2.8pt,yshift=-2pt}]([shift={(2.2ex,3.3ex)}]n#2);}
\newcommand{\toleftw}[2]{
\draw[implies] (n#1)edge[transform
canvas={xshift=-2.8pt,yshift=-2pt}]([shift={(-2.2ex,3.3ex)}]n#2);}

\newcommand{\torightww}[2]{
\draw[implies] (n#1)edge[transform
canvas={xshift=2.8pt,yshift=-2pt}]([shift={(2.2ex,3.3ex)}]n#2);}
\newcommand{\toleftww}[2]{
\draw[implies] (n#1)edge[transform
canvas={xshift=-2.8pt,yshift=-2pt}]([shift={(-2.2ex,3.3ex)}]n#2);}

\upperbraced{1}{$\lambda_1^1\vee u_1$}{-5,-2}
\upperbraced{2}{$\lambda_1^2\vee \lambda_1^3$}{-5,-3.0}
\fromright 12;
\upperbraced{3}{$\lambda_2^1\vee u_2$}{-3,-2}
\upperbraced{4}{$\lambda_2^2\vee \lambda_2^3$}{-3,-3.0}
\fromright 34;
\upperbraced{5}{$\lambda_3^1\vee u_3$}{-1,-2}
\upperbraced{6}{$\lambda_3^2\vee \lambda_3^3$}{-1,-3.0}
\fromright 56;
\upperbraced{7}{$\lambda_4^1\vee u_4$}{1,-2}
\upperbraced{8}{$\lambda_4^2\vee \lambda_4^3$}{1,-3.0}
\fromright 78;
\upperbraced{t1}{$\lambda_t^1\vee u_t$}{5,-2}
\upperbraced{t2}{$\lambda_t^2\vee \lambda_t^3$}{5,-3.0}
\fromright{t1}{t2};

\node[outer sep=0pt] (w1) at (0,0) {$w_1$};

\draw[implies] (w1)edge[out=180,in=90](n1);
\draw[implies] (w1)edge[out=195,in=90](n3);
\draw[implies] (w1)edge[out=250,in=90](n5);
\draw[implies] (w1)edge[out=340,in=90](n7);
\draw[implies] (w1)edge[out=0,in=90](nt1);

\underbracedw{w2}{$w_2\wedge w_3$}{0-.8,1}
\draw[implies] (nw2) edge ([shift={(-.25ex,1ex)}]w1);
\underbracedw{w4}{$w_4\wedge w_5$}{-.8-.8,2}
\underbracedw{w6}{$w_6\wedge w_7$}{.8-.8,2}
\toleftw{w4}{w2};
\torightw{w6}{w2};
\underbracedw{w8}{$w_8\wedge w_9$}{-2.65-.8,3}
\underbracedww{w10}{$w_{10}\wedge w_{11}$}{-1.05-.8,3}
\underbracedww{w12}{$w_{12}\wedge w_{13}$}{.9-.8,3}
\underbracedww{w14}{$w_{14}\wedge w_{15}$}{2.8-.8,3}

\draw[line width=1.5pt, line cap=round, dash pattern=on 
0pt off 2.6\pgflinewidth] (-4.5,3.9)--(-4.2,3.75);
\draw[line width=1.5pt, line cap=round, dash pattern=on 
0pt off 2.6\pgflinewidth] (-2.1,4)--(-1.8,3.75);
\draw[line width=1.5pt, line cap=round, dash pattern=on 
0pt off 2.6\pgflinewidth] (0,4.1)--(0,3.75);
\draw[line width=1.5pt, line cap=round, dash pattern=on 
0pt off 2.6\pgflinewidth] (2.1,4)--(1.8,3.75);
\draw[line width=1.5pt, line cap=round, dash pattern=on 
0pt off 2.6\pgflinewidth] (4.5,3.85)--(4.2,3.75);

\draw[line width=1.5pt, line cap=round, dash pattern=on 
0pt off 2.6\pgflinewidth] (3.4,-1.9)--(3.7,-1.9);

\draw[line width=1.5pt, line cap=round, dash pattern=on 
0pt off 2.6\pgflinewidth] (-3.7,5.55)--(-3.4,5.55);
\draw[line width=1.5pt, line cap=round, dash pattern=on 
0pt off 2.6\pgflinewidth] (-3.7,4.97)--(-3.4,4.97);

\draw[line width=1.5pt, line cap=round, dash pattern=on 
0pt off 2.6\pgflinewidth] (2.7,9.5)--(3.0,9.5);
\draw[line width=1.5pt, line cap=round, dash pattern=on 
0pt off 2.6\pgflinewidth] (2.7,4.97)--(3.0,4.97);

\node () at (-4.9,-.9) {$\Gamma_1$};
\node () at (-2.9,-.9) {$\Gamma_2$};
\node () at (-0.9,-.9) {$\Gamma_3$};
\node () at (2,-.9) {$\Gamma_4$};
\node () at (6,-.9) {$\Gamma_t$};

\node () at (-6,3.5) {$W$};

\node () at (-6.6,6.4) {$V$};

\node () at (0.6,13.4) {$X_i\cup Z_i$};
\node () at (4.6,13.4) {$X_2\cup Z_2$};
\node () at (7.1,13.4) {$X_1\cup Z_1$};

\draw[implies] (-2.6,2.58)--(-1.42,2.16);
\draw[implies] (-0.9,2.58)--(-0.43,2.16);
\draw[implies] (0.95,2.58)--(0.45,2.16);
\draw[implies] (2.8,2.58)--(1.4,2.16);
\underbracedwww{ws1}{$w_{2^s}\wedge w_{2^s+1}$}{-6-.8,5}
\underbracedwwwwww{ws3}{$w_{2^{s+1}-i-1}\wedge 
w_{2^{s+1}-i}$}{-1.2-.8,5}
\underbracedwwwww{ws4}{$w_{2^{s+1}-2}\wedge w_{2^{s+1}-1}$}{5.64-.8,5}

\draw[implies] (-6.5,5.8)--(-6.5,5.3);
\draw[implies] (-5.55,5.8)--(-5.55,5.3);
\draw[implies] (-1.72,5.8)--(-1.72,5.3);

\node[outer sep=10pt,anchor=west] (xi) at (-.5,6) {$x_i$};
\node[outer sep=10pt,anchor=west] (xi1) at (-.5,7) {$x_i^1$};
\node[outer sep=10pt,anchor=west] (xi2) at (-.5,8) {$x_i^2$};
\node[outer sep=10pt,anchor=west] (xi3) at (-.5,9) {$x_i^3$};
\node[outer sep=10pt,anchor=west] (xi2t1) at (-.5,11) {$x_i^{2t-1}$};
\node[outer sep=10pt,anchor=west] (xi2t) at (-.5,12) {$x_i^{2t}$};
\node[outer sep=10pt,anchor=west] (zi) at (.5,6) {$z_i$};
\node[outer sep=10pt,anchor=west] (zi1) at (.5,7) {$z_i^1$};
\node[outer sep=10pt,anchor=west] (zi2) at (.5,8) {$z_i^2$};
\node[outer sep=10pt,anchor=west] (zi3) at (.5,9) {$z_i^3$};
\node[outer sep=10pt,anchor=west] (zi2t1) at (.5,11) {$z_i^{2t-1}$};
\node[outer sep=10pt,anchor=west] (zi2t) at (.5,12) {$z_i^{2t}$};

\draw[implies] (.08,12.75)--(.08,12.25);
\draw[implies] (.08,11.75)--(.08,11.25);
\draw[implies] (.08,10.75)--(.08,10.25);
\draw[line width=1.5pt, line cap=round, dash pattern=on 
0pt off 2.6\pgflinewidth] (.08,9.85)--(.08,9.45);
\draw[implies] (.08,8.75)--(.08,8.25);
\draw[implies] (.08,7.75)--(.08,7.25);
\draw[implies] (.08,6.75)--(.08,6.25);
\draw[implies] (.08,5.75)--(.08,5.25);

\draw[implies] (1.08,12.75)--(1.08,12.25);
\draw[implies] (1.08,11.75)--(1.08,11.25);
\draw[implies] (1.08,10.75)--(1.08,10.25);
\draw[line width=1.5pt, line cap=round, dash pattern=on 
0pt off 2.6\pgflinewidth] (1.08,9.85)--(1.08,9.45);
\draw[implies] (1.08,8.75)--(1.08,8.25);
\draw[implies] (1.08,7.75)--(1.08,7.25);
\draw[implies] (1.08,6.75)--(1.08,6.25);
\draw[implies] (1-.08,5.75+.08)--(0.25,5.25);

\node[outer sep=10pt,anchor=west] (x2) at (3.5,6) {$x_2$};
\node[outer sep=10pt,anchor=west] (x21) at (3.5,7) {$x_2^1$};
\node[outer sep=10pt,anchor=west] (x22) at (3.5,8) {$x_2^2$};
\node[outer sep=10pt,anchor=west] (x23) at (3.5,9) {$x_3^2$};
\node[outer sep=10pt,anchor=west] (x22t1) at (3.5,11) {$x_2^{2t-1}$};
\node[outer sep=10pt,anchor=west] (x22t) at (3.5,12) {$x_2^{2t}$};
\node[outer sep=10pt,anchor=west] (z2) at (4.5,6) {$z_2$};
\node[outer sep=10pt,anchor=west] (z21) at (4.5,7) {$z_2^1$};
\node[outer sep=10pt,anchor=west] (z22) at (4.5,8) {$z_2^2$};
\node[outer sep=10pt,anchor=west] (z23) at (4.5,9) {$z_2^3$};
\node[outer sep=10pt,anchor=west] (z22t1) at (4.5,11) {$z_2^{2t-1}$};
\node[outer sep=10pt,anchor=west] (z22t) at (4.5,12) {$z_2^{2t}$};

\draw[implies] (4.08,12.75)--(4.08,12.25);
\draw[implies] (4.08,11.75)--(4.08,11.25);
\draw[implies] (4.08,10.75)--(4.08,10.25);
\draw[line width=1.5pt, line cap=round, dash pattern=on 
0pt off 2.6\pgflinewidth] (4.08,9.85)--(4.08,9.45);
\draw[implies] (4.08,8.75)--(4.08,8.25);
\draw[implies] (4.08,7.75)--(4.08,7.25);
\draw[implies] (4.08,6.75)--(4.08,6.25);
\draw[implies] (4.08+.08,5.75)--(5-.08,5.25);

\draw[implies] (5.08,12.75)--(5.08,12.25);
\draw[implies] (5.08,11.75)--(5.08,11.25);
\draw[implies] (5.08,10.75)--(5.08,10.25);
\draw[line width=1.5pt, line cap=round, dash pattern=on 
0pt off 2.6\pgflinewidth] (5.08,9.85)--(5.08,9.45);
\draw[implies] (5.08,8.75)--(5.08,8.25);
\draw[implies] (5.08,7.75)--(5.08,7.25);
\draw[implies] (5.08,6.75)--(5.08,6.25);
\draw[implies] (5.08,5.75)--(5.08,5.25);

\node[outer sep=10pt,anchor=west] (x1) at (6,6) {$x_1$};
\node[outer sep=10pt,anchor=west] (x11) at (6,7) {$x_1^1$};
\node[outer sep=10pt,anchor=west] (x12) at (6,8) {$x_1^2$};
\node[outer sep=10pt,anchor=west] (x13) at (6,9) {$x_1^3$};
\node[outer sep=10pt,anchor=west] (x12t1) at (6,11) {$x_1^{2t-1}$};
\node[outer sep=10pt,anchor=west] (x12t) at (6,12) {$x_1^{2t}$};
\node[outer sep=10pt,anchor=west] (z1) at (7,6) {$z_1$};
\node[outer sep=10pt,anchor=west] (z11) at (7,7) {$z_1^1$};
\node[outer sep=10pt,anchor=west] (z12) at (7,8) {$z_1^2$};
\node[outer sep=10pt,anchor=west] (z13) at (7,9) {$z_1^3$};
\node[outer sep=10pt,anchor=west] (z12t1) at (7,11) {$z_1^{2t-1}$};
\node[outer sep=10pt,anchor=west] (z12t) at (7,12) {$z_1^{2t}$};

\draw[implies] (6.58,12.75)--(6.58,12.25);
\draw[implies] (6.58,11.75)--(6.58,11.25);
\draw[implies] (6.58,10.75)--(6.58,10.25);
\draw[line width=1.5pt, line cap=round, dash pattern=on 
0pt off 2.6\pgflinewidth] (6.58,9.85)--(6.58,9.45);
\draw[implies] (6.58,8.75)--(6.58,8.25);
\draw[implies] (6.58,7.75)--(6.58,7.25);
\draw[implies] (6.58,6.75)--(6.58,6.25);
\draw[implies] (6.58,5.75)--(6.58,5.25);

\draw[implies] (7.58,12.75)--(7.58,12.25);
\draw[implies] (7.58,11.75)--(7.58,11.25);
\draw[implies] (7.58,10.75)--(7.58,10.25);
\draw[line width=1.5pt, line cap=round, dash pattern=on 
0pt off 2.6\pgflinewidth] (7.58,9.85)--(7.58,9.45);
\draw[implies] (7.58,8.75)--(7.58,8.25);
\draw[implies] (7.58,7.75)--(7.58,7.25);
\draw[implies] (7.58,6.75)--(7.58,6.25);
\draw[implies] (7.58-.08,5.75+.08)--(6.75,5.25);

\draw[dashed, rounded corners] (-5.2, -1.2) rectangle (-3.5,-3.5) {};
\draw[dashed, rounded corners] (-3.2, -1.2) rectangle (-1.5,-3.5) {};
\draw[dashed, rounded corners] (-1.2, -1.2) rectangle (+0.5,-3.5) {};
\draw[dashed, rounded corners] (+0.8, -1.2) rectangle (+2.5,-3.5) {};
\draw[dashed, rounded corners] (+4.8, -1.2) rectangle (+6.5,-3.5) {};

\draw[dashed, rounded corners] (-0.8, 5.18) rectangle (1.9,13);
\draw[dashed, rounded corners] (3.65, 5.18) rectangle (5.9,13);
\draw[dashed, rounded corners] (6.3, 5.18) rectangle (8.35,13);

\draw[dashed, rounded corners] (-7.0, 5.18) rectangle (-1.2,6.0);

\draw[dashed, rounded corners=20] 
(-6.8,4.6)--(7.7,4.6)--(-.1,-0.08)--cycle;

\end{tikzpicture}
\caption{Illustration of the implications encoded by the different 
sets of clauses used in the reduction of the proof of 
Theorem~\ref{thm:3-sus-complete}. Note that the clauses correspond to 
the arrows in the illustration, not the depicted variables. For
example, the clause $(\lneg w_2\vee \lneg w_3\vee w_1)$ is equivalent 
to the implication $(w_2\wedge w_3)\rightarrow w_1$ represented by 
the bottom-most arrow of $W$ in the illustration, the left-most arrow 
right below it represent implication $w_1\rightarrow (\lambda_1^1\vee
u_1)$ equivalent to the clause $(\lneg w_1\vee\lambda_1^1\vee u_1)$, and
the clause $(z_1^{2t})$ corresponds to arrow at the top right 
starting from nowhere, which is interpreted as an implication with a 
true antecedent, equivalent to the clause $(z_1^{2t})$.}
\label{fig:sus-implication-construction}
\end{figure}
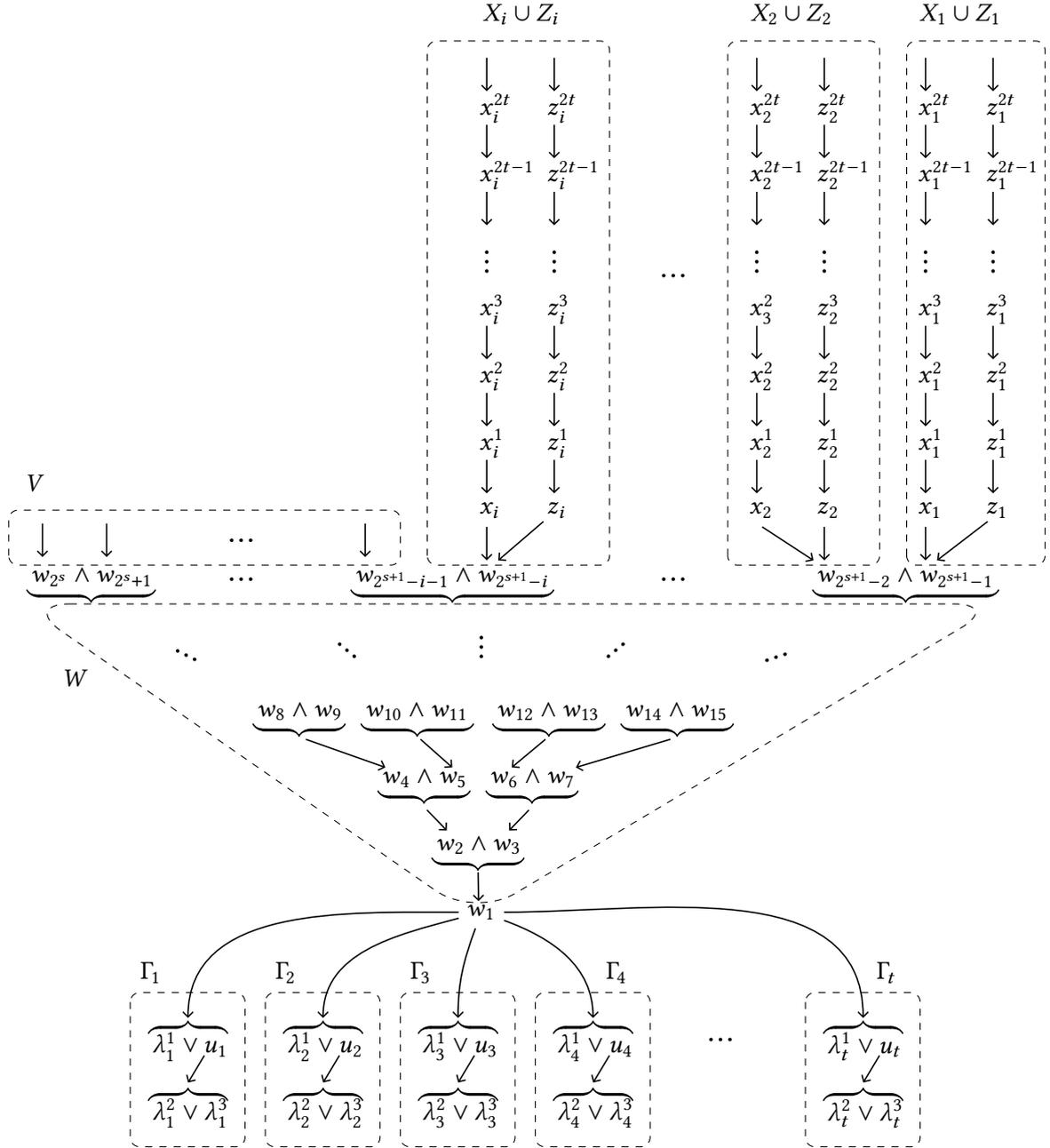

    This construction is clearly possible in logarithmic space. Note
    that $\psi$ contains exactly
    \[
        2^{s+1}-(m+1)
        + 2m(3 + 2t-1)
        + 2t
        =k+(2t+2)m
    \]
    clauses. Each of these clauses can be interpreted as an 
    implication as illustrated in 
    Figure~\ref{fig:sus-implication-construction}.
    We now prove the correctness of the reduction.
    Intuitively, all clauses of $\psi$ are selected
    with the only exception that, for each $i \in \numb{m}$,
    either $X_i$ or $Z_i$ is not included entirely.

    Observe that setting all variables to $1$ satisfies all clauses
    from $W\wedge V\wedge \bigwedge_{i=1}^m(X_i\wedge Z_i)$.
    It follows that any unsatisfiable formula includes a nonempty
    set of clauses from $\bigwedge_{j=1}^t\Gamma_i$.
    Moreover, since assigning $0$ to both $w_1$ and all $u_i$
    satisfies all of $\bigwedge_{j=1}^t\Gamma_i$, and $w_1$ is the
    only variable shared with the remaining clauses $W\wedge V\wedge
    \bigwedge_{i=1}^t(X_i\wedge Z_i)$, an unsatisfiable subformula
    also includes
    a subset of clauses from
    $W\wedge \bigwedge_{i=1}^m(X_i\wedge
    Z_i)$ such that any assignment satisfying this subset is
    forced to assign $1$ to $w_1$.
    It follows that an unsatisfiable subformula
    includes all of $W$. This
    is because $W$ contains clauses that represent a
    binary tree of implications of the form
    $w_{2i}\wedge w_{2i+1}\rightarrow w_i$ that needs to be
    included in full to force any assignment satisfying all selected
    clauses from $W\wedge V\wedge \bigwedge_{i=0}^m(X_i\wedge Z_i)$ to
    set $w_1$ at the root of the implication tree to $1$.
    Furthermore, any unsatisfiable subformula contains some clauses
    that forces any assignment satisfying them to
    set all the leaves of the mentioned binary tree, namely
    $w_{2^s},\dots,w_{2^{s+1}}$, to be set to $1$.
    Including the set $V$ of singleton clauses
    does this for all but the last $m$ leaves. For each remaining
    leaf $w_{2^{s+1}-1}$, either $X_i$ and $Z_i$ must be included.

    We have seen so far that any unsatisfiable
    subformula of $\psi$
    must contain a clause from $\bigwedge_{j=1}^t\Gamma_i$,
    all of $W\wedge V$, and, for each pair of implication
    chains $X_i$ and $Z_i$, one chain in its entirety.
    Including up to $2t$ clauses of the other chain of a pair cannot
    affect the satisfiability because only the last two of the $2t+2$
    clauses contain variables occurring outside of the chain, namely
    $x_i$ and $w_{2^{s+1}-i}$, and including such a clause and
    activating it via uninterrupted implication chain from the only
    unconditional clause at the head of the chain, either
    $(x_i^r)$ or $(z_i^r)$, requires at least $2t+1$ clauses.

    In addition to the $2^{s+1}-(m+1)+(2t+2)m$ clauses
    from $W\wedge V$ and,
    for each $i\in\numb m$, either $X_i$ or
    $Z_i$, we may
    include exactly $2t$ more clauses until the threshold
    of $k$ is reached,
    which exactly suffices to include all of
    $\bigwedge_{j=1}^t\Gamma_t$ but not any pertinent part of an
    additional implication chain, which would require at least
    $2t+1$ clauses.
    Therefore, $\psi$ has an unsatisfiable subformula
    of size at most $k$
    if and only if its subformula containing all
    clauses except, for
    each $i\in\numb m$, either $X_i$ or
    $Z_i$, is unsatisfiable.
    Any assignment that satisfies all clauses from $X_i$ or $Z_i$
    must assign $1$ to $x_i$ or $z_i$, respectively, while the other
    variable can
    be set to $0$ to satisfy as many clauses as possible from
    $\bigwedge_{i=1}^t\Gamma_i$. This corresponds
    to a consistent
    assignment of either $0$ or $1$ to each
    variable from
    $x_1,\dots,x_m$, whereas assignments to
    $y_1,\dots,y_n$ are not
    restricted in any way.

    Any assignment satisfying all included clauses except possibly
    those of $\bigwedge_{i=1}^t\Gamma_i$ assigns $1$ to $w_1$. Such
    an assignment satisfies the
    first clause of $\Gamma_i$ if
    its sets $\lambda_i^1$ or $u_i$ to $1$. If $u_i$ is set to $1$, then
    satisfying the
    second clause of $\Gamma_i$ requires $\lambda_i^2$ or $\lambda_i^3$
    to be set to $1$.
    For any assignment setting $w_1$ to $1$, two clauses of
    $\Gamma_i$ are thus equivalent to $C_i$ for each $i\in\numb m$
    and
    $\bigwedge_{i=1}^t\Gamma_i$ is equivalent to
    $\bigwedge_{i=1}^tC_i=\phi(x_1,\dots,x_m,y_1,\dots,y_n)$.

    This shows that $\phi$ is in
    \QthreeCNFtwo if and only if $\psi$ has an
    unsatisfiable
    subformula of size at most $k$.

    Note that we do not know the exact minimum size
    of an unsatisfiable
    subformula since not all clauses from
    $\bigwedge_{j=1}^t\Gamma_j$
    might be needed. But the threshold $k$ is sufficiently large
    for all of them to be included in any case.
\end{proof}

\subsection{\texorpdfstring%
{\boldmath\SigTwoP-Hardness for \TriUndelHitPack and \CycleUndelHitPack}
{Sigma2P-Hardness for Triangle-HitPack and Cycle-HitPack}}
\label{sec:sig2p:hardness:triangle}

We use the \SigTwoP-hardness of \SUS from \cref{thm:3-sus-complete}
to prove the \SigTwoP-hardness of \TriUndelHitPack.
The completeness for \SigTwoP
then directly follows from \cref{lem:sig2p:membership}.

\begin{theorem}\label{thm:sig2p:hardness:triangle}
    The problem \TriUndelHitPack is
    \SigTwoP-hard on tripartite graphs.
\end{theorem}
\begin{proof}
    We reduce from \SUS by mapping any given pair $(\phi,k)$,
    where
    $\phi=C_1\land\dots\land C_m$ is a formula with $m$ clauses on the
    variable set $\{x_1,\dots,x_n\}$, to an instance $(G,U,k,\ell)$ of \TriUndelHitPack
    as described
    below. The entire construction is
    shown in \cref{fig:triangle-hitpack} for a simple example.

    \subparagraph*{Construction.}
    The graph $G$ contains, for every variable $x_i$, a \emph{triangle
    cycle} of $2m$
    triangles, each conjoined to the next at one vertex.
    Formally, we have four sets of $m$ vertices each,
    \begin{equation*}
            V_i = \{      v_i^1, \dots,       v_i^m \}, \quad
        \lneg V_i = \{\lneg v_i^1, \dots, \lneg v_i^m \}, \quad
            W_i = \{      w_i^1, \dots,       w_i^m \}, \quad \text{ and } \quad
        \lneg W_i = \{\lneg w_i^1, \dots, \lneg w_i^m \},
    \end{equation*}
    and the edges of the triangles $\{w_i^j,v_i^j,\lneg w_i^j\}$ and
    $\{\lneg w_i^j,\lneg v_i^j,w_i^{j+1}\}$
    for every $j\in\numb{m}$%
    , identifying $w_i^{m+1}$ with $w_i^1$ for convenience.

    For every clause $C_j$, we add a \emph{clause triangle}
    $\{a_j,b_j,c_j\}$.
    Moreover,
    for every literal in this clause, we add two \emph{literal edges} from $a_j$
    and
    $b_j$ to a \emph{literal vertex}, namely $v_i^j$ if the literal is
    $x_i$ and $\lneg v_i^j$ if the
    literal is $\lneg x_i$ for some $i \in \numb{n}$.
    We call the vertices of this clause gadget the \emph{clause vertices}.

    Finally, we choose
    $\ell=mn+m$ as the number of triangles that should not be possible to pack
    and
    $U=V(G)\setminus\{c_1,\dots,c_m\}$ as the set of undeletable
    vertices.
    The number of vertices we are allowed to delete is $k$.
    This finishes the construction of the instance $(G, U, k, \ell)$.

    \subparagraph*{Correctness.}
    We now prove the correctness of the reduction.
    A maximum triangle packing of
    the graph $G$ consists, without loss of
    generality, of the $m$ clause triangles and every second triangle of
    each triangle cycle; these are $m+n\cdot m$ triangles in total.
    These triangles cover all vertices except, for each $i$, either $V_i$
    or $\lneg V_i$. We interpret $V_i$ being uncovered as assigning
    $1$ to $x_i$
    and $\lneg V_i$ being uncovered as assigning $0$ to $x_i$.

    Deleting $c_j$ means that we cannot keep the $j$th clause triangle in
    a maximum packing. But we can keep the edge $\{a_j,b_j\}$ and replace
    $c_j$ by any of its literal vertices unless they are already covered.
    Since either $V_i$ or $\lneg V_i$ is always covered, it follows
    that we can still pack $mn+m$ triangles into the graph from which a
    set $D$ of vertices has been deleted
    if and only if the \emph{corresponding subformula} consisting of the
    clauses $\{C_j\mid c_j\in D\}$ is satisfiable;
    otherwise, it is impossible to pack this many triangles.
    It follows immediately that from $G$ we can obtain a graph without a
    triangle packing of size $k=mn+m$ by deleting $k$ deletable
    vertices
    if and only if $\phi$ has an unsatisfiable subformula of size
    at most $k$.
\end{proof}

\begin{figure}[t]
    \centering
    \resizebox{\textwidth}{!}{%
    \begin{tikzpicture}[x=0.8cm,y=0.8cm,%
        packing/.style = {draw, ultra thick, fill=blue!20},
        withBorder/.style={draw=white,double=black},
        vertex/.style={draw,very
        thick,circle,
        fill=white,
        inner sep=.0pt, minimum
        size=0.62cm},font=\large,scale=1.03]
        \begin{scope}[xshift=-5cm,yshift=1.1cm,rotate=-30]
            \node[vertex] (w11) at (-2,0) {$w_1^1$};
            \node[vertex] (v11) at (-1.5,1) {$v_1^1$};
            \node[vertex] (ow11) at (-1,0) {$\lneg w_1^1$};
            \node[vertex,thin] (ov11) at (-.5,1) {$\lneg v_1^1$};
            \node[vertex] (w12) at (0,0) {$w_1^2$};
            \node[vertex] (v12) at (0.5,1) {$v_1^2$};
            \node[vertex] (ow12) at (1,0) {$\lneg w_1^2$};
            \node[vertex,thin] (ov12) at (1.5,1) {$\lneg v_1^2$};
            \node[vertex] (w13) at (2,0) {$w_1^3$};
            \draw
            (w11)--(v11)--(ow11)--(ov11)--(w12)--(v12)--(ow12)--(ov12)--(w13);
            \draw (w11)--(ow11)--(w12)--(ow12)--(w13);
            \draw (w11)--(ow11)--(w12)--(ow12)--(w13);
            \draw[dashed, rounded corners] (-2.5, -.5) rectangle
            (2.5,1.6) {};
            \node at (0,-1) {$x_1$};
        \end{scope}
        \begin{scope}[xshift=0cm,yshift=0cm,]
            \node[vertex] (w21) at (-2,0) {$w_2^1$};
            \node[vertex] (v21) at (-1.5,1) {$v_2^1$};
            \node[vertex] (ow21) at (-1,0) {$\lneg w_2^1$};
            \node[vertex] (ov21) at (-.5,1) {$\lneg v_2^1$};
            \node[vertex] (w22) at (0,0) {$w_2^2$};
            \node[vertex,thin] (v22) at (0.5,1) {$v_2^2$};
            \node[vertex] (ow22) at (1,0) {$\lneg w_2^2$};
            \node[vertex] (ov22) at (1.5,1) {$\lneg v_2^2$};
            \node[vertex] (w23) at (2,0) {$w_2^3$};
            \draw
            (w21)--(v21)--(ow21)--(ov21)--(w22)--(v22)--(ow22)--(ov22)--(w23);
            \draw (w21)--(ow21)--(w22)--(ow22)--(w23);
            \draw[dashed, rounded corners] (-2.5, -.5) rectangle
            (2.5,1.6) {};
            \node at (0,-1) {$x_2$};
        \end{scope}
        \begin{scope}[xshift=5cm,yshift=1.1cm,rotate=30]
            \node[vertex] (w31) at (-2,0) {$w_3^1$};
            \node[vertex,thin] (v31) at (-1.5,1) {$v_3^1$};
            \node[vertex] (ow31) at (-1,0) {$\lneg w_3^1$};
            \node[vertex] (ov31) at (-.5,1) {$\lneg v_3^1$};
            \node[vertex] (w32) at (0,0) {$w_3^2$};
            \node[vertex] (v32) at (0.5,1) {$v_3^2$};
            \node[vertex] (ow32) at (1,0) {$\lneg w_3^2$};
            \node[vertex] (ov32) at (1.5,1) {$\lneg v_3^2$};
            \node[vertex] (w33) at (2,0) {$w_3^3$};
            \draw
            (w31)--(v31)--(ow31)--(ov31)--(w32)--(v32)--(ow32)--(ov32)--(w33);
            \draw (w31)--(ow31)--(w32)--(ow32)--(w33);
            \draw[dashed, rounded corners] (-2.5, -.5) rectangle
            (2.5,1.6) {};
            \node at (0,-1) {$x_3$};
        \end{scope}
        \begin{scope}[xshift=-2.4cm,yshift=1cm,rotate=-0]
            \node[vertex] (a1) at (-.5,3) {$a_1$};
            \node[vertex] (b1) at (.5,3) {$b_1$};
            \node[vertex,thin] (c1) at (0,4) {$c_1$};
            \draw (a1) -- (b1) -- (c1) -- (a1);

            \draw[dashed, rounded corners] (-1.1, 2.4) rectangle
            (1.1,4.6) {};
            \node at (-1.6,4) {$C_1$};
        \end{scope}
        \draw (a1)edge[bend right=0,withBorder](v11);
        \draw (b1)edge[bend left=5,out=20,withBorder](v11);
        \draw (a1)edge[bend right=5,out=-20,withBorder](ov31);
        \draw (b1)edge[bend left=0,withBorder](ov31);
        \begin{scope}[xshift=2.4cm,yshift=1cm,rotate=0]
            \node[vertex] (a2) at (-.5,3) {$a_2$};
            \node[vertex] (b2) at (.5,3) {$b_2$};
            \node[vertex,thin,dashed] (c2) at (0,4) {$c_2$};
            \draw (a2) -- (b2) -- (c2) -- (a2);
            \draw (a2)edge[bend right=0,withBorder](ov12);
            \draw (b2)edge[bend left=10,withBorder](ov12);
            \draw (a2)edge[bend right=0,withBorder](v22);
            \draw (b2)edge[bend left=0,withBorder](v22);
            \draw[dashed, rounded corners] (-1.1, 2.4) rectangle
            (1.1,4.6) {};
            \node at (1.6,4) {$C_2$};
        \end{scope}
        \begin{scope}[on background layer]
            \draw[packing] (w11.center)--(v11.center)--(ow11.center)--cycle;
            \draw[packing] (w12.center)--(v12.center)--(ow12.center)--cycle;

            \draw[packing] (ow21.center)--(ov21.center)--(w22.center)--cycle;
            \draw[packing] (ow22.center)--(ov22.center)--(w23.center)--cycle;
            \draw[packing] (ow31.center)--(ov31.center)--(w32.center)--cycle;
            \draw[packing] (ow32.center)--(ov32.center)--(w33.center)--cycle;
            \draw[packing] (a1.center)--(b1.center)--(v21.center)--(a1.center);
            \draw[packing]
            (a2.center)--(b2.center)--(v32.center)edge[bend left=10,packing](a2.center);
        \end{scope}
    \end{tikzpicture}
    }
    \caption{Construction of $G(\phi)$ for the simple instance
    $(\phi,k)$ with $k=1$ and
    $\phi=(x_1
    \vee x_2 \vee \lneg x_3)\wedge(\lneg x_1 \vee x_2 \vee
    x_3)$. Note that within the variable gadgets, $w_1^3$, $w_2^3$,
    and
    $w_3^3$ are identified
    with
    $w_1^1$ $w_2^1$, and $w_3^1$, respectively, creating one
    \emph{triangle cycle} for each variable.
	Each clause has its own
    gadget consisting of a triangle and additional \emph{literal
    edges} connecting it to the variable gadgets.
    \\
	The dashed vertex $c_2$ is deleted, which corresponds to
	activating the second clause of the formula. The bold lines show
	a maximum triangle packing of the remaining graph. The remaining
	packed triangle correspond to the satisfying
	assignment $x_1\mapsto 0, x_2\mapsto 1, x_3\mapsto 1$.
	The maximum packing has a total size of $mn+m=8$, showing that
	this is a no-instance, which is trivial in this case since all
	subformulas of $\phi$ are satisfiable.
    }
    \label{fig:triangle-hitpack}
\end{figure}
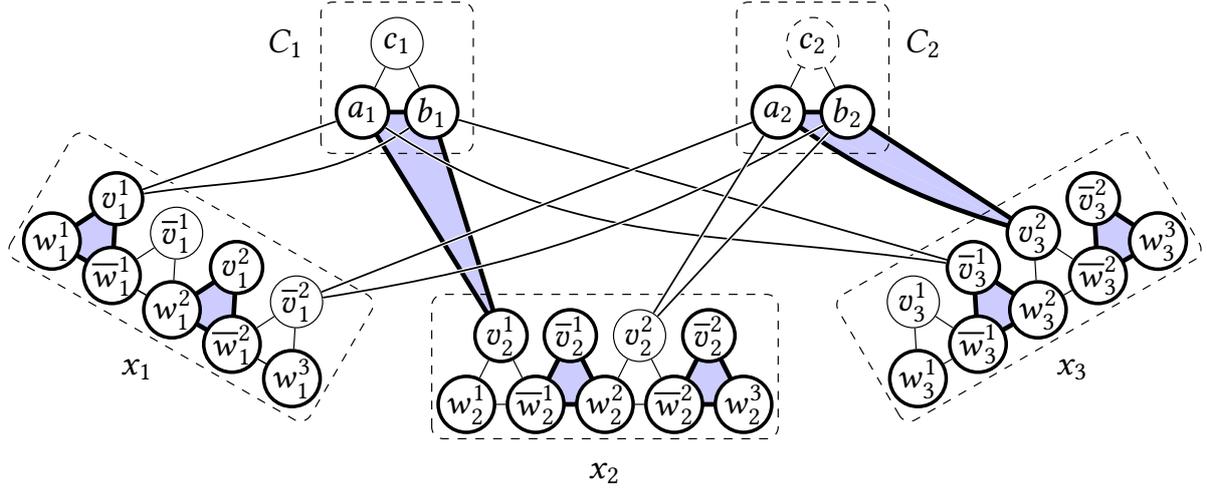

As a next step we extend the previous result to \CycleUndelHitPack.
For this we make use of the problem \TriHitPart
which is defined as follows.
\begin{definition}[\TriHitPart]

    The problem \TriHitPart asks,
    given a graph $G$, a vertex set $U\subseteq V(G)$, and an integer $k$,
    whether it is possible to delete \emph{exactly} $k$
    vertices from the
    graph $G$ such that the vertices of the remaining graph cannot be
    partitioned into triangles.
\end{definition}

Note that we require the deletion of \emph{exactly} $k$ vertices for
this partition version of the problem since allowing for the deletion
of less than $k$ vertices would render the problem trivial for any
positive $k$: If the number of vertices in the graph is divisible by
$3$, deleting one vertex; otherwise, leave it as it is,
gives a graph with no triangle-partition.
We could also require the deletion of exactly $k$ vertices for the
packing version of the problem. All of our results still hold because
deleting more vertices can never increase the packing number of the
remaining graph.

We will prove the following two theorems at once. 

\thmcyclepwLBSigP*\label\thisthm%
\begin{theorem}\label{thm:sig2p:hardness:tripart}
    The problem \TriHitPart is \SigTwoP-hard on tripartite graphs.
\end{theorem}

\begin{proof}[Proof of 
\cref{thm:sig2p:completeness:cycles,thm:sig2p:hardness:tripart}]
    We extend the construction used in the proof of
    \cref{thm:sig2p:hardness:triangle} as follows.
    We have already observed that for a no-instance all subgraphs of $G$
    resulting from
    deleting up to $k$ deletable vertices have a triangle packing of
    size $mn+m$
    where all uncovered vertices are from $V\cup \lneg V$.
    For $k$ deleted vertices, exactly
    $mn-k$ among these vertices remain uncovered.
    We extend $G$ to a new graph $G'$
    by adding, for each $r\in\numb{mn-k}$, a vertex pair
    $\{y_r,z_r\}$
    and, for each vertex $v\in V\cup \lneg V$, the edges of the
    triangle
    $\{v,y_r,z_r\}$.

    Each new vertex pair $\{y_r,z_r\}$ can be used for packing exactly one triangle
    that covers
    exactly one arbitrary vertex from $V\cup \lneg V$. This enables
    us to turn
    any of the previously considered packings of size $mn+m$ of $G$ after
    the deletion of $k$ vertices into a triangle partition of size
    $mn+m+(mn-k)=2mn+m-k$ of $G'$ after the deletion of the same vertices.
    Conversely, for any triangle partition of $G'$, which necessarily has
    size $2mn+m-k$, we obtain a packing
    of size exactly $mn+m$ for $G$ by removing the $mn-k$ triangles
    containing any of the new vertex
    pairs.

    This provides a direct reduction from \SUS both to
    \TriHitPart and, because a triangle partition is the
    most efficient way of packing as many cycles as possible into a
    graph, to
    \CycleUndelHitPack.
\end{proof}

\subsection{\texorpdfstring%
{\boldmath\SigTwoP-Hardness for \UndelHitPack{H}}
{Sigma2P-Hardness for H-HitPack}}
\label{sec:sig2p:hardness:H}

Now, we focus on the proof of the following statement.

\thmHSigTwoPLower*\label\thisthm

The containment in $\SigTwoP$ follows from~\cref{lem:sig2p:membership}.
The hardness result follows immediately from the hardness result
for \TriUndelHitPack, 
\cref{thm:sig2p:hardness:triangle},
together with a polynomial-time reduction from \TriUndelHitPack to \UndelHitPack{H},
formally stated in~\cref{lem:tw:triangleToAnyGraph}.
We actually prove a stronger result and additionally show
that this reduction is \emph{pathwidth-preserving},
by which we mean that for an input graph $G$ of \TriUndelHitPack,
the \UndelHitPack{H} oracle is called by the algorithm only on input graphs whose pathwidth exceeds the pathwidth of $G$ by at most some additive constant.
While we do not need this stronger version in the current section,
we use it in \cref{sec:pw:lower}.

\begin{lemma}
	\label{lem:tw:triangleToAnyGraph}
	Let $H$ be a fixed connected graph with at least three vertices.
	There is a polynomial-time pathwidth-preserving reduction from \TriUndelHitPack
	to \UndelHitPack{H}.
\end{lemma}

Towards proving~\cref{lem:tw:triangleToAnyGraph}, let us first give some intuition and introduce some definitions.

\subparagraph*{Intuition.}
The naive idea is to introduce, for each triangle $T$ of the given instance,
a copy of the graph $H$, say $H_T$,
and to identify three (arbitrary) vertices of $H_T$
with the vertices of the triangle $T$.
However, in this construction there is no guarantee that, in an $H$-packing of this modified instance,
$H_T$ is covered only by one copy of $H$.
Instead, it could be the case that one part of $H_T$
is covered by one copy of $H$, while another part of $H_T$ is covered by a different copy.
To avoid this situation, rather than identifying the vertices of the triangle $T$ with corresponding vertices in $H_T$ directly,
we introduce long chains (of gadgets)
between these vertices.
We create, for each vertex $v$ of a triangle $T$,
a chain $C_{T,v}$ consisting of modified copies of $H$
(which we later formally introduce as diamond gadgets)
that allow us to propagate the information whether $T$ was covered.
Again, as this propagation is not enforced,
there might be packings that do not properly propagate this information.
However, the construction of these diamond gadgets
and the length of the chain ensure that these undesired packings
have size far from the maximum. To make the construction with the chains work, we require that $H$ is connected, and in order to have distinct vertices in $H_T$ corresponding to the vertices of the triangle $T$, we need the requirement that $H$ have size at least $3$.

\subparagraph*{Diamond Gadgets.}
Inspired by Kirkpatrick and Hell~\cite{DBLP:journals/siamcomp/KirkpatrickH83},
we define the \emph{diamond gadget} as follows.
Let $H$ be a connected graph with at least two vertices.
Let $u$ and $u'$ be vertices with maximum distance in $H$. %
Since $H$ is connected, this distance is finite, and since $H$ has at least two vertices, $u$ and $u'$ are distinct.
We introduce a new vertex $\bar u$ that is a copy of $u$ in the sense that it is incident precisely to the neighbors of $u$.
The graph $H$ together with this new vertex $\bar u$
and the corresponding new edges forms the \emph{diamond gadget} (of $H$),
which we abbreviate by the symbol $\dimond$.
We say that $u$ and $\bar u$ are the \emph{connecting vertices} of $\dimond$.

We point out two choices of packing $H$ into $\dimond$ (there might be more):
(1) the packing that covers the original copy of $H$ (with the corresponding vertices), in particular, it covers $u$.
We say that this is the \emph{forward packing} for $\dimond$,
or (2) the packing that is identical to the forward packing, except that it covers $\bar u$ instead of $u$,
we say this is the \emph{backward packing} for $\dimond$.

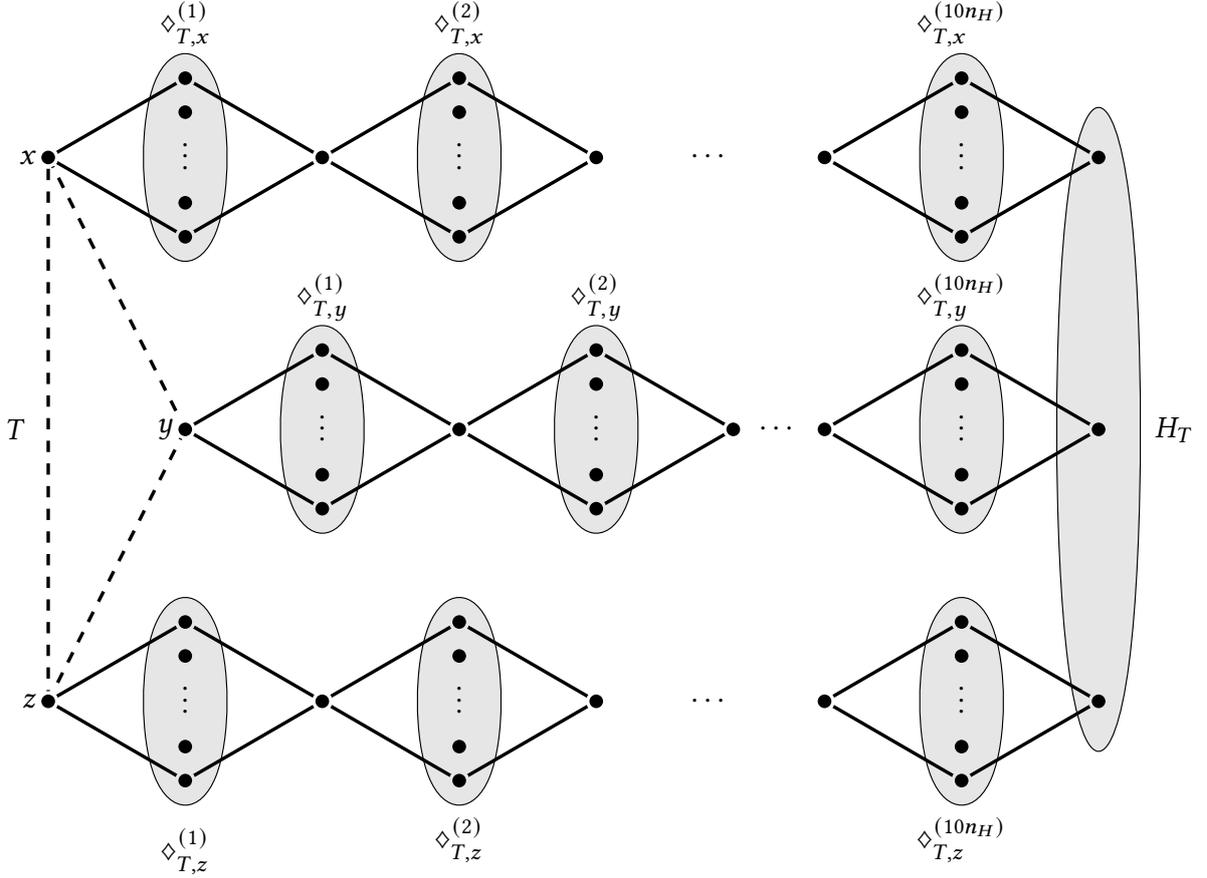
\begin{figure}[t]
	\centering
  \resizebox{\textwidth}{!}{%
  \tikzstyle{ipe stylesheet} = [
  ipe import,
  even odd rule,
  line join=round,
  line cap=butt,
  ipe pen normal/.style={line width=0.4},
  ipe pen heavier/.style={line width=0.8},
  ipe pen fat/.style={line width=1.2},
  ipe pen ultrafat/.style={line width=2},
  ipe pen normal,
  ipe mark normal/.style={ipe mark scale=3},
  ipe mark large/.style={ipe mark scale=5},
  ipe mark small/.style={ipe mark scale=2},
  ipe mark tiny/.style={ipe mark scale=1.1},
  ipe mark normal,
  /pgf/arrow keys/.cd,
  ipe arrow normal/.style={scale=7},
  ipe arrow large/.style={scale=10},
  ipe arrow small/.style={scale=5},
  ipe arrow tiny/.style={scale=3},
  ipe arrow normal,
  /tikz/.cd,
  ipe arrows, %
  <->/.tip = ipe normal,
  ipe dash normal/.style={dash pattern=},
  ipe dash dotted/.style={dash pattern=on 1bp off 3bp},
  ipe dash dashed/.style={dash pattern=on 4bp off 4bp},
  ipe dash dash dotted/.style={dash pattern=on 4bp off 2bp on 1bp off 2bp},
  ipe dash dash dot dotted/.style={dash pattern=on 4bp off 2bp on 1bp off 2bp on 1bp off 2bp},
  ipe dash normal,
  ipe node/.append style={font=\normalsize},
  ipe stretch normal/.style={ipe node stretch=1},
  ipe stretch normal,
  ipe opacity 10/.style={opacity=0.1},
  ipe opacity 30/.style={opacity=0.3},
  ipe opacity 50/.style={opacity=0.5},
  ipe opacity 75/.style={opacity=0.75},
  ipe opacity opaque/.style={opacity=1},
  ipe opacity opaque,
]
\definecolor{black}{rgb}{0,0,0}
\begin{tikzpicture}[ipe stylesheet]
  \filldraw[rgb color={fill=0.9 0.9 0.9}]
    (453.3333, 706.6667)
     .. controls (448, 665.3333) and (448, 582.6667) .. (453.3333, 541.3333)
     .. controls (458.6667, 500) and (469.3333, 500) .. (474.6667, 541.3333)
     .. controls (480, 582.6667) and (480, 665.3333) .. (474.6667, 706.6667)
     .. controls (469.3333, 748) and (458.6667, 748) .. cycle;
  \draw[ipe pen fat, ipe dash dashed]
    (96, 720)
     -- (144, 624);
  \draw[ipe pen fat, ipe dash dashed]
    (144, 624)
     -- (96, 528);
  \draw[ipe pen fat, ipe dash dashed]
    (96, 528)
     -- (96, 720);
  \filldraw[rgb color={fill=0.9 0.9 0.9}]
    (133.3333, 746.6667)
     .. controls (128, 733.3333) and (128, 706.6667) .. (133.3333, 693.3333)
     .. controls (138.6667, 680) and (149.3333, 680) .. (154.6667, 693.3333)
     .. controls (160, 706.6667) and (160, 733.3333) .. (154.6667, 746.6667)
     .. controls (149.3333, 760) and (138.6667, 760) .. cycle;
  \draw[ipe pen fat]
    (144, 748)
     -- (96, 720);
  \draw[ipe pen fat]
    (96, 720)
     -- (144, 692);
  \draw[ipe pen fat]
    (144, 692)
     -- (192, 720);
  \draw[ipe pen fat]
    (192, 720)
     -- (144, 748);
  \pic[ipe mark scale=6.0, draw=white, fill=black]
     at (96, 720) {ipe fdisk};
  \filldraw[rgb color={fill=0.9 0.9 0.9}]
    (229.3333, 746.6667)
     .. controls (224, 733.3333) and (224, 706.6667) .. (229.3333, 693.3333)
     .. controls (234.6667, 680) and (245.3333, 680) .. (250.6667, 693.3333)
     .. controls (256, 706.6667) and (256, 733.3333) .. (250.6667, 746.6667)
     .. controls (245.3333, 760) and (234.6667, 760) .. cycle;
  \draw[ipe pen fat]
    (240, 748)
     -- (192, 720);
  \draw[ipe pen fat]
    (192, 720)
     -- (240, 692);
  \draw[ipe pen fat]
    (240, 692)
     -- (288, 720);
  \draw[ipe pen fat]
    (288, 720)
     -- (240, 748);
  \filldraw[rgb color={fill=0.9 0.9 0.9}]
    (405.3333, 746.6667)
     .. controls (400, 733.3333) and (400, 706.6667) .. (405.3333, 693.3333)
     .. controls (410.6667, 680) and (421.3333, 680) .. (426.6667, 693.3333)
     .. controls (432, 706.6667) and (432, 733.3333) .. (426.6667, 746.6667)
     .. controls (421.3333, 760) and (410.6667, 760) .. cycle;
  \draw[ipe pen fat]
    (416, 748)
     -- (368, 720);
  \draw[ipe pen fat]
    (368, 720)
     -- (416, 692);
  \draw[ipe pen fat]
    (416, 692)
     -- (464, 720);
  \draw[ipe pen fat]
    (464, 720)
     -- (416, 748);
  \pic[ipe mark scale=6.0, rgb color={draw=0.9 0.9 0.9}, fill=black]
     at (464, 720) {ipe fdisk};
  \node[ipe node, anchor=base]
     at (144, 764) {$\diamondsuit^{(1)}_{T,x}$};
  \node[ipe node, anchor=base]
     at (240, 764) {$\diamondsuit^{(2)}_{T,x}$};
  \node[ipe node, anchor=base]
     at (416, 764) {$\diamondsuit^{(10 n_H)}_{T,x}$};
  \pic[ipe mark scale=6.0, draw=white, fill=black]
     at (368, 720) {ipe fdisk};
  \filldraw[rgb color={fill=0.9 0.9 0.9}]
    (133.3333, 554.6667)
     .. controls (128, 541.3333) and (128, 514.6667) .. (133.3333, 501.3333)
     .. controls (138.6667, 488) and (149.3333, 488) .. (154.6667, 501.3333)
     .. controls (160, 514.6667) and (160, 541.3333) .. (154.6667, 554.6667)
     .. controls (149.3333, 568) and (138.6667, 568) .. cycle;
  \draw[ipe pen fat]
    (144, 556)
     -- (96, 528);
  \draw[ipe pen fat]
    (96, 528)
     -- (144, 500);
  \draw[ipe pen fat]
    (144, 500)
     -- (192, 528);
  \draw[ipe pen fat]
    (192, 528)
     -- (144, 556);
  \pic[ipe mark scale=6.0, draw=white, fill=black]
     at (96, 528) {ipe fdisk};
  \filldraw[rgb color={fill=0.9 0.9 0.9}]
    (229.3333, 554.6667)
     .. controls (224, 541.3333) and (224, 514.6667) .. (229.3333, 501.3333)
     .. controls (234.6667, 488) and (245.3333, 488) .. (250.6667, 501.3333)
     .. controls (256, 514.6667) and (256, 541.3333) .. (250.6667, 554.6667)
     .. controls (245.3333, 568) and (234.6667, 568) .. cycle;
  \draw[ipe pen fat]
    (240, 556)
     -- (192, 528);
  \draw[ipe pen fat]
    (192, 528)
     -- (240, 500);
  \draw[ipe pen fat]
    (240, 500)
     -- (288, 528);
  \draw[ipe pen fat]
    (288, 528)
     -- (240, 556);
  \filldraw[rgb color={fill=0.9 0.9 0.9}]
    (405.3333, 554.6667)
     .. controls (400, 541.3333) and (400, 514.6667) .. (405.3333, 501.3333)
     .. controls (410.6667, 488) and (421.3333, 488) .. (426.6667, 501.3333)
     .. controls (432, 514.6667) and (432, 541.3333) .. (426.6667, 554.6667)
     .. controls (421.3333, 568) and (410.6667, 568) .. cycle;
  \draw[ipe pen fat]
    (416, 556)
     -- (368, 528);
  \draw[ipe pen fat]
    (368, 528)
     -- (416, 500);
  \draw[ipe pen fat]
    (416, 500)
     -- (464, 528);
  \draw[ipe pen fat]
    (464, 528)
     -- (416, 556);
  \pic[ipe mark scale=6.0, rgb color={draw=0.9 0.9 0.9}, fill=black]
     at (464, 528) {ipe fdisk};
  \node[ipe node, anchor=base]
     at (144, 472) {$\diamondsuit^{(1)}_{T,z}$};
  \node[ipe node, anchor=base]
     at (240, 476) {$\diamondsuit^{(2)}_{T,z}$};
  \node[ipe node, anchor=base]
     at (416, 476) {$\diamondsuit^{(10 n_H)}_{T,z}$};
  \pic[ipe mark scale=6.0, draw=white, fill=black]
     at (368, 528) {ipe fdisk};
  \filldraw[rgb color={fill=0.9 0.9 0.9}]
    (181.3333, 650.6667)
     .. controls (176, 637.3333) and (176, 610.6667) .. (181.3333, 597.3333)
     .. controls (186.6667, 584) and (197.3333, 584) .. (202.6667, 597.3333)
     .. controls (208, 610.6667) and (208, 637.3333) .. (202.6667, 650.6667)
     .. controls (197.3333, 664) and (186.6667, 664) .. cycle;
  \draw[ipe pen fat]
    (192, 652)
     -- (144, 624);
  \draw[ipe pen fat]
    (144, 624)
     -- (192, 596);
  \draw[ipe pen fat]
    (192, 596)
     -- (240, 624);
  \draw[ipe pen fat]
    (240, 624)
     -- (192, 652);
  \filldraw[rgb color={fill=0.9 0.9 0.9}]
    (277.3333, 650.6667)
     .. controls (272, 637.3333) and (272, 610.6667) .. (277.3333, 597.3333)
     .. controls (282.6667, 584) and (293.3333, 584) .. (298.6667, 597.3333)
     .. controls (304, 610.6667) and (304, 637.3333) .. (298.6667, 650.6667)
     .. controls (293.3333, 664) and (282.6667, 664) .. cycle;
  \draw[ipe pen fat]
    (288, 652)
     -- (240, 624);
  \draw[ipe pen fat]
    (240, 624)
     -- (288, 596);
  \draw[ipe pen fat]
    (288, 596)
     -- (336, 624);
  \draw[ipe pen fat]
    (336, 624)
     -- (288, 652);
  \pic[ipe mark scale=6.0, draw=white, fill=black]
     at (336, 624) {ipe fdisk};
  \filldraw[rgb color={fill=0.9 0.9 0.9}]
    (405.3333, 650.6667)
     .. controls (400, 637.3333) and (400, 610.6667) .. (405.3333, 597.3333)
     .. controls (410.6667, 584) and (421.3333, 584) .. (426.6667, 597.3333)
     .. controls (432, 610.6667) and (432, 637.3333) .. (426.6667, 650.6667)
     .. controls (421.3333, 664) and (410.6667, 664) .. cycle;
  \draw[ipe pen fat]
    (416, 652)
     -- (368, 624);
  \draw[ipe pen fat]
    (368, 624)
     -- (416, 596);
  \draw[ipe pen fat]
    (416, 596)
     -- (464, 624);
  \draw[ipe pen fat]
    (464, 624)
     -- (416, 652);
  \pic[ipe mark scale=6.0, rgb color={draw=0.9 0.9 0.9}, fill=black]
     at (464, 624) {ipe fdisk};
  \node[ipe node, anchor=base]
     at (192, 668) {$\diamondsuit^{(1)}_{T,y}$};
  \node[ipe node, anchor=base]
     at (288, 668) {$\diamondsuit^{(2)}_{T,y}$};
  \node[ipe node, anchor=base]
     at (416, 668) {$\diamondsuit^{(10 n_H)}_{T,y}$};
  \pic[ipe mark scale=6.0, draw=white, fill=black]
     at (368, 624) {ipe fdisk};
  \pic[ipe mark scale=6.0, rgb color={draw=0.9 0.9 0.9}, fill=black]
     at (416, 736) {ipe fdisk};
  \pic[ipe mark scale=6.0, rgb color={draw=0.9 0.9 0.9}, fill=black]
     at (416, 704) {ipe fdisk};
  \node[ipe node, anchor=base]
     at (416, 716) {$\vdots$};
  \pic[ipe mark scale=6.0, rgb color={draw=0.9 0.9 0.9}, fill=black]
     at (416, 556) {ipe fdisk};
  \pic[ipe mark scale=6.0, rgb color={draw=0.9 0.9 0.9}, fill=black]
     at (416, 544) {ipe fdisk};
  \pic[ipe mark scale=6.0, rgb color={draw=0.9 0.9 0.9}, fill=black]
     at (416, 512) {ipe fdisk};
  \pic[ipe mark scale=6.0, rgb color={draw=0.9 0.9 0.9}, fill=black]
     at (416, 500) {ipe fdisk};
  \node[ipe node, anchor=base]
     at (416, 524) {$\vdots$};
  \pic[ipe mark scale=6.0, rgb color={draw=0.9 0.9 0.9}, fill=black]
     at (416, 652) {ipe fdisk};
  \pic[ipe mark scale=6.0, rgb color={draw=0.9 0.9 0.9}, fill=black]
     at (416, 640) {ipe fdisk};
  \pic[ipe mark scale=6.0, rgb color={draw=0.9 0.9 0.9}, fill=black]
     at (416, 608) {ipe fdisk};
  \pic[ipe mark scale=6.0, rgb color={draw=0.9 0.9 0.9}, fill=black]
     at (416, 596) {ipe fdisk};
  \node[ipe node, anchor=base]
     at (416, 620) {$\vdots$};
  \pic[ipe mark scale=6.0, rgb color={draw=0.9 0.9 0.9}, fill=black]
     at (416, 748) {ipe fdisk};
  \pic[ipe mark scale=6.0, rgb color={draw=0.9 0.9 0.9}, fill=black]
     at (416, 692) {ipe fdisk};
  \pic[ipe mark scale=6.0, rgb color={draw=0.9 0.9 0.9}, fill=black]
     at (240, 748) {ipe fdisk};
  \pic[ipe mark scale=6.0, rgb color={draw=0.9 0.9 0.9}, fill=black]
     at (240, 736) {ipe fdisk};
  \pic[ipe mark scale=6.0, rgb color={draw=0.9 0.9 0.9}, fill=black]
     at (240, 704) {ipe fdisk};
  \pic[ipe mark scale=6.0, rgb color={draw=0.9 0.9 0.9}, fill=black]
     at (240, 692) {ipe fdisk};
  \node[ipe node, anchor=base]
     at (240, 716) {$\vdots$};
  \pic[ipe mark scale=6.0, rgb color={draw=0.9 0.9 0.9}, fill=black]
     at (240, 556) {ipe fdisk};
  \pic[ipe mark scale=6.0, rgb color={draw=0.9 0.9 0.9}, fill=black]
     at (240, 544) {ipe fdisk};
  \pic[ipe mark scale=6.0, rgb color={draw=0.9 0.9 0.9}, fill=black]
     at (240, 512) {ipe fdisk};
  \pic[ipe mark scale=6.0, rgb color={draw=0.9 0.9 0.9}, fill=black]
     at (240, 500) {ipe fdisk};
  \node[ipe node, anchor=base]
     at (240, 524) {$\vdots$};
  \pic[ipe mark scale=6.0, draw=white, fill=black]
     at (240, 624) {ipe fdisk};
  \pic[ipe mark scale=6.0, draw=white, fill=black]
     at (288, 720) {ipe fdisk};
  \pic[ipe mark scale=6.0, draw=white, fill=black]
     at (288, 528) {ipe fdisk};
  \pic[ipe mark scale=6.0, rgb color={draw=0.9 0.9 0.9}, fill=black]
     at (288, 652) {ipe fdisk};
  \pic[ipe mark scale=6.0, rgb color={draw=0.9 0.9 0.9}, fill=black]
     at (288, 640) {ipe fdisk};
  \pic[ipe mark scale=6.0, rgb color={draw=0.9 0.9 0.9}, fill=black]
     at (288, 608) {ipe fdisk};
  \pic[ipe mark scale=6.0, rgb color={draw=0.9 0.9 0.9}, fill=black]
     at (288, 596) {ipe fdisk};
  \node[ipe node, anchor=base]
     at (288, 620) {$\vdots$};
  \pic[ipe mark scale=6.0, rgb color={draw=0.9 0.9 0.9}, fill=black]
     at (144, 748) {ipe fdisk};
  \pic[ipe mark scale=6.0, rgb color={draw=0.9 0.9 0.9}, fill=black]
     at (144, 736) {ipe fdisk};
  \pic[ipe mark scale=6.0, rgb color={draw=0.9 0.9 0.9}, fill=black]
     at (144, 704) {ipe fdisk};
  \pic[ipe mark scale=6.0, rgb color={draw=0.9 0.9 0.9}, fill=black]
     at (144, 692) {ipe fdisk};
  \node[ipe node, anchor=base]
     at (144, 716) {$\vdots$};
  \pic[ipe mark scale=6.0, rgb color={draw=0.9 0.9 0.9}, fill=black]
     at (144, 556) {ipe fdisk};
  \pic[ipe mark scale=6.0, rgb color={draw=0.9 0.9 0.9}, fill=black]
     at (144, 544) {ipe fdisk};
  \pic[ipe mark scale=6.0, rgb color={draw=0.9 0.9 0.9}, fill=black]
     at (144, 512) {ipe fdisk};
  \pic[ipe mark scale=6.0, rgb color={draw=0.9 0.9 0.9}, fill=black]
     at (144, 500) {ipe fdisk};
  \node[ipe node, anchor=base]
     at (144, 524) {$\vdots$};
  \pic[ipe mark scale=6.0, draw=white, fill=black]
     at (144, 624) {ipe fdisk};
  \pic[ipe mark scale=6.0, draw=white, fill=black]
     at (192, 720) {ipe fdisk};
  \pic[ipe mark scale=6.0, draw=white, fill=black]
     at (192, 528) {ipe fdisk};
  \pic[ipe mark scale=6.0, rgb color={draw=0.9 0.9 0.9}, fill=black]
     at (192, 652) {ipe fdisk};
  \pic[ipe mark scale=6.0, rgb color={draw=0.9 0.9 0.9}, fill=black]
     at (192, 640) {ipe fdisk};
  \pic[ipe mark scale=6.0, rgb color={draw=0.9 0.9 0.9}, fill=black]
     at (192, 608) {ipe fdisk};
  \pic[ipe mark scale=6.0, rgb color={draw=0.9 0.9 0.9}, fill=black]
     at (192, 596) {ipe fdisk};
  \node[ipe node, anchor=base]
     at (192, 620) {$\vdots$};
  \node[ipe node, anchor=east]
     at (92, 720) {$x$};
  \node[ipe node, anchor=east]
     at (140, 624) {$y$};
  \node[ipe node, anchor=east]
     at (92, 528) {$z$};
  \node[ipe node, anchor=west]
     at (484, 624) {$H_T$};
  \node[ipe node, anchor=east]
     at (88, 624) {$T$};
  \node[ipe node, anchor=base]
     at (328, 720) {$\dots$};
  \node[ipe node, anchor=base]
     at (352, 624) {$\dots$};
  \node[ipe node, anchor=base]
     at (328, 528) {$\dots$};
\end{tikzpicture}
  }
	\caption{%
  An illustration of $H_T$ for $T=xyz$.
  The three rows of diamond gadgets correspond to
  the chains $C_{T,x}$, $C_{T,y}$, and $C_{T,z}$, respectively.
  The dashed edges of the triangle are deleted in the construction.
  }

	\label{fig:tw-lower-bound:general-H}
\end{figure}

\begin{proof}[Proof of \Cref{lem:tw:triangleToAnyGraph}]
  \newcommand{\vertH}{n_H}
  \newcommand{\Tris}{\mathcal T}

  Let $H$ be an arbitrary connected graph with at least three vertices.
  Given an instance $I_\triangle=(G, U, k, \ell)$ of \TriUndelHitPack,
  we define an instance $I_H$ of \UndelHitPack{H} as follows.
  We set $\vertH \deff |V(H)|$,
  and we denote by $\Tris$ the set of all triangles in $G$.

  For each triangle $T \in \Tris$ and each vertex $v$ of $T$
  we introduce a gadget $C_{T,v}$ (cf.~\cref{fig:tw-lower-bound:general-H}) as follows:
  \begin{itemize}
    \item
    For all $i\in \numb{10 \vertH}$,
    we create a copy of the gadget $\dimond$,
    denoted by $\dimond^{(i)}_{T,v}$.
    \item
    For all $i \in \numb{10\vertH-1}$,
    we identify vertex $\bar u$ of $\dimond^{(i)}_{T,v}$
    with vertex $u$ of $\dimond^{(i+1)}_{T,v}$.
    \item
    The gadget $C_{T,v}$ consists of the union of the diamond gadgets
    $\dimond^{(1)}_{T,v},\dots,\dimond^{(10\vertH)}_{T,v}$.
    \item
    If, for all $i \in \numb{10\vertH}$,
    we choose the forward packing for $\dimond^{(i)}_{T,v}$,
    then we say this is the \emph{forward packing} for $C_{T,v}$.
    \item
    If, for all $i \in \numb{10\vertH}$,
    we choose the backward packing for $\dimond^{(i)}_{T,v}$,
    then we say this is the \emph{backward packing} for $C_{T,v}$.
  \end{itemize}

  For each triangle $T \in \Tris$ and for every vertex $v \in T$,
  we identify $v$ with the connecting vertex $u$ of $\dimond^{(1)}_{T,v}$.
  For every triangle $T=xyz$ in $G$,
  we introduce a copy of $H$, denoted by $H_{T}$.
  We identify the connecting vertices $\bar u$ of
  $\dimond^{(10\vertH)}_{T,x}, \dimond^{(10\vertH)}_{T,y},
  \dimond^{(10\vertH)}_{T,z}$ with three (arbitrary) distinct vertices
  $v_p,v_q,v_r$ of $H_{T}$, respectively.
  In order to avoid the situation that the edges in the triangles of $G$ are used for packing $H$, we delete all edges which are originally in $G$. The construction for some triangle $T$ is illustrated in \cref{fig:tw-lower-bound:general-H}.

  Let $G'$ be the resulting graph.
  We define $U' \deff U \cup V(G') \setminus V(G)$
  as the set of undeletable vertices of $G'$,
  that is, the deletable vertices are precisely the deletable vertices from $G$.
  We set $k'\deff k$
  and $\ell' \deff 30\vertH \cdot \abs{\Tris} + \ell$.
  Let $I_H=(G', U', k', \ell')$ be the instance of \UndelHitPack{H}.
  This finishes the construction.

  Now we prove the correctness of this reduction.
  \begin{claim}
    \label{clm:tw-general:lower:reduction:soundness}
    If $I_\triangle$ is a \yes-instance of \TriUndelHitPack,
    then $I_H$ is a \yes-instance of \UndelHitPack{H}.
  \end{claim}
  \begin{claimproof}
    Let $S$ be some solution of $I_\triangle=(G, U, k, \ell)$.
    Since the deletable vertices of $G$ are also deletable in $G'$
    and by definition of $I_H$,
    we get $S \subseteq V(G') \setminus U'$ and $\abs{S} \le k=k'$.
    In the remainder we prove that all $H$-packings of  $G'\setminus S$
    have size less than $\ell'$.

    Assume, for the sake of contradiction that there is an $H$-packing $\mathcal P'$ for $G'-S$
    containing at least $\ell'$ copies of $H$.
    In the next steps we show that we can assume
    that $\mathcal P'$ is either the forward packing or the backward packing on all
    gadgets $C_{T,v}$. A priori, this is not clear as, while we delete the edges
    of the original graph $G$, each original vertex can be contained in multiple
    triangles of $G$, and therefore might now be connected to multiple gadgets
    $C_{T,v}$, $C_{T',v}$, and so on. So, a copy of $H$ in some packing could
    cover some part of $C_{T,v}$ and some part of $C_{T',v}$ (other than just
    $v$). We argue that without loss of generality, we can assume that this is not the case.

    \begin{enumerate}
	\item Consider %
	some diamond gadget $\dimond_{T,v}^{(i)}$. 
	Since $H$ is connected and has at least two vertices, this gadget has precisely two vertices that are connected to the rest of $G"$: the connecting vertices $u$ and $\bar u$. 
	Consequently, as $H$ is connected and $\dimond_{T,v}^{(i)}- \{u,\bar u\}$ only has $|V(H)|-1$ vertices, a copy of $H$ in the packing $\mathcal P'$ that covers some of the vertices of $\dimond_{T,v}^{(i)}$ has to cover at least one of $u$ or $\bar u$. Thus, there are at most two copies of $H$ that cover some vertices of $\dimond_{T,v}^{(i)}$. 
    \item
    Now suppose that in $\mathcal P'$ there is a copy of $H$
    that covers both connecting vertices of some diamond gadget. Then no other copy of $H$ covers vertices of this diamond gadget.
    Then we can replace this copy of $H$
    by the forward or backward packing for this diamond gadget
    without changing the size (or feasibility) of $\mathcal P'$.
    Thus, we can assume that in the $H$-packing $\mathcal P'$ there is no copy of $H$
    covering both connecting vertices of some diamond gadget.
    \item 
	Consider some diamond gadget $\dimond^{(i)}_{T,v}$ and the corresponding distinguished vertices $u$, $\bar u$, and $u'$.
	Suppose some copy of $H$ in $\mathcal P'$ covers some vertex $w$ outside of $\dimond^{(i)}_{T,v}$. We argue that this copy of $H$ cannot cover $u'$:
	First, as $H$ is connected, the vertices $u'$ and $w$ have to be connected in $G$. Consider a shortest path between $u'$ and $w$ in $G$. It has to go through one of the connecting vertices $u$ or $\bar u$, say through $u$. This implies that the distance between $u'$ and $w$ is strictly larger than the distance between $u'$ and $u$. Since the distance between $u'$ and $u$ equals the diameter of $H$, the graph $H$ cannot cover both $w$ and $u'$.
	So we have shown that a copy of $H$ that covers $u'$ covers no vertex outside of the corresponding gadget $\dimond^{(i)}_{T,v}$. By the previous deductions, we know that such a copy covers precisely one of $u$ or $\bar u$, and actually due to the fact that $\dimond^{(i)}_{T,v}$ has precisely $|V(H)|+1$ vertices, such a copy of $H$ then covers all vertices of $\dimond^{(i)}_{T,v}$ apart from one of $\bar u$ or $u$. So, summarizing, if $u'$ is covered, then we can assume that  $\dimond^{(i)}_{T,v}$ is covered according to the forward packing or backward packing, respectively.

    \item
    Suppose that
    for some $T$ and $v$,
    there is (at least) one gadget $\dimond^{(i)}_{T,v}$
    for which the corresponding vertex $u'$ is covered by $\mathcal P$.
    Let $i_{\min}\in \numb{10\vertH}$ be the smallest such index.
	From the previous point, we can assume that the copy of $H$ that covers $u'$ covers $\dimond^{(i_{\min})}_{T,v}$
    according to either the forward or the backward packing. 
    Suppose it is the backward packing --- the other case is analogous.
    Then $\bar u$ of $\dimond^{(i_{\min}+1)}_{T,v}$ is already covered by the backward packing of $\dimond^{(i_{\min})}_{T,v}$ and consequently there is at most one other copy of $H$ in the packing $\mathcal P$ that covers vertices of $\dimond^{(i_{\min}+1)}_{T,v}$. This copy then has to cover the vertex corresponding to $u$ (in $\dimond^{(i_{\min}+1)}_{T,v}$). Therefore, we can replace this copy of $H$ by the backward packing of $\dimond^{(i_{\min}+1)}_{T,v}$ without violating the feasibility of the packing. Using the argument iteratively, we can assume that for all $i\ge i_{\min}$ the corresponding gadget is covered according to the backward packing.
    Now let us look at $\dimond^{(i_{\min}-1)}_{T,v}$. By assumption, the corresponding vertex $u'$ is not covered. Let us consider a copy of $H$ that covers some vertices from $\dimond^{(i_{\min}-1)}_{T,v}$. We know that it covers precisely one of $u$ or $\bar u$ from $\dimond^{(i_{\min}-1)}_{T,v}$. Since it does not covers both, but it also does not cover $u'$, from the number of vertices of $H$, it follows that it has to cover some vertices outside of $\dimond^{(i_{\min}-1)}_{T,v}$. Thus, it cannot cover $\bar u$ since then it would have to cover other vertices in $\dimond^{(i_{\min})}_{T,v}$ --- however, all of these vertices are already covered. It follows that such a copy of $H$ would have to cover $u$, and importantly this means that there is only one such copy in $\mathcal P$ that covers vertices of $\dimond^{(i_{\min}-1)}_{T,v}$. Hence, we can replace it by the backward packing for $\dimond^{(i_{\min}-1)}_{T,v}$ without violating the feasibility of the packing. Again, we can use the argument iteratively for all $i<i_{\min}$.
    So we can assume the backward packing on $C_{T,v}$.
    Similarly we could assume the forward packing on $C_{T,v}$ if $\dimond^{(i_{\min})}_{T,v}$ were covered according to the forward packing.
    
    \item
    Finally, suppose there are $T$ and $v$ such that, for all $i\in \numb{10\vertH}$,
    the vertex $u'$ from $\dimond^{(i)}_{T,v}$ is not covered by the packing $\mathcal P'$.
    In this case the packing covers at most
    $10 (\vertH-1) \cdot \vertH$ vertices of $C_{T,v}$. Note also, that since $H$ is connected, each copy of $H$ in the packing that covers vertices within $C_{T,v}$ as well as outside, has to cover one of the three vertices of the triangle $T$. Consequently, there can be at most $3$ such copies of $H$ in the packing. Using this fact, together with the fact that at most $10 (\vertH-1) \cdot \vertH$ vertices of $C_{T,v}$ are covered, we conclude that at most $10(\vertH-1)+3$
    copies of $H$ cover vertices of $C_{T,v}$.
    Hence we can replace these copies of $H$ with
    the backward packing for $C_{T,v}$ and obtain an overall packing that is
    strictly larger
    as we now pack $10\vertH$ copies of $H$ to cover $C_{T,v}$.
    \end{enumerate}

    So we have shown that in $\mathcal P'$, every copy of $H$ either corresponds to a
    backward or forward packing for a diamond gadget
    or covers precisely the vertices of some $H_T$ (if it were to cover some vertices of $H_T$ and some vertices outside of $H_T$ then these outside vertices are part of some diamond gadget and this would contradict the fact that this diamond gadget is covered either by the forward or backward packing).

    Recall that we assumed for contradiction that $\abs{\mathcal P'} \ge \ell'=30\vertH \cdot \abs{\Tris} + \ell$.
    By the construction of the graph,
    the packing $\mathcal P''$ contains already $30\vertH \cdot \abs{\Tris}$
    copies of $H$ from the chains $C_{T,v}$.
    Then it contains at least an additional $\ell$ copies of $H$
    that cover $H_T$ for some triangle $T=xyz$.
    However, then the corresponding gadgets $C_{T,x}$, $C_{T,x}$, and $C_{T,x}$
    have to be covered according to the \emph{forward} packing
    as otherwise the vertices at the intersection of $H_T$
    and these chains would be covered twice.
    It follows that the vertices of the triangle $T$ are covered.
    So this would imply that there is a triangle-packing
    of size at least $\ell$ in $G\setminus S$,
    contradicting that $S$ is a solution for $I$.

    We conclude that
    $I_H=(G',U',k',\ell')$ is a \yes-instance of \UndelHitPack{H}
    which finishes the proof of the claim.
  \end{claimproof}

  Now we prove the reverse direction.

  \begin{claim}
    \label{clm:tw-general:lower:reduction:completeness}
    If $I_H=(G',U',k',\ell')$ is a \yes-instance of \UndelHitPack{H},
    then $I_\triangle=(G, U, k, \ell)$ is a \yes-instance of \TriUndelHitPack.
  \end{claim}
  \begin{claimproof}
    Let $S'$ be some solution for $I_H$.
    Since $I_H$ and $I_\triangle$ have the same set of deletable vertices,
    we have $S' \subseteq V(G) \setminus U$.
    Moreover, since $k' = k$,
    it also holds that $\abs{S'} \le k$.
    To show that $S'$ is a solution for $I_\triangle$,
    it remains to prove that $G-S'$ has no $\triangle$-packing of size at least $\ell$.

    Suppose for contradiction's sake that there is a $\triangle$-packing
    $\mathcal P$ of size at least $\ell$ in $G-S$.
    We define a corresponding  $H$-packing $\mathcal P'$ in $G'\setminus S'$ as follows.
    \begin{itemize}
      \item
      For every triangle $T \in \mathcal P$ and every vertex $v \in T$,
      we add to $\mathcal P'$ the forward packing for $C_{T,v}$
      (containing $10\vertH$ copies of $H$),
      and a copy of $H$ that covers precisely $H_T$.  Note that this is feasible as the forward packing does not cover any vertices from $H_T$.
      This gives a total of  $30\vertH+1$ copies of $H$ in the packing $\mathcal P'$ per triangle in $\mathcal P$.

      \item
      For every triangle $T \in \Tris\setminus \mathcal P$ and every vertex $v \in T$,
      we add the backward packing for $C_{T,v}$
      to $\mathcal P'$. This gives a total of $30\vertH$ copies of $H$ per triangle outside of $\mathcal P$.
    \end{itemize}
    It is straightforward to see that the resulting packing $\mathcal P'$ is indeed feasible, and it has size at least $30\vertH \abs{\Tris} + \ell=\ell'$.
    This finishes the proof of the claim.
  \end{claimproof}

  It remains to show that the pathwidth of $G'$ exceeds that of $G$ by at most some additive constant.
  \begin{claim}
    \label{clm:tw-general:lower:reduction:size}
    If $G$ has pathwidth $\pw$,
    then $G'$ has pathwidth $\pw+\Oh(1)$.
  \end{claim}
  \begin{claimproof}

    Suppose there is a path decomposition for $G$ of width $\pw$.
    For every triangle $T=xyz$ of $G$,
    there is at least one bag that contains all vertices of $T$.
    We replace it with a bag $X_T$ that contains all vertices of $C_{T,x}$, $C_{T,y}$, $C_{T,z}$, and of $H_T$.
    This gives a valid path decomposition of $G'$.
    Since the number of vertices in $C_{T,x}$, $C_{T,y}$, $C_{T,z}$, and $H_T$
    depends only on $H$ (which is fixed),
    the pathwidth of $G'$ is $\pw+\Oh(1)$.
  \end{claimproof}
  This finishes the proof of \cref{lem:tw:triangleToAnyGraph}.
\end{proof}

\section{Double-Exponential Lower Bounds Parameterized by Pathwidth}\label{sec:pw:lower}

In this section, we prove \cref{thm:tw:lower:H,thm:tw:lower:cycles}, which state double exponential lower bounds for \UndelHitPack{H} and \CycleUndelHitPack in terms of the pathwidth of the input graph. Note that this directly implies corresponding lower bounds in terms of treewidth.
Let us first consider \cref{thm:tw:lower:H} and restate it for convenience.
We consider \cref{thm:tw:lower:cycles} later in \cref{sec:pw:lower:cycles}.
\thmHpwLB*\label\thisthm

Often such lower bounds are stated in their strengthened version,
in which it is assumed that a path decomposition of width $\pw$
is given as part of the input.
However, this is no strengthening in our case
since the runtime lower bound is double exponential in $\pw$
and it is well-known that a path decomposition of $G$
can be computed in time exponential in $\pw^2$ and linear in $n$
\cite{KorhonenL23}.

Again, we first give a hardness result for \TriUndelHitPack.

\begin{lemma}
  \label{lem:tw:lower:reduction}
  Let $\phi$ be a 3-SAT instance with $n$ variables and $m$ clauses.
  In time $\Oh(n + m \log m)$,
  we can construct an instance $I=(G, U, k, \ell)$ of \TriUndelHitPack
  such that
  \begin{itemize}
    \item
    $\phi$ is satisfiable if and only if $I$ has a solution,
    \item
    $G$ has pathwidth at most $\Oh(\log m)$,
    and
    \item
    $G$ is of size $\Oh(m \log m)$.
  \end{itemize}
\end{lemma}

The proof of \cref{lem:tw:lower:reduction} is given in \cref{sec:tw:lower:triangle}.
As a second step, we then apply the pathwidth-preserving reduction from \TriUndelHitPack
to \UndelHitPack{H} (for any fixed, connected graph $H$), as stated in~\cref{lem:tw:triangleToAnyGraph}.
Here are the details.

\begin{proof}[Proof of \Cref{thm:tw:lower:H}]
  Assume that \UndelHitPack{H} could be solved in time $2^{2^{o(\pw)}} n^{\Oh(1)}$.
  Let $\phi$ be an instance of 3-SAT with $n$ variables and $m$ clauses.
  We use \cref{lem:tw:lower:reduction},
  to obtain an instance $I_\phi$ of \TriUndelHitPack
  with pathwidth $\Oh(\log(m))$ that has a solution if and only if $\phi$ has a solution. Now we can execute our hypothetical algorithm for \UndelHitPack{H} instead of each oracle call of the reduction from~\cref{lem:tw:triangleToAnyGraph} on input
  $I_\phi$. Note that this
  procedure solves a given instance $\phi$ of 3-SAT.
  Moreover, in total the procedure takes time
  \[
    2^{2^{o(\pw(I_\triangle) + \Oh(1))}} \cdot \abs{I_\triangle}^{\Oh(1)}
    \le 2^{2^{o(\log m)}} \cdot (n+m)^{\Oh(1)}
    \le 2^{o(n+m)}.
  \]
  This gives a contradiction, as the ETH
  (\cite{eth} combined with the Sparsification Lemma~\cite{ImpagliazzoPZ01})
  implies that $3$-SAT cannot be solved in
  $2^{o(n+m)}$ time~\cite[Theorem 14.4]{book1}.
\end{proof}

\subsection{\texorpdfstring%
{Lower Bound for \TriUndelHitPack}
{Lower Bound for Triangle-HitPack}}
\label{sec:tw:lower:triangle}

Before proving \cref{lem:tw:lower:reduction},
i.e., the lower bound for \TriUndelHitPack,
we first provide some intuition for the construction.
Let $\phi$ be the 3-SAT formula.
The graph of the instance $I$ of \TriUndelHitPack contains
a vertex set $Z$ (of size $\Oh(\log m)$)
that separates the graph into two halves,
referred to as the left and right half.
We use the vertices in $Z$ to represent (the \emph{identifier} of) some clause.
The right half encodes the assignment to the variables of $\phi$
by deleting certain vertices,
while the left half is used to select a clause
for which we want to check if it is satisfied by the assignment.
Then each clause will correspond to some maximal triangle packing in the constructed graph.

In order to ensure that each clause is satisfied, we want to \emph{rule out} the possibility that all literals of some clause are
\emph{unsatisfied}. This is done in the following way:
If the set of vertex deletions in the right half of the graph
specifies an assignment of variables that leaves all literals of some clause unsatisfied,
then the maximal triangle packing (after the deletions) that corresponds to that particular clause will have size at least $\ell$.
This way a variable assignment is not a solution if and only if
the corresponding deletion set is also not a solution
(as it allows too large of a packing).

To ensure that three unsatisfied literals of some clause yield a large triangle packing, we construct gadgets on the right half
that contribute a large number of triangles to the packing
whenever the literal is not satisfied by the clause.
Conversely, if the clause is satisfied by the literal,
the gadget contributes only a small number of triangles to the packing.
In terms of quantifiers and negations, we will check that \emph{for all} clauses it holds that it is \emph{not} the case that they contain three \emph{unsatisfied} literals (or equivalently that each clause has at least one satisfied literal).

\begin{proof}[Proof of \Cref{lem:tw:lower:reduction}]
  Let $\phi'$ be the given 3-SAT formula with $n'$ variables $x_1,\dots,x_{n'}$
  and $m'$ clauses.

  \subparagraph*{Preprocessing the Formula.}
  We assume in the following that no clause of $\phi'$ contains the same literal multiple times.
  Moreover, there is no clause that contains only one literal.
  Otherwise we could greedily set the value of this literal
  and simplify the formula accordingly.
  Additionally, we replace clauses containing exactly two literals.
  For two literals $\lambda_1$ and $\lambda_2$,
  let $C = (\lambda_1 \lor \lambda_2)$ be such a clause of $\phi'$
  and let $y$ be a fresh variable.
  We replace this clause $C$ by the two clauses
  \[
    (\lambda_1 \lor \lambda_2 \lor       y ) \land
    (\lambda_1 \lor \lambda_2 \lor \lneg{y})
    .
  \]
  After doing this replacement for all clauses with only two literals,
  let $\phi''$ be the resulting 3-SAT formula
  with $n'' \le n'+m'$ variables $x_1, \ldots, x_{n''}$ and $m'' \le 2m'$ clauses.
  It is easy to verify that $\phi'$ is satisfiable
  if and only if $\phi''$ is satisfiable.

  For a technical reason, which becomes clear later,
  we modify the formula $\phi''$ further.
  For each $i\in \numb{n''}$,
  we add the following four new clauses to the formula
  \[
    (x_i \lor \lneg{x}_i \lor       x_{i+1} ) \land
    (x_i \lor \lneg{x}_i \lor \lneg{x}_{i+1}) \land
    (x_i \lor \lneg{x}_i \lor       x_{i+2} ) \land
    (x_i \lor \lneg{x}_i \lor \lneg{x}_{i+2}),
  \]
  where we set $x_{n''+1} = x_1$ and $x_{n''+2} = x_2$
  to keep notation simple.

  Let $\phi$ be the resulting formula.
  Again, it can be easily verified that
  $\phi$ is satisfiable if and only if $\phi''$ is satisfiable.
  By duplicating clauses, we can also achieve that $\phi$
  contains exactly $2^c$ clauses for some appropriate $c \ge 4$.
  Hence, in the following we set $m=2^c \in \Oh(m')$
  as the number of clauses of $\phi$
  and $n = n'' \in \Oh(n'+m')$ as the number of variables.

\subparagraph*{Gadgets.}
For the construction of the gadgets we use an idea
which we refer to as a cycle of triangles.
To formalize this idea, we define the following gadget
which we later use to define the other gadgets.

\begin{claim}[Auxiliary Gadget]
  \label{clm:tw:lower:auxiliaryGadget}
  For all integers $r\geq 2$, there is a graph $\AuxGadget_r$ with $2r$
  distinguished vertices $v_1,\dots,v_r$ and $\bar v_1,\dots,\bar v_r$ such
  that:
  \begin{enumerate}[itemsep=0.1em]
    \item
    This graph has exactly two maximum triangle packings
    $P$ and $\bar P$ of size $r$.
    \item
    Packing $P$ covers $v_1, \dots, v_r$
    and none of the vertices $\bar v_1, \dots, \bar v_r$.
    \item
    Packing $\bar P$ covers $\bar v_1,\dots, \bar v_r$
    and none of the vertices $v_{1}, \dots, v_r$.
  \end{enumerate}
\end{claim}
\begin{proof}
  Intuitively, the gadget $\AuxGadget_r$ consists of $2r$ triangles arranged
  in a cycle such that two triangles share one endpoint.

  Formally, the graph $\AuxGadget_r$ consists of a cycle with $2r$ vertices
  $u_0,\dots,u_{2r-1}$ where, for all $i\in \numb{2r-1}$,
  $u_{i-1}$ is adjacent to $u_{i}$, $u_{2r-1}$ is connected to $u_0$,
  and no other edges exist between these vertices.
  Moreover, there are vertices $v_1,\dots,v_{r}$ and $\bar v_1,\dots,\bar v_{r}$
  such that, for all $i\in\numb{r}$,
  $v_i$ is connected to $u_{2i-2}$ and $u_{2i-1}$,
  and $\bar v_i$ is connected to $u_{2i-1}$ and $u_{2i\bmod 2r}$.

  We define the two packings $P$ and $\bar P$ as follows:
  \begin{displaymath}
      P \deff \{ (u_{2i-1}, v_i, u_{2i \bmod 2r}) \mid i \in \numb{r} \}
      \qquad
      \text{and}
      \qquad
      \bar P \deff \{ (u_{2i-2}, \bar v_i, u_{2i-1}) \mid i \in \numb{r} \}
  \end{displaymath}
  It is easy to check that these two sets actually define a packing
  of $r$ triangles with the properties (2) and (3).
  It remains to show that there is no other packing of $r$ (or more) triangles.
  But this follows directly from the fact that
  the gadget $\AuxGadget_r$ contains only triangles with vertex set
  $\{v_i,u_{2i-2},u_{2i-1}\}$ or $\{\bar v_i,u_{2i-1},u_{2i \bmod 2r}\}$
  for some $i\in \numb{r}$,
  since we assumed that $r \geq 2$.
  Hence, whenever $r$ (or more) triangles are contained in a packing
  that is different from $\bar P$ and $P$,
  then all of these triangles cannot be pairwise vertex disjoint
  as the neighboring triangles intersect.
\end{proof}

Based on this gadget $\AuxGadget$,
we define the following gadget, which we use to ``generate'' the
clause number
(see~\cref{fig:tw:gadget}).
Although this gadget is just a special case of the gadget $\AuxGadget$,
we provide it with a separate name to keep notation simple in the later proof.

\begin{claim}[Selector Gadget]
    \label{clm:tw:lower:reduction:gadget1}
    For all integers $r\ge 2$,
    there is a graph $\SelGadget_r$
    with $2r$ distinguished vertices
    $v_1,\dots,v_r$ and $\bar v_1,\dots,\bar v_r \in V$
    such that
    \begin{enumerate}
      \item
      there are exactly two packings of $3r$ triangles
      denoted by $P_0$ and $P_1$,
      \item
      $P_0$ covers $v_1, \dots, v_r$
      and none of the vertices $\bar v_1, \dots, \bar v_r$,
      \item
      $P_1$ covers $\bar v_1,\dots, \bar v_r$
      and none of the vertices $v_{1}, \dots, v_r$,
      \item
      there is no packing of more than $3r$ triangles,
      and
    \end{enumerate}
  \end{claim}
  \begin{claimproof}
    We use the graph $\AuxGadget_{3r}$ as the graph $\SelGadget_r$
    but only pick a third of the distinguished vertices of $\AuxGadget{3r}$
    as the distinguished vertices of $\SelGadget_r$
    (see~\cref{clm:tw:lower:auxiliaryGadget}).
    Let $v'_1, \dots, v'_{3r}$ and $\bar v'_1,\dots,\bar v'_{3r}$
    be the distinguished vertices of $\UniGadget_{3r}$.
    For all $i \in \numb{r}$, we choose $v_i \deff v'_{3i}$
    and $\bar v_i \deff \bar v'_{3i-2}$
    as the distinguished vertices of $\SelGadget_r$.
    The %
    properties of the gadget
    follow immediately from \cref{clm:tw:lower:reduction:gadget1}.
  \end{claimproof}

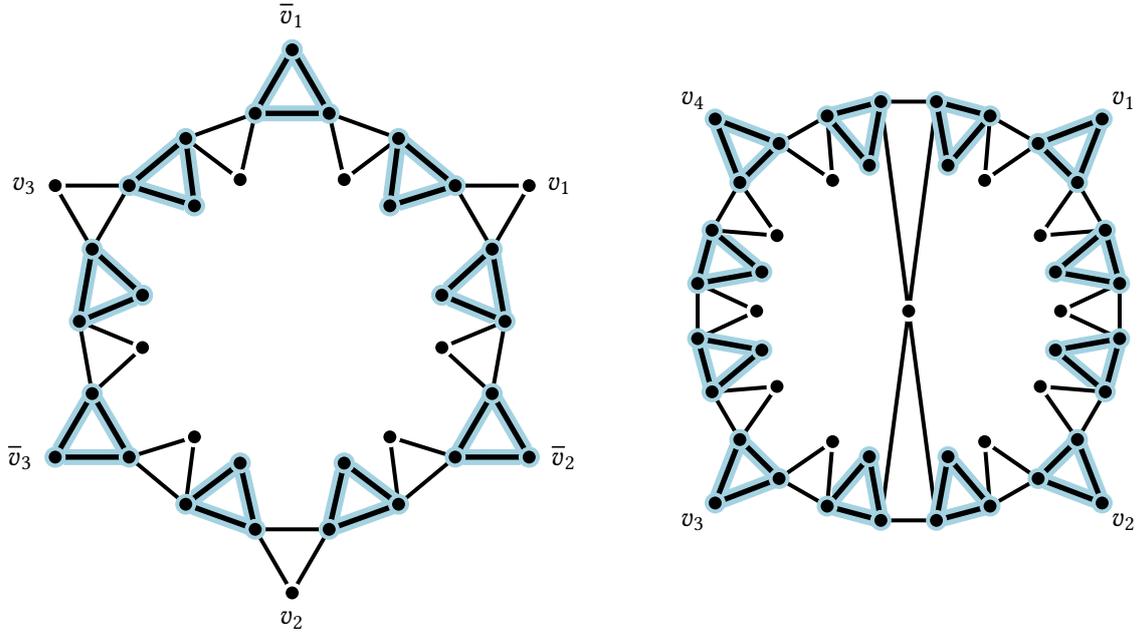
\begin{figure}[t]
    \centering
    \begin{subfigure}[1]{0.45\textwidth}

\begin{tikzpicture}[%
  scale=.8,
  ]
  \definecolor{selColor}{RGB}{164,210,225}
\definecolor{freeColor}{RGB}{243,180,48}

\tikzset{%
vertex/.style={
  draw=white,
  fill=black,
  label={below:{\large\ensuremath{#1}}},
  circle,
  line width = 1.5pt,
  inner sep = 2.25pt,
  outer sep = 0pt,
},
vertex/.default=\null,
vertexDist/.style={
  vertex = #1,
  diamond,
},
vertexDist/.default=\null,
edge/.style={
  line width=0.5pt,
  draw=white,
  double distance = 1.5pt,
  double = black,
  rounded corners,
},
selVtx/.style={
  draw = #1,
},
selVtx/.default=selColor,
freeVtx/.style={
  selVtx = freeColor,
},
selEdg/.style={
  line width = 2pt,
  double = black,
  double distance = 2pt,
  draw = #1,
  rounded corners,
},
selEdg/.default=selColor,
freeEdg/.style={
  selEdg = freeColor,
},
}

  \def\dist{1}
  \def\diam{3.5}
  \def\angle{60}
  \def\shift{\angle/6}
  \def\offset{150+\shift}

  \foreach \i in {1,...,6} {
    \coordinate (v\i0) at (\offset-\angle*\i:\diam);
    \coordinate (v\i1) at (\offset-\angle*\i-2*\shift:\diam);
    \coordinate (v\i2) at (\offset-\angle*\i-4*\shift:\diam);
    \coordinate (v\i)  at (\offset-\angle*\i-\shift:\diam+\dist);
    \coordinate (v\i3) at (\offset-\angle*\i-3*\shift:\diam-\dist);
    \coordinate (v\i4) at (\offset-\angle*\i-5*\shift:\diam-\dist);
    };

  \foreach \i in {1,3,5} {
    \node[vertex,selVtx] (v\i)  at (160-60*\i-10:\diam+\dist) {};
    \node[vertex] (v\i3) at (160-60*\i-30:\diam-\dist) {};
    \node[vertex,selVtx] (v\i4) at (160-60*\i-50:\diam-\dist) {};
    };
    \foreach \i in {2,4,6} {
      \node[vertex] (v\i)  at (160-60*\i-10:\diam+\dist) {};
      \node[vertex,selVtx] (v\i3) at (160-60*\i-30:\diam-\dist) {};
      \node[vertex] (v\i4) at (160-60*\i-50:\diam-\dist) {};

    };
  \node [above = .05 of v1] {$\bar v_1$};
  \node [right = .0  of v2] {$v_1$};
  \node [right = .05 of v3] {$\bar v_2$};
  \node [below = .0  of v4] {$v_2$};
  \node [left  = .05 of v5] {$\bar v_3$};
  \node [left  = .0  of v6] {$v_3$};

  \foreach \i in {2,4,6} {
    \ifthenelse{\i = 6}{
      \pgfmathsetmacro{\next}{1}
    }{%
      \pgfmathsetmacro{\next}{int(\i+1)}
    }
    \draw[edge] (v\i0) -- (v\i1) -- (v\i) -- (v\i0);
    \draw[edge,selEdg] (v\i1) -- (v\i2) -- (v\i3) -- (v\i1);
    \draw[edge] (v\i2) -- (v\i4) -- (v\next0) -- (v\i2);
  };
  \foreach \i in {1,3,5} {
    \pgfmathsetmacro{\next}{int(\i+1)}
    \draw[edge,selEdg] (v\i0) -- (v\i1) -- (v\i) -- (v\i0);
    \draw[edge] (v\i1) -- (v\i2) -- (v\i3) -- (v\i1);
    \draw[edge,selEdg] (v\i2) -- (v\i4) -- (v\next0) -- (v\i2);
  };
  \foreach \i in {1,...,6} {
    \node[vertex,selVtx] at (v\i0) {};
    \node[vertex,selVtx] at (v\i1) {};
    \node[vertex,selVtx] at (v\i2) {};
  };
  \foreach \i in {1,3,5} {
    \node[vertex,selVtx] at (v\i) {};
    \node[vertex] at (v\i3) {};
    \node[vertex,selVtx] at (v\i4) {};
  };
  \foreach \i in {2,4,6} {
    \node[vertex] at (v\i) {};
    \node[vertex,selVtx] at (v\i3) {};
    \node[vertex] at (v\i4) {};
  };

\end{tikzpicture}
    \end{subfigure}
    \hspace{1.5cm}
    \begin{subfigure}[1]{0.4\textwidth}

\begin{tikzpicture}[%
  scale=.8,
  ]
  \definecolor{selColor}{RGB}{164,210,225}
\definecolor{freeColor}{RGB}{243,180,48}

\tikzset{%
vertex/.style={
  draw=white,
  fill=black,
  label={below:{\large\ensuremath{#1}}},
  circle,
  line width = 1.5pt,
  inner sep = 2.25pt,
  outer sep = 0pt,
},
vertex/.default=\null,
vertexDist/.style={
  vertex = #1,
  diamond,
},
vertexDist/.default=\null,
edge/.style={
  line width=0.5pt,
  draw=white,
  double distance = 1.5pt,
  double = black,
  rounded corners,
},
selVtx/.style={
  draw = #1,
},
selVtx/.default=selColor,
freeVtx/.style={
  selVtx = freeColor,
},
selEdg/.style={
  line width = 2pt,
  double = black,
  double distance = 2pt,
  draw = #1,
  rounded corners,
},
selEdg/.default=selColor,
freeEdg/.style={
  selEdg = freeColor,
},
}

  \def\dist{1}
  \def\diam{3.5}
  \def\angle{45}
  \def\shift{\angle/6}
  \def\offset{90+\shift}

  \foreach \i in {1,...,8} {
    \coordinate (v\i0) at (\offset-\angle*\i:\diam);
    \coordinate (v\i1) at (\offset-\angle*\i-2*\shift:\diam);
    \coordinate (v\i2) at (\offset-\angle*\i-4*\shift:\diam);
  };
  \foreach \i in {1,3,5,7} {
    \coordinate (v\i)  at (\offset-\angle*\i-\shift:\diam+\dist);
    \coordinate (v\i3) at (\offset-\angle*\i-3*\shift:\diam-\dist);
    \coordinate (v\i4) at (\offset-\angle*\i-5*\shift:\diam-\dist);
  };
  \foreach \i in {2,4,6,8} {
    \coordinate (v\i3) at (\offset-\angle*\i-3*\shift:\diam-\dist);
    \coordinate (v\i4) at (\offset-\angle*\i-5*\shift:\diam-\dist);
  };
  \coordinate (v2)  at (\offset-\angle*2-\shift:\diam-\dist) {};
  \coordinate (v4)  at (\offset-\angle*4-\shift:0) {};
  \coordinate (v6)  at (\offset-\angle*6-\shift:\diam-\dist) {};
  \coordinate (v8)  at (\offset-\angle*8-\shift:0) {};

  \node [above right = .0  of v1] {$v_1$};
  \node [below right = .0  of v3] {$v_2$};
  \node [below left = .0  of v5] {$v_3$};
  \node [above left = .0  of v7] {$v_4$};

  \foreach \i in {2,4,6,8} {
    \ifthenelse{\i = 8}{
      \pgfmathsetmacro{\next}{1}
    }{%
      \pgfmathsetmacro{\next}{int(\i+1)}
    }
    \draw[edge] (v\i0) -- (v\i1) -- (v\i) -- (v\i0);
    \draw[edge,selEdg] (v\i1) -- (v\i2) -- (v\i3) -- (v\i1);
    \draw[edge] (v\i2) -- (v\next0) -- (v\i4) -- (v\i2);
  }
  \foreach \i in {1,3,5,7} {
    \pgfmathsetmacro{\next}{int(\i+1)}
    \draw[edge,selEdg] (v\i0) -- (v\i1) -- (v\i) -- (v\i0);
    \draw[edge] (v\i1) -- (v\i2) -- (v\i3) -- (v\i1);
    \draw[edge,selEdg] (v\i2) -- (v\i4) -- (v\next0) -- (v\i2);
  }

  \foreach \i in {1,...,8} {
    \node[vertex,selVtx] at (v\i0) {};
    \node[vertex,selVtx] at (v\i1) {};
    \node[vertex,selVtx] at (v\i2) {};
  };
  \foreach \i in {1,3,5,7} {
    \node[vertex,selVtx] at (v\i) {};
    \node[vertex] at (v\i3) {};
    \node[vertex,selVtx] at (v\i4) {};
  };
  \foreach \i in {2,4,6,8} {
    \node[vertex] at (v\i) {};
    \node[vertex,selVtx] at (v\i3) {};
    \node[vertex] at (v\i4) {};
  };

\end{tikzpicture}
    \end{subfigure}

    \caption{Left figure presents a $\SelGadget_4$ gadget.
    The highlighted triangles depict the triangle packing $\Ptrue$.
    Right figure presents $\LitGadget_4$ along with the
unique maximum triangle packing $\Pfalse$.}
    \label{fig:tw:gadget}
\end{figure}

  With the construction of $\SelGadget_r$ at hand,
  we have everything ready to define the gadgets for encoding the literals.
  For this gadget the size of the two largest solutions differs (see~\cref{fig:tw:gadget}).
  \begin{claim}[Literal Gadget]
    \label{clm:tw:lower:reduction:gadget2}
    For all integers $r\ge 4$,
    there is a graph $\LitGadget_r$
    with $r$ distinguished vertices $v_1,\dots,v_r \in V$
    such that
    \begin{enumerate}
      \item
      there is a packing $\Ptrue$ of $3r-1$ triangles
      that covers \emph{none} of $v_1,\dots,v_r$,
      \item
      there is exactly one packing $\Pfalse$ of $3r$ triangles
      that covers $v_1,\dots,v_r$, and
      \item
      there is no triangle packing with more than $3r$ triangles.
    \end{enumerate}
  \end{claim}
  \begin{claimproof}
    We define $\LitGadget_r$ based on the gadget $\SelGadget_r$ as follows.
    The vertex and edge set of $\LitGadget_r$
    is the same as that of $\SelGadget_r$
    with the only modification that we identify the vertices
    $\bar v_1$ and $\bar v_3$ with each other.
    We use the vertices $v_1, \dots, v_r$
    as the distinguished vertices of $\LitGadget_r$.

    First observe that identifying $\bar v_1$ with $\bar v_3$
    does not create any new triangles. %
    This directly implies the third property of the gadget.
    By this choice it directly follows that
    there is exactly one packing of $3r$ triangles
    which is precisely the packing $P_1$ for $\SelGadget_r$.
    Moreover, the only other packing of $3r$ triangles for $\SelGadget_r$, i.e., $P_0$,
    contains $\bar v_1$ and $\bar v_3$.
    As we identified these two vertices,
    this reduces the number of triangles by one. (There are now two choices for the packing $\Ptrue$.)
    Hence, the first and second property of the gadget follow.
  \end{claimproof}

\subparagraph*{Construction of the Instance.}

\begin{figure}[t]
    \centering
    \includegraphics[width=0.2\textwidth]{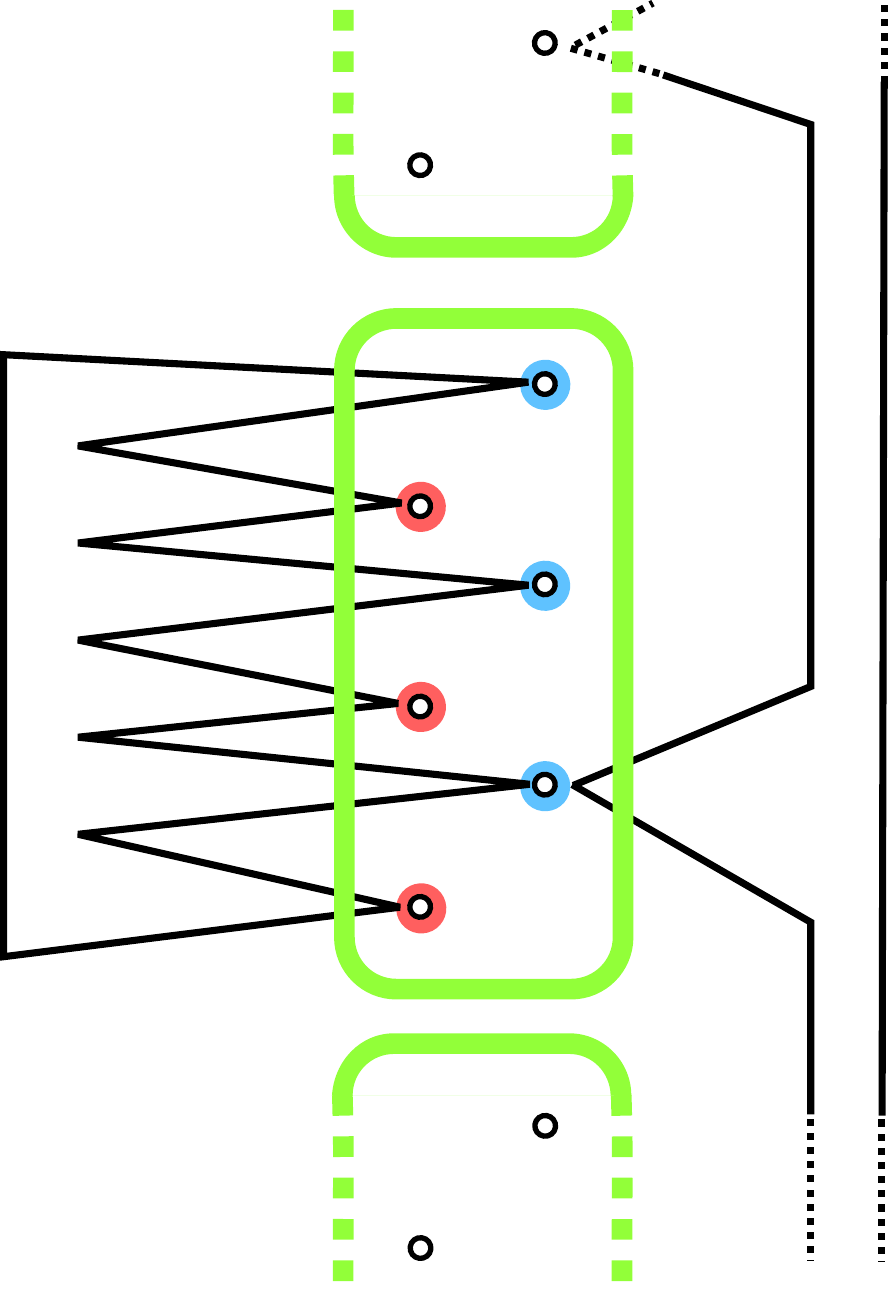}
    \caption{An overview of the lower-bound construction. On
        the left, we have the \emph{selector gadget}
        that can cover either blue vertices (on the right) or
        red vertices (on the left) in the green group.
        The middle layer consists of $\Oh(\log(m))$
        groups of $6$ vertices. Each group corresponds to the bit of the clause index. The right part represents the \emph{literal gadget},
        ensuring satisfaction of the corresponding literal in the clause. Importantly,
        the middle layer consists of only $\Oh(\log(m))$ vertices. Both the literal and
        selector gadgets have a constant treewidth, resulting in a created instance
    with a treewidth of $\Oh(\log(m))$.}
    \label{fig:overview-construction}
\end{figure}

  We present the overview of the construction in~\cref{fig:overview-construction}.
  Before defining the graph $G$,
  we introduce some more notation.
  For an integer $0 < N \le 2^c$,
  we define $\bin{N}$ as the binary encoding of $N-1$ with $c$ bits.
  For all $j \in \numb{c}$,
  we denote by $\bin{N}\pos{j}$
  the $j$th bit of the binary encoding of $N-1$.

  With this notation,
  we can now formally define the graph $G$ of the instance $I$.
  For all $\gamma\in\numb{c}$, $b\in\{0,1\}$, and $p \in \{1,2,3\}$,
  we create a vertex $z^{(\gamma)}_{p,b}$
  and define $Z \deff \{ z^{(\gamma)}_{p,b}\}_{\gamma,b,p}$
  as the set of these $6c$ vertices.

  The left side of the graph is defined as follows.
  For all $\gamma\in\numb{c}$,
  we create a copy $\SelGadget^{(\gamma)}$ of the gadget $\SelGadget_3$
  and identify the copies of the vertices $v_1,v_2,v_3$ with the vertices
  $z^{(\gamma)}_{1,0}, z^{(\gamma)}_{2,0}, z^{(\gamma)}_{3,0}$
  and identify the copies of the vertices $\bar v_1,\bar v_2,\bar v_3$
  with the vertices
  $z^{(\gamma)}_{1,1}, z^{(\gamma)}_{2,1}, z^{(\gamma)}_{3,1}$.
  There are no other vertices in the left side of the instance.

  For the right side we let $\Lambda = \{ x_1, \lneg{x}_1, \dots, x_n,
  \lneg{x}_n \}$ be the set of all literals
  corresponding to the variables of $\phi$.
  Note that from the preprocessing it follows that
  all of them appear in some clause.
  For each
  literal $\lambda \in \Lambda$, we define a set $C_\lambda \subseteq \numb{m}
  \times \{1,2,3\}$ such that $(j,p) \in C_\lambda$ if and only if literal
  $\lambda$ appears in the $j$th clause at position $p$.

  For all $\lambda \in \Lambda$ and all $(j,p)\in C_\lambda$,
  we create a vertex $d_\lambda$ and
  introduce a copy of the gadget $\LitGadget_{c+1}$
  denoted by $\LitGadget_{\lambda,j,p}$.
  For each such gadget $\LitGadget_{\lambda,j,p}$ and all $\gamma\in \numb{c}$,
  we identify the copy of the vertex $v_\gamma$ from $\LitGadget_{c+1}$
  with the vertex $z^{(\gamma)}_{p,1-\bin{j}\pos{\gamma}}$.
  The copy of the vertex $v_{c+1}$ from $\LitGadget_{c+1}$
  is identified with $d_\lambda$. Intuitively, the vertex $d_\lambda$ synchronizes the ``behavior'' of all gadget copies that belong to the same literal.
  This completes the definition of the graph $G$.

  Let $V$ be the set of vertices of $G$, then the set of undeletable vertices $U \deff V \setminus \{d_\lambda \mid \lambda \in \Lambda\}$, that is, only the $d_\lambda$'s can be deleted.
  To conclude the construction,
  we set $k \deff n$
  as the bound on the number of vertices we are allowed to delete
  and set $\ell \deff 3m (3c +2) + 9c + 3$
  as the strict upper bound on the number of triangles that we are allowed to pack after the vertex deletions.
  The resulting \TriUndelHitPack instance is $I=(G, U, k, \ell)$.

  \subparagraph*{Correctness.}
  In the next two steps we show that the construction is correct.
  We do this by handling both directions individually.
  \begin{claim}
    \label{clm:tw:lower:reduction:correct}
    If $\phi$ is satisfiable,
    then $I$ has a solution.
  \end{claim}
  \begin{claimproof}
    Let $\pi$ be an assignment for the variables $x_1,\dots,x_n$
    such that $\phi$ is satisfied under $\pi$.
    Based on $\pi$,
    we define the solution $S$ for $I$ as
    \[
    S \deff \{ d_{x_i}        \mid i \in\numb{n} \text{ and } \pi(x_i) = \true \}
    \cup \{ d_{\lneg{x}_i} \mid i \in\numb{n} \text{ and } \pi(x_i) = \false \}
      .
    \]
    By the definition of $S$, we get that $\abs{S}=n=k$.
    It remains to show that $S$ is indeed a solution for $I$.

    For the sake of a contradiction, assume that $S$ is not a solution
    and let $\mathcal{P}$ be a maximum size packing of triangles of size at least $\ell$.
    Let us analyze this packing $\mathcal{P}$.
    Let $q$ be the number of literal gadgets that contribute the maximum possible $3c+3$ triangles to $\mathcal P$. (These are those gadgets that contribute a packing according to $\Pfalse$.)
    From \cref{clm:tw:lower:reduction:gadget2} it follows that this triangle packing $\Pfalse$ covers precisely $c$ vertices from $Z$ (one for each $\gamma\in [c]$).
    As $Z$ has size $6c$ it follows that $q\le 6$, as otherwise some vertex from $Z$ would be covered twice.

    We strengthen this bound and argue that $q\le 3$.
    Assume otherwise and let $(\lambda_i,j_i,p_i)$ for $i\in\numb{4}$
    be four gadgets $\LitGadget_{\lambda_i,j_i,p_i}$, each contributing $3c+3$ triangles.
    As $p_i \in \numb{3}$,
    there must be distinct $i_1,i_2\in\numb{4}$ such that $p_{i_1}=p_{i_2}$.
    For ease of notation assume that $i_1=1,i_2=2$ and $p_1=1$.
    By \cref{clm:tw:lower:reduction:gadget2},
    the two gadgets cover $c$ vertices in $\{z^{(\gamma)}_{1,b}\}_{\gamma,b}$.
    As each vertex can only be covered by one triangle,
    we get that all vertices in $\{z^{(\gamma)}_{1,b}\}_{\gamma,b}$ are covered.
    However, if that is the case, then from the properties of \cref{clm:tw:lower:reduction:gadget1},
    it follows that, for all $\gamma\in\numb{c}$, the gadget $\SelGadget^{(\gamma)}$
    can contribute at most $8$ triangles to the packing.
    Now note that every literal gadget that does not contribute $3c+3$ triangles to the packing $\mathcal P$ contributes at most $3c+2$ triangles, according to \cref{clm:tw:lower:reduction:gadget2}.

    Hence, the size of the packing $\mathcal P$ is at most $3m(3c+2) + 8c + q$.
    Therefore, $\mathcal P$ can have size at least $\ell = 3m (3c+2) + 9c + 3$
    if and only if $q \ge c + 3$.
    Since we know that $q \le 6$, we get $c \le 3$
    which contradicts the assumption that $c \ge 4$.

    Combining the two above results, we have that $q\le3$.
    First suppose that $q\le 2$. Then the size of $\mathcal P$ is at most $3m(3c+2) + 9c + 2$, which contradicts the fact that its size should be at least $\ell$.
    Thus, it remains to consider the case $q=3$. Let $\lambda_i,j_i$ for $i\in\numb{3}$ such that
    $\LitGadget_{\lambda_i,j_i,p_i}$ contributes the maximum possible $3c+3$ triangles to $\mathcal{P}$
    for all $i\in\numb{3}$, i.e., these literal gadgets contribute triangles according to the packing $\Pfalse$.
    We claim that $j_1=j_2=j_3$.
    For the sake of a contradiction, assume otherwise
    and assume without loss of generality that $j_1\neq j_2$.
    Thus, there is some $\gamma\in\numb{c}$ such that
    $\bin{j_1}\pos{\gamma}\neq\bin{j_2}\pos{\gamma}$. This means that exactly one of $\LitGadget_{\lambda_1,j_1,p_1}$ or $\LitGadget_{\lambda_2,j_2,p_2}$ covers a vertex in  $z^{(\gamma)}_{1,0}, z^{(\gamma)}_{2,0}, z^{(\gamma)}_{3,0}$, whereas the other covers a vertex in $z^{(\gamma)}_{1,1}, z^{(\gamma)}_{2,1}, z^{(\gamma)}_{3,1}$.
    Thus, by the properties from \cref{clm:tw:lower:reduction:gadget1},
    $\SelGadget^{(\gamma)}$ contributes at most $8$ triangles to the packing.
    Hence, the size of the packing is bounded by $3m \cdot (3c+2) + (c-1) 9 + 8 + 3$
    which contradicts the assumption about the size of $\mathcal{P}$,
    i.e., $\abs{\mathcal{P}} \ge \ell$.
    Hence, we actually get that $j_1=j_2=j_3$.

    So, for $j=j_1=j_2=j_3$ and for all $i\in\numb{3}$,
    the gadgets $\LitGadget_{\lambda_i,j,i}$ contribute $3c+3$ triangles to $\mathcal{P}$. Then
    \cref{clm:tw:lower:reduction:gadget2} implies that $d_{\lambda_i}$
    is not deleted (since it is part of the packing $\Pfalse$),
    i.e., $d_{\lambda_i} \notin S$.
    Hence, by our definition of $S$ it must hold that $\pi(\lambda_i)=\false$
    for all $i\in\numb{3}$
    which contradicts the assumption that $\pi$ satisfies $\phi$.
  \end{claimproof}

  Now we prove the reverse direction.
  \begin{claim}
    \label{clm:tw:lower:reduction:complete}
    If $I$ has a solution,
    then $\phi$ is satisfiable.
  \end{claim}
  \begin{claimproof}
    Let $S \subseteq V\setminus U$ be a solution for $I$.
    The idea is to define a satisfying assignment $\pi$ for $\phi$ based on $S$.
    To show that this is possible, we first introduce some notation.
    We define, for all $i\in \numb{n}$,
    $S(i) \deff \abs{\{ d_{x_i}, d_{\lneg{x}_i}\} \cap S }$
    as the number of vertices selected for each variable.
    We first claim that, for all $i \in \numb{n}$,
    the solution $S$ satisfies $S(i) < 2$.

    \subparagraph*{The Solution is Good.}
    We say that a solution $S$ is \emph{good} if for each $i\in [n]$ we have $S(i)=1$. To simplify notation,
    we write $S(i + n) \deff S(i)$
    and similarly for the indices of the variables.
    Now assume for contradiction that there is some $i \in \numb{n}$
    such that $S(i) = 2$.
    Since we defined $k=n$,
    this implies that there must be some $b \in \numb{n}$
    such that $S(b) = 0$.
    Moreover, we claim that there must be some $b \in \numb{n}$
    such that $S(b) = 0$ and, additionally, $S(b+1) + S(b+2) < 4$.
    Again for contradiction assume that, for all $b \in \numb{n}$
    with $S(b) = 0$, it holds that $S(b+1) + S(b+2) = 4$.
    We define the following sets
    \begin{align*}
      N_0 & \deff \{ i \in \numb{n} \mid S(i) = 0 \land S(i+1) + S(i+2) = 4 \}\\
      N_1 & \deff \{ i \in \numb{n} \mid S(i) = 1 \}\\
      N_2 & \deff \{ i \in \numb{n} \mid S(i) = 2 \land i-1,i-2 \notin N_0 \}
      .
    \end{align*}
    Observe that each $i \in N_0$ contributes four selected vertices to $S$,
    each $i \in N_1$ contributes one vertex,
    and each $i \in N_2$ contributes two vertices to $S$.
    Moreover, this covers all vertices in $S$,
    as there is no index $i \in \numb{n}$
    with $S(i)=0$ and $S(i+1)+S(i+2) < 4$.
    Hence, we get that
    \[
      k
      = \sum_{i\in\numb{n}} S(i)
      = 4 \abs{N_0} + \abs{N_1} + 2 \abs{N_2}
      = n + \abs{N_0} + \abs{N_2}
      .
    \]
    However, since we assumed that $N_0 \neq \emptyset$,
    we arrive at a contradiction as $k=n$
    and conclude that there is some $b \in \numb{n}$
    such that $S(b)=0$ and $S(b+1)+S(b+2) < 4$.

    In the next step we show that this also leads to a contradiction, thereby disproving
    that there is some $i \in \numb{n}$ with $S(i)=2$.
    We do this by constructing a large packing
    which would imply that $S$ is not a solution.
    By our modification of $\phi$,
    for all $\lambda \in \{x_{b+1}, \lneg{x}_{b+1},
    x_{b+2}, \lneg{x}_{b+2} \}$,
    there is a clause $(x_b \lor \lneg{x}_b \lor \lambda)$ in $\phi$.
    By the above reasoning we know
    that there is at least one choice of $\lambda$
    such that $d_\lambda \notin S$.
    Assume without loss of generality that $\lambda = x_{b+1}$,
    the other cases follow analogously.
    Let $q \in \numb{m}$ be the number
    of the clause $(x_b \lor \lneg{x}_b \lor x_{b+1})$.
    We set $Q = \{(x_b, q, 1), (\lneg{x}_b, q, 2), (x_{b+1}, q, 3)\}$
    to simplify notation.

    For all $\lambda \in \Lambda$ and $(j, p) \in C_\lambda$
    where $(\lambda, j, p) \notin Q$,
    we can use the packing $\Ptrue$ from \cref{clm:tw:lower:reduction:gadget2}
    to find a packing of $3c+2$ triangles for the gadget $\LitGadget_{\lambda, j, p}$
    such that no vertex from $Z \cup \{ d_\lambda\}$ is covered.
    (Recall that $d_\lambda$ corresponds to the vertex $v_{c+1}$
    in the respective gadget, and the vertices in $Z$ are identified
    with vertices of the form $v_\gamma$.)
    This contributes $(3m -3) \cdot (3c+2)$ triangles.

    Also by \cref{clm:tw:lower:reduction:gadget2}, for every $(\lambda, j, p) \in Q$,
    we can find a packing of $3c+3$ triangles for $\LitGadget_{\lambda,j,p}$
    since $d_\lambda=v_{c+1}$ is not deleted in these gadgets.
    Observe that the vertex sets of these packings are disjoint
    because the position $p$ is different for each choice (and thus we consider different copies of the gadget $\LitGadget_{c+1}$).
    This contributes $3 (3c+3)$ additional triangles to the packing.

    Note that, by construction of $G$,
    the triangle packings of the three gadgets
    $\LitGadget_{x_b,j,1}$, $\LitGadget_{\lneg{x}_b,j,2}$, and $\LitGadget_{x_{b+1},j,3}$
    cover, for each $\gamma\in [c]$, exactly one of the sets
    $\{z_{1,0}^\gamma, z_{2,0}^\gamma, z_{3,0}^\gamma\}$ or
    $\{z_{1,1}^\gamma, z_{2,1}^\gamma, z_{3,1}^\gamma\}$,
    and they cover none of the vertices of the respective other set.
    Thus, for every $\gamma \in \numb{c}$,
    we can find a packing of $9$ triangles for the gadget $\SelGadget^{(\gamma)}$.
    This contributes $9c$ additional triangles to the packing.

    Hence, the final packing consists of
    \[
      (3m - 3) \cdot (3c + 2) + 3 (3c + 3) + 9c
      = 3m (3c + 2) + 9c + 3
    \]
    triangles which contradicts our assumption that
    $S$ is a solution because we defined $\ell = 3m\cdot (3c+2) + 9c + 3$.
    We conclude that it cannot happen that $S(i)=2$ for any $i \in \numb{n}$.

    \subparagraph*{Constructing and Verifying the Assignment.}
    Now we can define the assignment $\pi$ for $\phi$.
    For all $i\in \numb{n}$,
    we set
    \[
      \pi(x_i) \deff
        \begin{cases}
          \true, & \text{if } d_{x_i}\in S, \\
          \false, & \text{otherwise.}
        \end{cases}
    \]
    By the fact that the solution $S$ is good,
    this assignment is well-defined
    and each variable is assigned some truth-value.

    Now we argue that $\pi$ satisfies $\phi$.
    For the sake of a contradiction, assume that $\pi$ does not satisfy $\phi$.
    Then there is at least one clause that is not satisfied by $\pi$.
    Fix some arbitrary index $j \in \numb{m}$ of such an unsatisfied clause.
    Next we construct a packing of at least $\ell$ triangles in $G-S$.

    For all $\gamma\in\numb{c}$,
    there is, by \cref{clm:tw:lower:reduction:gadget1},
    some triangle packing $P^{(\gamma)}$ for the gadget $\SelGadget^{(\gamma)}$
    such that the vertices $z^{(\gamma)}_{1,\bin{j}\pos{\gamma}}, z^{(\gamma)}_{2,\bin{j}\pos{\gamma}}, z^{(\gamma)}_{3,\bin{j}\pos{\gamma}}$
    are covered by this packing (this is the packing $P_{\bin{j}\pos{\gamma}}$ of the corresponding gadget).
    By \cref{clm:tw:lower:reduction:gadget1},
    we know that each packing $P^{(\gamma)}$
    contains exactly $9$ triangles.

    Let $\lambda_1$, $\lambda_2$, and $\lambda_3$ be the literals
    such that the $j$th clause is $(\lambda_1 \lor \lambda_2 \lor \lambda_3)$.
    \begin{itemize}
    	\item For all literals
    	$\lambda \in \Lambda \setminus \{\lambda_1,\lambda_2,\lambda_3\}$, and for all $(j',p) \in C_\lambda$,
    	we use \cref{clm:tw:lower:reduction:gadget2}
    	to obtain a triangle packing $P_\lambda$ of size $3c+2$
    	for the gadget $\LitGadget_{\lambda,j,p}$.
    	As argued before, we can choose $P_\lambda=\Ptrue$ such that
    	it does not cover any vertex from $Z \cup \{ d_\lambda \}$.
    	Hence, we do not have to check whether the vertex $d_\lambda$
    	is contained in $S$, that is, to check whether it is deleted.
    	\item For the literals $\lambda_p$ with $p\in [3]$, and for all $(j',p)\in C_{\lambda_p}$ with $j'\neq j$, we also select a triangle packing of size $3c+2$
    	for the gadget $\LitGadget_{\lambda_p,j', p}$.
    	Again, this triangle packing does not use the vertex $d_{\lambda_p}$ or any vertex from $Z$.
    	\item Finally, for the literals $\lambda_p$ with $p\in [3]$, consider the gadget $\LitGadget_{\lambda_p,j,p}$.
    	Recall that $\pi(\lambda_p)=\false$
    	as we assume that the clause is unsatisfied.
    	Hence, $d_{\lambda_p} \notin S$.
    	Thus, by \cref{clm:tw:lower:reduction:gadget2},
    	we can choose the triangle packing that contains $d_{\lambda_p}$ for $\LitGadget_{\lambda_p,j,p}$, that is, the packing of size $3c+3$ that contains all of the vertices of the form $v_\gamma=z^{(\gamma)}_{p,1-\bin{j}\pos{\gamma}}$.
    \end{itemize}

    Let $P_{\lambda_p}$ be the combined packing
    for the gadgets $\LitGadget_{\lambda_p,j',p}$ for all $(j',p)\in C_{\lambda_p}$.
    We first argue why the final packing is actually a packing,
    that is, why the triangles from the right are disjoint
    from the triangle packing for the left side of the graph.
    For this fix some $\gamma\in\numb{c}$.
    We know that the gadget $\SelGadget^{(\gamma)}$ covers the vertices
    $z^{(\gamma)}_{p,\bin{j}\pos{\gamma}}$ for all $p\in\numb{3}$.
    By our choice above,
    for all $p\in\numb{3}$,
    the gadget $\LitGadget_{\lambda_p,j,p}$
    covers the vertices $z^{(\gamma)}_{p,1-\bin{j}\pos{\gamma}}$
    for all $\gamma\in\numb{c}$.
    Hence, the two vertex sets are actually disjoint
    which implies that the triangles from the packing are disjoint.

    It remains to analyze the size of the resulting packing.
    We know that $\abs{P^{(\gamma)}}=9$,
    $\abs{P_\lambda}=3c+2$
    for all $\lambda\in\Lambda\setminus\{\lambda_1,\lambda_2,\lambda_3\}$,
    and that
    $\abs{P_{\lambda_p}} = (\abs{C_{\lambda_p}}-1)\cdot (3c+2) + (3c+3)$
    for all $p\in\numb{3}$.
    Hence, the entire packing has size
    \begin{align*}
      &\ c \cdot 9
      + \sum_{\lambda \in \Lambda\setminus\{\lambda_1,\lambda_2,\lambda_3\}}
        \abs{C_\lambda} \cdot (3c+2)
        \\&
      + (3c+2) (\abs{C_{\lambda_1}}-1)
      + (3c+2) (\abs{C_{\lambda_2}}-1)
      + (3c+2) (\abs{C_{\lambda_3}}-1)
      + 3(3c+3)
        \\
      =&\ 9c
      + \sum_{\lambda \in \Lambda} (3c+2) \abs{C_\lambda}
      + 3\\
      =&\ 9c + 3m \cdot (3c+2) + 3
    \end{align*}
    which is equal to $\ell$.
    Since this is a contradiction to our assumption
    that we are given a solution for the \TriUndelHitPack instance,
    we conclude that $\pi$ actually satisfies the formula $\phi$.
  \end{claimproof}

  It remains to prove that the reduction
  also satisfies the required properties about the size and the pathwidth of $G$.
  \begin{claim}
    \label{clm:tw:lower:reduction:size}
    Graph $G$ has pathwidth at most $\Oh(\log m)$
    and the size of $G$ if bounded by $\Oh(m \log m)$.
  \end{claim}
  \begin{claimproof}
    The number of $\SelGadget^{(\gamma)}$ gadgets is precisely $c$
    and each gadget is of constant size.
    Each gadget $\LitGadget_{\lambda,p,j}$ is of size $\Oh(c)$ (by \cref{clm:tw:lower:reduction:gadget2}),
    and there are $3m$ such gadgets in total.
    As we already accounted for the vertices in $Z$ by the other gadgets,
    the size of the graph can be bounded above by
    \[
      \Oh(c) + 3m \cdot \Oh(c) = \Oh(m \log m).
    \]

    For the bound on the pathwidth we
    first observe that if we delete $Z$,
    then the graph consists of $c + 3m$ disjoint components.
    Moreover, each such component either corresponds
    to a gadget $\SelGadget^{(\gamma)}$, which has constant size,
    or to a gadget $\LitGadget_{\lambda, p, j}$,
    which has size $\Oh(c) = \Oh(\log m)$.
    Thus, the pathwidth of the entire graph can be bounded by $\Oh(\log m)$.
  \end{claimproof}

  We finish the proof by recalling that $\phi$ is satisfiable
  if and only if the original formula $\phi'$ is satisfiable.
  Moreover, based on our initial modifications to $\phi$,
  we get $m \in \Oh(m')$ which concludes the proof.
\end{proof}

\subsection{\texorpdfstring%
{\boldmath Lower Bound for \SquareUndelHitPack}%
{Lower Bound for Four-Cycle-HitPack}}
\label{sec:tw:lower:C4}

In this section we consider the problem \SquareUndelHitPack
and provide a lower bound matching the running time
of the general algorithm from \cref{sec:twUpper}
as stated in \cref{thm:twUpper:arbitrary}.
Formally, we prove \cref{thm:tw:lower:C4},
which we restate here for convenience.

\thmTwLowerCFour*\label\thisthm

Our reduction starts from 3-SAT and follows roughly the same outline
as the reduction for the general case.
However, to obtain this stronger bound with the logarithmic factor,
we have to construct an instance of \SquareUndelHitPack with smaller pathwidth.
More precisely, if we can find a reduction such that the pathwidth
is bounded by $\Oh(\log m / \log \log m)$,
then the claimed lower bound would follow immediately.

To understand how this can be achieved,
recall that the encoding of the clause number
was the limiting factor for the bound on the pathwidth in previous reduction.
By encoding the index of the clauses in binary,
the pathwidth of the graph became $\Oh(\log m)$.
When now considering \SquareUndelHitPack,
we can actually find a better way of representing the clause number.
Recall that in the previous construction,
each gadget (independent of its type) covered vertices
from their respective half and one additional vertex from the middle only.
In the following we allow that a cycle can go from the left side
to the right half of the graph.
Then the cycle can actually be described by pair of vertices form the middle.

We use this idea as follows while again designing a graph with two halves.
The left half contains gadgets corresponding to selecting the clause
and the gadgets on the right side correspond to setting the variables
and checking whether the literals are satisfied.
For the vertices in the middle,
we now have two groups (high and low) of $\Oh(t)$ vertices each
where $t$ is chosen such that $t! = m$
(this is for example possible by setting $t = \Oh(\log m / \log \log m)$).
The clause number is then encoded by a perfect matching
from the vertices of the high group to the vertices of the low group.
This gives us $t!$ possible perfect matchings
which then clearly suffices to represent all clauses.

We define the gadgets on the left such that
they connect each vertex from the high group to a vertex from the low group
by a path of length two (i.e., a half of a $C_4$)
which then induces a perfect matching.
The gadgets on the right side then have to complete these cycles
by another path of length two to form a $C_4$.
We define these literal gadgets such that this is only possible
if the corresponding literal was not satisfied.
The interpretation is again that a large $C_4$-packign corresponds to an unsatisfied clause and thus, also formula
while a small packing indicates a satisfying assignment.

We note that for this construction to work
we crucially rely on the fact that we can transfer information
from one half to the other half
for which we exploit that we are packing cycles of length four.

In the following we give the formal details.

\begin{proof}[Proof of \Cref{thm:tw:lower:C4}]

Let $\phi'$ be the given 3-SAT formula with
$n'$ variables $x_1,\dots,x_{n'}$ and $m'$ clauses. We preprocess the formula in
a similar way to that in \cref{sec:tw:lower:triangle}, and we get a formula
$\phi$ with $n$ variables and $m$ clauses. Let $t$ be the minimum integer such
that $t!$ is at least the number of clauses of the formula. Next, we duplicate
arbitrarily selected clause to guarantee that the final formula
has exactly $t!$~clauses.

\subparagraph*{Gadgets.}
As a next step we define two types of gadgets,
one for the literals and one for selecting the clause.

\begin{claim}[Literal Gadget for $C_4$]
	\label{clm:tw:lower:C4:gadget1}
	For all integers $r\ge 2$,
	there is a graph $\clit_r$
	with $2r$ distinguished vertices
	$v_1,\dots,v_r$ and $\bar v_1,\dots,\bar v_r \in V$
	such that
	\begin{enumerate}
		\item
		there are exactly two packings of $C_4$'s of size $3r$
		denoted by $P_0$ and $P_1$,
		\item
		$P_0$ covers $v_1, \dots, v_r$
		and none of the vertices $\bar v_1, \dots, \bar v_r$,
		\item
		$P_1$ covers $\bar v_1,\dots, \bar v_r$
		and none of the vertices $v_{1}, \dots, v_r$,
		\item
            there is no packing of $C_4$'s of size greater than $3r$.
	\end{enumerate}
\end{claim}

  \begin{claimproof}
      Intuitively, our literal gadget $\clit_r$ is the graph obtained by replacing every triangle in $\UniGadget_r$ with a $C_4$.

      First, we create $6r$ copies of $C_4$, which we denote as
      $a_1b_1c_1d_1, \dots, a_{6r}b_{6r}c_{6r}d_{6r}$. We identify $d_i$ with
      $b_{i+1 \bmod 6r}$ for every $i \in [6r]$. For every $i \in [r]$, we name
      $a_{6i-5}$ as $v_i$. Next, for every $i \in [r]$ we name vertex
      $a_{6i-2}$ as $\bar{v}_i$. This concludes the construction of graph
      $\clit_r$.

      Let $P_0=\{a_{2i+1}b_{2i+1}c_{2i+1}d_{2i+1} \mid i\in [3r]\}$ and
      $P_1=\{a_{2i}b_{2i}c_{2i}d_{2i} \mid i\in [3r]\}$
      be the two distinguished packings.
	 We can easily verify that properties 2 and 3 of this claim are true.

     Observe, that any $C_4$ in the graph is of the form $a_ib_ic_id_i$ for some
     $i\in [6r]$.
     This means that an arbitrary $C_4$ in any packing in the graph must cover
     both vertices of $\{b_id_i\}$ for some $i\in [6r]$.
     It follows that properties 1 and 4 are true as
     $\{b_{2i+1},d_{2i+1}\} \mid i\in [3r]\}
     =\{\{b_{2i},d_{2i}\} \mid i\in [3r]\}$,
     every $C_4$ in a packing must contain $b_i,d_i$ for some $i$
     and both sets have size $3r$ only.
  \end{claimproof}

We also introduce the selector gadget $\csel$. Intuitively, $\csel$ is a \emph{chain}
of five $4$-cycles. Formally, $\csel$ is constructed from $5$ copies of $C_4$, denoted as
$a_1b_1c_1d_1, \dots, a_5b_5c_5d_5$.
Then, we identify vertices $d_i$ with
$b_{i+1}$ for every $i\in \{1,2,3,4\}$. This concludes the construction of
$\csel$.

\subparagraph*{Construction of the Instance.}
Now we describe how to construct the instance of \UndelHitPack{C_4} from a processed instance of 3-SAT.
For all $\gamma\in\numb{t}$, $g\in\{0,1\}$, and $p \in \{1,2,3\}$,
we create a vertex $z^{(\gamma)}_{p,g}$
and define $Z \deff \{ z^{(\gamma)}_{p,g}\}_{\gamma,g,p}$
as the set of vertices in the middle.

The left side of the graph is as follows.
For each pair $(i,j)\in [t]\times [t]$, we create a copy of the selector gadget, denoted by $\csel^{(i,j)}$.
For each gadget $\csel^{(i,j)}$, we connect $a_1$, $a_3$, and $a_5$
(where $a_i$ is the vertex in the $i$th $C_4$ of $\csel^{(i,j)}$)
as follows to six vertices in $Z$:
\begin{itemize}
    \item
    We make $a_1$ adjacent to the two vertices $z^{i}_{1,0}$ and $z^{j}_{1,1}$.
    \item
    We make $a_3$ adjacent to the two vertices $z^{i}_{2,0}$ and $z^{j}_{2,1}$.
    \item
    We make $a_5$ adjacent to the two vertices $z^{i}_{3,0}$ and $z^{j}_{3,1}$.
\end{itemize}

The right side of the graph is as follows.
Let $\Lambda = \{ x_1, \lneg{x}_1, \dots, x_n,
\lneg{x}_n \}$ be the set of all literals
corresponding to the variables of $\phi$.
Note that from the preprocessing it follows that
all of them appear in some clause.
For each
literal $\lambda \in \Lambda$, we define a set $\CID_\lambda \subseteq \numb{m}
\times \{1,2,3\}$ such that $(j,p) \in \CID_\lambda$ if and only if literal
$\lambda$ appears in the $j$th clause at position $p$.

We create two vertices $d_{i}^{(\true)}$ and $d_{i}^{(\false)}$ for each $i\in [n]$.
For all $\lambda \in \Lambda$ and all $(j,p)\in \CID_\lambda$,
we introduce a copy of the gadget $\clit_{t+1}$
denoted by $\clit_{\lambda,j,p}$.
Depending on the literal $\lambda$ we consider two cases.
\begin{itemize}
    \item
    The literal $\lambda$ is a positive literal,
    that is, $\lambda = x_i$ for some variable $x_i$:

    Then, we identify $v_{t+1}$ of the gadget $\clit_{\lambda,j,p}$ with $d_{i}^{(\true)}$.
    Moreover,
    we identify $\bar v_{t+1}$ of the gadget $\clit_{\lambda,j,p}$ with $d_{i}^{(\false)}$.

    \item
    The literal $\lambda$ is a negative literal,
    that is, $\lambda = \lneg x_i$ for some variable $x_i$:

    Then we identify $v_{t+1}$ of the gadget $\clit_{\lambda,j,p}$ with $d_{i}^{(\false)}$.
    Moreover,
    we identify $\bar v_{t+1}$ of the gadget $\clit_{\lambda,j,p}$ with $d_{i}^{(\true)}$.
\end{itemize}

Let $\sigma_1, \sigma_2,\ldots,\sigma_m$ be $m$ different permutation functions
of $t$ elements, chosen in arbitrary order.
Note that there are $t!$ different permutation functions of $t$ elements and ${t!}=m$.
For each gadget $\clit_{\lambda,j,p}$ and each $\gamma \in [t]$, we connect $\bar v_{\gamma}$ and $z^{(\gamma)}_{p,0}$ with an edge,
and connect $\bar v_{\gamma}$ and $z^{(\sigma_j(\gamma))}_{p,1}$ with an edge.
This concludes the construction of the graph $G$.

Let $V$ be the set of vertices of $G$,
then the set of undeletable vertices is
$U \deff V \setminus \{d_i^{(\true)},d_i^{(\false)} \mid i\in [n]\}$.
We set $k \deff n$ as the bound on the number of vertices we are allowed to
delete and set $\ell \deff 3t^2 + 9m(t+1) + 2t$.
The constructed \UndelHitPack{C_4} instance is $I=(G, U, k, \ell)$.
This concludes the description of the construction. Clearly, the runtime of the
reduction is polynomial. It remains to prove its correctness and analyze the
pathwidth of $G$.

\subparagraph*{Correctness.}
Now, we show that the construction is correct.
We do this by handling both directions individually.
  \begin{claim}
	\label{clm:tw:lower:C4:correct}
	If $\phi$ is satisfiable, then $I$ has a solution.
\end{claim}
\begin{claimproof}
	Let $\pi$ be an assignment for the variables $x_1,\dots,x_n$
	such that $\phi$ is satisfied under $\pi$.
	Based on $\pi$,
	we define the solution for $I$ as
	    \[
	S \deff \{ d_{i}^{(\true)} \mid i \in\numb{n} \text{ and } \pi(x_i) = \true \}
	\cup \{ d_{i}^{(\false)} \mid i \in\numb{n} \text{ and } \pi(x_i) = \false \}
	.
	\]
    It remains to show that $S$ is indeed a solution for $I$, i.e., there is no
    packing of $\ell$ vertex-disjoint $4$-cycles in $G\setminus S$.

	For the sake of contradiction, assume that $\mathcal{P}$ is a maximum-size
    packing of $4$-cycles of size at least $\ell$.
	Let us analyze this packing $\mathcal{P}$.
	First we check the right side of the graph $G$, which is the set of all literal gadgets.
    Each of $\clit_{\lambda,j,p}$ can contribute at most $3(t+1)$ many $4$-cycles to
    $\mathcal{P}$ (if we ignore the $4$-cycles of $\mathcal{P}$ in the middle that are incident with vertices of the literal gadgets).
	There are $3m$ literal gadgets and thus the literal gadgets can contribute
    at most $3m(3t+3)$ many $4$-cycles in total.

    It remains to check the left part and the middle part of $G$.  Since there
    are $6t$ vertices in the middle, the middle part can contribute at most $3t$
    many $4$-cycles to $\mathcal{P}$.  We show that the middle part contributes
    exactly $3t$ many $4$-cycles to $\mathcal{P}$.  Otherwise, the size of
    $\mathcal{P}$ is strictly less than $3t+3(t^2-t)+2t+3m(3t+3)=\ell$, a
    contradiction to our assumption.

    Next, we argue that there are exactly $t$ selector gadgets which are
    incident with at least one $C_4$ of $\mathcal{P}$ contributed by the middle
    part.  Suppose that there are $q>t$ selector gadgets which are incident with
    at least one $C_4$ of $\mathcal{P}$ contributed by the middle part.  Then
    the size of $\mathcal{P}$ is at most $3t+3(t^2-q)+2q+3m(3t+3)<\ell$, a
    contradiction to our assumption.  Thus it holds that there are exactly $t$
    selector gadgets which are incident with at least one $C_4$ of $\mathcal{P}$
    contributed by the middle part.
    Since there are $6t$ vertices in the
    middle part,  it follows that there are exactly $t$ selector gadgets each of
    which is incident with three $4$-cycles of $\mathcal{P}$ contributed by the
    middle part.  Thus the left part can contribute at most $2t+(t^2-t)\cdot
    3=3t^2-t$ many $4$-cycles to $\mathcal{P}$.  By the construction, these $t$
    selector gadgets actually define a permutation of $t$ elements, say
    $\sigma_{h}$ ($h\in [m]$), according to the  vertices of $Z$ to which they
    are adjacent.

    Note that to pack $3t$ many $4$-cycles in the middle part and $3m(3t+3)$ many
    $4$-cycles on the right part (to make sure that $|\mathcal{P}|\geq \ell$),
    the only possible situation is that each of the three literal gadgets
    corresponding to the $h$th clause $C_h$ are incident with $t$ many $4$-cycles
    of $\mathcal{P}$ in the middle part.  Moreover, each of the three literal
    gadgets also contributes $3(t+1)$ many $4$-cycles to $\mathcal{P}$.  By the
    definition of the solution $S$, all three literals of $C_h$ are
    set to $\false$ by the assignment $\pi$.
    Thus, the clause $C_h$ is not
    satisfied by $\pi$, contradicting that $\pi$ is a satisfying assignment.

	As a result, there is no packing of $\ell$ vertex-disjoint $4$-cycles in $G\setminus S$ and $S$ is a solution to $I$.
\end{claimproof}

Now we prove the reverse direction.
\begin{claim}
	\label{clm:tw:lower:C4:complete}
	If $I$ has a solution, then $\phi$ is satisfiable.
\end{claim}
\begin{claimproof}
	Let $S \subseteq V\setminus U$ be a solution for $I$.
	By a very similar analysis to that in
    \cref{clm:tw:lower:reduction:complete}, we can show that the solution $S$ is
    \emph{good},
    that is, for no variable $x_i$,
    both vertices $d_i^{(\true)}$ and $d_i^{(\false)}$ are deleted
    (see~\cref{sec:tw:lower:triangle} for more details).
	We define the assignment $\pi$ for $\phi$ as follows.
	For all $i\in \numb{n}$,
	we set
	\[
	\pi(x_i) \deff
	\begin{cases}
		\true, & \text{if } d_{i}^{(\true)}\in S, \\
		\false, & \text{otherwise.}
	\end{cases}
	\]
	By the fact that the solution $S$ is good,
	this assignment is well-defined
	and each variable is assigned some boolean value.

    Now we argue that $\pi$ satisfies $\phi$.  For the sake of a contradiction,
    assume that $\pi$ does not satisfy $\phi$.
    Hence, there is at least one
    clause that is not satisfied by $\pi$ and let $j \in \numb{m}$ be the index
    of such an unsatisfied clause.
    Next we construct a packing $\mathcal{P}$ of at
    least $\ell$ many $4$-cycles in $G-S$.

	Let $C_{j}=(\lambda_1 \lor \lambda_2 \lor \lambda_3)$
    be this specific clause.
	Since $\pi(\lambda_p)=\false$ for each $p\in [3]$,
    vertex $v_{t+1}$ is not deleted in each of $\clit_{\lambda_p,j,p}$ by the construction.
	By \cref{clm:tw:lower:C4:gadget1},
	we can add $t+1$ many $4$-cycles (the packing $P_0$ covering $v_1,\ldots,v_{t+1}$) for the literal gadget $\clit_{\lambda_p,j,p}$ to $\mathcal{P}$ for each $p\in \numb{3}$.
	Consider a vertex $\bar v_{\gamma}$ (where $\gamma \in [t]$) of the literal
    gadget $\clit_{\lambda_p,j,p}$ (where $p\in \numb{3}$).
	Recall that $\bar v_{\gamma}$ is adjacent to $z^{(\gamma)}_{p,0}$ and  $z^{(\sigma_j(\gamma))}_{p,1}$.
	Let $a_{2p-1}$ be the vertex of $\csel^{(\gamma,\sigma_j(\gamma))}$ which is adjacent to both $z^{(\gamma)}_{p,0}$ and  $z^{(\sigma_j(\gamma))}_{p,1}$.
    Then, we get the $C_4$ that covers $\bar v_{\gamma}$, $a_{2p-1}$, $z^{(\gamma)}_{p,0}$ and $z^{(\sigma_j(\gamma))}_{p,1}$.
    Thus we get a packing $\mathcal{P}_{\text{mid}}$ of $3t$ many $4$-cycles of
    this form and by our construction these cycles are vertex-disjoint.
    We add all of the $4$-cycles in $\mathcal{P}_{\text{mid}}$ to $\mathcal{P}$.

	For any literal gadget $\clit_{\lambda,j',p}$ such that $j'\neq j$,
	by \cref{clm:tw:lower:C4:gadget1}, we can add $t+1$ many $4$-cycles to $\mathcal{P}$
	as only one vertex of $v_{t+1}$ and $\bar v_{t+1}$ is deleted in the gadget.
	For a selector gadget $\csel^{(i,j)}$, note that it is incident with either
    three or none of the $4$-cycles of $\mathcal{P}_{\text{mid}}$ by our construction.
    If it is incident with three $4$-cycles of $\mathcal{P}_{\text{mid}}$,
	then we add two $4$-cycles of $\csel^{(i,j)}$ which does not cover any vertex
    of $\mathcal{P}_{\text{mid}}$ to $\mathcal{P}$.
	Otherwise we add three $4$-cycles of $\csel^{(i,j)}$ to $\mathcal{P}$.

	We can verify that the $4$-cycles in $\mathcal{P}$ are pairwise vertex-disjoint.
	The size of $\mathcal{P}$ is $3m(3t+3)+3t+2t+3(t^2-t)=\ell$.
	This contradicts the fact that $S$ is a solution for $I$.
	Thus we can conclude that the constructed assignment $\pi$ satisfies $\phi$.
\end{claimproof}

It remains to prove that the reduction also satisfies the required properties about the size and the pathwidth of $G$.
\begin{claim}
	\label{clm:tw:C4:reduction:size}
    Graph $G$ has pathwidth at most $\Oh(t)$.
\end{claim}
\begin{claimproof}
	Observe that if we delete $Z$,
	then the graph consists of $t^2 + 3m$ disjoint components.
	Moreover, each such component either corresponds
	to a gadget $\csel^{(i,j)}$, which has constant size,
	or to a gadget $\clit_{\lambda, j, p}$,
	which has size $\Oh(t) = \Oh(t)$.
	Thus, the pathwidth of the entire graph can be bounded by $\Oh(t)$.
\end{claimproof}

Recall that ${t!} = m$, which means that $m = 2^{\Oh(t \log{t})} = 2^{\Oh(\pw
\log{\pw})}$.
Observe that this concludes the proof of~\cref{thm:tw:lower:C4} because a $2^{2^{o(\pw
\log{\pw})}} \cdot |V(G)|^{\Oh(1)}$ algorithm for \UndelHitPack{C_4} would yield
a
$2^{o(m)}$ algorithm for 3-SAT with $m$ clauses.
\end{proof}

\subsection{Lower Bound for \CycleUndelHitPack}\label{sec:pw:lower:cycles}
In this section we prove \cref{thm:tw:lower:cycles}, which we restate for convenience.

\thmcyclepwLB*\label\thisthm

We note that this result does not follow from the statement of
\cref{thm:tw:lower:C4} although a $C_4$ is obviously a cycle.
The main reason is
that there might be a large cycle packing despite the fact that there is no
large $C_4$ packing (e.g., the graph itself is a long cycle).

However,
in the following we argue
that our reduction from \cref{thm:tw:lower:C4}
was already stated such that it works for \CycleUndelHitPack.
The crucial piece is the following lemma.

\begin{lemma}
  \label{lem:tw:lower:propertyCycles}
  Let $I=(G, U, k, \ell)$ be the \UndelHitPack{C_4} instance
  from \cref{thm:tw:lower:C4}.
  Let $S \subseteq V(G) \setminus U$ be a set of vertices
  and $\packs$ be an arbitrary cycle-packing for $G-S$
  containing at least one cycle that is not a $C_4$.
  Then, there is a $C_4$-packing $\packs'$ for $G-S$
  such that $\abs{\packs'} \ge \abs{\packs}$.
\end{lemma}
\begin{proof}

Recall that the constructed graph is based on two different types of gadgets.
The gadgets of the form $\csel$ were used to select some clause number and the
gadgets of the form $\clit_{t+1}$ encoded the literals appearing in the formula.

Consider some cycle $C$ of length at least $5$ in the packing $\packs$.
This cycle $C$ cannot be contained in the gadget $\csel$
since such a gadget does not contain any other cycles
besides the five $4$-cycles which are arranged as a path by construction.
Suppose $C$
is entirely contained in some gadget $\clit_{t+1}$.  Let us denote
this gadget by $F$ for ease of notation.
Then, by the structure of these gadgets, $C$ intersects all
$4$-cycles in $F$.
In particular, there is no other cycle in $\packs$ that is
entirely contained in $F$.  Also, as a \emph{non}-distinguished vertex of $F$ is
not part of any edge outside of $F$, such a vertex cannot be contained in any
cycle of $\packs$, other than possibly $C$.
Hence, removing the cycle $C$ and then
adding some maximal packing of $C_4$ in $F$ that does \emph{not} contain any
distinguished vertex gives another feasible packing of size at least that of
$\packs$.

Now consider some cycle $C$ of length at least $5$ in the packing $\packs$ that
is \emph{not} entirely contained in some gadget.  In this case we know that
$C$ enters some gadget $F$ at some distinguished vertex and leaves $F$ at
some other distinguished vertex. However, from the construction of the
gadgets it is straightforward to see that this path that is the intersection
of $C$ with the gadget $F$ goes through at least two $C_4$'s in $F$ all of
whose vertices are non-distinguished vertices.
Therefore, removing $C$ from the packing $\packs$ and adding one of these
$C_4$ instead gives another feasible packing of size at least that of
$\packs$ (and fewer cycles of length at least $5$).
\end{proof}

Now we can directly prove \cref{thm:tw:lower:cycles}.
\thmcyclepwLB*
\begin{proof}[Proof of \Cref{thm:tw:lower:cycles}]
  It clearly suffices to prove that the construction
  gives a reduction from 3-SAT to \CycleUndelHitPack
  with the properties from \cref{lem:tw:lower:reduction}.
  Then, the result follows by the same arguments as \cref{thm:tw:lower}.

  Since we did not modify our construction,
  it suffices to prove the first property,
  that is, the 3-SAT formula $\phi$ and the \CycleUndelHitPack instance $I$
  are equivalent.
  If $\phi$ is satisfiable,
  we know that $I$ has a solution for \TriUndelHitPack,
  that is, after removing at most $k$ vertices
  there is no large packing of $C_4$.
  \Cref{lem:tw:lower:propertyCycles} implies
  that there is also no large packing of cycles
  as otherwise there would also be a large packing of $C_4$.

  For the converse direction assume that at most $k$ vertices of $I$
  can be deleted such that there is no large cycle-packing.
  Since every $C_4$-packing is also a cycle-packing,
  this implies that there is no large $C_4$ packing.
  Hence, formula $\phi$ must be satisfiable.
\end{proof}

\clearpage
\phantomsection
\addcontentsline{toc}{section}{References}
\bibliographystyle{plainurl}
\bibliography{bib}

\clearpage
\tableofcontents
\label{toc}
\end{document}